\RequirePackage{amsthm}
\documentclass[sn-mathphys-num]{sn-jnl}

\usepackage{graphicx}%
\usepackage{multirow}%
\usepackage{bbm}
\usepackage{amsmath,amssymb,amsfonts}%
\usepackage{mathrsfs}%
\usepackage[title]{appendix}%
\usepackage{xcolor}%
\usepackage{textcomp}%
\usepackage{manyfoot}%
\usepackage{booktabs}%
\usepackage{algorithm}%
\usepackage{algorithmicx}%
\usepackage{algpseudocode}%
\usepackage{listings}%

\usepackage{fullpage}

\usepackage{graphicx}
\usepackage{mathrsfs}
\usepackage{float}
\usepackage{comment}
\usepackage{xfrac} 
\newcommand{\um}{\sfrac{1}{2}}

\usepackage{tabularx}
\usepackage{multirow}
\usepackage{makecell}

\usepackage{mdframed}

\usepackage{mathtools}

	\newcommand{\bra}[1]{\ensuremath{\left\langle#1\right|}}
	\newcommand{\ket}[1]{\ensuremath{\left|#1\right\rangle}}

	\newcommand{\matrixel}[3]{\ensuremath{\left\langle #1 \middle| #2 \middle| #3 \right\rangle}}
	\newcommand{\smatrixel}[3]{\ensuremath{\langle #1 | #2 | #3 \rangle}}
	
	\newcommand{\id}{\mathbbm{1}}

\renewcommand{\[}{\begin{equation}}
\renewcommand{\]}{\end{equation}}

\allowdisplaybreaks[1] 

\theoremstyle{thmstyleone}%
\newtheorem{theorem}{Theorem}[section]
\newtheorem{proposition}[theorem]{Proposition}
\newtheorem{corollary}[theorem]{Corollary}
\newtheorem{lemma}[theorem]{Lemma}
\newtheorem{assumption}[theorem]{Assumption}

\theoremstyle{thmstyletwo}%
\newtheorem{remark}{Remark}[section]%

\theoremstyle{thmstylethree}%
\newtheorem{definition}{Definition}[section]%

\numberwithin{equation}{section}

\newcommand{\modifica}[1]{#1}               

\newcommand{\rimodifica}[1]{#1}               

\begin{document}

\title[Trained quantum neural networks are Gaussian processes]{Trained quantum neural networks are Gaussian processes}

\author*[1,2,3]{\fnm{Filippo} \sur{Girardi}}\email{filippo.girardi@sns.it}

\author[4]{\fnm{Giacomo} \sur{De Palma}}\email{giacomo.depalma@unibo.it}

\affil[1]{\orgdiv{Korteweg--de Vries Institute for Mathematics}, \orgname{University of Amsterdam}, \orgaddress{\street{Science Park 105-107}, \city{Amsterdam}, \postcode{1098 XG}, \country{The Netherlands}}}

\affil[2]{\orgname{QuSoft}, \orgaddress{\street{Science Park 123}, \city{Amsterdam}, \postcode{1098 XG}, \country{The Netherlands}}}

\affil[3]{\orgname{Scuola Normale Superiore}, \orgaddress{\street{Piazza dei Cavalieri 7}, \postcode{56126}, \city{Pisa (PI)}, \country{Italy}}}

\affil[4]{\orgdiv{Department of Mathematics}, \orgname{University of Bologna}, \orgaddress{\street{Piazza di Porta San Donato 5},  \postcode{40126}, \city{Bologna (BO)}, \country{Italy}}}

\abstract{We study quantum neural networks made by parametric one-qubit gates and fixed two-qubit gates in the limit of infinite width, where the generated function is the expectation value of the sum of single-qubit observables over all the qubits. First, we prove that the probability distribution of the function generated by the untrained network with randomly initialized parameters converges in distribution to a Gaussian process whenever each measured qubit is correlated only with few other measured qubits. Then, we analytically characterize the training of the network via gradient descent with square loss on supervised learning problems. We prove that, as long as the network is not affected by barren plateaus, the trained network can perfectly fit the training set and that the probability distribution of the function generated after training still converges in distribution to a Gaussian process. Finally, we consider the statistical noise of the measurement at the output of the network and prove that a polynomial number of measurements is sufficient for all the previous results to hold and that the network can always be trained in polynomial time.}


\maketitle
\tableofcontents

\section{Introduction}

The last years have witnessed the surge of artificial intelligence powered by deep neural networks, which have achieved an enormous success in several fields such as speech recognition, computer vision and natural language processing, up to the recent development of large language models such as GPT-4 \cite{mnih2015human,lecun2015deep,radford2015unsupervised,schmidhuber2015deep,Goodfellow-et-al-2016,vaswani2017attention}.
Quantum neural networks constitute \modifica{a} quantum counterpart of deep neural networks.
The functions that they generate are the expectation value of a quantum observable (such as the number of qubits in the $1$ state) measured on the output of a quantum circuit made by parametric one-qubit and two-qubit gates.
The parameters are usually trained variationally with stochastic gradient descent.
Quantum neural networks have the promise to add the power of quantum computers to the capabilities of deep neural networks \cite{Cerezo_review,Schuld_2021}.
Indeed, embedding classical data into a Hilbert space whose dimension is exponential in the number of qubits may give access to computations that are classically hard to perform \cite{Havl_ek_2019, Schuld_2019a}. In particular, \cite{Liu_2021} provided an example of a family of datasets -- based on the discrete logarithm -- that only a quantum classifier can learn efficiently.
However, in all the known examples where quantum neural networks show quantum advantage, the right parameters need to be known a priori and put in by hand, and it is unknown whether they could be found variationally.
Indeed, there is no problem of practical interest yet where variationally trained quantum neural networks can provably outperform classical neural networks.
Furthermore, the training of quantum neural networks can suffer from bad local minima \cite{you2021exponentially,Anschuetz2022} or gradients whose size decreases exponentially with the number of qubits \cite{McClean_2018}, a phenomenon called ``barren plateaus''.
The possible solutions that have been proposed to avoid the issue of barren plateaus consist in changing the cost function of the network.
Ref. \cite{Cerezo_2021} considered untrained quantum neural networks with random parameters and proved that local cost functions (\emph{i.e.}, cost functions that are made by the sum of local terms acting on few qubits each) do not suffer from the issue, provided that the depth of the network is at most logarithmic in the number of qubits \cite{napp2022quantifying}.
Ref. \cite{Kiani_2022} proposed more general cost functions that arise from the quantum generalization of optimal mass transport and of the Lipschitz constant of \cite{De_Palma_2021} and are robust with respect to quantum operations acting on few qubits. However, a rigorous proof of the efficacy of these solutions is still lacking, and it has been conjectured that all the architectures for quantum neural networks that do not suffer from barren plateaus can be simulated efficiently on a classical computer \cite{cerezo2023does}.

A fundamental breakthrough in the mathematical understanding of the functioning of classical deep neural networks has been the proof that, in the limit of infinite width, the training is always able to perfectly fit the training examples while at the same time avoiding overfitting.
More precisely, a series of works \cite{Neal1996,lee2018deep,Lee2020,Han18,Hanin,Han21a,Han21b} has proved that in the limit of infinite width, the probability distribution of the function generated by a deep neural network trained with stochastic gradient descent on a supervised learning problem converges to a Gaussian process.
Such Gaussian process always perfectly fits the training examples, and its mean and covariance can be computed analytically given the training data and the architecture of the network \cite{yang2020scaling,yang2021tensorI,yang2020tensorII,yang2021tensorIIb,yang2021tensorIII,yang2022featureIV,yang2022tensorV,yang2023tensorIVb,yang2023tensorVI}.
These results prove that the limit of infinitely many parameters is smooth and analytically solvable, and provide the first rigorous mathematical foundation to classical machine learning.
The proof of these results is based on the following intermediate steps:
\begin{itemize}
\item The probability distribution of the function generated by an untrained deep neural network with random parameters converges to a Gaussian process;
\item The variation of any single parameter during the whole training is infinitesimal, but since the number of parameters is infinite, all the infinitesimal variations add up to a finite change in the generated function. Such regime is called ``lazy training''.
\end{itemize}

This classical breakthrough has stimulated the study of overparameterized quantum neural networks, which has been initiated in \cite{Larocca_2023}, while their lazy-training regime has been studied in \cite{liu2022representation}.
Ref. \cite{QLazy} proved that the training of quantum neural networks with constant depth happens in the lazy regime in the limit of infinite width and is capable of perfectly fitting the training examples.
However, quantum neural networks with constant depth do not have any hope of quantum advantage since the past light-cone of each measured qubit has constant size and its classical simulation is efficient.
Ref. \cite{liu2023analytic} considers the opposite limit, where the depth is high enough such that the number of parameters grows exponentially with the number of qubits and the network can reproduce any unitary, and proves the lazy training and the trainability of the network in such a regime.
Ref. \cite{garciamartin2023deep} proves that, in the same regime, the probability distribution of the function generated at initialization converges to a Gaussian process.

\subsection{Our results}

In this paper, we generalize to the quantum setting the breakthrough in classical machine learning presented above.
We consider quantum neural networks trained on supervised learning problems where the generated function is the expectation value of the sum of single-qubit observables over all the qubits\footnote{Our model differs from \cite{you2023analyzing}, where a single single-qubit observable, or more generally an observable with eigenvalues $\pm1$, is chosen: in that setting the convergence of the training has been proved to be slower than in the classical case. In our case, on the contrary, the convergence is exponential in the training time.}
and prove for the first time trainability in the limit of infinite width in any regime where the depth is allowed to grow with the number of qubits as long as barren plateaus do not arise. \rimodifica{To be precise, our conclusions are still valid even if only a fraction of the derivatives of the model function is exponentially suppressed, while the norm of the gradient (or, equivalently, the covariance of the model function at initialization, as in Lemma \ref{doppiastima}) is not exponentially suppressed.}
We stress that quantum circuits that are deep enough to forbid a naive efficient classical simulation but not too deep such that they do not suffer from barren plateaus do exist.
\modifica{
Indeed, barren plateaus are not present whenever the measured observable is local and the depth of the circuit scales logarithmically with the number of qubits \cite{napp2022quantifying}.
Such depth still allows the size of the past light-cone of each measured qubit to scale as a power of the number of qubits (see \autoref{combinazioni} for more details).
In this case, a na\"ive classical simulation would require exponential time.
}

Our first result states that the probability distribution of the function generated by a randomly initialized quantum neural network in the limit of infinite width converges in distribution to a Gaussian process when the parameters on which each measured qubit depends influence only a small number of other measured qubits (Theorem \ref{init}).
 
We then consider quantum neural networks trained in continuous time with gradient flow, and prove that the training happens in the lazy regime and is able to perfectly fit the training set (Theorem \ref{gradfl}).
Contrarily to \cite{QLazy}, our results allow for a depth growing with the number of qubits and are thus valid in regimes with hope of quantum advantage (see \autoref{sec:Qlazy} for a detailed comparison).
Furthermore, we prove that the dependence of the model on the parameters can be approximated with its linearized version around the value of the parameters at initialization (Theorems \ref{gronwall} and \ref{cfevolution}).
The corresponding linearized evolution equation has an analytical solution that determines the probability distribution of the function generated by the trained network, which is a Gaussian process whose mean and covariance can be computed analytically  (Corollary \ref{corgp}).
Building on such results, we prove that the probability distribution of the function generated by the trained network converges in distribution to the Gaussian process above (Theorem \ref{qnngp}).

We then focus on the issues related to real physical setting of the training of the network. On the one hand, we need to consider that the training takes place in discrete time steps, and we discuss the convergence of the gradient descent. On the other hand, we have to take into account that the function generated by a quantum neural network is the expectation value of a quantum observable, which can only be estimated via measurements. In this regard, we prove that a number of measurement polynomial in the number of qubits is enough to ensure all the previous results (Theorems \ref{unbgraddesc} and \ref{qnngpn}). In particular, the probability distribution of the function generated by the quantum neural network converges in distribution to a Gaussian process also in the presence of statistical noise.
The above results on the evolution of quantum neural networks in the noisy setting are our most important achievement.
\modifica{
Indeed, they prove the trainability of the model in polynomial time with respect to the number of qubits, which is a necessary condition for quantum advantage, and prove that the limit of infinitely many parameters is smooth.
Furthermore, thanks to the above results, the Seeger's generalization bounds for Gaussian processes \cite{seeger2002pac} can be applied to any trained quantum neural network satisfying the hypotheses of our paper in the limit of infinite width.
In particular, our results prove that such limit does not trivialize the generalization properties of quantum neural networks.
}
\rimodifica{Indeed, Seeger's generalization bound can be nontrivial even if, being infinitely wide, the network has infinitely many parameters with a finite training set\footnote{Of course, depending on the specific training set and on the specific neural tangent kernel of the considered quantum architecture, Seeger's generalization bound can become trivial, but this does not happen systematically for the only fact that the parameters are more than the training examples.}.}

Finally, differently from the main works which have studied the initialization of quantum neural networks \cite{garciamartin2023deep,rad2023deep}, we go beyond the standard assumption that the feature space is a finite set. Indeed, we rigorously prove that the convergence to a Gaussian process, both at initialization and during the training, is valid also when we consider an infinite feature space (Theorem \ref{final}). This achievement not only makes the mathematical treatment of our results complete, but also gives a precise physical insight: our proof ensures the existence of a lower bound to the rate of convergence to the Gaussian process which does not depend on the particular discretization of the feature space. In particular, the rate of convergence will not trivialize to zero as the the discretization is refined.

\modifica{
Summarizing, this paper identifies the limit of infinite width as a regime in which quantum neural networks are provably efficiently trainable whenever they do not suffer from barren plateaus.
Therefore, despite results such as \cite{you2021exponentially,Anschuetz2022} prove that problems other than barren plateaus can affect the trainability of variational quantum algorithms, we prove that barren plateaus are the only obstacle that can affect the trainability of quantum neural networks (at least in the limit of infinite width).
Therefore, our results constitute the first step to prove that useful and provably trainable quantum neural networks do exist.
}

\rimodifica{Having proved that sufficiently wide quantum neural networks that satisfy our hypotheses are efficiently trainable, we hope that our paper will stimulate the practitioners of quantum machine learning to consider wide architectures. We stress that such architectures do not actually require a quantum computer with as many qubits as their width. Indeed, as we explain in more detail in \autoref{advantages} (see, in particular, Remark \ref{rem:lab}), the expectation value of each measured one-qubit observable can be estimated independently, and such estimate requires only the qubits in the past light-cone of the measured qubit. We further stress that our hypotheses are extremely general, and any sequence of quantum neural networks with increasing width and with not too large depth will satisfy them. What are the best wide architectures for a given supervised learning problem is a fundamental question that will need to be answered by experiments.}

The article is organized as follows.

In \autoref{ch4}, we fix the setup and the notation of the paper.
In \autoref{ch5}, we prove that the function generated by a quantum neural network with random parameters in the limit of infinite width is a Gaussian process. In \autoref{ch6}, we consider quantum neural networks trained in continuous time and prove that the trained model is able to perfectly fit the training set and that the probability distribution of the generated function generated by the trained network is a Gaussian process.
In \autoref{ch7}, we consider quantum neural networks trained with discrete gradient descent taking into account the statistical noise, and prove that all the previous results are still valid.
In \autoref{infinite}, we extend our result to the general case of an infinite input space.
We conclude in \autoref{concl}.
In \autoref{ch5-5} we show, by means of a counterexample, that a circuit violating our hypotheses can generate a function that is not distributed as a Gaussian process. In \autoref{architecture}, we compute some architecture-independent bounds for the cardinalities of the light cones.

In order to improve the readability of the paper, \autoref{ch5}, \autoref{ch6}, \autoref{ch7} and \autoref{infinite} start with the presentation of the results, whose proofs are deferred to the second part of the section.

\begin{table}[t]
    \caption{Notation concerning the properties of the circuit at initialization.}
    \label{table1}
    \begin{tabularx}{\textwidth}{p{0.1\textwidth}X>{\raggedleft\arraybackslash}l} 
    \toprule
    Symbol & Description & Introduced in \\
    \midrule
      $m$ & number of qubits in the parameterized quantum circuit & \autoref{thecircuit}\\ 
      $L$ & number of layers in the parameterized quantum circuit& Def. \ref{deflayer}\\
      $[\ell \, m]$ & layer-qubit representation & Def. \ref{lqrep}\\
      $\Theta\in \mathscr{P}$ & vector of the parameters of the quantum circuit, belonging to the parameter space; typically $\mathscr{P}=[0,\pi]^{Lm}$ & \autoref{thecircuit}\\
      $|\Theta|$ & number of parameters $|\Theta|:=\dim\mathscr{P}=Lm$ & \textit{ibid} \\
      $x\in\mathcal{X}$ & input vector, belonging to the input (or feature) space $\mathcal{X}$; typically $\mathcal{X}=[0,\pi]^{\dim\mathcal{X}}$ & \textit{ibid.}\\
      $U(\Theta,x)$ & parameterized quantum circuit (unitary operator) & \textit{ibid.}\\
      $f(\Theta,x)$ & function generated by the quantum neural network, often called \textit{original model} & \textit{ibid.}\\
      $N(m)$ & normalization factor of the model & \textit{ibid.}\\
      \multicolumn{1}{l}{$\mathcal{M}_i$} & \multicolumn{1}{l}{(extended) future light cone of the parameter $i$} & \multirow{2}{*}{\hspace{-1em}\bigg\} 
      \parbox[l]{8.5em}{\raggedright Def. \ref{deflc}, Def. \ref{lightcones} and Cor. \ref{extlc}} }\\
      \multicolumn{1}{l}{$\mathcal{N}_k$} & \multicolumn{1}{l}{(extended) past light cone of the observable $k$} & \\
      $|\mathcal{M}|$ & maximal cardinality of a future light cone in the circuit & Def. \ref{defmax}\\
      $|\mathcal{N}|$ & maximal cardinality of a past light cone in the circuit & \textit{ibid.}\\
      \bottomrule
    \end{tabularx}
\end{table}
\begin{table}[t]
    \caption{Notation concerning the training of the circuit.}
    \label{table2}
    \begin{tabularx}{\textwidth}{p{0.1\textwidth}X>{\raggedleft\arraybackslash}l} 
    \toprule
    Symbol & Description & Introduced in \\
    \midrule
      $n$ & number of examples in the training set  & \autoref{datasetandexamples}\\
      $\mathcal{D}$ & training set, whose elements are denoted by $(x^{(i)},y^{(i)})$ for $i=1,\dots,n$ & \textit{ibid.} \\
      $X$ & vector containing the inputs of the training set & \textit{ibid.}\\
      $Y$ & vector containing the outputs of the training set corresponding to $X$ & \textit{ibid.}\\
      $f^{\mathrm{lin}}(\Theta,x)$ & linearized model & \autoref{analyticsol}\\
      $ \hat K_{\Theta}(x,x')$ & empirical neural tangent kernel & Def. \ref{entk}\\
      $N_K(m)$ & normalization factor of the empirical neural tangent kernel & \textit{ibid.}\\
      $ K(x,x')$ & analytic neural tangent kernel & Def. \ref{anNTK}\\
      $\bar K(x,x')$ & (uniform) limit of the analytic neural tangent kernel in the limit of infinitely many qubits & Ass. \ref{assNTK}\\
      $ \bar K$ & abbreviation for $\bar K(X,X^T)$ & \textit{ibid.}\\
      $\lambda_{\min}^K$ & smallest eigenvalue of $\bar K$ &\textit{ibid.} \\
      \bottomrule
    \end{tabularx}
\end{table}
\begin{table}[t]
    \caption{Notation concerning the training of the circuit (continuation).}
    \label{table3}
    \begin{tabularx}{\textwidth}{p{0.1\textwidth}X>{\raggedleft\arraybackslash}l} 
    \toprule
    Symbol & Description & Introduced in \\
    \midrule
      $t$ & continuous or discrete training time & \autoref{trqnn}\\
      $F(t)$ & vector containing the model function evaluated in the inputs of the dataset, i.e., $F(t)=f(\Theta_t,X)$ & \autoref{6-2}\\
      $F^{\mathrm{lin}}(t)$ & vector containing the linearized model function evaluated in the inputs of the dataset, i.e., $F^{\mathrm{lin}}(t)=f^{\mathrm{lin}}(\Theta^{\mathrm{lin}}_t,X)$ & \autoref{analyticsol} \\
      $\eta$ & learning rate, which enters the gradient flow equation and is a function of $m$ &  \makecell[tl]{\autoref{trqnn}\\ and Ass. \ref{assETA}}\\
      $\eta_0$ & resized learning rate (which does not depend on $m$) & Ass. \ref{assETA}\\
      $\mathcal{L}(\Theta)$ & cost function for the original model according to the training set & \makecell[tl]{Def. \ref{defmse}\\ and \autoref{analyticsol}} \\
      $\Theta_t$ & parameter vector evolving via gradient flow (or gradient descent) according to the cost function & \autoref{trqnn} \\
      $\mathcal{L}^{\mathrm{lin}}(\Theta)$ & cost function for the linearized model, sometimes called \textit{linearized cost function} & \autoref{analyticsol} \\
      $\Theta_t^{\mathrm{lin}}$ & parameter vector evolving via gradient flow (or descent) according to the linearized cost function & \textit{ibid.}\\
      \bottomrule
    \end{tabularx}
\end{table}

 \section{Quantum neural networks and light cones}\label{ch4}
Quantum neural networks constitute the main paradigm in quantum machine learning \cite{Cerezo_review}. In order to give a quantitative description of their behaviour at initialization and during training, in this section we will fix a rather generic architecture (see \autoref{circuit}), which covers a very large class of situations of interest. Moreover, most of our results can be straightforwardly generalized to slightly different architectures\footnote{For instance, one could consider two-qubit parameterized gates instead of one-qubit ones, or consider local observable acting on $O(1)$ qubits instead of one-qubit observables. This would produce different constants in the bounds we will present.}. 

In the whole paper, we consider a sequence of quantum neural networks with increasing number of qubits $m$, and we study the limit of infinite width $m\to\infty$.
The key point of this section is to identify some global quantities which describe the physical features of this generic architecture. We will state all our results in terms of such quantities, without introducing further assumptions on the circuit. Conversely, if one wants to check if a specific circuit satisfies our theorems, only such global quantities have to be computed\footnote{If, on the one hand, dealing only with global quantities that ignore the detailed structure is a great simplification, on the other hand we should not deduce that all the circuits which are indistinguishable in these quantities can be considered equivalent for the machine-learning problem. Indeed, the specific structure of the circuit typically may improve (or hinder) the learning process and the generalization capacity of the model.}: 
\begin{align*}
m &= \text{ number of qubits in the circuit},\\
L &= \text{ number of layers of the circuit},\\
|\mathcal{M}| &= \text{ maximal dimension of the future light cone of a parameter},\\
|\mathcal{N}| &= \text{ maximal dimension of the past light cone of an observable},\\
N(m) &= \text{ normalization constant of the model}.
\end{align*}
All these quantities may depend on the numer of qubits $m$ and, as we will discuss in \autoref{combinazioni}, their mutual growth as $m\to\infty$ will control the good behaviour of the circuit.
The physical meaning of the light cones is depicted in \autoref{lico1}, \autoref{lico2} and \autoref{geoloc}.

The main notations used in this article are summarized in \autoref{table1}, \autoref{table2} and \autoref{table3}.

\subsection{The dataset of examples}\label{datasetandexamples}
Let $\mathcal{X}\subseteq \mathbb{R}^{\dim\mathcal{X}}$ be the feature space (or input space) and let $\mathcal{Y}\subseteq\mathbb{R}$ be the output space.
Given a training set $\mathcal{D}=\{(x^{(i)},y^{(i)})\}_{i=1,\dots,n}\subseteq \mathcal{X}\times\mathcal{Y}$, we will call $n=|\mathcal{D}|$ the number of examples and we will represent it in a vectorized form as follows
\begin{align}
X=\begin{pmatrix} x^{(1)}\\x^{(2)}\\\vdots\\x^{(n)} \end{pmatrix},\qquad 
Y=\begin{pmatrix} y^{(1)}\\y^{(2)}\\\vdots\\y^{(n)} \end{pmatrix}.
\end{align}

\modifica{Given any function $g:\mathbb{R}\to\mathbb{R}$, we will often use the following notation:
\begin{align}
    g(X)\coloneqq \begin{pmatrix} g(x^{(1)})\\g(x^{(2)})\\\vdots\\g(x^{(n)}) \end{pmatrix},\qquad g(X^T)\coloneqq \begin{pmatrix} g(x^{(1)})&g(x^{(2)})&\cdots& g(x^{(n)}) \end{pmatrix}
\end{align}
Similarly, for any bivariate function $K:\mathbb{R}\times\mathbb{R}\to\mathbb{R}$ we will write $K(X,X^T)$ to indicate the $n\times n$ matrix with entries $\left(K(X,X^T)\right)_{ij}\coloneqq K(x^{(i)},x^{(j)})$ for $1\leq i,j\leq n$.}

\begin{assumption}\label{convex}
We assume that the set of the possible inputs $\mathcal{X}$ is finite.
\end{assumption}
\begin{remark}[The case of infinite input space]\label{infremark}
    In a concrete setting we will typically have a finite amount of digits to codify our inputs (over a bounded interval), so it is reasonable to suppose $\mathcal{X}$ to be finite. Even though we will present all our results under the assumption that $\mathcal{X}$ is finite -- this allows to provide simpler and less technical proofs -- all the final statements holds even when $\mathcal{X}$ is infinite. The reader interested in this technical treatment of the infinite case can consult \autoref{infinite}, where we provide complete proofs for this more general setting.
\end{remark}

\begin{definition}[Mean squared error]\label{defmse} As a cost function associated to the training set $\mathcal{D}$, we will consider the \textit{mean squared error}:
\begin{align}
\nonumber \mathcal{L}(\Theta)&=\frac{1}{n}\sum_{i=1}^n\frac{1}{2} \left(f(\Theta,x^{(i)})-y^{(i)}\right)^2\\ &=\frac{1}{2n}\|f(X)-Y\|_2^2.
\end{align}
\end{definition}

\subsection{The model function: a quantum circuit}
\subsubsection{Layers, parameters and observables}\label{thecircuit}
Following the notations of \cite{QLazy},
let $L$ be the number of parameterized layers of the circuit, which might depend on the number of qubits $m$.\\ Let $\theta_1,\dots,\theta_{Lm}$ be the parameters of the circuit, which we will compactly write in the vector
\begin{align}
\Theta=\begin{pmatrix} \theta_1\\\theta_2\\\vdots\\\theta_{Lm} \end{pmatrix},
\end{align}
whose dimension will often be denoted as $|\Theta|$.
A circuit is a finite composition of parameterized layers. We need to clearly define the structure of a layer and how the parameterized gates are arranged.

\begin{definition}\label{deflayer}
A \textit{layer} is the unitary operation $U_\ell(\Theta,x)\in \mathcal{L}(\mathcal{H})$ ($1\leq \ell\leq L$)  resulting from
\begin{enumerate}
\item the application on each qubit of a different parameterized single-qubit gate $W_i(\Theta)\in \mathcal{L}(\mathbb{C}^2)$; each parameterized gate depends on a single parameter $\theta_i$, which is different for each gate,
\end{enumerate}
followed by
\begin{enumerate}
\setcounter{enumi}{1}
\item a set of one-qubit and two-qubit gates acting on disjoint qubits, \emph{i.e.}, each qubit can be acted on by at most one gate; each gate may depend only on the input $x$; the resulting unitary operation will be called $V_\ell(x)\in\mathcal{L}(\mathcal{H})$.
\end{enumerate}
\end{definition}

\begin{figure}[ht]
\centering
\includegraphics[width=0.55\textwidth]{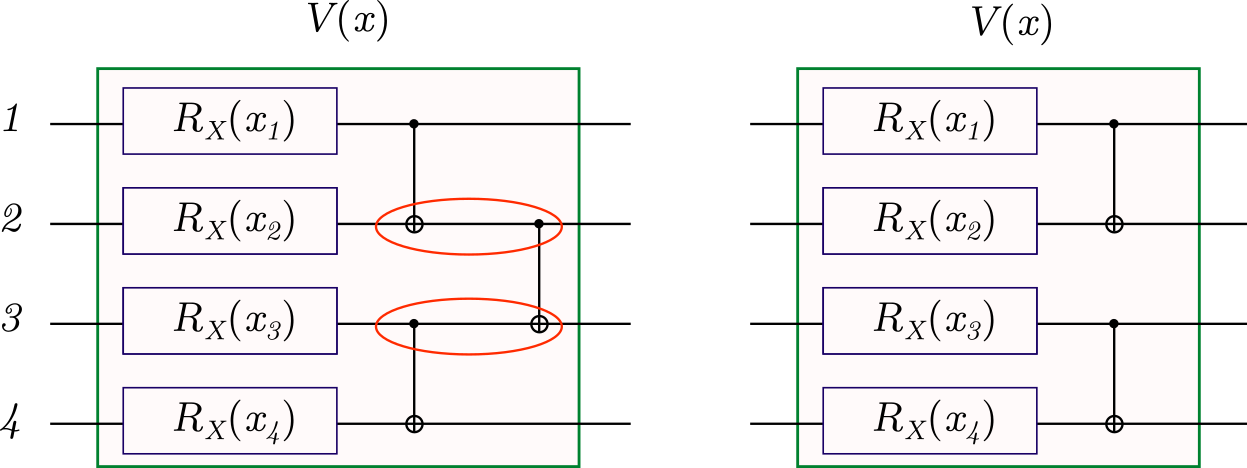}
\caption{On the left, an internal structure of $V(x)$ which is forbidden according to Definition \ref{deflayer}: \modifica{the entangling gate acting on the qubits 2 and 3 must be moved to a new layer.} The \modifica{structure} on the right is allowed.}
\label{2layer}
\end{figure}
To clarify the second point, it is useful to look at \autoref{2layer} as an example: the internal structure of $V(x)$ on the left is not allowed according to our definition of layer, since qubit 3 interacts with both qubit 2 and qubit 4. On the contrary, the structure on the right is allowed, since each qubit interacts with at most one different qubit.

In \autoref{circuit} a circuit composed of $L$ layer is represented.
\begin{figure}[ht]
\centering
\includegraphics[width=\textwidth]{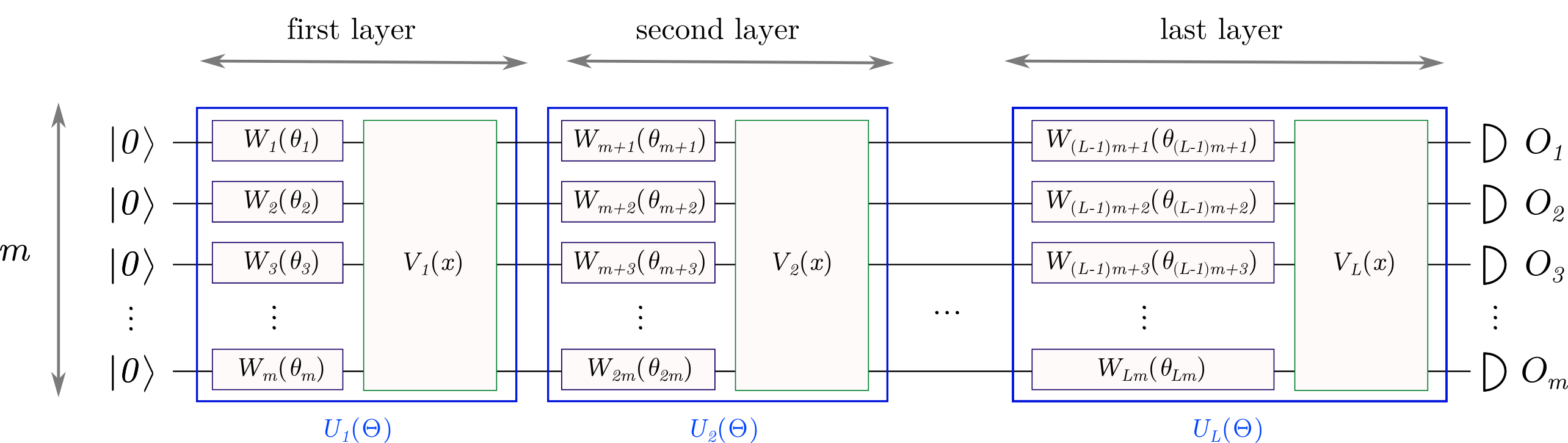}
\caption{Our parameterized quantum circuit.}
\label{circuit}
\end{figure}

We introduce a convenient notation for the indices of the parameters.

\begin{definition}[Layer-qubit representation]\label{lqrep}
Each parameter index $i\in\{1,\dots,Lm\}$ can be represented in the form $i=m(\ell-1)+k$, with $\ell\in\{1,\dots, L\}, k\in\{1,\dots,m\}$. $k$ refers to the qubit involved in the single-qubit gate parameterized by $\theta_i$, while $\ell$ refers to the layer in which such gate acts, as in  \autoref{lq}. The following compact notation, which we call \textit{layer-qubit representation} of the parameter index $i$, simplifies the above form:
\begin{align}
i=[\ell m]\equiv m(\ell-1)+k.
\end{align}
\end{definition}

\begin{figure}[ht]
\centering
\includegraphics[width=\textwidth]{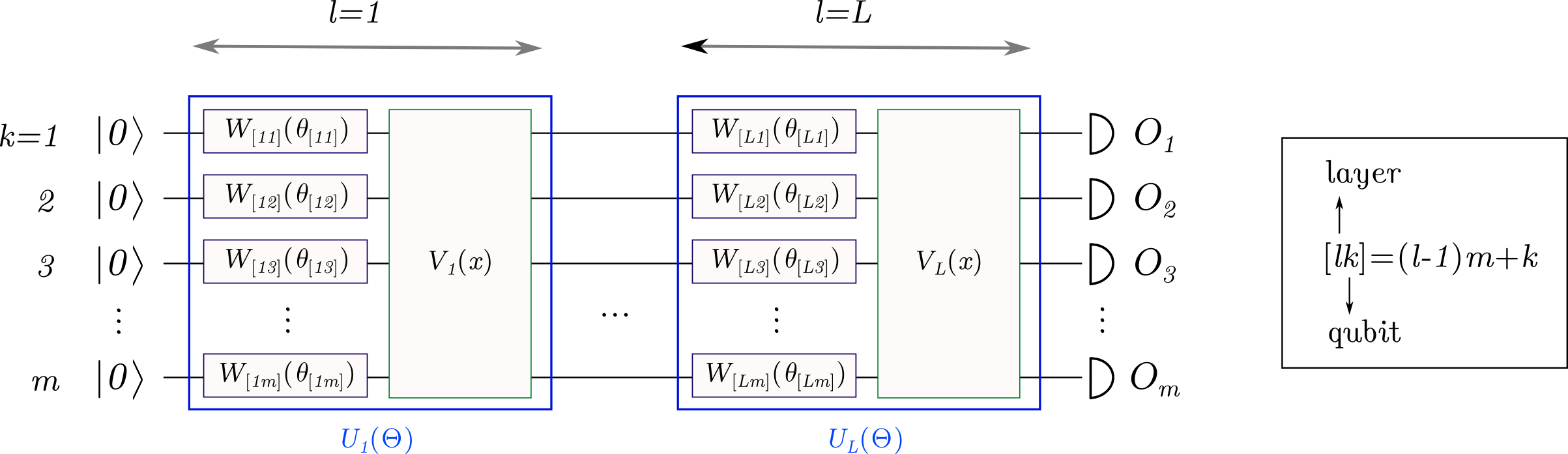}
\caption{The layer-qubit representation.}
\label{lq}
\end{figure}

Therefore, a layer $U_\ell$ can be written as
\begin{align}
\nonumber U_\ell(\Theta,x) &= V_\ell(x)\left(W_{[\ell 1]}\otimes W_{[\ell 2]}\otimes \cdots \otimes W_{[\ell m]}\right)(\Theta)\\
&=V_\ell(x)W_\ell(\Theta),
\end{align}
where
\[W_\ell(\Theta)=\left(W_{[\ell 1]}\otimes W_{[\ell 2]}\otimes \cdots \otimes W_{[\ell m]}\right)(\Theta).\]
The action of the circuit on the initial state $\ket{\psi_0}$ is described by the unitary operation
\[U(\Theta,x) = U_L(\Theta,x)\cdots U_2(\Theta,x)U_1(\Theta,x),\qquad  \ket{\psi_{out}}=U(\Theta,x)\ket{\psi_0}.\]

Let $\mathcal{O}\in\mathcal{L}(\mathcal{H})$ be the observable of the circuit, i.e., the measurement we are interested to perform on $\ket{\psi_{out}}$.
\begin{assumption}
\label{domain} 
$\mathcal{O}$ is a local observable given by the sum of single qubit observables $O_k$:
\begin{align}
\nonumber \mathcal{O}=\sum_{k=1}^m\mathcal{O}_k&=O_1\otimes \id_2\otimes\cdots\otimes \id_m\\
\nonumber &+\id_1\otimes O_2\otimes\id_3\otimes\cdots\otimes \id_m\\
\nonumber &+\dots+\\&+\id_1\otimes\cdots\otimes\id_{m-1}\otimes O_m.
\end{align}
For each $k$, we assume that $\mathrm{Tr}[O_k]=0$ and that the spectrum of $O_k$ is contained in the interval $[-1,+1]$, as in the case of Pauli observables.
We further assume that the parameterized single qubit gates of the circuit $W_i(\theta_i)$ can be written in terms of hermitian operators $\mathcal{G}_i$ with spectrum $\{-1,+1\}:$
\begin{align}
W_i(\theta_i)=e^{-i\mathcal{G}_i\theta_i}.
\end{align}
Therefore, the function $f(\,\cdot\,, x)$ is componentwise periodic with period $\pi$ for any $x\in\mathcal{X}$.
\end{assumption}

\begin{remark}\label{pauli}
Whenever the spectrum of $\mathcal{G}_i$ is $\{-1,+1\}$ for any $i=1,\dots, Lm$ (as in the case of Pauli rotation gates), in the basis diagonalizing $\mathcal{G}_i$ we have that
\[ W_i(\theta_i)= \exp\left[-i\begin{pmatrix} 1 & \\ & -1\end{pmatrix}\theta_i\right]=\begin{pmatrix} e^{-i\theta_i} & \\ & e^{i\theta_i}\end{pmatrix}\]
is a periodic function of $\theta_i$, with period\footnote{up to an irrelevant phase.} $\pi$. Therefore, the domain of $\Theta$ can be identified with the compact set $\mathscr{P}=[0,\pi]^{Lm}$.
\end{remark}

\begin{remark}
    We stress that our results strongly depends on choosing $\mathcal{O}$ as the sum of single-qubit observables over all the qubits.
    Indeed, if $\mathcal{O}$ is given by a single single-qubit observable, or more generally by an observable with eigenvalues $\pm1$, the training does not in general happen in the lazy regime \cite{you2023analyzing}.
\end{remark}

We are going to consider, as a model function, the \textit{expected value} of the measurement of $\mathcal{M}$ on the initial state $\ket{0^{m}}\equiv \otimes_{k=1}^m\ket{0}_k\in\mathcal{H}$, up to a normalization $N(m)$:
\begin{align}
\nonumber f^{(m)}(\Theta,x)&=\frac{1}{N(m)}\smatrixel{0^{m}}{U^\dagger(\Theta,x)\mathcal{O}U(\Theta,x)}{0^{m}}\\
&=\frac{1}{N(m)}\sum_{k=1}^m f^{(m)}_k(\Theta,x)\label{deflocobs},
\end{align}
where we defined
\[f^{(m)}_k(\Theta,x)=\smatrixel{0^{m}}{U^\dagger(\Theta,x)\mathcal{O}_kU(\Theta,x)}{0^{m}}.\]
$N(m)$ is determined by the covariance function of the model at initialization, i.e., when the parameters $\Theta$ are randomly initialized before the training (see Assumption \ref{zeromean}). A circuit suffering of the problem of barren plateaus would have a normalization $N(m)$ exponentially decreasing as a function of $m$.
\begin{remark}
Because of Assumption \ref{domain}, $f^{(m)}_k(\Theta,x)$ are bounded functions:
\[|f^{(m)}_k(\Theta,x)|\leq 1.\]
We will use $f(\Theta,x)$ instead of $f^{(m)}(\Theta,x)$, and $f_k(\Theta,x)$ instead of $f^{(m)}_k(\Theta,x)$ in order to have a cleaner notation.
\end{remark}

\subsubsection{The encoding of the input}\label{encinput}

The internal structure of $V_\ell(x)$ can be arbitrary as long as the second point of Definition \ref{deflayer} is ensured, i.e., each qubit interacts at most with one different qubit. This general requirement will be enough for the discussion of the case with finite $\mathcal{X}$, so the following restriction is meaningful only for the reader interested in the general case of infinite $\mathcal{X}$ (\autoref{infinite}).

Let $S_\ell$ be the set of the couples $\{i,j\}$ of qubits $i,j$ interacting in the layer $\ell$. If a qubit $i$ does not interact with any other qubit in the layer $\ell$ then we will just add the singlet $\{i\}$ to $S_\ell$. Given this definition we will suppose that the unitaries $V(x)$ have the form
\[V_\ell(x)=\prod_{E\in S_\ell}\prod_{k=1}^{\dim\mathcal{X}} U_{E,k}^{(\ell)} \prod_{j\in E} e^{-i x_k\mathcal{K}_k^{\ell,j}}\label{Vx}\]
where
\begin{enumerate}
    \item $U_{E,k}^{(\ell)}$ is an arbitrary non parameterized unitary acting only on the qubits belonging to $E$;
    \item $\mathcal{K}^{\ell,j}_k$ acts only on the qubit $j$;
    \modifica{
    \item Either $\text{Spec}(\mathcal{K}^{\ell,j}_k)=\{-1,+1\}$ (in this case, the $\ell$-th layer encodes $x_k$ nontrivially in the qubit $j$) or $\mathcal{K}^{\ell,j}_k=0$ (in this case, the $\ell$-th layer does not encode $x_k$ in the qubit $j$).
    }
\end{enumerate}
\begin{figure}[ht]
\centering
\includegraphics[width=0.88\textwidth]{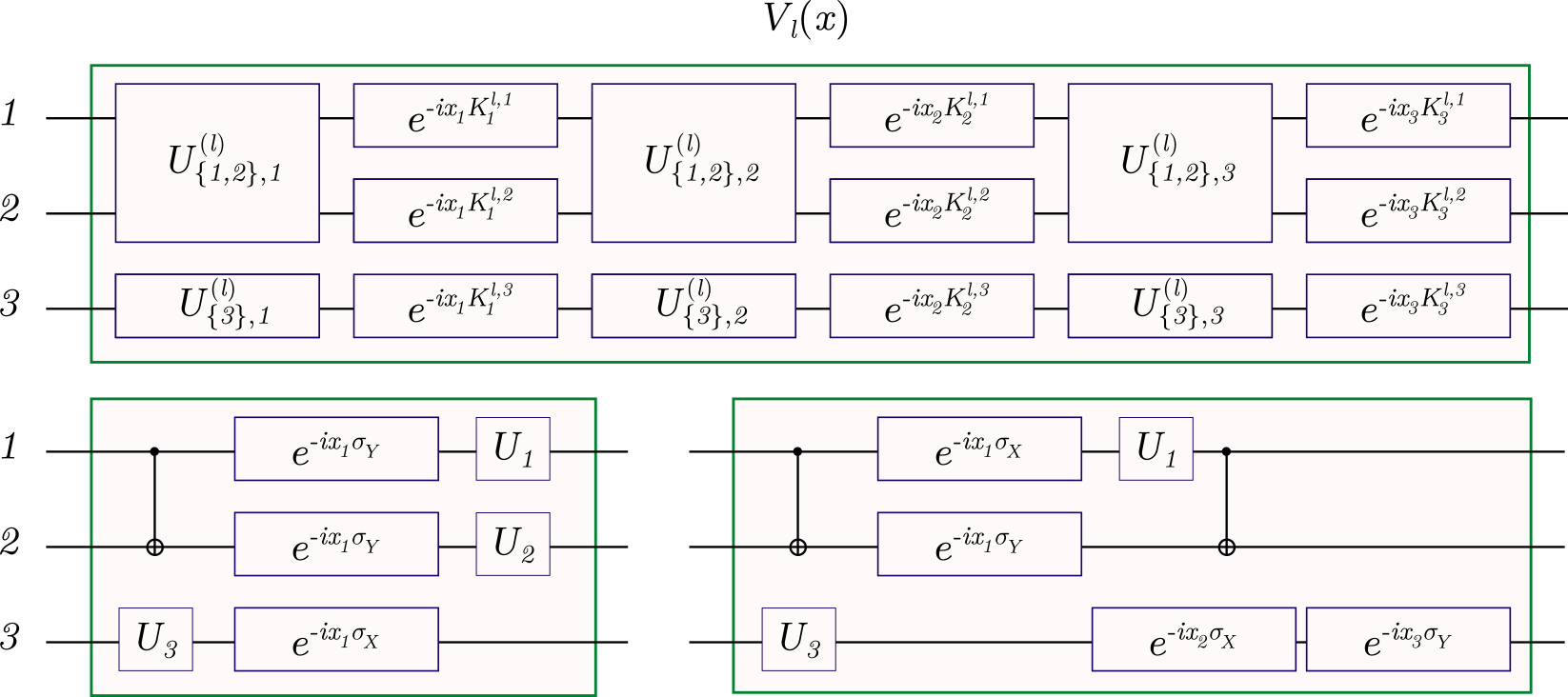}
\caption{The general structure of a feature encoding layer (above) for $m=3$, $\dim\mathcal{X}=3$ and $S_\ell = \big\{ \{1,2\},\{3\}\big\}$ according to (\ref{Vx}) with a couple of examples.}
\label{Vxfigure}
\end{figure} 
As a consequence of the third point, the model function $f(\Theta,\,\cdot\,)$ is componentwise periodic, with period $\pi$ for any $\Theta\in\mathcal{X}$, so we will consider $\mathcal{X}=[0,\pi]^{\dim\mathcal{X}}$. In conclusion,
\[\begin{dcases}
    U(\Theta,x)= V_L(x)W_L(\Theta)\cdots V_1(x)W_1(\Theta),\\
    f(\Theta,x)= \frac{1}{N(m)}\smatrixel{0^{m}}{U^\dagger(\Theta,x)\mathcal{O}U(\Theta,x)}{0^{m}}.
\end{dcases}\label{modelloconx}\]

\begin{remark}
    In Definition \ref{deflayer} we asked each qubit to be acted on \textit{by at most one gate} per layer. Even if in (\ref{Vx}) it might appear that more than one gate per qubit is acting, it is sufficient to notice that $V_\ell(x)$ can be rewritten in the form
        \[V_\ell(x)=\bigotimes_{E\in S_\ell} U_E^{(\ell)}(x),\]
    where $U_E^{(\ell)}(x)$ are single-qubit or two-qubit unitaries fixed by (\ref{Vx}), to conclude that Definition \ref{deflayer} holds.
\end{remark}

\subsection{Light cones}
\subsubsection{Light cones and their extensions}
\begin{figure}[ht]
\centering
\includegraphics[width=0.83\textwidth]{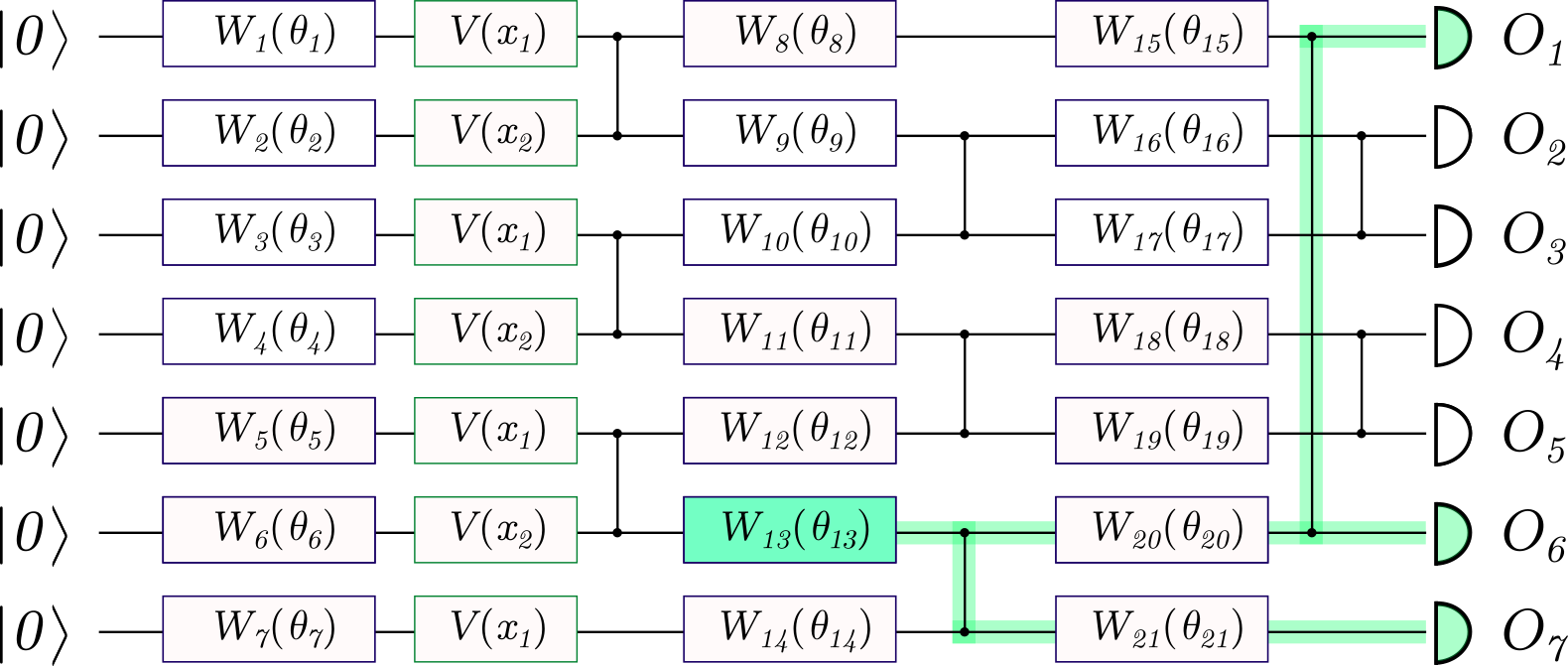}
\caption{Extended light cone $\mathcal{M}_{13}=\{1,6,7\}$ of the parameter $\theta_{13}$ for the circuit in the figure. Here $m=7$, $\mathcal{X}=[0,\pi]^2$, $|\Theta|=21$. Informally, the set $\mathcal{M}_{13}$ is the answer to the question: what are all the observables that may depend on the parameter $\theta_{13}$?}
\label{lico1}
\end{figure} 

\begin{figure}[ht]
\centering
\includegraphics[width=0.83\textwidth]{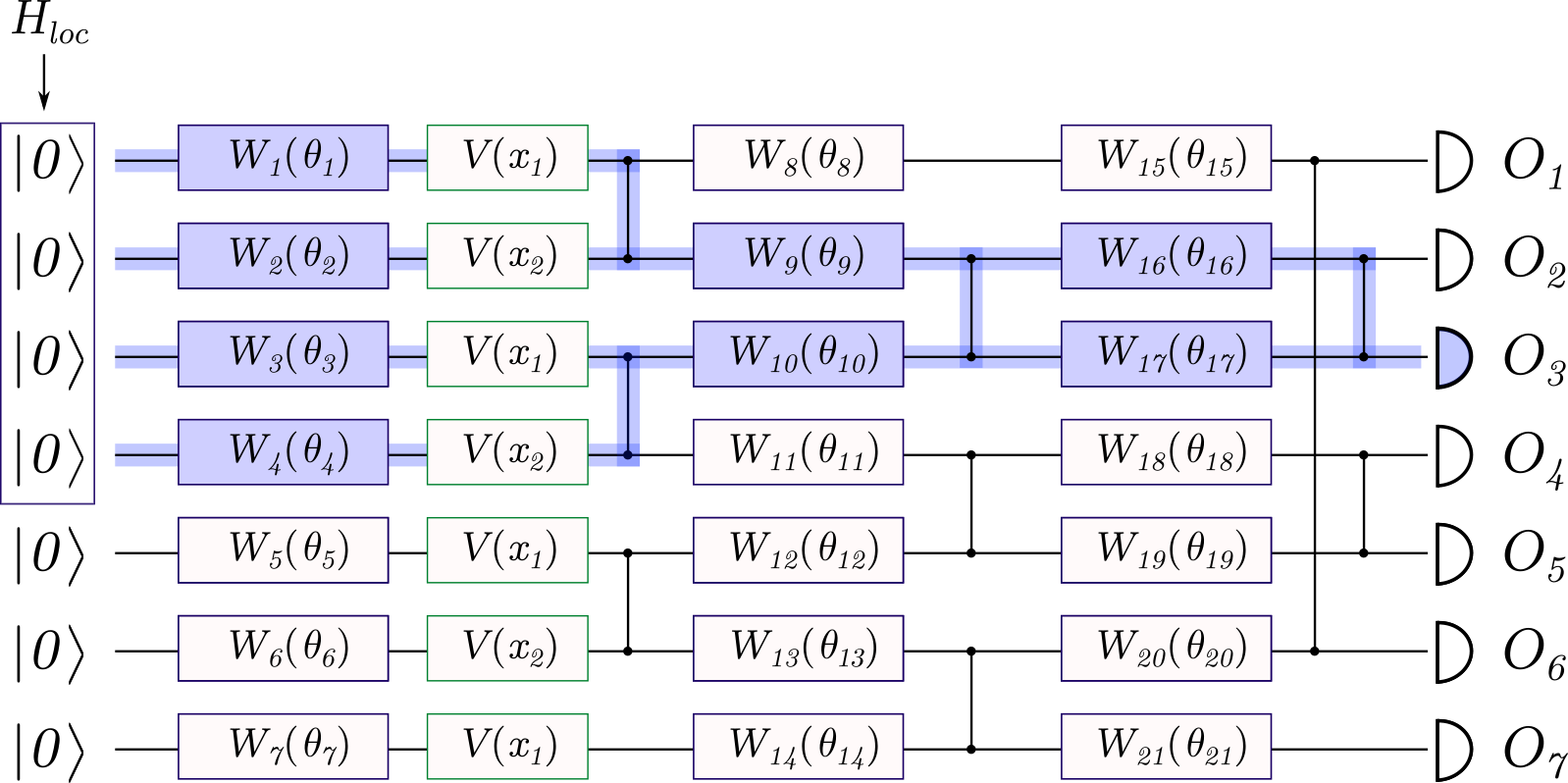}
\caption{Extended light cone $\mathcal{N}_3=\{1,2,3,4,9,10,16,17\}$ of the observable $f_3(\Theta,x)$ for the circuit in the figure. Here $m=7$, $\mathcal{X}=[0,\pi]^2$, $|\Theta|=21$. $\mathcal{H}_{\mathrm{loc}}$ is the local Hilbert space associated to $f_3(\Theta,x)$. Informally, the set $\mathcal{N}_3$ is the answer to the question: what are all the parameters that may influence the observable $f_3(\Theta,x)$?}
\label{lico2}
\end{figure}

Light cones are objects which, together with $N(m)$, characterize the \textit{global} properties of the architecture of the circuit and will be crucial to understand the interplay between
\begin{enumerate}
\item conditions for trainability,
\item conditions for quantum advantages.
\end{enumerate}
In this subsection we are going to introduce the light cone of a parameter and of a local observable, and we will define the local Hilbert space of an observable. We start from a simple definition, which is useful to understand the meaning of these sets (Definition \ref{deflc}). Then, we will give an operative definition to build some sets containing the light cones: we will call them \textit{extended light cones} (Definition \ref{lightcones}). Since the second definition allows cleaner proofs and upper bounds, we will use it in all the computations of the following sections.

\begin{definition}[Light cones]\label{deflc}
For any parameter index $i\in\{1,\dots,|\Theta|\}$ we define the \textit{future light cone} $\mathscr{L}^f_i$ \textit{of the parameter} $\theta_i$ as the subset
\begin{align}
\mathscr{L}^f_i&=\{k\in\{1,\dots,m\} : f_k(\Theta,x) \text{ depends on } \theta_i\}.
\end{align}
Besides, for any qubit index $k\in\{1,\dots,m\}$ we define the \textit{past light cone} $\mathscr{L}^p_k$ \textit{of the qubit} $k$ as the subset
\begin{align}
\mathscr{L}^p_k&=\{i\in\{1,\dots,|\Theta|\} : f_k(\Theta,x) \text{ depends on } \theta_i\}.
\end{align}
\end{definition}

Both $\mathscr{L}^f_i$ and $\mathscr{L}^p_k$ are useful to track the dependence of the observables on the parameters. More precisely, we want to understand which observables are affected by the random initialization of a parameter and by its evolution during the training. The following definition generalizes these sets (see Corollary \ref{extlc}), giving an operative way to build them by the only knowledge of the interactions between the qubits of the circuit.

For any circuit $U$, we need to define some auxiliary sets:
\begin{align}
\mathcal{I}_{\ell,k}&=\{k'\in\{1,\dots,m\}: \text{ the qubit } k \text{ interacts with} 
\\ & \phantom{=\{k'\in\{1,\dots,m\}:\;\;}\text{ the qubit } k' \text{ in the layer } \ell\}\cup \{k\}, \nonumber\\
\mathcal{J}^\ell_k&=
\begin{dcases}
\quad\mathcal{I}_{L,k}&\ell=L\\
\bigcup_{k'\in\mathcal{J}^{\ell+1}_k}\mathcal{I}_{\ell,k'}&\ell<L\\
\end{dcases},\label{J}\\
\mathcal{N}^\ell_k&=\bigcup_{k'\in\mathcal{J}^{\ell}_k}\{[\ell\, k']\}.\label{Nl}
\end{align}
\modifica{In particular $\mathcal{J}_k^1$ is the set of qubits in the past light cone of the observable $k$, i.e., the qubits involved in the computation of its expectation value.} The \textit{family of interactions} $\{\mathcal{I}_{\ell,k}\}_{\ell=1,\dots,L;k=1,\dots,m}$ will be denoted as $\mathcal{I}_U$.
When we will consider the sets (\ref{J}) and (\ref{Nl}) referred to a different circuit $U'$, we will use the notation
\[\mathcal{J}^\ell_k\Big|_{U'},\qquad \mathcal{N}^\ell_k\Big|_{U'}.\]

Now we have all the ingredients to define a generalization of the light cones.

\begin{definition}[Extended light cones]\label{lightcones}
Given any qubit index $k\in\{1,\dots,m\}$ we define the \textit{extended past light cone} $\mathcal{N}_k$ \textit{of the qubit} $k$ as the subset of the parameter indices $\{1,\dots,|\Theta|\}$ given by
\[\mathcal{N}_k=\bigcup_{\ell=1}^L\mathcal{N}^\ell_k.\]
Furthermore, given any parameter index $i\in\{1,\dots,|\Theta|\}$, we define the \textit{extended future light cone} $\mathcal{M}_i$ \textit{of the parameter} $\theta_i$ as the subset of the qubit indices $\{1,\dots,m\}$ with the following property:
\[\mathcal{M}_i=\{k\in\{1,\dots,m\}:i\in\mathcal{N}_k\}.\]
\end{definition}

\begin{lemma}[Constructive definition of the extended future light cones]\label{constructive}
As a consequence of Definition \ref{lightcones}, we have a procedure to construct the extended future light cones of the parameters using the family of interactions $\mathcal{I}_U$:
\[
\begin{dcases}
\mathcal{M}_{[L k]}=\mathcal{I}_{L,k}&\quad \ell=L,\\
\mathcal{M}_{[\ell k]}=\bigcup_{k'\in\mathcal{I}_{\ell,k}}\mathcal{M}_{[(\ell+1)\,k']} & \quad \ell<L.
\end{dcases}
\]
\end{lemma}
\begin{proof}
Using (\ref{Nl}),
\[[\ell k]\in\mathcal{N}_{k'} \quad \iff \quad [\ell k]\in\mathcal{N}_{k'}^\ell \iff \quad k\in\mathcal{J}^\ell_{k'},\]
and by Definition \ref{lightcones},
\[\mathcal{M}_{[\ell k]}=\{k':[\ell k]\in\mathcal{N}_{k'}\}=\{k':[\ell k]\in\mathcal{N}_{k'}^\ell\}=\{k':k\in\mathcal{J}^\ell_{k'}\}.\]
By (\ref{J}), if $\ell< L$,
\[k \in \mathcal{J}^\ell_{k'} \hspace{0.5em} \iff \hspace{0.5em} k \in \mathcal{I}_{\ell,\bar k} \text{ for some } \bar k\in\mathcal{J}^{\ell+1}_{k'}\hspace{0.5em} \iff \hspace{0.5em} \bar k \in \mathcal{I}_{\ell, k} \text{ for some } \bar k\in\mathcal{J}^{\ell+1}_{k'},\]
so
\[\mathcal{M}_{[\ell k]}=\{k':k\in\mathcal{J}^\ell_{k'}\}=\bigcup_{\bar k\in\mathcal{I}_{\ell,k}}\{k':\bar k \in \mathcal{J}^{\ell+1}_{k'}\}=\bigcup_{\bar k\in\mathcal{I}_{\ell,k}}\mathcal{M}_{[(\ell+1)\,\bar k]}.\]
Finally,
\[\mathcal{M}_{[L k]}=\{k':k\in\mathcal{I}_{L,k'}\}=\{k':k'\in\mathcal{I}_{L,k}\}=\mathcal{I}_{L,k}.\]
\end{proof}

\begin{definition}\label{defmax} The \textit{maximal cardinalities} of the extended light cones will be denoted as
\begin{align}
|\mathcal{M}|=\max_i|\mathcal{M}_i|,\qquad |\mathcal{N}|=\max_k|\mathcal{N}_k|
\end{align}
and, for any integer $k\geq 1$, the sum of the $k$-th power of the cardinalities of $\mathcal{M}_i$ will be called
\begin{align}
\Sigma_k=\sum_{i=1}^{|\Theta|}|\mathcal{M}_i|^k.
\end{align}
\end{definition}
\begin{remark}
Clearly, the simplest upper bound to $\Sigma_k$ in terms of $|\mathcal{M}|$ is
\[\Sigma_k\leq Lm|\mathcal{M}|^k.\]
\end{remark}

\subsubsection{The circuit seen by a local observable}

\begin{figure}[ht]
\centering
\includegraphics[width=0.87\textwidth]{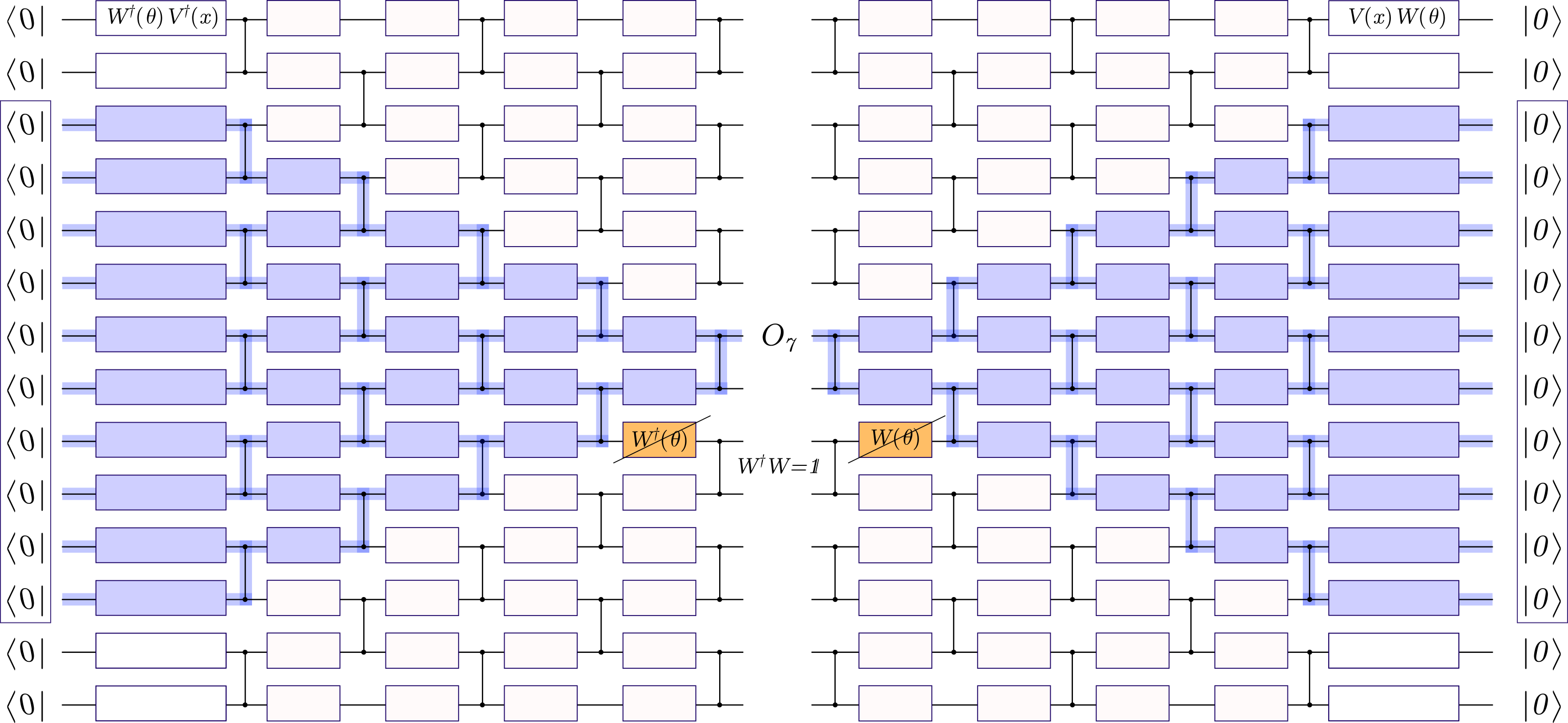}\\
\vspace{1em}
\includegraphics[width=0.87\textwidth]{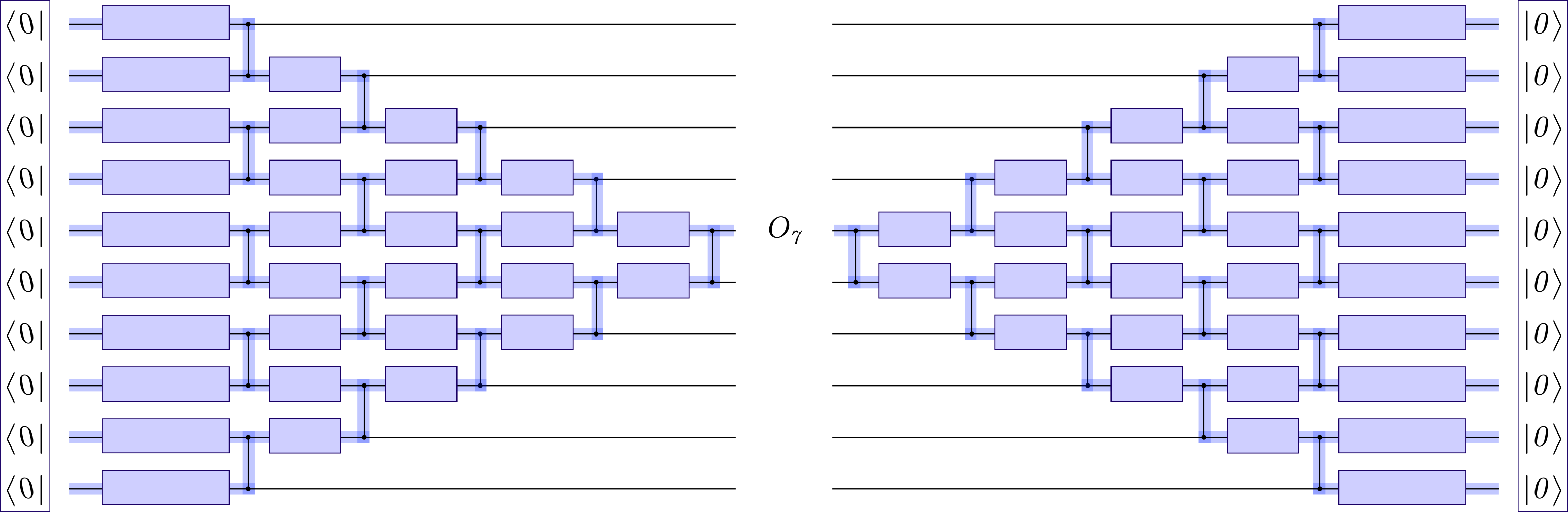}
\caption{Above: the entire circuit. Below: the part of the circuit which contributes to the computation of $\smatrixel{0^m}{U^\dagger(\Theta,x)\mathcal{O}_7U(\Theta,x)}{0^m}$. We will call it the ``pruned'' circuit.}
\label{figsimple}
\end{figure} 

The aim of this subsubsection is to determine the gates in the circuit on which the expected value of each local observable actually depends (see \autoref{figsimple}). We start from the following remark.

\begin{remark}
According to Definition \ref{deflayer}, we can rewrite
\[V_\ell(x)=\prod_{\{(k,k'): k'\in\mathcal{I}_{\ell,k}\}/\sim}V_\ell^{(k,k')},\]
where $V_\ell^{(k,k')}$ is a unitary operation acting only on the qubit $k$ and on the qubit $k'$,
and $\sim$ is the equivalence relation $(k,k')\sim (k',k)$.
The order of the product of the unitaries $V_\ell^{(k,k')}$ is irrelevant, since they commute as they act on different Hilbert spaces.
\end{remark}

\begin{definition}[Pruning of a circuit] \label{defpruning}
The \textit{pruning} operation $[\,\cdot\,]_k$ ($k=1,\dots,m$) of the circuit $U$ to which the family $\mathcal{I}_U$ is associated,
\[ \big[U\big]_k=\big[U_LU_{L-1}\dots U_1\big]_k=\left[\prod_{\ell=0}^LV_\ell(x)W_\ell(\Theta)\right]_k,\]
is the circuit obtained from $U$ by replacing, for any $k',k''\in\{1,\dots,m\}$ and $\ell\in\{1,\dots,L\}$,
\begin{align}
W_{[\ell k']} \quad\to\quad &\id\qquad\iff \qquad [\ell k']\notin\mathcal{N}^\ell_k,\\
V^{(k',k'')}_\ell \quad\to\quad &\id\qquad\iff \qquad k''\in\mathcal{I}_{\ell,k'}\text{ and }\{k',k''\}\cap\mathcal{J}^\ell_k=\emptyset.
\end{align}
\end{definition}

\begin{lemma}[Fundamental property of the pruning operation]\label{reduction}
Given any circuit $U(\Theta,x)$ according to Definition \ref{deflayer} and a local observable $\mathcal{O}_k$ acting only on the qubit $k$, for any $\ket{\psi}\in\mathcal{H}$ the following identity holds
\[\smatrixel{\psi}{U^\dagger(\Theta,x)\mathcal{O}_kU(\Theta,x)}{\psi}=\smatrixel{\psi}{\big[U(\Theta,x)\big]_k^\dagger\mathcal{O}_k\big[U(\Theta,x)\big]_k}{\psi}.\label{tesired}\]
\end{lemma}

\begin{proof}
We prove \eqref{tesired} by induction on $L$. If $L=0$, $U=\id$ and $[U]_k=\id$, the claim holds. Let us now consider a circuit composed of $L>0$ layers
\[U=U_LU_{L-1}\dots U_2 U_1\]
and assume the claim for $L-1$ layers.
The unitary 
\[U'=U_LU_{L-1}\dots U_2\]
represents a circuit of $L-1$ layers. Since the sets $\mathcal{J}^\ell_k$ and $\mathcal{N}^\ell_k$ are defined starting from the last layer, the circuits $U$ and $U'$ have in common
\[\mathcal{J}^\ell_k\Big|_U=\mathcal{J}^{\ell-1}_k\Big|_{U'}\quad \text{and} \quad \mathcal{N}^\ell_k\Big|_U=\mathcal{N}^{\ell-1}_k\Big|_{U'}\qquad \forall \, \ell>1 \]

Let $\ket{\psi'}=U_1\ket{\psi}$; by inductive hypothesis
\begin{align}
\nonumber\smatrixel{\psi'}{(U')^\dagger\mathcal{O}_kU'}{\psi'}
&=\smatrixel{\psi'}{\big[U'\big]_k^\dagger\mathcal{O}_k\big[U'\big]_k}{\psi'}\\
\nonumber &=\smatrixel{\psi}{U_1^\dagger\big[U'\big]_k^\dagger\mathcal{O}_k\big[U'\big]_kU_1}{\psi}\\
&=\smatrixel{\psi}{\left(V_1 W_1
\right)^\dagger\big[U'\big]_k^\dagger\mathcal{O}_k\big[U'\big]_k
\left(V_1 W_1\right)}{\psi}.
\end{align}
Suppose $\{k',k''\}\cap\mathcal{J}^\ell_k=\emptyset$ where $k''\in\mathcal{I}_{1,k'}$; then
\[ V_1^{(k',k'')}\quad \text{commutes with} \quad \big[U'\big]_k\]
since $\big[U'\big]_k$ acts as the identity on the qubits $k',k''$, due to the definition of the pruning operation combined with the elementary property $\mathcal{J}^{\ell}_k\subseteq\mathcal{J}^{\ell+1}_k$ for any $\ell=1,\dots,L-1$.
Furthermore, because $k\in \mathcal{J}^\ell_k$, both $k'$ and $k''$ are different from $k$, whence
\[ \left(V_1^{(k',k'')}\right)^\dagger \mathcal{O}_k V_1^{(k',k'')}=\mathcal{O}_k\left(V_1^{(k',k'')}\right)^\dagger V_1^{(k',k'')} = \mathcal{O}_k.
\]
Similarly, if $[\ell k']\notin\mathcal{N}^\ell_k$, then
\[ W_{[1 k']}\quad \text{commutes with} \quad \big[U'\big]_k \quad \text{and} \quad \mathcal{O}_k.\]
As a consequence, 
\begin{align}
\nonumber\smatrixel{\psi'}{(U')^\dagger\mathcal{O}_kU'}{\psi'}
&=\smatrixel{\psi}{\left(V_1 W_1
\right)^\dagger\big[U'\big]_k^\dagger\mathcal{O}_k\big[U'\big]_k
\left(V_1 W_1\right)}{\psi}\\
\nonumber&=\smatrixel{\psi}{\big[U'V_1 W_1\big]_k^\dagger\mathcal{O}_k\big[ U'V_1 W_1\big]_k}{\psi}\\
&=\smatrixel{\psi}{\big[U\big]_k^\dagger\mathcal{O}_k\big[ U\big]_k}{\psi},
\end{align}
which proves (\ref{tesired}) for $L$ layers.
\end{proof}

As a corollary, we can show the property anticipated above.

\begin{corollary}[Extended light cones generalize the light cones]\label{extlc}
The following relations hold
\[\mathscr{L}^p_k\subseteq \mathcal{N}_k,\qquad \mathscr{L}^f_i\subseteq \mathcal{M}_i.\]
\end{corollary}

\begin{proof}
If $f_k(\Theta,x)$ depends on $\theta_{[\ell k']}$, then, by Lemma \ref{reduction} the parameter must appear also in the pruned circuit
\[\big[U(\Theta,x)\big]_k.\]
By definition of pruning, the dependence on the parameter $\theta_{[\ell k']}$ is not removed if and only if $[\ell k']\in\mathcal{N}_k^\ell$.
Since 
\[ \mathcal{N}_k=\bigcup_{\ell=1}^L\mathcal{N}_k^\ell\]
then $\theta_{[\ell k']}\in \mathcal{N}_k$, whence
\[ \mathcal{N}_k\supseteq \mathscr{L}^p_k.\]
This implies that
\[ \mathscr{L}^f_i=\{k:i\in\mathscr{L}^p_k\}\subseteq\{k:i\in\mathcal{N}_k\}=\mathcal{M}_i.\]
\end{proof}

\begin{remark}
Since the extended light cones are immediate to construct\footnote{Instead, the original light cones may be more difficult to construct, when some cancellations of the dependencies occurs or when $\mathcal{G}_i=0$ for some $i$. The knowledege of the family $\mathcal{I}_U$ is insufficient to define $\mathscr{L}^p_k$ and $\mathscr{L}^f_i$, while is the only object involved in the construction of $\mathcal{N}_k$ and $q\mathcal{M}_i$.} using Definition \ref{lightcones}, from now on we will consider only them and, with a slight abuse of notation, we will simply refer to them as ``light cones''.
\end{remark}

In the following sections, when we will need to emphasize the dependence of a local observable $f_k(\Theta,x)$ on the subset of parameter $\{\theta_i:i\in\mathcal{N}_k\}$, we will use the notation (borrowed from \cite{QLazy}) $f_k(\Theta_{\mathcal{N}_k},x)$ instead of $f_k(\Theta,x)$.
More generally, if we will need to underline the dependence of a function $g(\Theta)$ on a subset $S\subseteq \{\theta_1,\dots,\theta_{Lm}\}$, we will write $g(\Theta_S)$ instead of $g(\Theta)$.

\subsubsection{Light cones, local Hilbert spaces and classical simulability}\label{advantages}

As we can clearly notice by looking at \autoref{figsimple} and as Theorem \ref{reduction} implies, the Hilbert space involved in the computations of the output $f_k(\Theta,x)$ of a local observable is smaller that $\mathcal{H}$. This motivates the following definition.

\begin{definition}[Local Hilbert space]
If we denote as $\mathcal{H}_1,\mathcal{H}_2,\dots,\mathcal{H}_m$ the Hilbert space associated to the qubits $1,2,\dots, m$, the \textit{local Hilbert space} associated with a local observable $\mathcal{O}_k$ is defined as
\[\mathcal{H}_{\mathrm{loc}}^k=\bigotimes_{k'\in\mathcal{J}^1_k}\mathcal{H}_{k'}.\label{locH}\]
\end{definition}
In particular
\[\dim\mathcal{H}_{\mathrm{loc}}^k=2^{|\mathcal{J}^1_k|}.\label{dimHlocJ}\]

\begin{lemma}[Sufficiency of the local Hilbert space] 
The computation of $f_k(\Theta,x)$ requires only linear-algebra operations in $\mathcal{H}_{\mathrm{loc}}^k$.
\end{lemma}

\begin{proof}
Because of Theorem \ref{reduction}, the only qubits involved in the computation of $f_k(\Theta,x)$ are the ones on which $[U]_k$ acts nontrivially. In particular, $[U]_k$ is the identity outside the tensor product of the Hilbert spaces of the qubits
\[\bigcup_{\ell=1}^L\mathcal{I}_k^\ell=\mathcal{I}_k^1,\]
where we used that $\mathcal{I}_k^{\ell+1}\subseteq \mathcal{I}_k^\ell$ for any $\ell=1,\dots,L-1$.
Such tensor product is (\ref{locH}).
\end{proof}

A necessary condition to achieve a quantum advantage with respect to classical computations is the requirement that the dynamics of a quantum circuit must be hard to simulate classically.
The linear-algebra operations required to compute the model function $f(\Theta,x)$ are of the order of $m$ times the operations required to compute a local observable $f_k(\Theta,x)$. The order of the number of these operations is $O(\mathrm{poly}(\dim \mathcal{H}_{\mathrm{loc}}))$.

\begin{lemma}[Necessary condition for quantum advantages]\label{neccondqa} 
We can estimate the dimension of the local Hilbert space:
\[2^{|\mathcal{N}|/L}\leq \max_k\dim\mathcal{H}^k_{\mathrm{loc}}\leq 2^{|\mathcal{N}|}\]
Therefore, a necessary condition to achieve a quantum advantage is that $|\mathcal{N}|$ grows superlogarithmically with the number of qubits.
\end{lemma}

Any logarithmic or sublogarithmic growth of $|\mathcal{N}|$ excludes any superpolynomial quantum advantage.

\begin{proof}
By Definition \ref{lightcones}, and using the fact that $\mathcal{N}^1_k\supseteq\mathcal{N}^\ell_K$ for any $\ell = 1,\dots,L$,
\[\mathcal{N}_k=\bigcup_{\ell=1}^L\mathcal{N}^\ell_k\qquad \to\qquad |\mathcal{N}_k|\leq L|\mathcal{N}^1_k|\quad \text{and}\quad |\mathcal{N}^1_k|\leq |\mathcal{N}|,\]
Since $|\mathcal{N}^1_k|=\log_2\dim\mathcal{H}_{\mathrm{loc}}^k$,
\[\frac{|\mathcal{N}_k|}{L}\leq\log_2\dim\mathcal{H}_{\mathrm{loc}}^k\leq |\mathcal{N}|.\]
Taking the maximum over $k$,
\[\frac{|\mathcal{N}|}{L}\leq \max_k\log_2\dim\mathcal{H}^k_{\mathrm{loc}}\leq |\mathcal{N}|.\]
\end{proof}

\rimodifica{
\begin{remark}\label{rem:lab}
    The fact that the outcome of any local observable only depends on the state reduced to the local Hilbert space has operational consequences when the circuit has to be experimentally realized and gives further insights about the meaning of the infinite width limit in quantum neural networks.
    \begin{enumerate}
        \item Concerning the implementation, wide architectures do not actually require a quantum computer with as many qubits as their widths: the number of qubits which are needed in the laboratory is not $m$, but $\max_k\log\dim\mathcal{H}_{\rm loc}^k=\max_k |\mathcal{J}^1_k|$ (see \eqref{dimHlocJ}), since the expectation value of each measured one-qubit observable can be estimated independently, and such estimate requires only the qubits in the past light-cone of the measured qubit. Without producing architectures which can be naively classically simulated, it is possible to choose $|\mathcal{J}_k^1|=O((\epsilon \log m)^d)$ with $d\geq 2$ and $\epsilon>0$ smaller than a certain threshold, as we will discuss in detail in \autoref{combinazioni}. This means that the effective number of qubits needed to run a quantum neural network of size $m$ is much smaller than $m$.
        \item Since number of qubits involved in the light cones is the key determining the effective size of the implemented network and the classical (non-)simulability properties, we should think about $\max_k |\mathcal{J}^1_k|$ as the effective width of the circuit.
    \end{enumerate}
\end{remark}

}

\subsubsection{Architecture-independent bounds and geometrically local circuits}\label{sec:lattice}
Is it possible to estimate the cardinalities of the light cones without knowing the specific structure of the circuit? Yes, it is, but the bounds are very weak:
\[ |\mathcal{M}|\leq 2^L\qquad |\mathcal{N}|\leq 2^{L+1}\qquad \Sigma_n\leq 2m2^{nL}\]
For the derivation of these bounds, see \autoref{architecture}.

\begin{figure}[ht]
\centering
\includegraphics[width=0.47\textwidth]{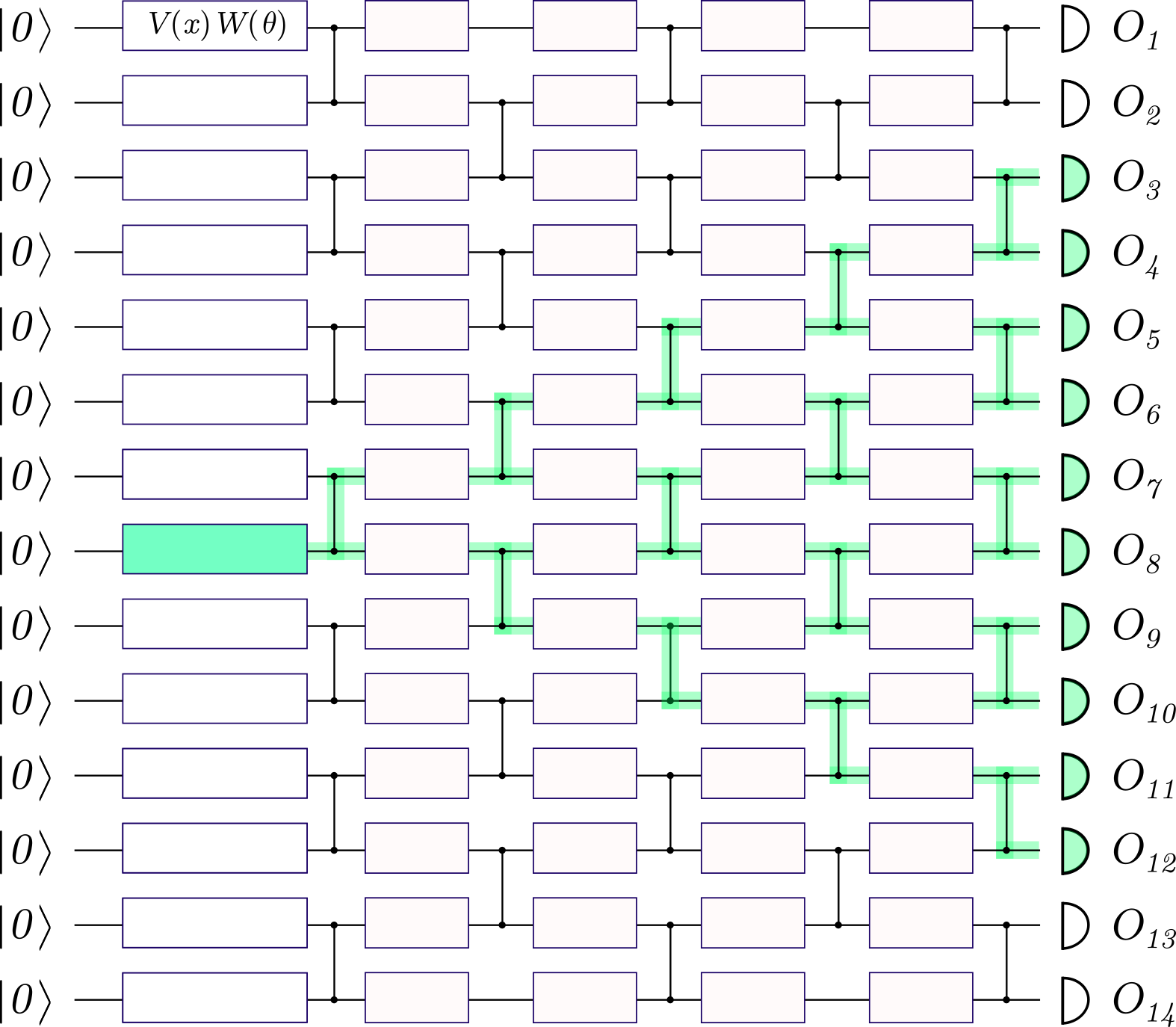}\hspace{1em}
\includegraphics[width=0.47\textwidth]{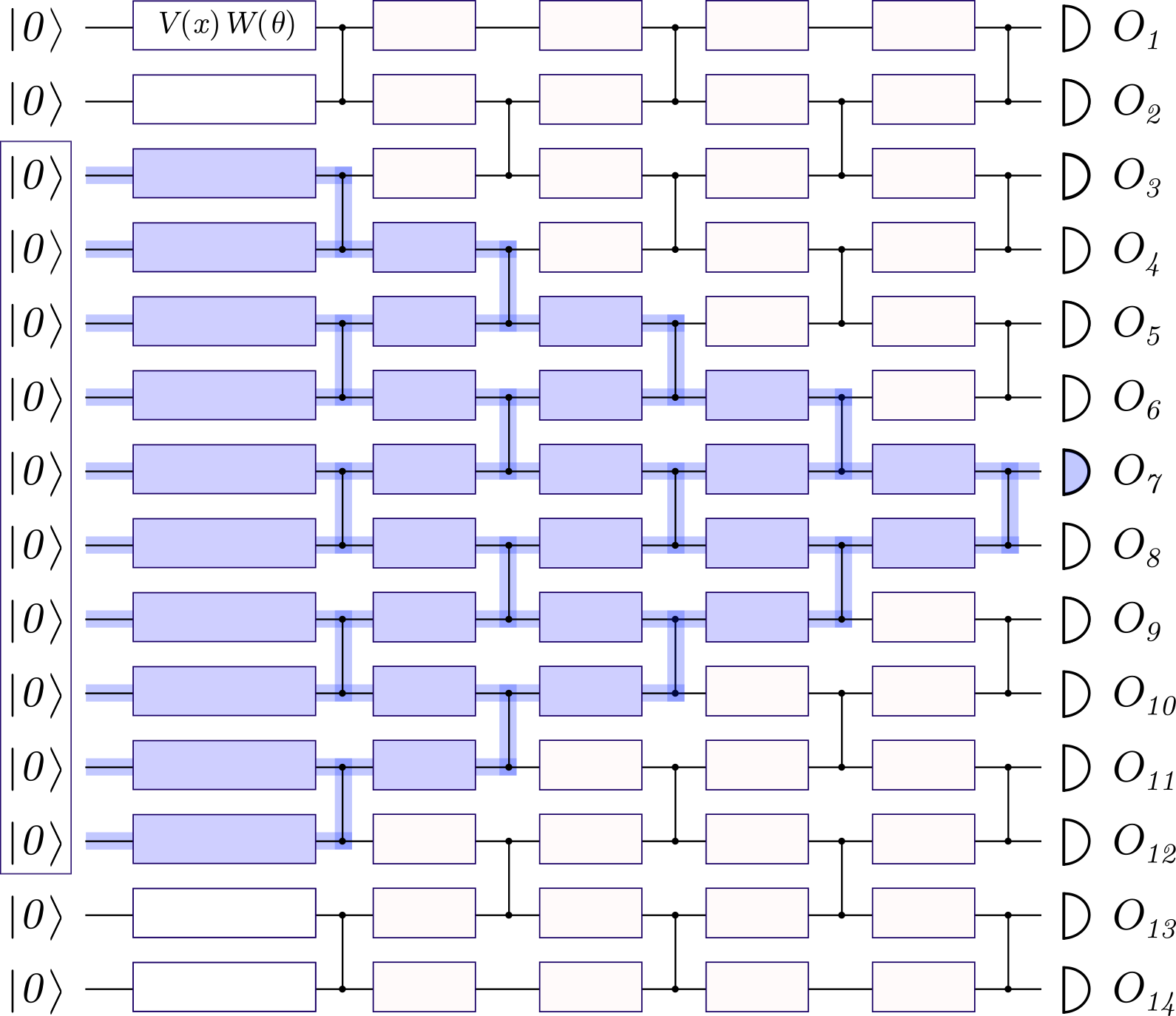}
\caption{Light cones $\mathcal{M}_{[1\,8]}$ and $\mathcal{N}_7$ for a geometrically local circuit. Since each qubit interacts only with neighbouring qubits, the growth of the light cones is polynomial in the number of layers $L$.}
\label{geoloc}
\end{figure}

All the bounds obtained for the most generic architecture using 2-qubit gates are exponential in $L$, which in general is not a good dependence if we are interested in the case $L=L(m)$. These estimates -- which accounts also for the worst case -- are typically too loose since, in a physical implementation with many qubits, the reasonable assumption that only neighbouring qubits interact\footnote{Such architectures are called \textit{geometrically local} \cite{QLazy}.} turns this exponential behaviour into a polynomial growth. This can be informally realized by looking at \autoref{geoloc}. In particular, the dependence on the number of qubits is of the order
\[\text{dim}\,\mathcal{H}_{\mathrm{loc}}=2^{O(L^d)},
\qquad
|\mathcal{N}|=O(L\cdot L^d),
\qquad
|\mathcal{M}|=O(L^d),\]
where $d$ is the dimension of the lattice on which qubits are arranged. We will not enter into the details of these circuits since in our theorems we never assume geometrical locality, and all the theorems will be stated in terms of the maximal cardinalities of the light cones.

We will come back to the problem of the exponential dependence on the number of layers in \autoref{combinazioni}.

\subsection{Random initialization of quantum circuits}
We say that a circuit is randomly initialized when the parameters $\theta_i$ are sampled from independent (not necessarily identical) distributions. At the moment, we do not add any assumption on these distributions other than the independence and the following requirement.

\begin{assumption}\label{finallayer}
The final layer is chosen so that, for any $k=1,\dots,m$, 
\[\mathbb{E}[f_k(\Theta,x)]=0.\label{2.61}\]
\end{assumption}

Now we show that any circuit $U(\Theta,x)$ with traceless local observables satisfies \modifica{(\ref{2.61})} by means of a specific final layer of local parameterized gates. To understand how this works, let us consider the simpler case in \autoref{zz} where the observables are Pauli $Z$.

\begin{figure}[ht]
\centering
\includegraphics[width=0.98\textwidth]{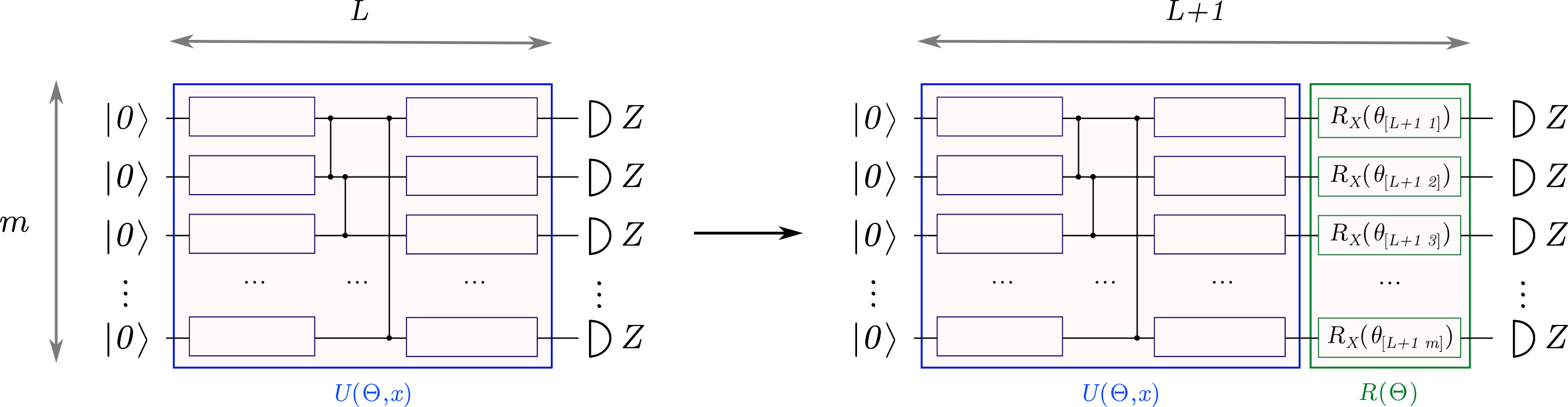}
\caption{The zero mean hypothesis can be obtained by adding a final layer.}
\label{zz}
\end{figure}

Let us add a final layer 
\[ U_{L+1}(\Theta)=R_X(\theta_{[(L+1)\, 1]})\otimes R_X(\theta_{[(L+1)\, 2]})\otimes\cdots\otimes R_X(\theta_{[(L+1)\, m]}),\]
where
\[ R_X(\theta)=e^{-i\theta \sigma_X}\qquad \sigma_X= \begin{pmatrix}  & 1\\1 & \end{pmatrix}.\]
Therefore
\begin{align}
\nonumber f'_k(\Theta,x)&=\bra{0^m}U^\dagger(\Theta,x)U^\dagger_{L+1}(\Theta)\\
\nonumber&\phantom{=\bra{0^m}U^\dagger U}
\id_1\otimes \cdots \otimes \id_{k-1}\otimes Z_k \otimes \id_{k+1}\otimes \cdots \otimes \id_m\\
\nonumber&\phantom{=\bra{0^m}U^\dagger U
\id_1\otimes \cdots \otimes \id_{k-1}\otimes Z_k \otimes \id_{k+1}\otimes}
U_{L+1}(\Theta)U(\Theta,x)\ket{0^m}\\
\nonumber&= \bra{0^m}U^\dagger(\Theta,x)\\
\nonumber&\phantom{=\bra{0^m}U}
\id_1 \otimes \cdots \otimes e^{i\theta_{[(L+1)\,k]}\sigma_X}Z_k e^{-i\theta_{[(L+1)\,k]}\sigma_X} \otimes \cdots \otimes \id_m\\
&\phantom{=\bra{0^m}U
\id \otimes \cdots \otimes e^{i\theta_{[(L+1)\,k]}X_k}Z_k e^{-i\theta_{[(L+1)\,k]}X_k} \otimes}
U(\Theta,x)\ket{0^m},
\end{align}
where $X_k$ and $Z_k$ are respectively $\sigma_X$ and $\sigma_Z$ acting on the qubit $k$.
Using the fact that
\[\sigma_i\sigma_j=\delta_{ij}+i\epsilon_{ijk}\sigma_k\]
and
\[e^{-i\theta \hat n\cdot \vec \sigma}=\id \cos \theta-i\hat n \cdot \vec\sigma\sin\theta \quad \text{with}\quad  |\hat n|=1,\]
we see that
\begin{align}
\nonumber e^{i\theta \sigma_X}\sigma_Z e^{-i\theta \sigma_X}&=
(\id \cos\theta+i\sigma_X\sin\theta)\sigma_Z(\id \cos\theta-i\sigma_X\sin\theta)\\
\nonumber&= (\id \cos\theta+i\sigma_X\sin\theta)(\sigma_Z \cos\theta+\sigma_Y\sin\theta)\\
\nonumber&= \modifica{\sigma_Z \cos^2\theta + \sigma_Y \cos\theta\sin\theta +\sigma_Y\sin\theta\cos\theta -\sigma_Z\sin^2\theta}\\
\nonumber&= \modifica{\sigma_Z(\cos^2\theta-\sin^2\theta)+\sigma_Y2\sin \theta\cos\theta}\\
&=\modifica{\sigma_Z \cos 2\theta+\sigma_Y\sin 2\theta}.
\end{align}
If $\theta$ is uniformly distributed in $[0,2\pi]$, then
\[\mathbb{E}[\cos 2\theta]=0\qquad\modifica{\text{and}\qquad\mathbb{E}[\sin 2\theta]=0}.\]
So, if, for all $k=1,\dots,m$, $\theta_{[(L+1)\, k]}$ is uniform in $[0,2\pi]$ and independent of $\theta_1,\dots,\theta_{Lm}$, then
\rimodifica{\begin{align}
\nonumber f'_k(\Theta,x) &= 
\bra{0^m}U^\dagger(\Theta,x)
\id_{1,\dots, k-1} \otimes Z_k \cos 2\theta_{[(L+1)\, k]} \otimes \id_{k+1,\dots, m}
U(\Theta,x)\ket{0^m}\\
&\quad +
\bra{0^m}U^\dagger(\Theta,x)
\id_{1,\dots, k-1} \otimes Y_k \sin 2\theta_{[(L+1)\, k]} \otimes \id_{k+1,\dots, m}
U(\Theta,x)\ket{0^m}\\
\nonumber&= f_k(\Theta,x) \cos 2\theta_{[(L+1)\, k]}+g_k(\Theta,x) \sin 2\theta_{[(L+1)\, k]}; \\[12pt]
\mathbb{E}[f'_k(\Theta,x)]&=\mathbb{E}[f_k(\Theta,x)]\mathbb{E}[\cos 2\theta_{[(L+1)\, k]}]+\mathbb{E}[g_k(\Theta,x)]\mathbb{E}[\sin 2\theta_{[(L+1)\, k]}]=0,
\end{align}
for all $k\in\{1,\dots,m\}$ where
\[
\begin{aligned}
    f_k(\Theta,x)&=\smatrixel{0^m}{U^\dagger(\Theta,x)\left(\id_{1,\dots, k-1} \otimes Z_k \otimes\id_{k+1,\dots, m}\right) U(\Theta,x)}{0^m}\\
    g_k(\Theta,x)&=\smatrixel{0^m}{U^\dagger(\Theta,x)\left(\id_{1,\dots, k-1} \otimes Y_k \otimes\id_{k+1,\dots, m}\right) U(\Theta,x)}{0^m}
\end{aligned}\]}
is the output of the original circuit, which does not depend on $\theta_{[(L+1)\, k]}$.
The general case is analogous.

The following lemma will be crucial to prove -- in terms of the cardinalities of the light cones -- that randomly initialized quantum neural networks yield model functions distributed as Gaussian processes.

\begin{lemma}[Dependent observables at random initialization]
\label{PM}
For any $k=1,\,\ldots,\,m$, let \begin{align}
\nonumber \mathcal{P}_k&=\{k'\in\{1,\dots,m\} : f_{k'}(\Theta,x) \text{ is not independent of }\phantom{\}} \\
&\phantom{=\{k'\in\{1,\dots,m\} :\}}f_{k}(\Theta,x) \text{ at random initialization}\}.
\end{align}
Then
\[|\mathcal{P}_k|\leq|\mathcal{M}||\mathcal{N}|.\]
\end{lemma}

\begin{proof} 
If the observables $f_{k}(\Theta,x)$ and $f_{k'}(\Theta,x)$ are not independent, then there is a parameter on which both the qubit $k$ and the qubit $k'$ depend:
\begin{align}
\nonumber \mathcal{P}_k&\subseteq \{k'\in\{1,\dots,m\}: \mathcal{N}_k\cap\mathcal{N}_{k'}\neq \emptyset\}\\
\nonumber &= \{k'\in\{1,\dots,m\}: \exists i\in\mathcal{N}_k\cap\mathcal{N}_{k'}\}\\
\nonumber&= \bigcup_{i\in\mathcal{N}_k}\{k'\in\{1,\dots,m\}: i\in\mathcal{N}_{k'}\}\\
\nonumber&= \bigcup_{i\in\mathcal{N}_k}\{k'\in\{1,\dots,m\}: k'\in\mathcal{M}_{i}\}\\
&= \bigcup_{i\in\mathcal{N}_k}\mathcal{M}_i.
\end{align}
Therefore \[|\mathcal{P}_k|\leq |\mathcal{N}_k|\max_i|\mathcal{M}_i|\leq |\mathcal{N}||\mathcal{M}|.\]
\end{proof}

\subsection{Some examples and the estimate of \texorpdfstring{$N(m)$}{N(m)}}\label{combinazioni}
In the theorems of convergence and of trainability that we are going to prove in the next sections, the hypotheses will have the form
\begin{align}\label{condizione}
\lim_{m\to\infty}\frac{L^\alpha m^\beta |\mathcal{M}|^\gamma|\mathcal{N}|^\delta}{N(m)}=0
\end{align}
for some $\alpha,\beta,\gamma,\delta>0$ according to the precise statement.
It is important to provide concrete examples of architectures for which it is possible to estimate $N(m),|\mathcal{M}|$ and $|\mathcal{N}|$ so that the existence of valid candidates of circuits is ensured.

\begin{figure}[ht]
\centering
\includegraphics[width=0.95\textwidth]{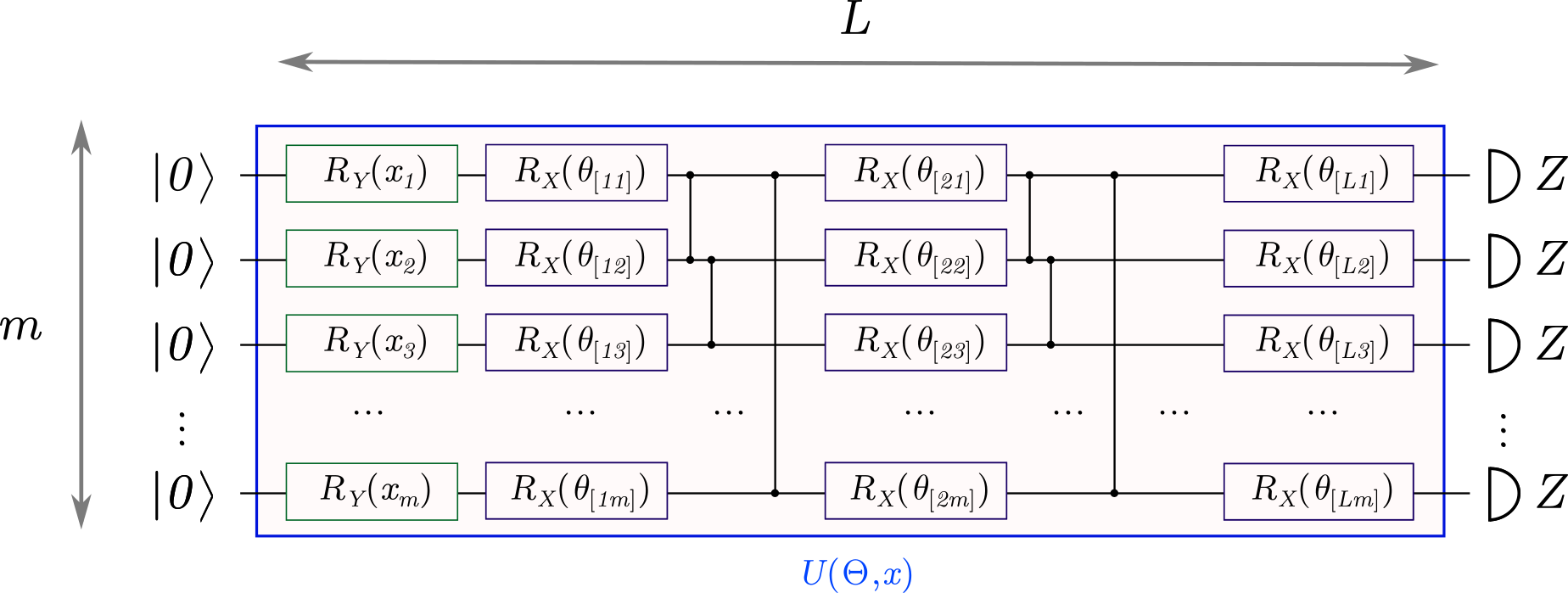}
\caption{The example described in \cite{QLazy}.}
\label{circuitQL}
\end{figure}

The example provided by E. Abedi \& al. in \cite{QLazy} (see \autoref{circuitQL}) is a circuit in which
\[ N(m)=\sqrt m, \quad |\mathcal{M}|=O(L), \quad  |\mathcal{N}|=O(L^2), \quad L \text{ is fixed}.\]
As we will see, in general $\beta<\frac{1}{2}$, therefore this kind of circuit always satisfies all our theorems. However, since $\dim\mathcal{H}_{\mathrm{loc}}=2^{O(L)}=O(1)$, such circuit is classically simulable, so we do not expect to achieve quantum advantages with such architecture.\\
A suitable class of circuits for investigating quantum advantages is the one studied by J. C. Napp in \cite{napp2022quantifying}, where the following bound is provided
\[ N(m)\geq\frac{\sqrt m}{2^{CL}}\quad \text{for some fixed } C. \label{eq:Napp}\]
\modifica{In a generic setting, the number of qubits in the past light cone of any observable $O_k$ can grow as
\[ |\mathcal{J}_k^1|=O(2^L).\]
Under the hypothesis of geometrical locality (i.e., each qubit can interact only with the nearest neighbour qubits), a $d$-dimensional lattice of qubits has
\[ |\mathcal{J}_k^1|=O(L^d).\]
Let us assume to choose $L=\epsilon\log_2 m$. Without assumptions on the geometrical locality, $|\mathcal{J}_k^1|=O(m^{\epsilon})$, which is an upper bound, so any growth $|\mathcal{J}_k^1|=\Theta(m^{\epsilon'})$ with $\epsilon'\leq\epsilon$ can be achieved by an appropriate choice of the interactions. In this case, each condition of the form (\ref{condizione}) can be satisfied for $\epsilon,\epsilon'$ small enough  
\[ \lim_{m\to\infty}\frac{(\log_2m)^\alpha m^{\beta+(\gamma+\delta)\epsilon'}}{m^{1/2-C\epsilon}}=0 \]
since in all our results the hypothesis appears with $\beta<1/2$. However, the local Hilbert spaces have a superpolynomial dimension $2^{m^{\epsilon'}}$.
In the geometrically local setting, we can choose $|\mathcal{J}_k^1|=\Theta((\epsilon\log_2m)^{d})$; if $d\geq 2$, then the local Hilbert spaces still have a superpolynomial dimension, since $2^{(\epsilon \log_2 m)^d}=m^{ \epsilon^d\log^{d-1}_2 m}$. Besides, each condition of the form (\ref{condizione}) can be satisfied for $\epsilon$ small enough.}

We should mention that the exponential decrease of $N(m)$ on the number of layers is a manifestation of the phenomenon of \textit{barren plateaus}, which will be discussed in \autoref{bp}.
We will not further investigate the particular architectures which satisfy our theorems, but our statements will be motivated by the existence of examples such as the ones we have just presented.

As a final remark, we explain the deep reason to avoid \modifica{in general} the further simplification given by Lemma \ref{cardinalities}. This lemma would reduce the hypothesis to the simpler (but stronger) form
\[ \lim_{m\to\infty}\frac{L^\alpha m^\beta 2^{(\gamma+\delta)L}}{N(m)}=0 \]
which requires only the knowledge of $N(m)$ and $L$. The exponential dependence on $L$ (combined with the behaviour $N(m)=O(\sqrt m)$, which in general is not tight) constrains $L$ to be at most logarithmically dependent on $m$. The same bottleneck appears in the case of barren plateaus. Even though this constraint does not impede a priori the achievement quantum advantages, such limit on the growth of $L$ with respect to $m$ might not be intrinsic: in particular, in the case of specific architectures avoiding barren plateaus, the simplification of Lemma \ref{cardinalities} would fictitiously rule out more interesting behaviours of $L(m)$. This is the reason why it is important to keep the explicit dependence on the maximal cardinalities of the light cones in the hypothesis of the theorems proved in the following sections.

 \section{Quantum neural networks with random parameters are Gaussian processes}\label{ch5}
In this section we are going to consider the random initialization of the parameters of the model function: $\{\theta_i\}_{i=1,\dots,Lm}$, are therefore random variables, which we assume to be independent and identically distributed.
The goal of this section is to study the conditions under which the function generated by the quantum neural network is distributed as a Gaussian process as in the classical case (Theorem \ref{init}) and to fix reasonable assumptions on the output of the randomly initialized circuits (see Assumption \ref{zeromean}).

\begin{figure}[ht]
\centering
\includegraphics[width=0.45\textwidth]{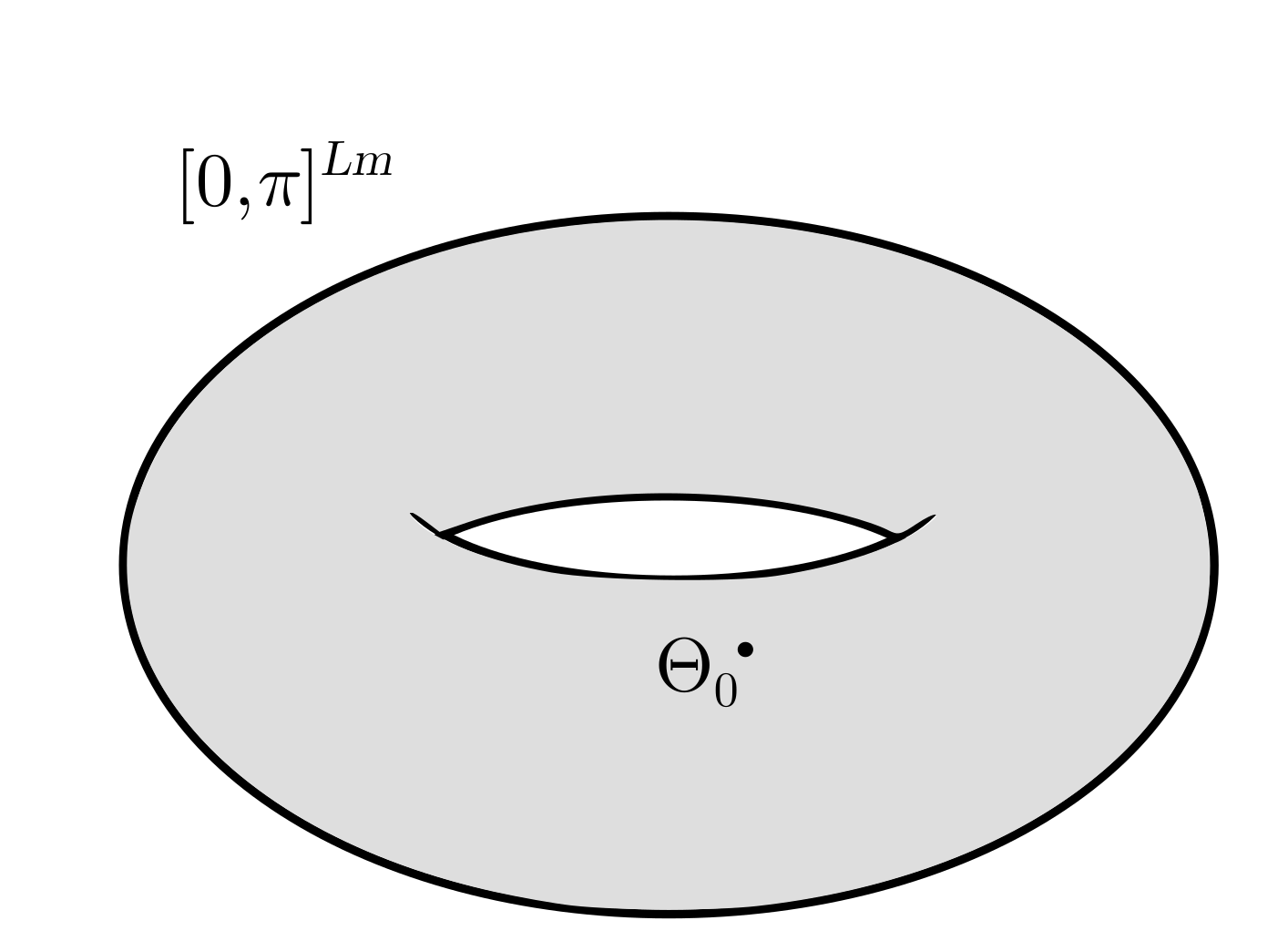}
\caption{The uniform random initialization in the parameter space $\mathscr{P}=[0,\pi]^{Lm}$.}
\label{torus}
\end{figure}

As discussed in Remark \ref{pauli}, a typical parameter space could be $\mathscr{P}=[0,\pi]^{Lm}$ with periodic boundary conditions -- which is a torus. Besides, a typical initialization could be the uniform one (with respect to Lebesgue measure on the subset $[0,\pi]^{Lm}\subseteq \mathbb{R}^{Lm}$). These particular choices are a convenient example to fix ideas on a possible initialization, but the assumptions and the theorem we will state in this section will be more general.

Differently from the classical neural networks \cite{lee2018deep}, we cannot follow an inductive procedure by performing the limit $m\to\infty$ layer by layer: in our case, the number of qubits is the same in all the layers.
Moreover, as discussed in \autoref{advantages}, the number of layer must grow with the number of qubits when we look for quantum advantages. Hence, if we proved that the output is distributed as a Gaussian process for any fixed $L$, this would be insufficient: the limit $L\to\infty$ could be performed only after the limit $m\to\infty$, but the dependence $L(m)$ requires a unique limit $m\to\infty$ in both the spatial (number of qubits) and temporal (number of layers) directions.

Our approach to face these issues is to consider the global properties of the random outputs and, in particular, their dependency relations. On the one hand, the strategy and the physical intuition behind the central limit theorem are important (see Theorem \ref{clt}). 
On the other hand, looking back at the discussion about light cones and dependence between the local measurements, it is easy to realize that the model function is a sum of functions which, in general, are not independent. However, let us fix $k\in\{1,\dots, m\}$ and consider the observable $f_k(\Theta,x)$ after a random initialization of the parameters, which are \textit{independent} random variables. For narrow enough light cones, the number of observables $f_{k'}(\Theta,x)$ whose output is not independent of the outcome of $f_k(\Theta,x)$ will be slowly growing with respect to the total number of qubits. This situation is depicted in \autoref{narrow}. Such behaviour can be thought as an approximate version of the hypothesis of the central limit theorem. If the deviations from the standard hypothesis are somehow ``small'', the result will still be valid.

Our results improve and generalize the discussion of \cite{QLazy} by means of quantitative bounds which are valid for any dependence $L(m)$ and which will be used to determine the behaviour of the circuit in the limit $m\to\infty$.

The outline of this section is the following: first, we will introduce the basic notions about the Gaussian processes and the theory of cumulants, and we will state some technical theorems from the literature to support quantitatively the intuition just presented (see \autoref{ch5-1}). Then, we will fix and discuss some assumptions about the output of the circuit (see \autoref{ch5-2}). Combining the technical results with these assumptions, we will prove the main theorem of the section (Theorem \ref{init}): a sufficient condition of the form (\ref{condizione}) to ensure the convergence to the output to a Gaussian process. Finally, we will show, by means of a counterexample, that a circuit violating the hypothesis we will have just stated does not produce, in general, an output distributed as a Gaussian process (see \autoref{ch5-5}). Even though we expect that out hypotheses can be relaxed and the statement of the convergence theorem can be extended to a larger class of circuits using a different strategy for the proof, this observation excludes a generalization to any variational circuit.

\subsection{Central limit theorem and generalizations}\label{ch5-1}
In this subsection we recall some fundamental definitions and results which serve as an introduction to our results and strategies. We start defining a Gaussian process (see \cite{Rasmussen2005} for more details).

\subsubsection{Gaussian processes and the central limit theorem}

\begin{definition}[Stochastic process]\label{stp} 
A \textit{stochastic process} is a collection of random variables $\{X_\alpha\}_{\alpha\in A}$, which defined on a common probability space.
\end{definition}
When two random variables $X$ and $Y$ have the same distribution, we write $X\sim Y$.\\
Given $\mu\in\mathbb{R}$ and $\sigma^2>0$, we write $X\sim\mathcal{N}(\mu,\sigma^2)$ when $X$ is a random variable with the following Gaussian probability density function:
\[p(x)=\frac{1}{\sqrt{2\pi\sigma^2}}e^{-\frac{(x-\mu)^2}{2\sigma^2}}, \quad x\in\mathbb{R}.\]
If $\mu\in\mathbb{R}^N$ is a vector and $\mathcal{K}\in\mathbb{R}^{N\times N}$ is a positive-definite matrix, we write $X\sim\mathcal{N}(\mu,\modifica{\mathcal{K}})$ when $X$ is a random vector of $\mathbb{R}^N$ with the following multivariate Gaussian probability density function:
\[p(x)=\frac{1}{(2\pi)^{N/2}}\frac{1}{\sqrt {\det\mathcal{K}}}e^{-\frac{1}{2}(x-\mu)^T\mathcal{K}^{-1}(x-\mu)}, \quad x\in\mathbb{R}^N.\]

\begin{definition}[Gaussian process]\label{defgp}
A stochastic process $\{X_\alpha\}_{\alpha\in A}$ is said to be a \textit{Gaussian process} (GP) if every finite collection of random variables has a multivariate Gaussian distribution, i.e., for any finite subset $F\subseteq A$ there exist a vector $\mu_F\in\mathbb{R}^{|F|}$ and a positive semidefinite covariance matrix $\mathcal{K}_F\in \mathbb{R}^{|F|\times|F|}$ such that
\[\{X_\alpha\}_{\alpha\in F}\sim \mathcal{N}(\mu_F,\mathcal{K}_F).\]
\end{definition}

In particular, if we define the mean function $\mu:A\to \mathbb{R}$ and a covariance function $\mathcal{K}:A\times A\to\mathbb{R}$ as
\begin{align}
\mu(\alpha)&=\mathbb{E}\left[X_\alpha\right],\\
\mathcal{K}(\alpha,\alpha')&=\mathbb{E}\left[\big(X_\alpha-\mu(\alpha)\big)\big(X_{\alpha'}-\mu(\alpha')\big)\right],
\end{align}
then $\mu_F$ and $K_F$ can be seen as the vector and the matrix produced by the evaluations of these functions:
\[\mu_F=\mu(F),\qquad \mathcal{K}_F=\mathcal{K}(F,F^T).\]
The mean function and the covariance function characterize the Gaussian process. Therefore, we will refer to the distribution of $\{X_\alpha\}_{\alpha\in A}$ by writing
\[\{X_\alpha\}\sim GP(\mu,\mathcal{K}).\]
As an example, the random output of a function $f:[0,1]\to\mathbb{R}$ can be distributed as a Gaussian process. In this case, we are setting $A=[0,1]$, and we are considering $\mu:[0,1]\to \mathbb{R}$ and $\mathcal{K}:[0,1]\times[0,1]\to \mathbb{R}$ so that
\[\{f(x)\}_{x\in[0,1]}\sim GP(\mu,\mathcal{K}).\]
Each sample of this distribution consists in a particular function $f$ which is defined at each $x\in[0,1]$. If we repeated many times the sampling, we would notice that the shape of the function is statistically affected by $\mu(x)$ and $\mathcal{K}(x,x')$. In particular, if we fixed a point $x_1\in [0,1]$ and we isolated \modifica{the} value $f(x_1)$ of each sample, we would notice that
\[\mathbb{E}[f(x_1)]=\mu(x_1).\]
A similar discussion can be done for the covariance of the values of the function and, more generally, for the distribution of any finite subset of values of the functions in the sample. The notion of Gaussian process, as in the classical setting, will be necessary to describe the output of the quantum circuit in the limit of infinitely many qubits.

In order to state the central limit theorem, we recall one of the equivalent definitions of the convergence in distribution (see e.g.\cite{vaart_1998, Kallenberg2021}).

\begin{definition}[Convergence in distribution]\label{defconvdistr}
Let $X$ be a real valued random variable and let $X_1, X_2, X_3, \dots$ be a sequence of real-valued random variables. The sequence $X_1, X_2, X_3, \dots$  \textit{converges in distribution} to $X$ if
\[\lim_{k\to\infty}\mathbb{E}[f(X_k)]=\mathbb{E}[f(X)]\]
for any bounded and continuous function $f:\mathbb{R}\to\mathbb{R}$.
We will use the notation
\[X_k\xrightarrow{d}X.\]
\end{definition}

\begin{lemma}[Equivalent definition of the convergence in distribution]\label{lemmaconvdistr}
Let $X$ be a real-valued random variable and let $X_1, X_2, X_3, \dots$ be a sequence of real-valued random variables. Then
\[X_k\xrightarrow{d} X \quad \text{as}\quad k\to\infty\]
is equivalent to 
\begin{align}
\nonumber\lim_{k\to\infty}\mathbb{P}\big(X_k\leq x\big)=\mathbb{P}\big(X\leq x\big) \quad&\text{for any } x\in\mathbb{R} \text{ such that} \\
&x\mapsto \mathbb{P}\big(X\leq x\big) \text{ is continuous.}
\end{align}
\end{lemma}

\begin{proof}
See e.g. \cite{vaart_1998}.
\end{proof}

\begin{theorem}[Central limit theorem]\label{clt}
Let $X_1, X_2, X_3, \dots$ be a sequence of i.i.d. random variables with
\[ \mathbb{E}[X_i]=0 \qquad \mathrm{Var}[X_i]=\sigma^2\]
and let $X\sim\mathcal{N}(0,\sigma^2)$. Then
\[ \frac{1}{\sqrt m}\sum_{k=1}^m X_k\xrightarrow{d}X \quad \mathrm{ as }\quad m\to\infty.\]
\end{theorem}

\begin{proof} See \cite{Kallenberg2021}.
\end{proof}

The following result is used in the proof of the central limit theorem and we will need it to prove Theorem \ref{init}.
\begin{theorem}[Lévy's continuity theorem]\label{levy}
Let $X_1, X_2, X_3, \dots$ be a sequence of random variable with characteristic functions
\[ \phi_k(t)=\mathbb{E}\left[e^{itX_k}\right]\qquad k\geq 1.\]
If there exists a random variable $X$ whose characteristic function $\phi(t)$ is the pointwise limit of $\phi_k(t)$, i.e.
\[\lim_{k\to\infty}\phi_k(t)=\phi(t)\qquad \forall t\in\mathbb{R},\]
then $X_k$ converges in distribution to $X$:
\[X_k\xrightarrow{d} X.\]
\end{theorem}

\begin{proof} See e.g. \cite{williams_1991, Fristedt1996-ps, Kallenberg2021}.
\end{proof}

Since, in our case, $f_k(\Theta,x)$ is not a sequence of i.i.d. random variables -- the local observables are neither identically distributed nor independent -- we need a strategy to generalize Theorem \ref{clt}. This is the aim of the following paragraph.

\subsubsection{Cumulants and dependency graphs}\label{cumulants}

The strategy to prove Theorem \ref{init} is based on the theory of cumulants of a random variable, which we briefly recall (see \cite{doring2021method}):
\begin{definition}[Cumulants]
Given a real-valued random variable $X$, its \textit{cumulants} $\kappa_j\equiv\kappa_j(X)$, $j\in\mathbb{N}$, are the defined as
\[\kappa_j(X)=(-i)^j\frac{d^j}{dt^j}\bigg|_{t=0}\log\mathbb{E}\left[e^{itX}\right],\]
provided that the derivative exists.
\end{definition}
It is easy to prove that
\[\kappa_1(X)=\mathbb{E}[X],\qquad \kappa_2(X)=\text{Var}[X],\]
when they exist.
If the characteristic function is $r$-times differentiable, an equivalent characterization of the cumulants of order $j=1,\dots, r$ is provided by the Taylor expansion:
\[\log\mathbb{E}\left[e^{itX}\right]=\sum_{j=1}^r\frac{\kappa_j}{j!}(it)^j+o(t^r) \qquad \text{as} \quad t\to 0.\]
Now we introduce the theory of dependency graphs following \cite{ModPHI}. It will be the key tool to support and quantify the intuition that a circuit in which two local observables are ``often'' independent yields an output behaving as one would expect by the central limit theorem, even though the hypotheses are not satisfied.

\begin{definition}\label{dep}
Given a family of random variables $\{X_\alpha\}_{\alpha\in A}$, a graph $\mathscr{G}$ with vertex set $V$ is called a \textit{dependency graph} for the family if the following property holds:
whenever $V_1$ and $V_2$ are disjoint subsets of $V$ such that there are no edges in $\mathscr{G}$ with one end in $V_1$ and one in $V_2$, the subfamilies of random variables $\{X_\alpha\}_{\alpha\in V_1}$ and $\{X_\alpha\}_{\alpha\in V_2}$ are \modifica{mutually independent.
We stress that there is no assumption of independence within each of the sets  $\{X_\alpha\}_{\alpha\in V_1}$ and $\{X_\alpha\}_{\alpha\in V_2}$.}
\end{definition}

The dependency graph of a family of random variables might not be unique: the complete graph is always a valid dependecy graph. However, as we will see in the next theorems, we look for dependency graph having as few edges as possible.

\begin{figure}[ht]
\centering
\includegraphics[width=0.98\textwidth]{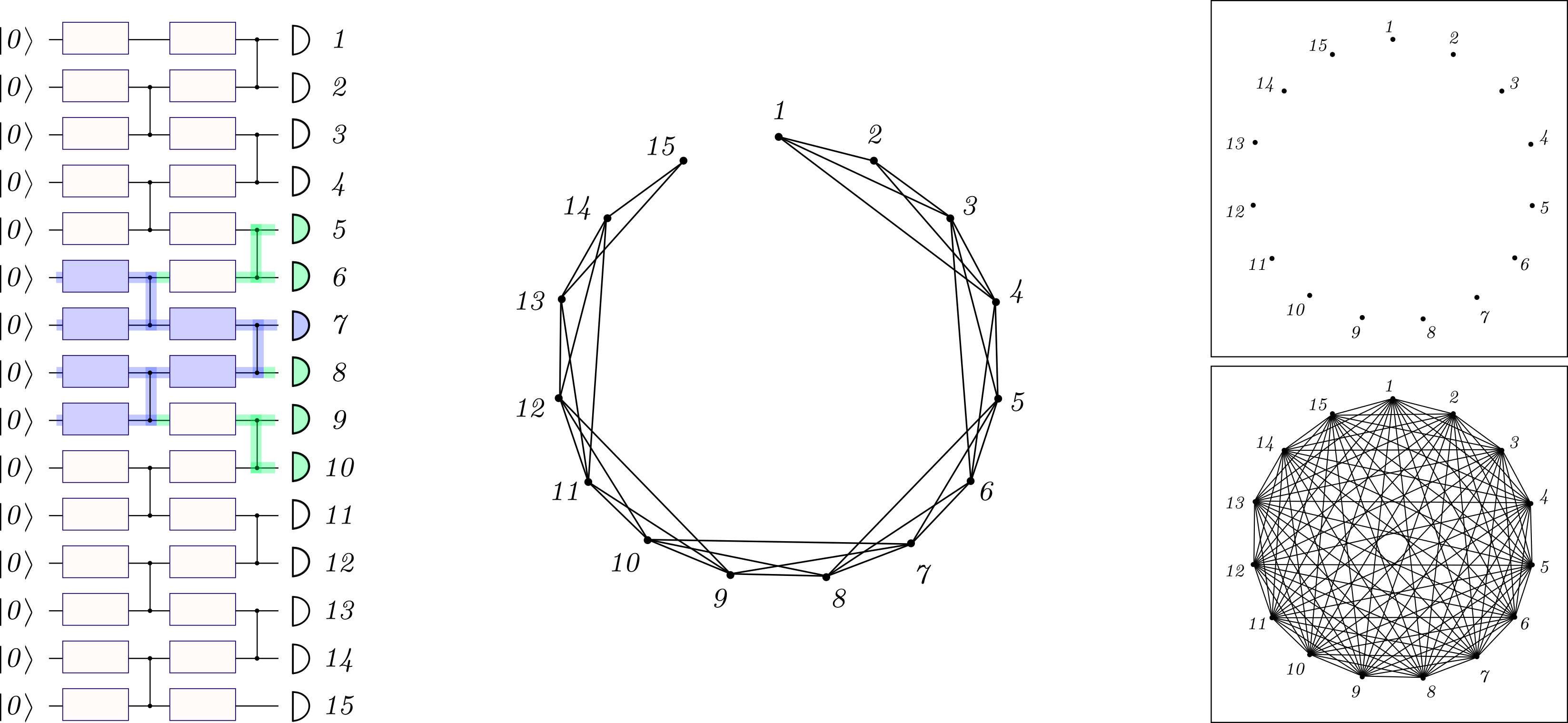}
\caption{A dependency graph for the circuit on the left compared with the trivial graph with no edges (independent observales) and the complete graph (all the variables are pairwise dependent). Informally speaking, the dependency graph of the circuit seems somehow closer to the trivial graph than the complete graph, since ``far'' observables are independent.}
\label{graph}
\end{figure}

In \autoref{graph} we see an example of a dependency graph for a quantum circuit. The edges connect the index of the observables which are not mutually independent. We should prove that, \textit{for our case}, this procedure to construct a dependency graph is compatible with Definition \ref{dep}. 
Indeed, this statement is not generalizable to any family of random variables. An example of a triple of random variables which are independent in pairs but not globally independent can be found in \cite{Hogg2004, Romano2017}.
The case of the quantum circuit we are considering is not affected by this problem.
\begin{lemma}\label{ok}
    Let us consider a quantum circuit as in \autoref{ch4} with $m$ qubits. Let $\mathscr{G}=(V,A)$ be the graph with vertices $V=\{1,\dots,m\}$ and edges $E$ defined by 
    \begin{align}
    \nonumber (k,k')\in E\qquad \iff\qquad &f_k(\Theta,x) \text{ and } f_{k'}(\Theta,x) \\
    \nonumber&\text{are not independent at random initialization}\\
    &\text{i.e., } k'\in \mathcal{P}_k.
    \end{align}
    Then, $\mathscr{G}$ is a dependency graph for $\{f_k(\Theta,x)\}_{k\in V}$ according to Definition (\ref{dep}).
\end{lemma}
\begin{remark}
As already asserted in \autoref{ch4}, in all our discussion we assume that the parameters $\theta_1,\dots, \theta_m$ are \textit{independent} random variables. We will not repeat every time this assumption, which will be almost always implicit.
\end{remark}
\begin{proof}
Recalling that
\begin{align}
\nonumber\mathcal{P}_k&=\{k'\in\{1,\dots,m\} : f_{k'}(\Theta,x) \text{ is not independent of }\phantom{\}} \\
\nonumber&\phantom{=\{k'\in\{1,\dots,m\} :\}}f_{k}(\Theta,x) \text{ at random initialization}\}\\
\nonumber&\subseteq \{k'\in\{1,\dots,m\}: \mathcal{N}_k\cap \mathcal{N}_{k'}\neq \emptyset\}\\
&=\{k'\in\{1,\dots,m\}: \exists\, i\in\{1,\dots,Lm\} \text{ such that } k\in\mathcal{M}_i \text{ and } k'\in\mathcal{M}_{i}\}.
\end{align}
Therefore, if $k'\notin\mathcal{P}_k$, then the observables
\[f_k(\Theta_{\mathcal{N}_k},x)\qquad\text{and}\qquad f_{k'}(\Theta_{\mathcal{N}_{k'}},x)\]
are functions of different parameters. Therefore, if $V_1$ and $V_2$ are subsets of $V$ such that
\[ \forall\, k\in V_1 \quad \nexists\, k'\in V_2 \quad \text{such that} \quad k'\in\mathcal{P}_k,  \]
then the sets
\[\bigcup_{k\in V_1}\mathcal{N}_k \quad \text{and} \quad \bigcup_{k'\in V_2}\mathcal{N}_{k'}\] 
are disjoint:
\[ \left(\bigcup_{k\in V_1}\mathcal{N}_k\right) \cap \left(\bigcup_{k'\in V_2}\mathcal{N}_{k'}\right)=\bigcup_{\substack{k\in V_1\\k'\in V_2}}\mathcal{N}_k\cap \mathcal{N}_{k'}=\emptyset\]
since $\mathcal{N}_k\cap \mathcal{N}_{k'}=\emptyset$ for any $k\in V_1,k'\in V_2$.
Noticing that the family $\{f_k(\Theta,x)\}_{k\in V_1}$ depends on $\bigcup_{k\in V_1}\mathcal{N}_k$
and the family $\{f_{k'}(\Theta,x)\}_{k'\in V_2}$ depends on $\bigcup_{k'\in V_2}\mathcal{N}_{k'}$, we conclude that they are independent, 
because they depend on different independent random variables.
\end{proof}

We define the \textit{maximal degree} $D$ of a graph $\mathscr{G}=(V,E)$ as the maximum number of edges incidenting on any fixed vertex
\[D=\max_{k\in V}\deg k = \max_{k\in V} |\{k'\in V: (k,k')\in E\}|.\]

Now that we are able to construct a dependency tree for our circuit, let us state the following technical theorem:

\begin{theorem}[\cite{ModPHI}]\label{janson}
 For any integer $r\geq 1$, there exists a constant $C_r$ with the following property.
Let $\{X_\alpha\}_{\alpha\in V}$ be a family of random variables with dependency graph $\mathscr{G}$. We denote with $N=|V|$ the number of vertices of $\mathscr{G}$ and $D$ the maximal degree of $\mathscr{G}$. Assume that the variables $X_\alpha$ are
uniformly bounded by a constant $A$. Then, if 
\[X = \sum_{\alpha\in V}X_\alpha,\]
one has
\begin{align}
|\kappa_r(X)|\leq C_rN(D+1)^{r-1}A^r
\label{kappa}
\end{align}
with
\begin{align}
C_r=2^{r-1}r^{r-2}.
\label{cr}
\end{align}
\end{theorem}

The upper bound (\ref{kappa}) was stated by Janson in \cite{10.1214/aop/1176991903}. D\"oring and Eichelsbacher established in \cite{D_ring_2012} that the bound, according to the original proof, holds with $C_r=(2e)^r(r!)^3$. Féray, Méliot and Nikeghbali \cite{ModPHI} found the new bound (\ref{cr}) by giving a new proof of (\ref{kappa}).

\subsection{Convergence to a Gaussian process}\label{ch5-2}

The following assumption fixes the expectation value of the model function at random initialization and the asymptotic behaviour of the normalization $N(m)$.
\begin{assumption}
\label{zeromean}
Where not otherwise specified, we will assume that the distribution of the independent random variables $\theta_1,\dots,\theta_{|\Theta|}$, the architecture of the quantum circuit and the normalization $N(m)$ chosen are such that
\[\mathbb{E}[f_k(\Theta,x)]=0\qquad \forall\,x\in \mathcal{X}\label{mean}\]
and
\[\lim_{m\to\infty}\sup_{x,x'\in\mathcal{X}} \left|\mathbb{E}[f(\Theta,x)f(\Theta,x')]-\mathcal{K}(x,x')\right|=0, \label{convergence}\]
where $\mathcal{K}:\mathcal{X}\times\mathcal{X}\to \mathbb{R}$ is an arbitrary bivariate function from the feature space to the real numbers with strictly positive diagonal elements, i.e.
\[\mathcal{K}(x,x)>0 \qquad\forall\,x\in\mathcal{X}. \label{eq:positive} \]
\end{assumption}

Therefore $N(m)$ is chosen so that the limit $m\to \infty$ yields a finite nontrivial covariance function. This assumption does not fix a unique normalization, since a multiplicative constant or asymptotically negligible corrections to $N(m)$ produce other valid normalizations. Such ambiguity is irrelevant in the following discussion: after choosing the circuit we want to study, we will always assume to have fixed a possible normalization among the ones satisfying Assumption \ref{zeromean}.

\modifica{
\begin{remark}
Assumption \ref{zeromean} implies that the operator norm of $\frac{\mathcal{O}}{N(m)}$ does in general scale with $m$.
Indeed, let us consider the toy model where the expectation values of each $\mathcal{O}_k$ are independent.
Then, to get a finite nonzero variance we need $N(m)=\sqrt{m}$, making the operator norm of $\frac{\mathcal{O}}{N(m)}$ scale as $\sqrt{m}$.
\end{remark}
}

\begin{remark}
The architecture we introduced in \autoref{ch4} satisfies (\ref{mean}) by Assumption \ref{finallayer}. 
\end{remark}
\modifica{\begin{remark}\label{rem3.3}
The reader may be worried that the existence of the limit \eqref{convergence} may be a very restrictive hypothesis.
However, we will now show that such limit always exists for any sequence of quantum neural networks satisfying the remaining hypotheses upon taking a suitable subsequence.
\begin{enumerate}
    \item Due to the multiplicative degree of freedom in the definition of $N(m)$ for each $m$, we can suppose to fix a constant $c$ (not depending on $m$) that uniformly bounds the matrix elements $\left|\mathbb{E}[f(\Theta,x)f(\Theta,x')]\right|\leq c$ for any $x,x'\in\mathcal{X}$.
    \item Since we assumed that the set of input is finite (Assumption \ref{convex}), by a standard compactness argument it is possible to identify a subsequence $(m_k)_{k\in\mathbb{N}}$ such that the limit
    \begin{align}
        \mathcal{K}(x,x')\coloneqq\lim_{k\to\infty}\mathbb{E}[f^{(m_k)}(\Theta,x)f^{(m_k)}(\Theta,x')]
    \end{align}
    exists.
    \item We can now study the properties of the circuits defined by this subsequence. To simplify the notation, in the rest of the paper we will not refer to any subsequence, but all our results can easily be generalized to such case.
\end{enumerate}
\end{remark}
}

\begin{figure}[ht]
\centering
\includegraphics[width=0.45\textwidth]{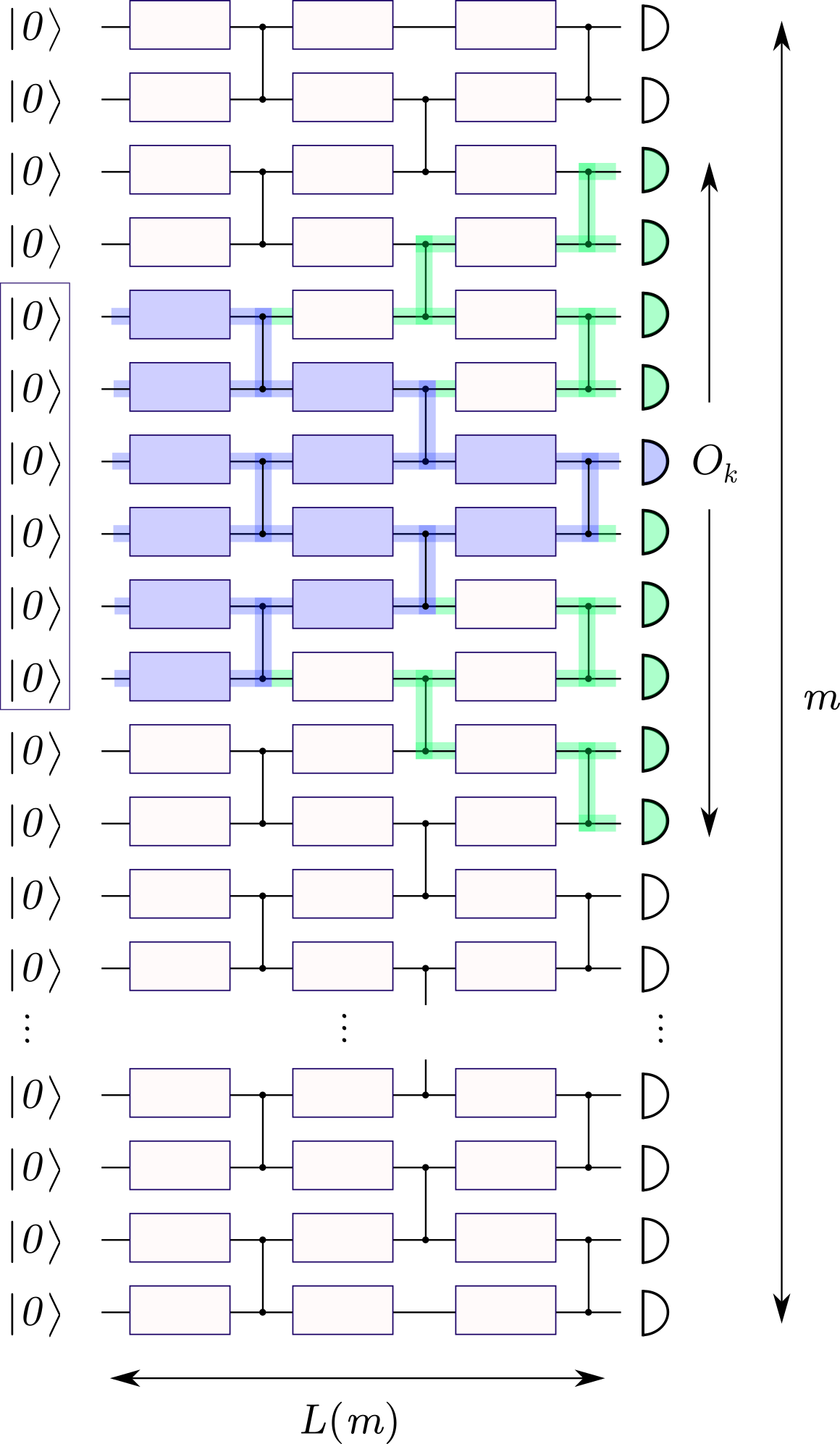}
\caption{For narrow enough light cones, the number of observables which are not independent of $f_k(\Theta,x)$ is small with respect to the total number of observables in the limit $m\to\infty$.}
\label{narrow}
\end{figure}

Now that for any $k\in\{1,\dots,m\}$ we can estimate the number of the observables which are not independent of $f_k(\Theta,x)$ at random initialization (Lemma \ref{PM}) and that we know how to construct the dependency graph $\mathscr{G}$ of the quantum circuit (Lemma \ref{ok}), we can give an upper bound to the maximal degree of $\mathscr{G}$ and apply then Theorem \ref{janson}. Combining the result with a strategy analogous to the proof of the central limit theorem, we obtain the following Theorem, which is the main result of this section.
\begin{mdframed}
\begin{theorem}[Gaussian process at initialization]
\label{init}
 Given $\theta_1,\dots,\theta_{|\Theta|}$ independent random variables, suppose that Assumption \ref{zeromean} holds and that
\begin{align}\label{hpgrowth}
\lim_{m\to\infty}\frac{m|\mathcal{M}|^2|\mathcal{N}|^2}{N^3(m)}=0.
\end{align}
Then, as $m\to \infty$, $x\mapsto f^{(m)}(\Theta,x)$ converges in distribution to a mean zero Gaussian process with covariance function $\mathcal{K}$ defined in Assumption \ref{zeromean}:
\begin{align}
f^{(m)}(\Theta,x)\xrightarrow{d}f^{(\infty)}(x)\sim GP(0,\mathcal{K})\qquad{m\to \infty}. 
\end{align}
\end{theorem}
\end{mdframed}

\subsection{Proof of Theorem \ref{init}}
Before proving Theorem \ref{init}, we need to show that the hypothesis (\ref{hpgrowth}) implies a control on the growth of the light cones with respect to the normalization constant (see Corollary \ref{controllo}) due to an upper bound of the normalization constant (Lemma \ref{Nmax}).

\begin{lemma}\label{Nmax}
Assumption \ref{zeromean} implies the following bound: there exists $\bar m\in\mathbb{N}$ for any $m\geq \bar m$.
\[N(m)\leq C\sqrt m\sqrt{|\mathcal{M}||\mathcal{N}|}.\]
\end{lemma}

\begin{proof}
We recall that we assumed that $f_k(\Theta,x)$ are bounded observables:
\[|f_k(\Theta,x)|\leq 1\qquad \forall k\in\{1,\dots,m\}\,\forall\,x\in\mathcal{X},\,\forall\,\Theta\in\mathscr{P}.\]

By Assumption \ref{zeromean} there exists $\bar m_x\in\mathbb{N}$ such that, for any $m\geq \bar m_x$
\[ 0<\frac{1}{2}\mathcal{K}(x,x)\leq \mathbb{E}[f^2(\Theta,x)]. \]
Hence, for $m\geq \bar m_{\modifica{x}}$, $\mathcal{K}(x,x)$ can be estimated as
\begin{align}
\nonumber\frac{1}{2}\mathcal{K}(x,x)&\leq\frac{1}{N^2(m)}\sum_{k,k'=1}^m\mathbb{E}[f_k(\Theta,x)f_{k'}(\Theta,x)]\\
\nonumber&=\frac{1}{N^2(m)}\sum_{k=1}^m\left(\sum_{k'\in\mathcal{P}_k}\mathbb{E}[f_k(\Theta,x)f_{k'}(\Theta,x)]
+\sum_{k'\notin\mathcal{P}_k}\mathbb{E}[f_k(\Theta,x)]\mathbb{E}[f_{k'}(\Theta,x)]\right)\\
\nonumber&=\frac{1}{N^2(m)}\sum_{k=1}^m\sum_{k'\in\mathcal{P}_k}\mathbb{E}[f_k(\Theta,x)f_{k'}(\Theta,x)]\\
\nonumber&\leq \frac{1}{N^2(m)}\sum_{k=1}^m\sum_{k'\in\mathcal{P}_k}\mathbb{E}[|f_k(\Theta,x)||f_{k'}(\Theta,x)|]\\
&\leq \frac{1}{N^2(m)}\sum_{k=1}^m\sum_{k'\in\mathcal{P}_k}1\leq \frac{1}{N^2(m)}m\max_k|\mathcal{P}_k|\leq \frac{1}{N^2(m)}m|\mathcal{M}||\mathcal{N}|.
\end{align}\\
where the last inequality follows from Lemma \ref{PM}.
Let
\[C_x=\max\left\{\frac{\sqrt 2}{\sqrt{\mathcal{K}(x,x)}}, \frac{N(m)}{\sqrt m\sqrt{|\mathcal{M}||\mathcal{N}|}}\bigg|_{m=1,\dots,\bar m_x-1}\right\}.\]
\modifica{With such a definition of $C_x$, the following inequality holds not only for $m\geq\bar m$, but more in general for any $m\geq 1$:} 
\[N(m)\leq C_x\sqrt{m}\sqrt{|\mathcal{M}||\mathcal{N}|}\qquad\forall\,x\in\mathcal{X},\]
so
\[N(m)\leq C\sqrt{m}\sqrt{|\mathcal{M}||\mathcal{N}|}\qquad \forall\,m\geq 1,\]
where 
\[ C=\inf_{x\in\mathcal{X}} C_x.\]
\end{proof}

\begin{corollary}\label{controllo}
If
\[ \lim_{m\to\infty }\frac{m|\mathcal{M}|^2|\mathcal{N}|^2}{N^3(m)}=0,\]
then
\[ \lim_{m\to\infty}\frac{|\mathcal{M}||\mathcal{N}|}{N(m)}=0.\]
\end{corollary}
\begin{proof} 
First, by Lemma \ref{Nmax},
\[ \frac{m|\mathcal{M}|^2|\mathcal{N}|^2}{N^3(m)}\geq \frac{m\mathcal{|M|}^2|\mathcal{N}|^2}{C^3m^{3/2}|\mathcal{M}|^{3/2}|\mathcal{N}|^{3/2}}=\frac{1}{C^3}\frac{\sqrt{\mathcal{|M|}|\mathcal{N}|}}{\sqrt m}\]
which implies
\[\lim_{m\to\infty}\frac{\mathcal{|M|}|\mathcal{N}|}{m}=0.\]
Besides, using again Lemma \ref{Nmax}
\[1\leq C\frac{\sqrt{|\mathcal{M}||\mathcal{N}|}\sqrt m}{N(m)},\]
therefore
\begin{align}
\nonumber\frac{|\mathcal{M}||\mathcal{N}|}{N(m)}&=\left(\frac{|\mathcal{M}|^2|\mathcal{N}|^2}{N^2(m)}\right)^{1/2}\leq 
\left(C\frac{\sqrt{|\mathcal{M}||\mathcal{N}|}\sqrt m}{N(m)}\frac{|\mathcal{M}|^2|\mathcal{N}|^2}{N^2(m)}\right)^{1/2}\\
&=\sqrt C\left(\frac{|\mathcal{M}||\mathcal{N}|}{m}\right)^{1/4}\left(\frac{m|\mathcal{M}|^2|\mathcal{N}|^2}{N^3(m)}\right)^{1/2},
\end{align}
whence
\begin{align}
\lim_{m\to\infty}\frac{|\mathcal{M}||\mathcal{N}|}{N(m)}
&\leq\sqrt C\lim_{m\to\infty}\left(\frac{|\mathcal{M}||\mathcal{N}|}{m}\right)^{1/4}\lim_{m\to\infty}\left(\frac{m|\mathcal{M}|^2|\mathcal{N}|^2}{N^3(m)}\right)^{1/2}=0.
\end{align}
\end{proof}

We now have all the ingredients to prove Theorem \ref{init}.\\
We fix a finite collection $x_A=(x^{(\alpha)},\alpha\in A)$. Let $\xi\in\mathbb{R}^{|A|}$ and $\|\xi\|_1=\sum_\alpha|\xi_\alpha|$. The characteristic function of the finite collection of random variables $f^{(m)}(\Theta,x_A)$ is
\begin{align}
\nonumber\phi_{m}(\xi)&=\mathbb{E}\left[ \exp\left\{if^{(m)}(\Theta,x_A)\cdot\xi_A\right\}\right]\\
&=\mathbb{E}\left[\exp \left\{i\frac{1}{N(m)}\sum_{k=1}^m f_k(\Theta,x_A)\cdot\xi_A\right\}\right],
\end{align}
where
\begin{align}
f_k(\Theta,x_A)\cdot\xi_A\equiv\sum_{\alpha\in A}f_k(\Theta,x^{(\alpha)})\cdot\xi_\alpha.
\end{align}
We rename
\begin{align}
Y^{(m)}\equiv\frac{1}{N(m)}\sum_{k=1}^m f_k(\Theta,x_A)\cdot\xi_A
\end{align}
and we introduce a fictitious variable $t$ such that $\phi_{m}(\xi)=\phi_{m}(\xi,t=1)$:
\begin{align}
\nonumber \phi_{m}(\xi,t)&=\mathbb{E}\left[ \exp(itf^{(m)}(\Theta,x_A)\cdot\xi_A)\right]\\
&=\mathbb{E}\left[\exp(it Y^{(m)})\right].
\end{align}
As we saw in paragraph \ref{cumulants}, $\log\phi_{m}(\xi,t)$ can be expanded in terms of the cumulants of $Y^{(m)}$:
\begin{align}
\kappa_r^{(m)}\equiv(-i)^n\frac{d^r}{dt^r}\log\phi_m(\xi,t)\qquad\to\qquad
\log\phi_m(\xi,t)=\sum_{r=1}^\infty \frac{\kappa_r^{(m)}}{r!}(it)^r.
\end{align}

As we proved in Lemma \ref{ok}, a valid dependency graph $\mathscr{G}=(V,E)$ of the quantum circuit is
\[V=\{1,2,\dots,m\}\qquad E=\{(k,k')\in V\times V: k'\in\mathcal{P}_{k}\},\]
where
\begin{align}
\nonumber \mathcal{P}_k&=\{k'\in\{1,\dots,m\} : f_{k'}(\Theta,x) \text{ is not independent of }\phantom{\}} \\
&\phantom{=\{k'\in\{1,\dots,m\} :\}}f_{k}(\Theta,x) \text{ at random initialization}\}.
\end{align}
Calling
\begin{align}
D&=\max_{k\in V}\deg k,
\end{align}
by means of Lemma \ref{PM} we can bound
\begin{align}
D&=\max_{k\in V}\big| \{(k,k')\in V\times V\,:\, k'\in\mathcal{P}_{k}\}\big|=\max_{k\in I}|\mathcal{P}_k|\leq |\mathcal{M}||\mathcal{N}|.
\end{align}
The random variable $Y^{(m)}$ is a sum of uniformly bounded random variables by Assumption \ref{domain}:
\begin{align}
\left|\frac{1}{N(m)}f_k(\Theta,x_A)\cdot\xi_A\right|\leq\frac{\|\xi\|_1}{N(m)}
\end{align}
Therefore, Theorem \ref{janson} can be applied:
\begin{align}
|\kappa_r^{(m)}|&\leq 2^{r-1}r^{r-2}m(D+1)^{r-1}\left(\frac{\|\xi\|_1}{N(m)}\right)^r\leq \frac{m}{D}\left(\frac{4\|\xi\|_1D}{N(m)}\right)^rr^r.
\end{align}
Using the following inequality for the factorial \cite{stirling}
\[e\frac{r^r}{e^r}\leq r!\qquad \to \qquad r^r\leq e^rr!,\]
we further estimate
\begin{align}
|\kappa_r^{(m)}|&\leq \frac{m}{D}\left(\frac{4e\|\xi\|_1D}{N(m)}\right)^rr!.
\end{align}
In this way, we can control the tail $T(t,\xi)$ of the Taylor expansion of $\log\phi_m(\xi,t)$:
\begin{align}
T(t,\xi):=\sum_{r=3}^\infty \frac{\kappa_r^{(m)}}{r!}(it)^r,&\\
\nonumber|T(t,\xi)|=\left|\sum_{r=3}^\infty \frac{\kappa_r^{(m)}}{r!}(it)^r\right|&
\leq\sum_{r=3}^\infty \frac{|\kappa_r^{(m)}|}{r!}t^r\\
\nonumber&\leq \sum_{r=3}^\infty\frac{m}{D}\left(\frac{4e\|\xi\|_1tD}{N(m)}\right)^r\\
\nonumber&= \modifica{(4e\|\xi\|_1t)^3}\frac{mD^2}{N^3(m)}\sum_{r=0}^\infty\left(\frac{4e\|\xi\|_1tD}{N(m)}\right)^r\\
&\leq \modifica{(4e\|\xi\|_1t)^3}\frac{m|\mathcal{M}|^2|\mathcal{N}|^2}{N^3(m)}\sum_{r=0}^\infty\left((4e\|\xi\|_1t) \frac{|\mathcal{M}||\mathcal{N}|}{N(m)}\right)^r.
\end{align}
The following bounds hold for fixed $t$ and $\xi$. If, by hypothesis, \[\lim_{m\to \infty}\frac{m|\mathcal{M}|^2|\mathcal{N}|^2}{N^3(m)}=0,\]
then, by Corollary \ref{controllo}, we must also have
\[\lim_{m\to \infty}\frac{|\mathcal{M}||\mathcal{N}|}{N(m)}=0,\]
which implies, for some $m_0$ (which depends on $t$ and $\xi$),
\[ (4e\|\xi\|_1t) \frac{|\mathcal{M}||\mathcal{N}|}{N(m)}\leq \frac{1}{2} \qquad \forall \, m\geq m_0.\]
Therefore
\begin{align}
|T(t,\xi)|\leq \frac{\modifica{(4e\|\xi\|_1t)^3}}{2}\frac{m|\mathcal{M}|^2|\mathcal{N}|^2}{N^3(m)} \qquad \forall \, m\geq m_0
\qquad\to\qquad \lim_{m\to\infty}T(t,\xi)=0.
\end{align}
So we can expand
\begin{align}
\log\phi_m(\xi)&=\log\phi_m(\xi,1)=\kappa_1^{(m)}(\xi)-\frac{1}{2}\kappa_2^{(m)}(\xi)+T(\xi,1).
\end{align}
Recalling that $A$ is finite and using the summation convention over $\alpha$ and $\beta$, we can compute
\begin{align}
\kappa_1^{(m)}&=\mathbb{E}[Y^{(m)}]=0 \text{ because of the } (\ref{mean}),\\
\kappa_2^{(m)}&=\mathbb{E}[(Y^{(m)})^2]-\mathbb{E}[Y^{(m)}]^2=\frac{1}{N^2(m)}\sum_{k,k'=1}^m\xi_\alpha\,\mathbb{E}[f_k(\Theta,x^{(\alpha)})f_{k'}(\Theta,x^{(\beta)})]\,\xi_\beta,\\
\to \kappa_2^{(m)}&\to \xi_\alpha\,\mathcal{K}(x^{(\alpha)},x^{(\beta)})\,\xi_\beta \quad \text{ as }m\to \infty \text{ because of the } (\ref{convergence}).
\end{align}
Since the logarithm is a continuous function, we have pointwise convergence:
\begin{align}
\lim_{m\to\infty}\phi_m(\xi)=\exp\left\{-\frac{1}{2}\xi_{\alpha}\mathcal{K}(x^{(\alpha)},x^{(\beta)})\xi_\beta\right\}.
\end{align}
This is the the characteristic function of a multivariate random variable (whose distribution is jointly Gaussian), so we can conclude by Lévy's theorem \ref{levy}.

\begin{remark}[A counterexample for large light cones]
    The hypotheses under which Theorem \ref{init} holds are very general, since the proof is based only on the dependency relations between the observables. Many studied examples of circuits respecting the most general structure of a variational circuit could be considered, even with arbitrary (and artificial) parameter dependence and distribution. It is interesting to notice that, if we relax the hypothesis on the growth of the light cones, we can exhibit a circuit that violates Theorem \ref{init}. An explicit construction is done in \autoref{ch5-5}.
\end{remark}

 \section{Trained quantum neural networks are Gaussian processes}\label{ch6}
\subsection{Introduction}
\subsubsection{Training quantum neural networks}\label{trqnn}
\begin{figure}[ht]
\centering
\includegraphics[width=0.65\textwidth]{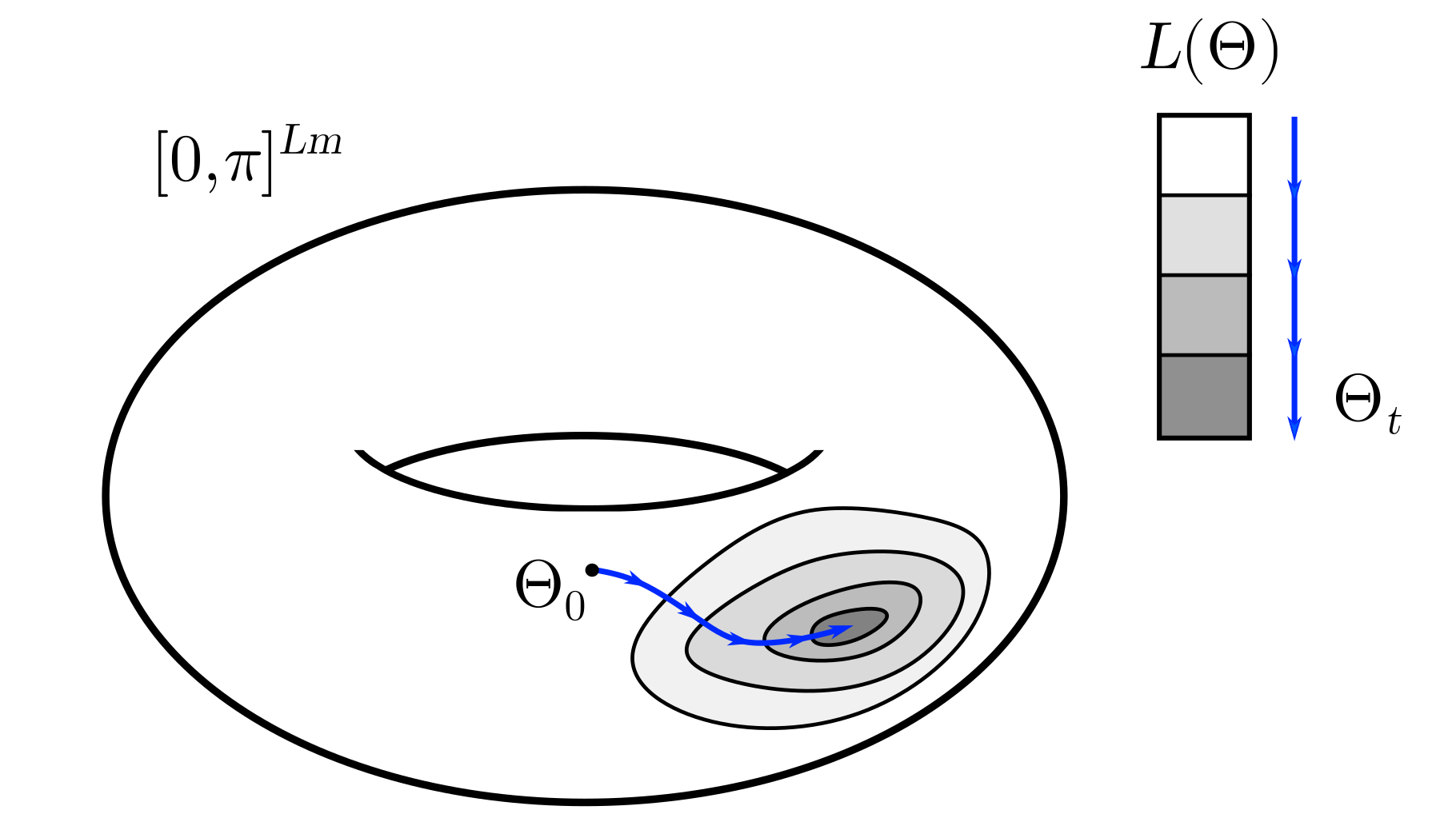}
\caption{Gradient flow in the parameter space}
\label{decreasing}
\end{figure}
As in the classical setting, the aim of the training is the minimization of a cost function $\mathcal{L}(\Theta,\mathcal{D})$ (from now on, we will omit the dependence on $\mathcal{D}$) based on the discrepancy between the model function $f(\Theta,x^{(i)})$ and the training labels $y^{(i)}$. This can be done in continuous time by the gradient flow method. The vector of the parameters evolves according to the differential equation
\[ \dot\Theta_t=-\eta \nabla_\Theta\mathcal{L}(\Theta_t),\]
where $\eta>0$ is a constant which resizes the time. This equation ensures that the cost is decreasing during the evolution (see \autoref{decreasing})
\[\frac{d}{dt}\mathcal{L}(\Theta_t)=\dot \Theta_t\cdot \nabla_\Theta\mathcal{L}(\Theta_t)=-\eta\|\nabla_\Theta\mathcal{L}(\Theta_t)\|_2^2\leq 0.\]
The problem of the minimization of the cost is, in general, non-convex, so the convergence to a global minimum is not ensured. We are going to discuss this issue in \autoref{bp}.

Since the computation of the model function in a concrete setting happens in discrete time, the minimizing procedure will be by described by the gradient descent method, which consists in the update of the parameters according to the rule
\[\Theta_{t+1}-\Theta_t=-\eta \nabla_\Theta\mathcal{L}(\Theta_t).\label{6.3}\]
Besides, we should notice that the right-hand side of (\ref{6.3}) is exactly known only performing an infinite amount of measurements, since the model function is defined as an expectation value (and, at the moment, we do not have a procedure to compute the derivatives of the model function). In a real setting, we will be able to perform only a finite amount of measurements of the output of the circuit. Therefore, we will need to substitute the RHS with an estimator of the gradient.\\
In this section, in order to understand the behaviour of the circuit during the training procedure and to focus on the main properties of the limit $m\to\infty$, we will start considering the gradient flow equation without noise.
Later, in the next section, we will understand how to deal with the discrete time and the noise, and we will prove that the limit $m\to\infty$ allows us to solve a non-convex optimization problem.

\modifica{
\begin{remark}
We stress that the parametric unitary operator $U(\Theta,x)$ encodes both the input $x$ and the trainable parameters $\Theta$.
Therefore, the training does not generate a single state of $m$ qubits, but rather a parametric family of states $\left\{U(\Theta,x)|0^m\rangle\right\}_{x\in\mathcal{X}}$ that depend on the input $x$.
For this reason, training the quantum neural networks considered here is not equivalent to finding a low-energy state of a given Hamiltonian, and the two tasks may have very different complexities.
Therefore, results on the classical easiness of finding low-energy states such as \cite{brandao2016product} may not apply in the scenario considered in this paper.
\end{remark}
}

\subsubsection{Trainability: the loss landscape and barren plateaus}\label{bp}
Before entering into the details of the equations of our model, let us discuss briefly the possible landscapes that may appear in a generic optimization task and the consequences on the gradient-flow training. We will also mention the main problems discussed in the current literature.

\begin{figure}[H]
\centering
\includegraphics[width=0.73\textwidth]{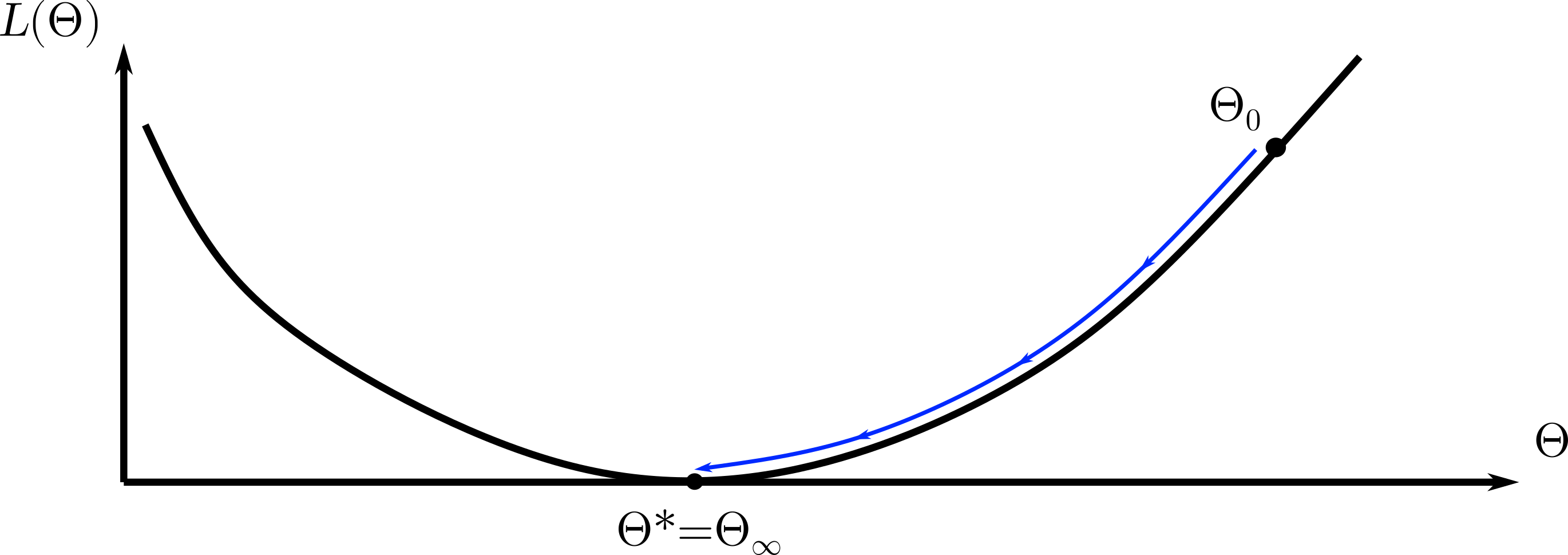}
\caption{Convex $\mathcal{L}(\Theta)$.}
\label{figconv}
\end{figure}

When $\mathcal{L}$ is a convex\footnote{There are various shades of convexity, which can provide different results concerning the rate of convergence of the minimization procedure. For a survey on the main definitions, see \cite{PLBook}.} function, then the convergence to the minimum is guaranteed, whatever the result of the initialization is (see \autoref{figconv}).

\begin{figure}[H]
\centering
\includegraphics[width=0.73\textwidth]{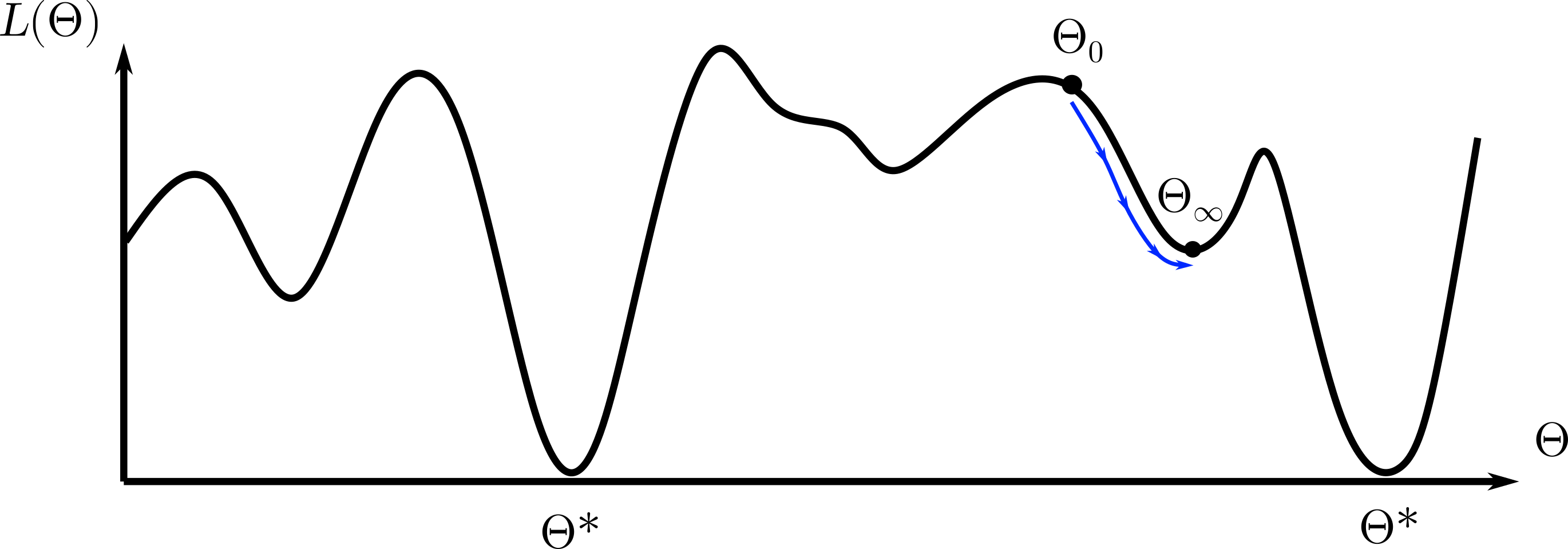}
\caption{Non-convex $\mathcal{L}(\Theta)$.}
\label{nonconv}
\end{figure}

On the contrary, a generic non-convex cost function may be characterized by maxima, saddle points, local minima which are stationary points for the gradient flow evolution. In particular, after a random initialization in which the initial set of parameters is sampled far from a global minimum, the gradient flow evolution could lead to a local minimum, where the naive algorithm stops: the cost function cannot be further minimized by a continuous evolution of the parameters. Therefore, the result is not the one expected (see \autoref{nonconv}).

\begin{figure}[H]
\centering
\includegraphics[width=0.73\textwidth]{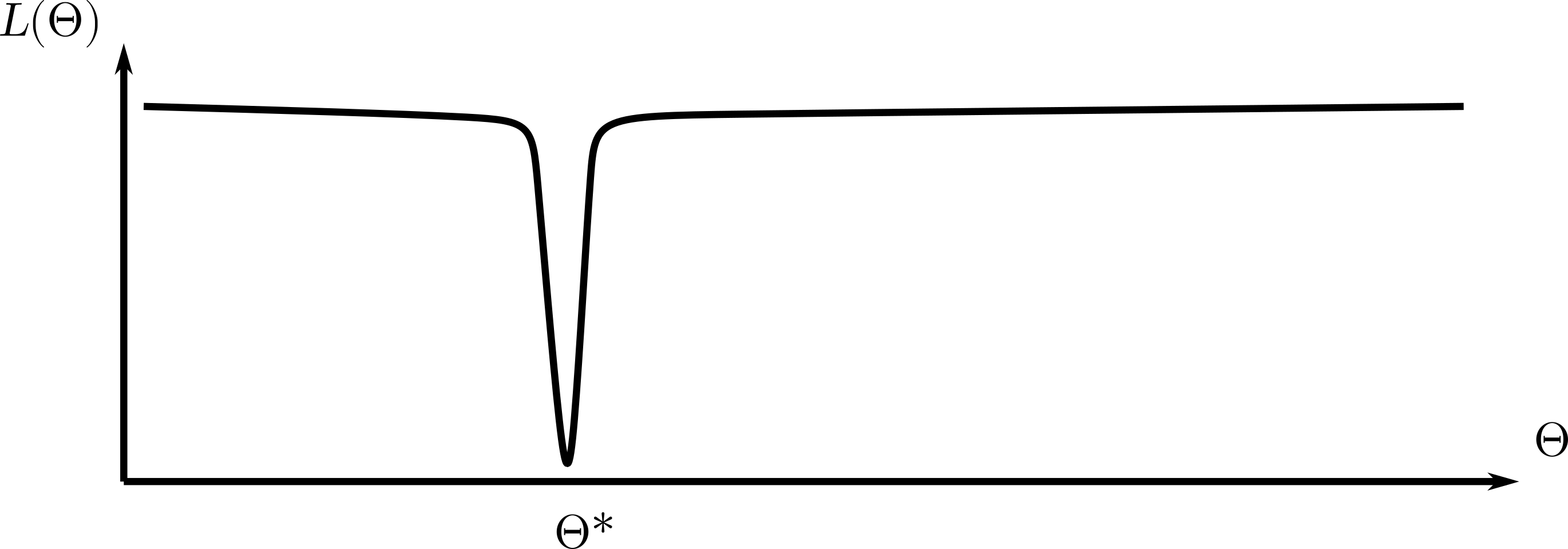}
\caption{Barren plateaus scenario.}
\label{barpl}
\end{figure}

Different problems may arise in the training of a quantum circuit. The main problem is constituted by barren plateaus (see e.g. \cite{McClean_2018, Cerezo_2021, marrero2021entanglement, napp2022quantifying}), depicted in \autoref{barpl}: in some classes of circuits, typically when they are very deep, after a random initialization of the parameters the optimization landascape is characterized by gradients with a norm that decays exponentially with the number of qubits. This means that the convergence of the training is slow and that the estimate of the expectation value of the gradient, in order to be precise enough, requires an exponentially large amount of measurements. This would hinder any potential speedup of the quantum algorithm.

\begin{figure}[H]
\centering
\includegraphics[width=0.73\textwidth]{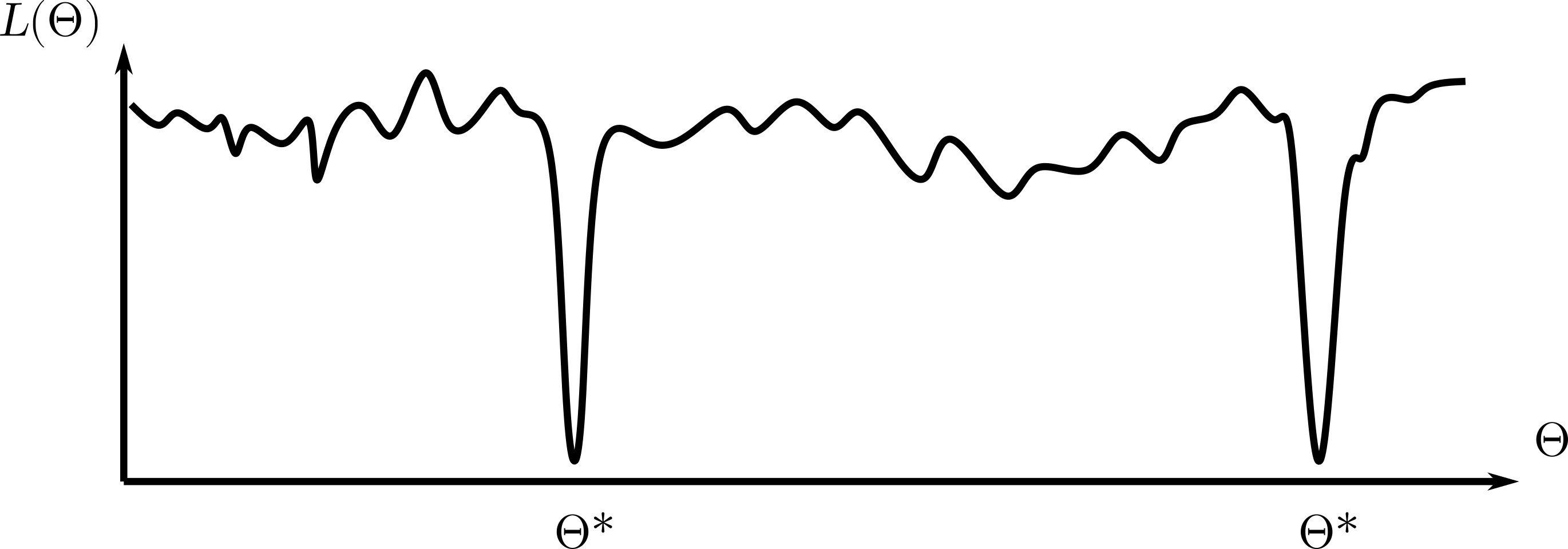}
\caption{Swampland of local minima.}
\label{swamp}
\end{figure}

There are other possible obstacles: in some class of shallow circuits which does not exhibit barren plateaus, the loss landscape can be swamped with many local minima which are far from the global minimum \cite{Anschuetz2022} (see \autoref{swamp}). If the local minima are rare with respect to global minima, the initialization procedure and the consequent training can be iterated until the global minimum is reached. But if the landscape is dense of local minima, we may have a too large number of iterations necessary to find a starting point for which the training converges to optimal parameters.

\begin{figure}[h]
\centering
\includegraphics[width=0.73\textwidth]{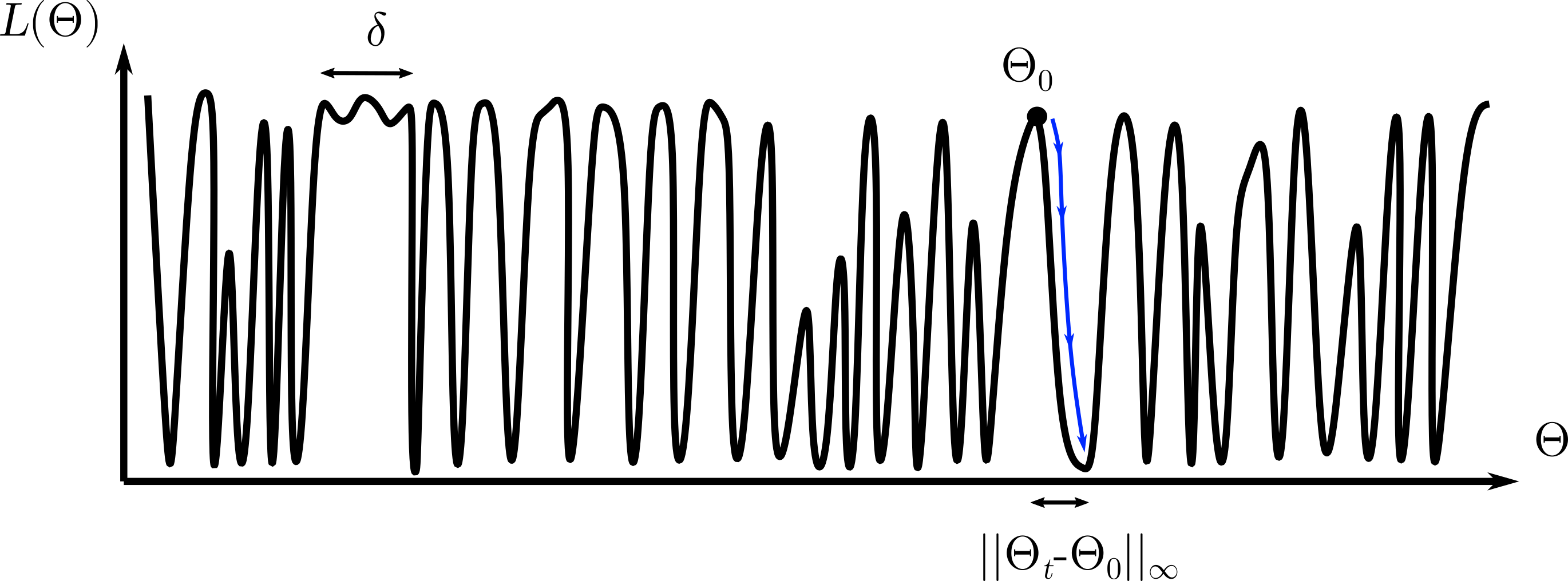}
\caption{Our scenario: lazy training.}
\label{ourscenario}
\end{figure}

We are going to prove that the limit $m\to\infty$ yields a very peculiar landscape (see \autoref{ourscenario}), which has many features in common with the one of wide classical neural networks. First, the minimum of the empirical cost approaches to zero as the number of qubit increases: the model becomes able to fit the dataset. Second, fixed any $\delta>0$, with probability at least $1-\delta$ over random initialization, the parameters evolve succesfully from the starting point $\Theta_0$ to a minimizer $\Theta^\ast$ which is increasingly close to $\Theta_0$ as $m\to\infty$ (lazy training). The closeness of $\Theta_0$ and $\Theta^\ast$ allows to compare the model with its first-order Taylor expansion around $\Theta_0$. This comparison is motivated by the fact that an analytic solution for the linearized model is computable and its probability distribution during training is known.\

In order to achieve these results -- and to generalize them to the noisy setting in \autoref{ch7} --, we follow a strategy which combines some ideas from \cite{Lee2020}, \cite{QLazy} and \cite{you2022convergence}.
\begin{enumerate}
\item In \autoref{6-2}, we translate the gradient flow equation (\ref{defgrfl}) into a differential equation for the model function (\ref{evolution}) and we define the empirical neural tangent kernel (Definition \ref{entk}), which is a bivariate symmetric function of the input space appearing in (\ref{evolution}) and describing the evolution of the randomly initialized network during the training. 
\item In \autoref{6-3}, we define the analytic neural tangent kernel as the expectation value of the empirical neural tangent kernel (Definition \ref{anNTK}) and we study its properties.
\item In \autoref{analyticsol}, we introduce the linearized version -- which is the first order Taylor expansion -- of the model (\ref{linearmodel}) and we compute its evolution equations under gradient flow.
\item In the first part of \autoref{6-5}, we show that the empirical neural tangent kernel with the parameters initialized randomly concentrates on its expectation value (Theorem \ref{ntkconv}), using a strategy which generalizes the proof of \cite{QLazy}. Differently from \cite{QLazy}, our bounds are quantitative and they hold not only in the case with a fixed number of layers and with a geometrically local architecture, but for any dependence of the number of layers with respect to the number of qubits and for any architecture. As a consequence of our concentration inequality, we show that the linear model is distributed as a Gaussian process during all the evolution in the limit of infinitely many qubits (Corollary \ref{corgp}).
\item In the second part of \autoref{6-5}, we generalize the classical proof of \cite{Lee2020} to the quantum setting in order to show that, as the number of qubits increases, the trained model converges to the examples -- i.e., the cost function converges to zero -- and the parameters are increasingly close to their value at initialization (Theorem \ref{gradfl}). 
\item In the last part of \autoref{6-5}, with some computations involving the size of the light cones, we show that the discrepancy between the model and its linearized version is small if the parameters are close to their value at initialization (Theorem \ref{powerful}). As a consequence, we show that the model converges to its linearized version in the limit of infinitely many qubits (Theorem \ref{gronwall}). This result is a strong generalization of Theorem \ref{qlazyth} of \cite{QLazy}: not only we consider the more general setting of a non-fixed number of layers and of generic architectures, but we also extend the convergence of the model to the linearized version from the dataset to all the input space. Furthermore, in \cite{QLazy} the convergence to the linearized model at any time can be achieved only by performing first the limit of infinitely many qubits, while for a circuit with finite size their bound diverges as $t^2$; clearly it is a fictitious divergence, since the training is convergent. Our bounds are valid for any time also in the (physical) finite size case.
\item As a consequence of the previous results, we conclude that the model function converges to a Gaussian process in the limit of infinitely many qubits (Theorem \ref{qnngp}). This is the main result of this section.
\end{enumerate}
Most of these strategies are valid also in the noisy setting with a discrete time evolution, provided that the variance of the estimator of the gradient is sufficiently small. We will prove in \autoref{ch7} that there is a polynomial\footnote{with respect to the number of qubits} upper bound on the variance which ensures an exponentially fast convergence of the training as in the simpler setting of this section.

\subsubsection{Lazy training in the quantum setting}\label{sec:Qlazy}
Before presenting our results, let us briefly review the results of a recent work about lazy training for quantum neural networks. In particular, we will focus on the statements and the bounds that we need to improve in order to achieve our convergence results.

E. Abedi, S. Beigi and L. Taghavi study in \cite{QLazy} a generalization to the quantum setting of the results of \cite{JGH18} for the classical case. In particular, they consider a geometrically local circuit\footnote{More precisely, \cite{QLazy} requires that \textit{the qubits are arranged on nodes of a bounded-degree graph}, that \textit{the two-qubit gates can be applied only on pair of
qubits connected by an edge} and that \textit{the observable is a sum of terms, each of which acts only on a constant number of neighboring qubits}.} having a fixed number of layers. They prove the following results, which hold in the limit of infinitely many qubits:
\begin{enumerate}
\item the neural tangent kernel (see (\ref{ntkqlazy})) concentrates on its expectation value at random initialization (Theorem \ref{qulazyth0});
\item the evolution of the parameters is small during the training, so the linearized model behaves similarly to the original model (Theorem \ref{qlazyth}).
\end{enumerate}

Let us quote the precise statement of the theorems proved by \cite{QLazy}. Differently from the notation we used so far, the parameters are labelled $\theta_1,\dots,\theta_p$ in \cite{QLazy}.
\begin{theorem}[\cite{QLazy}]\label{qulazyth0}
Let $f(\Theta, x)$ be a model function associated to a geometrically local parameterized quantum circuit on $m$ qubits with a fixed number of layers
\[ f(\Theta,x)=\frac{1}{\sqrt m}\sum_{k=1}^mf_k(\Theta,x),\qquad f_k(\Theta,x)=\smatrixel{0^m}{U^\dagger (\Theta,x)\mathcal{O}_k U(\Theta,x)}{0^m},\label{modelqlazy}\]
such that
\[\left\| \frac{\partial}{\partial \theta_j}U(\Theta,x)\right\|\leq c,\qquad \left\| \frac{\partial^2}{\partial\theta_i\partial \theta_j}U(\Theta,x)\right\|\leq c\qquad \forall i,j\in\{1,\dots, |\Theta|\}.\]
Suppose that the observable $\mathcal{O}$ is also geometrically local given by
\[\mathcal{O}=\sum_{k=1}^m\mathcal{O}_k,\]
where $\mathcal{O}_k$ acts on
the $k$-th qubit and possibly on a constant number qubits in its neighborhood, and satisfies $\|\mathcal{O}_k\|_\mathcal{L}\leq 1$. In this case the model function is given by (\ref{modelqlazy}). Suppose that
$\theta_1,\dots,\theta_p$ are chosen independently at random. Then, for any $x,x'\in\mathbb{R}^d$ we have
\[\mathbb{P}\left(|\hat K_\Theta(x,x')-\mathbb{E}[\hat K_\Theta(x,x')|\geq \epsilon\right)\leq \exp\left(-\Omega\left(\frac{m^2\epsilon^2}{pc^4}\right)\right),\label{eq1lt}\]
where
\[\hat K_\Theta(x,x'):=\nabla_\Theta f(\Theta,x)\cdot\nabla_\Theta f(\Theta,x'),\label{ntkqlazy}\]
\end{theorem}

\begin{remark}
    We will introduce a normalization $N_K(m)$, depending on the number of qubits, in the definition of the (empirical) neural tangent kernel $\hat K_\Theta(x,x')$ (see Definition \ref{defntk}).
\end{remark}

\begin{theorem}[\cite{QLazy}]\label{qlazyth}
Let $f(\Theta,x)$ be a model function associated with a parameterized quantum circuit satisfying the assumptions of Theorem \ref{qulazyth0}. Suppose that a dataset $\mathcal{D}$
\[\mathcal{D}=\{(x^{(1)},y^{(1)}),\dots,(x^{(n)},y^{(n)})\}\quad \text{with}\quad x^{(i)}\in\mathbb{R}^d \text{ and } y^{(i)}\in\mathbb{R}\]
is given. Assume that at initialization we choose $\Theta_0 = (\theta_1(0),\dots ,\theta_p(0))$ independently at random, and apply the gradient flow to update the parameters in time by
\[\dot\Theta_t=-\nabla_\Theta\mathcal{L}(\Theta_t),\]
where $\mathcal{L}(\Theta)$ is the mean squared error. Then, the followings hold.
\begin{enumerate}
\item For any $1\leq j\leq p$ we have
\[|\partial_t\theta_j(t)|=O\left(\sqrt{\frac{\mathcal{L}(\Theta_0)}{m}}\right).\label{lt1}\]
\item For any $x,x'$ we have
\[|\partial_t \hat K_{\Theta_t}(x,x')|=O\left(\sqrt{\frac{\mathcal{L}(\Theta_0)}{m}}\right).\label{lt2}\]
\item Let $f^{\mathrm{lin}}(\Theta,x)$ be the function associated to the linearized model, i.e.
\[f^{\mathrm{lin}}(\Theta,x)=f(\Theta_0,x)+\nabla_\Theta f(\Theta_0,x)\cdot (\Theta-\Theta_0).\]
Suppose we train the linearized model with its associated loss function
\begin{align}
\begin{dcases}
\frac{d}{dt}\Theta_t^{\mathrm{lin}}=-\nabla_\Theta\mathcal{L}^{\mathrm{lin}}(\Theta^{\mathrm{lin}}_t)\\
\Theta^{\mathrm{lin}}_0=\Theta_0,
\end{dcases}\qquad
\mathcal{L}^{\mathrm{lin}}(\Theta_t)=\frac{1}{n}\sum_{i=1}^n\left(f^{\mathrm{lin}}(\Theta,x^{(i)})-y^{(i)}\right)^2.
\end{align}
Then, for all $t$ we have
\[\left(\frac{1}{n}\sum_{i=1}^n\left(f(\Theta_t,x^{(i)})-f^{\mathrm{lin}}(\Theta^{\mathrm{lin}}_t,x^{(i)})\right)^2\right)^{1/2}=O\left(\frac{\mathcal{L}(\Theta_0)t^2}{\sqrt m}\right).\label{qt1}\]
and
\[|\mathcal{L}(\Theta_t)-\mathcal{L}^{\mathrm{lin}}(\Theta^{\mathrm{lin}}_t)|=O\left(\frac{\mathcal{L}(\Theta_0)^{3/2}t^2}{\sqrt m}\right).\label{qt2}\]
\end{enumerate}
\end{theorem}

We improve both these results:
\begin{enumerate}
\item We will consider a more general setting in which the number of layers is not fixed and can grow with the number of qubits. This is a necessary condition to achieve a quantum advantage, as we discussed in \autoref{ch4}.
Furthermore, our discussion will not be restricted to the class of geometrically local circuits: the only locality assumption we will state will concern the observable measured at the end of the circuit. Therefore, we will generalize (\ref{eq1lt}) to our larger class of circuits and we will provide a quantitative bound (Theorem \ref{ntkconv}): the constant $c$ appearing in (\ref{quantbound}) can be explicitly computed following the proof we give.
\item As remarked in \cite{QLazy}, the bounds of Theorem \ref{qlazyth} are effective when the cost function at initialization does not depend on the number of qubits. More \modifica{precisely}, it is enough that $\mathcal{L}(\Theta_0)=o(m^{1/3})$ in order not to trivialize (\ref{lt2}), (\ref{lt2}), (\ref{qt1}) and (\ref{qt2}) as $m\to\infty$. \cite{QLazy} claims to expect this property when the model at initialization approaches a Gaussian process, but no proof concerning the probability distribution of the function is provided. 
In \autoref{ch5} we have proved that, in the limit of infinitely many qubits, the function generated by a quantum neural network with randomly initialized parameters converges to a Gaussian process when the parameters on which each qubit depends influence only a small number of other qubits (Theorem \ref{init}). Therefore, the cost function at initialization is bounded, with high probability, by a constant (Corollary \ref{corollaryR}).
\item We will prove that the whole training takes place in the lazy regime: the displacement of each parameter from its initial value is bounded by quantity which is convergent as $t\to\infty$ for $m$ fixed, and which is arbitrarily small as $m\to\infty$ (Theorem \ref{gradfl}), i.e.
\[\sup_t\|\Theta_t-\Theta_0\|_\infty \leq C(m),\qquad \text{with}\qquad \lim_{m\to\infty}C(m)=0.\label{strong}\]
A bound on the derivative of the form (\ref{lt1}) does not imply the stronger property (\ref{strong}).
The same improvement will be done for (\ref{lt2}) in Corollary \ref{freezntk}: we will prove that the neural tangent kernel is frozen on its value at initialization.
\item Furthermore, the bounds (\ref{qt1}) and (\ref{qt2}) become trivial as $t\to\infty$ with $m$ fixed. We will prove a quantitative bound for the distance between the original model and its linear approximation \textit{over the entire input space} with no divergences for $t\to\infty$ and $m$ fixed (Theorem \ref{gronwall}).
\end{enumerate}

\subsection{Gradient flow and neural tangent kernel}\label{6-2}
We call $\Theta_t$ the evolution of $\Theta_0$ under gradient flow. We will sometimes use the compact notation
\begin{align}
F(t)=\begin{pmatrix} f(\Theta_t,x^{(1)})\\f(\Theta_t,x^{(2)})\\ \vdots\\f(\Theta_t,x^{(n)})\\\end{pmatrix}, \quad\text{where}\quad (x^{(i)},\cdot)\in\mathcal{D}\quad n=|\mathcal{D}|
\end{align}
or similar vectorized forms of the model function.
Using the gradient flow equation 
\begin{align}
\dot \Theta_t = -\eta\nabla_\Theta\mathcal{L}(\Theta_t),
\label{defgrfl}
\end{align}
and the chain rule, the equations for the evolution of the parameters and the model function can be rewritten as follows
\begin{align}\label{evolution0}
\begin{dcases}
\hspace{3.3em}\dot \Theta_t&\hspace{-0.5em}=-\eta\nabla_\Theta f(\Theta_t,X^T) \nabla_{f(\Theta_t,X)}\mathcal{L}(\Theta_t)\\
\frac{d}{dt} f(\Theta_t,x)&\hspace{-0.5em}=\big(\nabla_\Theta f(\Theta_t,x)\big)^T \dot \Theta_t=-\eta \big(\nabla_\Theta f(\Theta_t, x)\big)^T \nabla_\Theta f(\Theta_t,X^T) \nabla_{f(\Theta_t,X)}\mathcal{L}(\Theta_t)
\end{dcases}.
\end{align}
\begin{definition}[Empirical NTK]\label{entk} We define the \textit{empirical neural tangent kernel} (NTK) as
\begin{align}
\hat{K}_\Theta(x,x')=\frac{1}{N_K(m)}\big(\nabla_\Theta f(\Theta,x)\big)^T \nabla_\Theta f(\Theta,x'),
\label{defntk}
\end{align}
where $N_K(m)$ is a normalization factor, which will be specified by Assumption \ref{assNTK}.
\end{definition}
The set of equations (\ref{evolution0}) becomes
\begin{align}
\begin{dcases}
\hspace{3.3em}\dot \Theta_t&\hspace{-0.5em}=-\eta\nabla_\Theta f(\Theta_t,X^T) \nabla_{f(\Theta_t,X)}\mathcal{L}(\Theta_t)\\
\frac{d}{dt} f(\Theta_t,x)&\hspace{-0.5em}=-\eta \,N_K(m)\,\hat K_{\Theta_t}(x,X^T) \nabla_{f(\Theta_t,X)}\mathcal{L}(\Theta_t)
\label{evolution}
\end{dcases}.
\end{align}

\subsection{Assumptions and properties of the NTK}\label{6-3}
In the previous subsection we introduced the empirical NTK.  Now, we introduce its expectation value over random initialization of the parameters and we discuss some properties.
\begin{definition}[Analytic NTK] \label{anNTK}
The \textit{analytic NTK} $K$ is defined as the expectation value of the empirical NTK $\hat K_\Theta$ over random initialization of the parameters:
\begin{align}
K(x,x')\equiv \mathbb{E}\left[\hat K_\Theta(x,x')\right].
\end{align}
\end{definition}
\begin{assumption}
\label{assNTK} 
We assume that we can choose a normalization $N_K(m)$ such that there exists a limit function $\bar K:\mathcal{X}\times\mathcal{X}\to\mathbb{R}$
\begin{align}
\label{limKbar}
\lim_{m\to\infty} \sup_{x,x'\in\mathcal{X}}|K(x,x')-\bar K(x,x')|=0
\end{align}
not identically zero. We also assume that the minimum eigenvalue of the finite\footnote{The number of examples in $\mathcal{D}$ is assumed to be finite.} matrix $\bar K \equiv \bar K(X,X^T)$ is strictly positive: $\lambda_{\min}^K>0$.
\end{assumption}

\modifica{
\begin{remark}
We stress that, provided that Assumptions \ref{zeromean} and \ref{assNTK} are satisfied, the architecture of the quantum neural networks can be completely arbitrary.
In particular, we do not need to assume that the qubits are arranged on a regular lattice with $2$-qubit gates acting only between neighboring qubits.
\end{remark}

\begin{remark}
Proceeding as in Remark \ref{rem3.3}, the limit \eqref{limKbar} always exists upon taking a suitable subsequence.
\end{remark}
}

Since $\bar K$ is a finite matrix, its spectrum is bounded. In particular, we will call $\lambda_{\max}^K$ its maximum eigenvalue.

As in the case of $N(m)$, we are left with some degrees of freedom in the choice of $N_K(m)$ under Assumption \ref{assNTK}: a global factor or negligible corrections do not change that requirement. As we will see, a multiplicative factor $\alpha>0$ in
\[K\quad \to \quad \alpha K\]
would rescale
\[\lambda_{\min}^K\quad\to\quad \alpha \lambda_{\min}^K.\]
Since $\lambda_{\min}^K$ will regulate the rate of convergence of the training (see Theorem \ref{gradfl}), in the face of this rescaling one may question the well-definedness of this quantity. As we will see, an arbitrary parameter $\eta_0$ will always multiply $\lambda_{\min}^K$ in the expression of the rate of convergence: therefore, any rescaling constant can be absorbed in $\eta_0$, and this simply represents a reparametrization of the continuous time in the gradient flow equation.

\begin{assumption} 
\label{uniform}
From now on, we will assume that the parameters $\{\theta_i\}_{i=1,\dots,Lm}$ are distributed as independent uniform random variables in $[0,\pi]$, i.e., the initialization of the circuit is uniform in the parameter space $\mathscr{P}=[0,\pi]^{Lm}$. 
\end{assumption}
\rimodifica{
Under Assumption \ref{uniform}, we can more explicitly write
\begin{align}
    K(x,x')=\frac{1}{N_K(m)}\int_{[0,\pi]^{Lm}}\frac{d^{Lm}\theta}{(\pi)^{Lm}}\sum_i \partial_{\theta_i}f(\Theta,x)\partial_{\theta_j}f(\Theta,x).
\end{align}
\begin{remark}
    As we will prove later (see Theorem \ref{qnngp}), the analytic NTK, together with the training set and the covariance at initialization, fully characterizes the probability distribution of the function generated by the quantum neural network in the limit of an infinite number of qubits.
    In the classical setting, there are some standard architectures which have been extensively studied (e.g. fully connected or convolutional neural networks). For such architectures, a closed formula is known for the corresponding analytic NTK. In the quantum setting, there are no standard architectures which have been proved to be more suitable for a quantum advantage yet. Therefore, as mentioned in the previous sections, we will never fix any specific architecture of the quantum circuit in our work, but we will only study the general properties holding for arbitrary overparameterized quantum neural networks. As a consequence, we will not provide any closed formula for the analytic NTK in the quantum setting. However, computing a closed formula for the analytic NTK for any finite size circuit is theoretically always possible by a direct (and inefficient) calculation, as discussed in \autoref{app:NTK}. 
\end{remark}
}

As we saw in the previous sections, $N(m)$ has to grow rapidly enough in order to ensure that a hypothesis of the form (\ref{condizione}) is verified. The main limit to this request is due to barren plateaus. Since, as we will see, (\ref{condizione}) is going to be generalized to 
\[\lim_{m\to\infty}\frac{L^\alpha m^\beta |\mathcal{M}|^\gamma|\mathcal{N}|^\delta}{\big(N_K(m)\big)^\mu N(m)}=0 \label{condizione2}\]
\modifica{for some exponents $\alpha,\beta,\gamma,\delta$ and $\mu$ that will be specified later}, we may wonder if there is an analogous problem with $N_K(m)$. The aim of next lemma is to show that -- differently from $N(m)$ -- $N_K(m)$ does not suffer from the problem of a potential too weak growth\footnote{Because of the form of (\ref{condizione2}), a too weak growth would mean $N_K(m)=o(1)$ as $m\to\infty$.}. Furthermore, we will prove that a condition (\ref{condizione2}) can always be translated into the slightly stronger form (\ref{condizione}).
\begin{lemma}[Bounding the NTK normalization] 
\label{boundnorm}
If Assumption \ref{zeromean} and Assumption \ref{assNTK} are satisfied, then
\[\Omega(1)\leq N_K(m)<O\left(|\mathcal{N}|\right),\]
which means that there exist $c_1,c_2>0$ and $m_0\in\mathbb{N}$ such that, for any $m\geq m_0$
\[c_1 \leq N_K(m) \leq c_2|\mathcal{N}|.\label{boundnormeq}\]
\end{lemma}

\begin{proof}
See \autoref{proofboundnorm}.
\end{proof}

\begin{remark}
Since we are interested in the case $L=O(\text{poly}(m))$, we have that
\[ |\mathcal{N}|\leq |\Theta|=Lm=O(\text{poly}(m))\quad\to\quad N_K(m)=O(\text{poly}(m))\]
\modifica{independently} of the behaviour of $N(m)$. Therefore, a superpolynomial suppression of $N(m)$ in (\ref{condizione2}) cannot be solved by the growth of $N_K(m)$. However, due to the lower bound on $N_K(m)$, we see that the problem of barren plateaus is rooted in the normalization $N(m)$ of the function and does not involve $N_K(m)$.
\end{remark}

\begin{remark}\label{semplif}
As a consequence of this lemma, the estimates we will present in the following can be simplified if we set
\[\frac{1}{N_K(m)}\leq c\]
or, without loss of generality,
\[\frac{1}{N_K(m)}\leq 1 \label{riscalamento}\]
 by a rescaling of the normalization $N_K(m)$.
This of course may weaken the upper bounds or make the hypotheses of the theorems stronger, but it allows us to present the inequalities (\ref{condizione2}) in the cleaner form (\ref{condizione}) which does not involve the computation of the asymptotic behaviour of $N_K(m)$.
\end{remark}

\subsection{Analytic solution for the linearized model}\label{analyticsol}
As we will prove, the first order Taylor expansion around the values of the parameters at initialization is an approximation of the model which becomes increasingly accurate as the number of qubits is large. As in the classical case, this linear approximation allows an exact solution of the evolution equations in the case of the mean squared error. Therefore, when the number of qubits is large, the analytical solution for the linear model becomes a precious result to understand the evolution of the original model, which cannot be computed analytically.
So, let us start introducing the linearized model function:
\begin{align}\label{linearmodel}
f^{\mathrm{lin}}(\Theta_t^{\mathrm{lin}},x)=f(\Theta_0,x)+\nabla_\Theta f(\Theta_0,x)^T(\Theta^{\mathrm{lin}}_t-\Theta_0).
\end{align}
The set of equations (\ref{evolution}) for the linearized model reads
\begin{align}
\begin{dcases}
\hspace{4.3em}\dot \Theta_t^{\mathrm{lin}}&\hspace{-0.5em}=-\eta\nabla_\Theta f(\Theta_0,X^T) \nabla_{f^{\mathrm{lin}}(\Theta_t,X)}\mathcal{L}^{\mathrm{lin}}(\Theta^{\mathrm{lin}}_t)\\
\frac{d}{dt} f^{\mathrm{lin}}(\Theta^{\mathrm{lin}}_t,x)&\hspace{-0.5em}=-\eta \,N_K(m)\,\hat K_{\Theta_0}(x,X^T) \nabla_{f^{\mathrm{lin}}(\Theta^{\mathrm{lin}}_t,X)}\mathcal{L}^{\mathrm{lin}}(\Theta^{\mathrm{lin}}_t)
\label{linevolution}
\end{dcases},
\end{align}
where $\mathcal{L}^{\mathrm{lin}}$ is the loss function computed on $\mathcal{D}$ using the linearized model instead of the original model.
We recall that $n=|\mathcal{D}|$. In the case of the mean squared error, the solution of (\ref{linevolution}) can be analytically computed. We recall that
\begin{align}
\nonumber \mathcal{L}(\Theta)&=\frac{1}{n}\sum_{i=1}^n\frac{1}{2} \left(f(\Theta,x^{(i)})-y^{(i)}\right)^2\\ &=\frac{1}{2n}\|f(X)-Y\|_2^2.
\end{align}
After writing the gradient of the loss function with respect to the function,
\[\nabla_{f(\Theta_t,X)}\mathcal{L}(\Theta_t)=\frac{1}{n}f(\Theta_t,X)-Y,\]
 the equations of both the original and the linearized model become
\[
\begin{dcases}
\hspace{3.3em}\dot \Theta_t&\hspace{-0.5em}=-\frac{\eta}{n} \nabla_\Theta f(\Theta_t,X^T) (f(\Theta_t,X)-Y)\\
\frac{d}{dt} f(\Theta_t,x)&\hspace{-0.5em}=-\eta\frac{N_K(m)}{n}\hat K_{\Theta_t}(x,X^T) \left(f(\Theta_t,X)-Y\right)
\end{dcases},\]
\[
\label{eqlin}
\begin{dcases}
\hspace{4.3em}\dot \Theta^{\mathrm{lin}}_t&\hspace{-0.5em}= -\frac{\eta}{n} \nabla_\Theta f(\Theta_0,X^T) (f^{\mathrm{lin}}(\Theta^{\mathrm{lin}}_t,X)-Y)\\
\frac{d}{dt} f^{\mathrm{lin}}(\Theta^{\mathrm{lin}}_t,x)&\hspace{-0.5em}=-\eta\frac{N_K(m)}{n}\hat K_{\Theta_0}(x,X^T) \left(f^{\mathrm{lin}}(\Theta^{\mathrm{lin}}_t,X)-Y\right)
\end{dcases}.\]
Using the notation introduced above, the equations of evolution for the model and its linearized version evalued at the dataset inputs read
\begin{align}
\dot F(t)&=-\eta\frac{N_K(m)}{n}\hat K_{\Theta_t}(F(t)-Y),\\
\dot F^{\mathrm{lin}}(t)&=-\eta\frac{N_K(m)}{n}\hat K_{\Theta_0}(F^{\mathrm{lin}}(t)-Y),
\label{evolF}
\end{align}
where $\hat K_{\Theta_0}\equiv \hat K_{\Theta_0}(X,X^T)$ and $F^{\mathrm{lin}}(t)=f^{\mathrm{lin}}(\Theta^{\mathrm{lin}}_t,X)$. In order to write the analytical solution for a new input, we assume $\hat K_{\Theta_0}$ to be invertible\footnote{This is ensured asymptotically by Assumption \ref{assNTK} and by the statements (\ref{ntkconv}) and (\ref{lambdamin}).}.
The solution of (\ref{evolF}) gives the evolution of the linearized function evalued at the inputs of the dataset. By definition $F^{\mathrm{lin}}(0)=F(0)$:
\begin{align}\label{convergenzaesempi}
F^{\mathrm{lin}}(t)=e^{-\eta \frac{N_K(m)}{n}t\hat K_{\Theta_0}}(F(0)-Y)+Y.
\end{align}
For a new input $x\in\mathcal{X}$
\begin{align}
\nonumber \frac{d}{dt} f^{\mathrm{lin}}(\Theta^{\mathrm{lin}}_t,x)&=-\eta\frac{N_K(m)}{n}\hat K_{\Theta_0}(x,X^T)\cdot (F^{\mathrm{lin}}(t)-Y) \\
&=-\eta\frac{N_K(m)}{n}\hat K_{\Theta_0}(x,X^T)e^{-\eta\frac{N_K(m)}{n}\hat K_{\Theta_0}t}(F(0)-Y).
\end{align}
Therefore
\begin{align}
f^{\mathrm{lin}}(\Theta^{\mathrm{lin}}_t,x)=f(\Theta_0,x)-\hat K_{\Theta_0}(x,X^T)\hat K^{-1}_{\Theta_0}\left(\id-e^{-\eta \frac{N_K(m)}{n} \hat K_{\Theta_0}t}\right)(F(0)-Y).
\label{solutionevol}
\end{align}
\begin{assumption}\label{assETA}
As $m$ increases, we fix a constant $\eta_0>0$ and we rescale the learning rate \[\eta=\frac{n}{N_K(m)}\eta_0.\] Hence, (\ref{solutionevol}) reads
\begin{align}
f^{\mathrm{lin}}(\Theta^{\mathrm{lin}}_t,x)=f(\Theta_0,x)-\hat K_{\Theta_0}(x,X^T)\hat K^{-1}_{\Theta_0}\left(\id-e^{-\eta_0 \hat K_{\Theta_0}t}\right)(F(0)-Y). \label{solutionevol2}
\end{align}
\end{assumption}
The theorems in the following subsection show that, in the limit $m\to\infty$, $\{f^{\mathrm{lin}}(\Theta^{\mathrm{lin}}_t,x)\}_{x\in\mathcal{X}}$ is a Gaussian process and $f(\Theta_t,x)$ converges to $f^{\mathrm{lin}}(\Theta^{\mathrm{lin}}_t,x)$.

\subsection{Convergence to the linearized model}\label{6-5}

The following theorem is a generalization of \cite[Theorem 1]{QLazy} to quantum circuits which are not geometrically local and whose number of layer may depend on $m$.

\begin{theorem}[NTK concentration at initialization] 
\label{ntkconv}
Calling $m$ the number of qubits and $L(m)$ the number of layers, let us assume that
\begin{align}
\label{weakassumption}
\lim_{m\to\infty}\frac{1}{N_K^2(m)}\,\frac{\Sigma_2|\mathcal{M}|^2|\mathcal{N}|^2}{N^4(m)}=0.
\end{align}
When a parameterized quantum circuit is initialized at random, i.e., the parameters $\theta_i$ are independent random variables, the empirical NTK converges in probability to the analytic NTK as $m\to \infty$. In particular, there exists a constant $c$ such that, for any $x,x'\in\mathcal{X}$, we have
\begin{align}\nonumber
\mathbb{P}\left[|\hat K_\Theta(x,x')-K(x,x')|\geq\epsilon\right]&\leq \exp\left[-cN_K^2(m)\frac{N^4(m)}{\Sigma_2|\mathcal{M}|^2|\mathcal{N}|^2}\,\epsilon^2\right]\\
&\leq \exp\left[-cN_K^2(m)\frac{N^4(m)}{Lm|\mathcal{M}|^4|\mathcal{N}|^2}\,\epsilon^2\right].\label{quantbound}
\end{align}
\end{theorem}
\begin{proof}
See \autoref{proofntk}.
\end{proof}

\begin{remark}
The hypothesis (\ref{weakassumption}) can be simplified with the (slightly stronger) request
\[\lim_{m\to\infty}\frac{Lm|\mathcal{M}|^4|\mathcal{N}|^2}{N^4(m)}=0,\]
which may be easier to be operationally verified, since the computation of $\Sigma_2$ and $N_K(m)$ are not required.
\end{remark}

Denoting $\xrightarrow{p}$ the convergence in probability, and recalling that the convergence in probability implies the convergence in distribution \cite{vaart_1998}, we are going to prove, informally speaking, the following implication:
\[
\left. \begin{aligned}
     f(\Theta_0,\,\cdot\,)&\xrightarrow{d} f^{(\infty)}(\,\cdot\,) \\
     \hat K_{\Theta_0}(x,x') &\xrightarrow{p} K(x,x') \to \bar K(x,x')
\end{aligned}
 \quad \right\} \quad\to \quad 
 \begin{gathered}
     \lim_{m\to \infty} f^{\mathrm{lin}}(\Theta^{\mathrm{lin}}_t,x) \\
     \text{ can be written in terms of } \\
     f^{(\infty)}(\,\cdot\,) \text{ and } \bar K(\,\cdot\,, \,\cdot\,).
\end{gathered}
\]
The advantage of writing $\lim_{m\to \infty} f^{\mathrm{lin}}(\Theta^{\mathrm{lin}}_t,x)$ in terms of $f^{(\infty)}(\,\cdot\,)$ and $\bar K(\,\cdot\,, \,\cdot\,)$ is that the first object is a random variable with a known Gaussian process distribution, while the second object is a deterministic function. So, the distribution of $\lim_{m\to \infty} f^{\mathrm{lin}}(\Theta_t,x)$ can be defined by means of the following corollaries of Lemma \ref{ntkconv}.
\begin{corollary} 
\label{convsol}
If (\ref{weakassumption}) and Assumption \ref{assNTK} hold, the empirical NTK appearing in the solutions (\ref{solutionevol2}) of the equations of the evolution under gradient flow converges in probability to the limit function $\bar K$ of the analytical NTK. Therefore
\begin{align}
f^{\mathrm{lin}}(\Theta^{\mathrm{lin}}_t,x)\xrightarrow{d}f^{(\infty)}(x)-\bar K(x,X^T)\bar K^{-1}\left(\id-e^{-\eta_0\bar K t} \right)(F^{(\infty)}-Y)
\label{sol-limit}
\end{align}
as $m\to\infty$. At the RHS, $f^{(\infty)}(x)$ and $F^{(\infty)}=f^{(\infty)}(X)$ are the random variables obtained in the limit $m\to \infty$ over random initialization according to Theorem \ref{init}.
\end{corollary}
\begin{proof}
See \autoref{proofcorgp}.
\end{proof}

\begin{corollary}[The linearized model is a GP during all the evolution] 
\label{corgp}
As $m\to\infty$, $\{f^{\mathrm{lin}}(\Theta^{\mathrm{lin}}_t,x)\}_{x\in\mathcal{X}}$ converges in distribution to a Gaussian process $\{f^{(\infty)}_t(x)\}_{x\in\mathcal{X}}$ with mean and covariance
\begin{align}
\mu_t(x)&=\, \bar K(x,X^T)\bar K^{-1}\left(\id-e^{-\eta_0\bar K t}\right)Y,\\
\nonumber\mathcal{K}_t(x,x')&=\mathcal{K}_0(x,x')\\
\nonumber&\phantom{=}- \bar K(x,X^T)\bar K^{-1}\left(\id-e^{-\eta_0\bar K t}\right) \mathcal{K}_0(X,x')\\
\nonumber&\phantom{=}-\bar K(x',X^T)\bar K^{-1}\left(\id-e^{-\eta_0\bar K t}\right) \mathcal{K}_0(X,x) \\
&\phantom{=}+\bar K(x,X^T)\bar K^{-1}\left(\id-e^{-\eta_0\bar K t}\right)\mathcal{K}_0(X,X^T)\left(\id-e^{-\eta_0\bar K t}\right)\bar K^{-1} K(X,x').
\end{align}
\end{corollary}
\begin{proof}
See \autoref{proofcorgp}.
\end{proof}

Now we state the theorems which prove that $f(\Theta_t,x)$ and $f^{\mathrm{lin}}(\Theta^{\mathrm{lin}}_t,x)$ are asymptotically close as $m\to \infty$, adapting the strategy presented in \cite{Lee2020} to the quantum setting.\\
First, we show that the model converges to the examples as $t\to\infty$ and that it enters the \textit{lazy regime} as $m\to\infty$, i.e.
\begin{enumerate}
\item both $\sup_{t\geq 0} \|\Theta_t-\Theta_0\|_\infty$ and $\sup_{t\geq 0} \|\Theta^{\mathrm{lin}}_t-\Theta_t\|_\infty$ converge to zero in the limit of infinitely many qubits;
\item as a consequence, the linear model becomes an increasingly good approximation of the original model during the training.
\end{enumerate}
Therefore, the distribution of the trained model with many qubits will be close to the distribution of the linearized model.

\begin{theorem}[Convergence to the examples and lazy training]
\label{gradfl} Let us assume that the hypotheses of Theorem \ref{init} and of Theorem \ref{ntkconv} are satisfied and that Assumption \ref{assNTK} holds.
For any $\delta>0$, there exist $\bar m\in\mathbb{N}$ and some constants $R_0,R_1$ such that, when applying gradient flow with learning rate $\eta=\frac{n}{N_K(m)}\eta_0$, for all $m\geq \bar m$ the following inequalities hold with probability at least $1-\delta$ over random initialization for all $t\geq 0$:
\begin{align}
\label{grad1}\|F(t)-Y\|_2&\leq R_0\,\sqrt{n\log(2n)}\, e^{-\frac{\eta_0\lambda^K_{\min}}{3}t},\\
\label{grad2}\|\Theta_t-\Theta_0\|_\infty&\leq\frac{R_1}{\lambda^K_{\min}}\,n\sqrt{\log(2n)}\frac{|\mathcal{M}|}{N_K(m)N(m)}\left(1-e^{-\frac{1}{3}\eta_0\lambda^K_{\min}t}\right).
\end{align}
\end{theorem}
\begin{proof}
See \autoref{proofgradfl}.
\end{proof}
\begin{remark}
The (\ref{grad1}) is equivalent to
\[\mathcal L (\Theta_t)\leq  \frac{1}{2}R^2_0\log(2n)e^{-\frac{2}{3}\eta_0\lambda^K_{\min}t}.\]
\end{remark}

As in the classical case, the consequence of the lazy training is that the NTK is ``freezed'' on its initialization value.

\begin{corollary}[Freezing of the NTK]\label{freezntk} 
Let us assume that the hypotheses of Theorem \ref{init} and of Theorem \ref{ntkconv} are satisfied and that Assumption \ref{assNTK} holds. For any $\delta>0$, there exist $\bar m\in\mathbb{N}$ and a constant $R_2$ such that, when applying gradient flow with learning rate $\eta=\frac{n}{N_K(m)}\eta_0$, the following inequality holds for all $m\geq \bar m$ with probability at least $1-\delta$ over random initialization:
\begin{align}
\label{grad3}\sup_t\sup_{x,x'\in\mathcal{X}}|\hat K_{\Theta_0}(x,x')-\hat K_{\Theta_t}(x,x')|&\leq \frac{R_2}{\lambda_{\min}^K}\,n\sqrt{\log(2n)}\,\frac{\Sigma_1|\mathcal{M}|^3|\mathcal{N}|}{N_K^2(m)N^3(m)}.
\end{align}
\end{corollary}
\begin{proof}
See \autoref{prooffreezntk}.
\end{proof}

As we will prove in \autoref{lemmalin} (see, in particular, Theorem \ref{powerful}), the lazy behaviour of the \modifica{training} ensures that the second order corrections becomes negligible in the limit of infinitely many qubits provided that the parameter at which the model is evaluated are the same:
\[|f(\Theta,x)-f^{\mathrm{lin}}(\Theta,x)|\leq \frac{Lm |\mathcal{M}|^2|\mathcal{N}|}{N(m)}\|\Theta-\Theta_0\|^2_\infty.\]
This upper bound, combined with the lazy training estimate (\ref{grad2}), will lay the foundations for the bound which takes into account the difference in the trajectories $\Theta_t$ and $\Theta_t^{\mathrm{lin}}$ (see Lemma \ref{secondorder} and Theorem \ref{cfevolution}):
\[ |f(\Theta_t,x)-f^{\mathrm{lin}}(\Theta_t^{\mathrm{lin}},x)|\leq \left(\frac{C}{(\lambda_{\min}^K)^3}+C'\right) n^3\log(2n) \frac{L^2m^2|\mathcal{M}|^6|\mathcal{N}|^2}{N^5(m)}\log N(m),\]
for some $C$ and $C'$, with high probability.

Let us see more in detail the statements that lead to the convergence of the evolution of the original model to the evolution of the linearized model, which is a Gaussian process.  
\begin{theorem}[Uniform convergence to the linearized model]
\label{gronwall}Let us assume that the hypotheses of Theorem \ref{init} and of Theorem \ref{ntkconv} are satisfied and that Assumption \ref{assNTK} holds.
For any $\delta>0$, there exist a constant $R_1>0$ and $\bar m\in\mathbb{N}$ such that the following inequality holds for all $m\geq \bar m$ with probability at least $1-\delta$ over random initialization:
\begin{align}\label{fflin}
\sup_t\sup_{x\in\mathcal{X}}|f(\Theta_t,x)-f^{\mathrm{lin}}(\Theta_t,x)|\leq \left(\frac{R_1}{\lambda^K_{\min}}\right)^2n^2\log(2n)\frac{Lm|\mathcal{M}|^4|\mathcal{N}|}{N_K^2(m)N^3(m)}.
\end{align}
\end{theorem}
\begin{proof}
See \autoref{prooffreezntk}.
\end{proof}

\begin{remark}
Since the LHS of (\ref{fflin}) contains a supremum over $x\in\mathcal{X}$, when
\[\lim_{m\to\infty}\frac{Lm|\mathcal{M}|^4|\mathcal{N}|}{N_K^2(m)N^3(m)}=0,\]
the model converges to its linear approximation \textit{uniformly} in $x$.
\end{remark}

Now we have to take into account the fact that the equation of evolution for the parameters $\Theta_t$ and $\Theta_t^{\mathrm{lin}}$ involve different cost functions -- the original and the linearized ones -- so they are in general different after $t=0$. We therefore need a control on their discrepancy.

\begin{lemma}[The parameters differ at the second order]\label{secondorder} Let us assume that the hypotheses of Theorem \ref{init} and of Theorem \ref{ntkconv} are satisfied and that Assumption \ref{assNTK} holds. Then, for any $\delta>0$ there exists $C_1, C_2$ and an integer $m_0\in \mathbb{N}$ such that, for any $m\geq m_0$,
\[\big\|\Theta_t-\Theta_t^{\mathrm{lin}}\big\|_\infty\leq \left(\frac{C_1}{(\lambda_{\min}^K)^3}+C_2\right) n^3\log(2n) \frac{Lm|\mathcal{M}|^5|\mathcal{N}|^2}{N^4(m)}\log N(m)\]
with probability at least $1-\delta$.
\end{lemma}
\begin{proof}
    See \autoref{proofsecondorder}.
\end{proof}
Thanks to the previous lemma, we can now control the discrepancy between the two models taking into account the different trajectories in the parameter space.
\begin{theorem}[The original evolution is close to the linear evolution]\label{cfevolution} Let us assume that the hypotheses of Theorem \ref{init} and of Theorem \ref{ntkconv} are satisfied and that Assumption \ref{assNTK} holds. Then, for any $\delta>0$ there exists $C_1, C_2$ and an integer $m_0\in \mathbb{N}$ such that, for any $m\geq m_0$
\[ |f(\Theta_t,x)-f^{\mathrm{lin}}(\Theta_t^{\mathrm{lin}},x)|\leq \left(\frac{C_3}{(\lambda_{\min}^K)^3}+C_4\right) n^3\log(2n) \frac{L^2m^2|\mathcal{M}|^6|\mathcal{N}|^2}{N^5(m)}\log N(m)\]
with probability at least $1-\delta$.
\end{theorem}
\begin{proof}
    See \autoref{proofcfevolution}.
\end{proof}
The following theorem is the main result of this section.
\begin{mdframed}
\begin{theorem}[Quantum neural networks as Gaussian processes] \label{qnngp}
Let us assume that a circuit satisfying Assumption \ref{assNTK} and such that
\[\lim_{m\to\infty}\frac{L^2m^2|\mathcal{M}|^6|\mathcal{N}|^3}{N^5(m)}\log N(m)=0\]
is randomly initialized according to Assumption \ref{zeromean} and is trained using the gradient flow evolution (\ref{defgrfl}). Then, for any $t\ge0$, in the limit of infinitely many qubits $m\to\infty$, $\{f(\Theta_t,x)\}_{x\in\mathcal{X}}$ converges in distribution to a Gaussian process $\{f^{(\infty)}_t(x)\}_{x\in\mathcal{X}}$ with mean and covariance
\begin{align}
\mu_t(x)&=\bar K(x,X^T)\bar K^{-1}\left(\id-e^{-\eta_0\bar K t}\right)Y,\\
\nonumber\mathcal{K}_t(x,x')&=\mathcal{K}_0(x,x')\\
\nonumber&\phantom{=}- \bar K(x,X^T)\bar K^{-1}\left(\id-e^{-\eta_0\bar K t}\right) \mathcal{K}_0(X,x')\\
\nonumber&\phantom{=}-\bar K(x',X^T)\bar K^{-1}\left(\id-e^{-\eta_0\bar K t}\right) \mathcal{K}_0(X,x) \\
&\phantom{=}+\bar K(x,X^T)\bar K^{-1}\left(\id-e^{-\eta_0\bar K t}\right)\mathcal{K}_0(X,X^T)\left(\id-e^{-\eta_0\bar K t}\right)\bar K^{-1} \bar K(X,x').
\end{align}
\end{theorem}
\end{mdframed}
\begin{proof}
See \autoref{proofqnngp}.
\end{proof}

\subsection{Some useful lemmas}
In this subsection we motivate and prove some lemmas which will be frequently used in many of the proof concerning the convergence of the training.

\subsubsection{Loss function at initialization}

The first statement is a corollary of Theorem \ref{init} and consists in a bound on the mean squared loss at initialization. The role of this theorem is to fix -- with high probability -- the highest value of the mean squared error that the randomly initialized model could produce in the limit of infinitely many qubits. This result is not immediate, since, even though the \textit{local} observables are bounded and, if $m$ is fixed, also the model is bounded, the model function could in principle diverge for some inputs in the limit of infinitely many qubits. Therefore, the loss at initialization could be arbitrarily large. Furthermore, this result lays the groundwork for a proof of an exponential convergence of the form
\[\mathcal{L}(\Theta_t)\leq Ce^{-t/\tau}\qquad \forall\, m\geq m_0 \qquad \text{for some}\qquad C,\tau,m_0>0. \label{expconv}\]
Indeed, (\ref{expconv}) implies an upper bound to the loss at initialization.
\begin{corollary} 
\label{corollaryR}
Let us assume that the hypotheses of Theorem \ref{init} are satisfied. Then, for any $\delta>0$, there are a constant $R>0$ and integer $m_0\in \mathbb{N}$ such that
\begin{align}
\mathbb{P}(\|F(0)-Y\|_2\leq R)\geq 1-\delta\qquad \forall\, m\geq m_0
\end{align}
over random initialization, where 
\[R=R_0(\delta,\max_i |y^{(i)}|,\mathcal{K})\sqrt{n\log(2n)}.\]
\end{corollary}
\begin{remark}
If $\mathcal L$ is the mean squared error, the claim is equivalent to
\begin{align}
\forall\, m\geq m_0\qquad \mathbb{P}\left(\mathcal{L}(\Theta_0)\leq \frac{1}{2}R_0^2\right)\geq 1-\delta.
\end{align}
\end{remark}

\begin{proof}
We recall that $\mathcal{K}(x,x)>0$ for any $x\in\mathcal{X}$ by Assumption \ref{zeromean}. Let
\begin{align}
\sigma^2_i=\mathcal{K}(x^{(i)},x^{(i)}), \quad  \sigma=\max_{1 \leq i\leq n}\sigma_i\quad \text{and}\quad y=\max_{1 \leq i\leq n} |y^{(i)}|.
\end{align}
By Theorem \ref{init}, the random variables $\{f(\Theta_0,x^{(i)})\}_{i=1,\dots,n}$ converge in distribution to $\{f^{(\infty)}(x^{(i)})\}_{i=1,\dots,n}$, which are Gaussian random variables with mean zero and variance $\sigma_i^2$.
Given a constant $R$ to be fixed later, for any $i=1,\dots,n$ we can compute
\begin{align}\nonumber
\mathbb{P}\bigg[|f^{(\infty)}&(x^{(i)})-y^{(i)}|\geq \frac{R}{\sqrt n} \bigg]\\
&=\mathbb{P}\left[f^{(\infty)}(x^{(i)})\geq \frac{R}{\sqrt n}+y^{(i)} \right]+
\mathbb{P}\left[-f^{(\infty)}(x^{(i)})\geq \frac{R}{\sqrt n}-y^{(i)} \right], 
\end{align}
which can be estimated using Chernoff bound \cite{chernoff}:
\begin{align}\nonumber
\mathbb{P}\bigg[|f^{(\infty)}&(x^{(i)})-y^{(i)}|\geq \frac{R}{\sqrt n} \bigg]\\
\nonumber &\leq \exp\left[-\frac{1}{2\sigma_i^2}\left(\frac{R}{\sqrt n}+y^{(i)}\right)^2 \right]+
\exp\left[-\frac{1}{2\sigma_i^2}\left(\frac{R}{\sqrt n}-y^{(i)}\right)^2 \right]\\
&\leq 2\exp\left[-\frac{1}{2\sigma_i^2}\left(\frac{R}{\sqrt n}-\big|y^{(i)}\big|\right)^2 \right].\label{cher}
\end{align}

Let
\[R_0(\delta,y,\mathcal{K}) =\left(\sqrt{1+\frac{1}{\log 2}\log\frac{2}{\delta}}+\frac{y}{\sqrt 2\sigma \sqrt{\log 2}}\right)\sqrt 2 \sigma\]
and
\[R=R_0\sqrt{ n\log (2n)}.\]
Then
\begin{align}
R&= \sqrt n\left[\sqrt 2 \sigma  \sqrt{\log (2n)+\frac{\log (2n)}{\log 2}\log\frac{2}{\delta}}+y\sqrt{\frac{\log (2n)}{\log 2}}\right]\\
&\geq \sqrt n\left[\sqrt 2 \sigma  \sqrt{\log (2n)+\log\frac{2}{\delta}}+y\right]\\
&\geq \sqrt n\left[\sqrt 2 \sigma_i  \sqrt{\log \frac{4n}{\delta}}+|y^{(i)}|\right]\qquad \text{ for any } i=1,\dots, n.
\end{align}
Since 
\[\frac{R}{\sqrt n}\geq \sqrt 2 \sigma_i  \sqrt{\log \frac{4n}{\delta}}+|y^{(i)}|\geq |y^{(i)}|,\]
we can bound (\ref{cher}) as follows
\begin{align}\nonumber
\mathbb{P}\bigg[|f^{(\infty)}(x^{(i)})-y^{(i)}|\geq \frac{R}{\sqrt n} \bigg] &\leq 2\exp\left[-\frac{1}{2\sigma_i^2}\left(\frac{R}{\sqrt n}-\big|y^{(i)}\big|\right)^2 \right]\\
\nonumber &\leq 2\exp\left[-\frac{1}{2\sigma_i^2}  \left(\sqrt 2 \sigma_i\sqrt{\log \frac{4n}{\delta}}\right)^2 \right]\\
 &\leq 2\exp\left[-\log \frac{4n}{\delta}\right]=\frac{\delta}{2n}.
\end{align}
Since the random variables $\{f(\Theta_0,x^{(i)})\}_{i=1,\dots,n}$ converge in distribution to $\{f^{(\infty)}(x^{(i)})\}_{i=1,\dots,n}$, there exists $m_0\in\mathbb{N}$ such that
\[ \mathbb{P}\bigg[|f(\Theta_0,x^{(i)})-y^{(i)}|\geq \frac{R}{\sqrt n} \bigg]\leq 2 \mathbb{P}\bigg[|f^{(\infty)}(x^{(i)})-y^{(i)}|\geq \frac{R}{\sqrt n} \bigg]\qquad \forall\, m\geq m_0\]
because of Lemma \ref{lemmaconvdistr}.
Therefore,
\[  \mathbb{P}\bigg[|f(\Theta_0,x^{(i)})-y^{(i)}|\geq \frac{R}{\sqrt n} \bigg]\leq\frac{\delta}{n}.\]
The union bound ensures that
\[  \mathbb{P}\bigg[|f(\Theta_0,x^{(i)})-y^{(i)}|\leq \frac{R}{\sqrt n}\quad \forall i\in\{1,\dots,n\} \bigg]\geq 1-\delta.\]
So, with probability at least $1-\delta$,
\[\|F(0)-Y\|_2^2 = \sum_{i=1}^n \left(f(\Theta_0,x^{(i)})-y^{(i)} \right)^2\leq \sum_{i=1}^n\frac{R^2}{n}=R^2.\]
\end{proof}

\subsubsection{Lipschitzness of the model}\label{sectionlipmodel}
The aim of this subsection is to compute the Lipschitz constant of the model, of its gradient and of the NTK. In particular, we will prove that
\begin{align}
\label{daqui}|f(\Theta,x)-f(\Theta',x)|&\leq 2\frac{\Sigma_1}{N(m)}\|\Theta-\Theta'\|_\infty, \\
\|\nabla_\Theta f(\Theta,x)-\nabla_{\Theta} f(\Theta',x)\|_\infty&\leq 4\,\frac{|\mathcal{M}|^2|\mathcal{N}|}{N(m)}\|\Theta-\Theta'\|_\infty,\\
 |\hat K_{\Theta}(x,x')-\hat K_{\Theta_0}(x,x')|
\label{aqui}&\leq 16\frac{\Sigma_1|\mathcal{M}|^2|\mathcal{N}|}{N_K(m)N^2(m)} \|\Theta-\Theta_0\|_\infty.
\end{align}
(\ref{daqui})-(\ref{aqui}) imply
\begin{align}
|f(\Theta,x)-f(\Theta',x)|&\leq 2\frac{Lm|\mathcal{M}|}{N(m)}\|\Theta-\Theta'\|_\infty,\\
\|\nabla_\Theta f(\Theta,x)-\nabla_{\Theta} f(\Theta',x)\|_\infty&\leq 4\,\frac{|\mathcal{M}|^2|\mathcal{N}|}{N(m)}\|\Theta-\Theta'\|_\infty,\label{eq4.88}\\
 |\hat K_{\Theta}(x,x')-\hat K_{\Theta_0}(x,x')|\label{lipderiv}
&\leq 16\frac{Lm|\mathcal{M}|^3|\mathcal{N}|}{N^2(m)} \|\Theta-\Theta_0\|_\infty.
\end{align}

Even though it is not suitable to discuss quantum advantages and barren plateaus, the example of \cite{QLazy}, which we already introduced in \autoref{combinazioni}, turns out to be particularly useful to describe the behavior of the lazy training in in the new light of the previous inequalities. Indeed, recalling that, in this example,

\[L \text{ is fixed}\qquad\text{and}\qquad N(m)=\sqrt m,\]
we have
\begin{align}
|f(\Theta,x)-f(\Theta',x)|&\leq c_1(L)\sqrt m \,\|\Theta-\Theta'\|_\infty, \\
\|\nabla_\Theta f(\Theta,x)-\nabla_{\Theta} f(\Theta',x)\|_\infty&\leq \frac{c_2(L)}{\sqrt m}\|\Theta-\Theta'\|_\infty,\\
 |\hat K_{\Theta}(x,x')-\hat K_{\Theta_0}(x,x')|
&\leq c_3(L)\|\Theta-\Theta_0\|_\infty.
\end{align}

The fact that the Lipschitz constant of the model diverges with $m$ is a necessary condition for the lazy training; otherwise, any evolution yielding a small change in the parameters
\[\|\Theta_t-\Theta_0\|_\infty=o(1)\qquad \text{as}\qquad m\to\infty\]
would asymptotically be uneffective to modify the output from the initial value to the expected label according to the dataset.
Also the decreasing Lipschitz constant of the gradient is meaningful: a lazy change in the parameters induces a small change in the gradient, which therefore remains closer to the one of the linearized model.
Finally, from the Lipschitz constant of the kernel we understand that it is the lazy regime to ensure the freezing of the kernel to its initial value.

The starting point to prove the \modifica{Lipschitzness} lemmas is the following statement.
\begin{lemma}[Uniform bounds on the derivative of the observables]
\label{unifbounds}
The function $\partial_{\theta_i}\partial_{\theta_j}f_k(\Theta,x)$ is continuous with respect to $\Theta$ and the following bounds hold for any $i,j\in \{1,\dots,Lm\}$:
\begin{align}
|\partial_{\theta_i}f_k(\Theta,x)|&\leq 2,\label{eq4.95}\\
|\partial_{\theta_i}\partial_{\theta_j}f_k(\Theta,x)|&\leq 4,\\
|\partial_{\theta_i}\partial_{\theta_j}f(\Theta,x)|&\leq  4\,\frac{|\mathcal{M}_i\cap\mathcal{M}_j|}{N(m)}.
\end{align}
\end{lemma}

\begin{proof} 
\modifica{
Let $\|\cdot\|$ be the norm of the Hilbert space of the physical states, and let $\|\cdot\|_\mathcal{L}$ be the operator norm
\[\|O\|_\mathcal{L}=\sup_{\|v\| \leq 1}\|Ov\|.\]  By Assumption \ref{domain} $\|\mathcal{O}_k\|_{\mathcal{L}}\leq 1$  for all $k\in\{1,\dots,m\}$, so
\begin{align}
|f_k(\Theta,x)|&=\left|\langle0|U^\dagger(\Theta,x)\mathcal{O}_kU(\Theta,x)|0\rangle\right|\leq 1.
\end{align}
Furthermore, by Assumption \ref{domain}, the generators of the unitaries encoding the trainable parameters have spectrum $\{-1,+1\}$. In quantum circuits like these, the derivatives of the expectation value of any observable with respect to any parameter can be written in terms of other expectation values computed for ``shifted'' parameters. These identities are indeed known as parameter shift rules (see (14) in \cite{Schuld_2019}):
\begin{align}
    \partial_{\theta_i}f_k(\Theta,x)=f_k(\Theta+\Delta^{(i)},x)-f_k(\Theta-\Delta^{(i)},x)\qquad \Delta^{(i)}_j\coloneqq\begin{cases}
    \frac{\pi}{4} & j=i \\ 0 & j\neq i.
    \end{cases}
\end{align}
Therefore, we immediately see that
\begin{align}
    |\partial_{\theta_i}f_k(\Theta,x)|\leq |f_k(\Theta+\Delta^{(i)},x)|+|f_k(\Theta-\Delta^{(i)},x)|\leq 2
\end{align}
Similarly, using a parameter shift rule for $f_k(\Theta\pm\Delta^{(i)},x)$ we get
\begin{align}
    \nonumber
    \partial_{\theta_j}\partial_{\theta_i}f_k(\Theta,x)&=\partial_{\theta_j}f_k(\Theta+\Delta^{(i)},x)-\partial_{\theta_j}f_k(\Theta-\Delta^{(i)},x)\\
    \nonumber
    &=f_k(\Theta+2\Delta^{(i)},x)-f_k(\Theta,x)-f_k(\Theta,x)+f_k(\Theta-2\Delta^{(i)},x),\\
    &=f_k(\Theta+2\Delta^{(i)},x)-2f_k(\Theta,x)+f_k(\Theta-2\Delta^{(i)},x),
\end{align}
whence
\begin{align}
    |\partial_{\theta_j}\partial_{\theta_i}f_k(\Theta,x)|\leq |f_k(\Theta+2\Delta^{(i)},x)|+2|f_k(\Theta,x)|+|f_k(\Theta-2\Delta^{(i)},x)|\leq 4.
\end{align}
}
\end{proof}

We need this further result to identify the Lipschitz constants we are looking for. 
\begin{lemma}[Bounded differential on a convex implies Lipschitz]
\label{compactconvex}
Let $f:\mathcal{U}\to \mathbb{R}^n$ be continuously differentiable, where $\mathcal{U}\subseteq \mathbb{R}^m$ is convex. If the differential of $f$ is bounded, i.e.
\[K=\max_{x\in\mathcal{U}}\|df(x)\|_\mathcal{L}< \infty,\]
then
\begin{align}
\|f(y)-f(x)\|\leq K\|y-x\| \qquad\forall\,x,y\in\mathcal{U}.
\end{align}
\end{lemma}
\begin{proof}
See e.g. \cite{giaquinta2010mathematical}.
\end{proof}
We finally need the following technical estimate.

\begin{lemma}\label{maxMM} The following upper bound holds:
\[\max_{1\leq j\leq |\Theta|}\sum_{i=1}^{|\Theta|} |\mathcal{M}_i\cap \mathcal{M}_j|\leq |\mathcal{M}|^2|\mathcal{N}|.\]
\end{lemma}

\begin{proof}
Let us define
\[S_{j}=\{i\in\{1,\dots,|\Theta|\}: \mathcal{M}_i\cap \mathcal{M}_j\neq\emptyset\},\]
so that we can write
\begin{align}
\sum_{i=1}^{|\Theta|} |\mathcal{M}_i\cap \mathcal{M}_j|&=\sum_{i\in S_{j}} |\mathcal{M}_i\cap \mathcal{M}_j|\leq \sum_{i\in S_{j}} |\mathcal{M}_j|,\\
\max_{1\leq j\leq |\Theta|}\sum_{i=1}^{|\Theta|} |\mathcal{M}_i\cap \mathcal{M}_j|&\leq
\left(\max_{j}|S_{j}|\right)\left(\max_{j}|\mathcal{M}_j|\right).
\end{align}
An estimate of the cardinality of $S_j$ can be done by means of the following rewriting
\begin{align}
S_{j}&=\{i\in\{1,\dots,|\Theta|\}: \mathcal{M}_i\cap \mathcal{M}_j\neq\emptyset\}\\
\nonumber&=\{i\in\{1,\dots,|\Theta|\}: \exists k \in{1,\dots, m}: i\in \mathcal{N}_k \land j\in \mathcal{N}_k\}\\
\nonumber&=\{i\in\{1,\dots,|\Theta|\}: \exists k \in\mathcal{M}_j: i\in \mathcal{N}_k \land j\in \mathcal{N}_k\}\\
&=\{i\in\{1,\dots,|\Theta|\}: \exists k \in\mathcal{M}_j: i\in \mathcal{N}_k\},\\[8pt]
|S_j|&\leq |\mathcal{M}_j|\max_k|\mathcal{N}_k|=|\mathcal{M}_j\|\mathcal{N}|.
\end{align}
Therefore,
\begin{align}
\max_{1\leq j\leq |\Theta|}\sum_{i=1}^{|\Theta|} |\mathcal{M}_i\cap \mathcal{M}_j|&\leq
\left(\max_{j}|\mathcal{M}_j|\right)^2|\mathcal{N}|\leq |\mathcal{M}|^2|\mathcal{N}|.
\end{align}
\end{proof}

\begin{lemma}[Lipschitzness of the gradient]
The following inequalities hold:
\label{lemma} 
\begin{align}
\label{lemma1}|\partial_{\theta_i}f(\Theta,x)|&\leq 2\frac{|\mathcal{M}_i|}{N(m)},\\
\label{lemma2}\|\partial_{\theta_i}f(\Theta,X)\|_2&\leq 2\sqrt n \,\frac{|\mathcal{M}|}{N(m)},\\
\label{lemma3}\|\nabla_{\Theta}f(\Theta,x)\|_1&\leq \frac{2\Sigma_1}{N(m)} \leq 2L\frac{m}{N(m)}|\mathcal{M}|, \\
\label{lemma4}\|\nabla_\Theta f(\Theta,x)-\nabla_{\Theta} f(\Theta',x)\|_\infty&\leq 4\,\frac{|\mathcal{M}|^2|\mathcal{N}|}{N(m)}\|\Theta-\Theta'\|_\infty.
\end{align}
\end{lemma}

\begin{proof} (Lemma \ref{lemma}) 
We will use the bounds of Lemma \ref{unifbounds} and the notation introduced in the proof of Theorem \ref{ntkconv}:
\begin{align}
\partial_{\theta_i}f(\Theta,x)&=\frac{1}{N(m)}\sum_{k=1}^m\partial_{\theta_i}f_k(\Theta_{\mathcal{N}_k},x)=\frac{1}{N(m)}\sum_{k\in\mathcal{M}_i}\partial_{\theta_i}f_k(\Theta,x),\\
|\partial_{\theta_i}f(\Theta,x)|&\leq\frac{1}{N(m)}\sum_{k\in\mathcal{M}_i}|\partial_{\theta_i}f_k(\Theta,x)|\leq 2\frac{|\mathcal{M}_i|}{N(m)},
\end{align}
which is (\ref{lemma1}). Moreover,
\begin{align}
\|\partial_{\theta_i}f(\Theta,X)\|_2&=\sqrt{\sum_{j=1}^{n}\left(\partial_{\theta_i}f(\Theta,x^{(j)}))\right)^2}
\leq 2\sqrt{n} \,\frac{|\mathcal{M}_i|}{N(m)}\leq 2\sqrt{n} \,\frac{|\mathcal{M}|}{N(m)},
\end{align}
which proves (\ref{lemma2}). We can also estimate
\begin{align}
\nonumber \|\nabla_\Theta f(\Theta,x)\|_1&=\sum_{i=1}^{Lm}|\partial_{\theta_i}f(\Theta,x)| \leq \frac{2}{N(m)}\sum_{i=1}^{Lm}|\mathcal{M}_i|\\
& = \frac{2\Sigma_1}{N(m)} \leq 2L\frac{m}{N(m)}|\mathcal{M}|,
\end{align}
so that (\ref{lemma3}) is verified.\\
In a similar way
\begin{align}
|\partial_{\theta_i}\partial_{\theta_j}f(\Theta,x)|&\leq\frac{1}{N(m)}\sum_{k\in\mathcal{M}_i\cap \mathcal{M}_j}|\partial_{\theta_i}\partial_{\theta_j}f_k(\Theta,x)|
\leq 4\, \frac{|\mathcal{M}_i\cap \mathcal{M}_j|}{N(m)}.
\end{align}
If we fix any $x\in\mathcal{X}$ and we call
\begin{align}
h_x(\Theta)= \nabla_\Theta f(\Theta,x),
\end{align}
Lemma \ref{unifbounds} ensures that $h_x(\Theta)$ is a continuously differentiable function.
Given any $v\in \mathbb{R}^{|\Theta|}$, the differential of $h_x(\Theta)$ (with respect to $\Theta$) acts as
\begin{align}
\modifica{
dh_x(\Theta)v=\sum_{i=1}^{|\Theta|}\partial_{\theta_i}\nabla_\Theta f(\Theta,x)\,v_i.}
\end{align}
The operator norm of $dh_x(\Theta): (\mathbb{R}^{|\Theta|},\ell^\infty)\to (\mathbb{R}^{|\Theta|},\ell^\infty)$ is given by
\begin{align}
\nonumber \|dh_x(\Theta)\|_{\ell^\infty \to \ell^\infty}&=\sup_{\|v\|_\infty\leq 1}\|dh(\Theta)v\|_\infty\\
\nonumber&\leq\sup_{\|v\|_\infty\leq 1}\Bigg\|\sum_{i=1}^{|\Theta|}\partial_{\theta_i}\nabla_\Theta f(\Theta,x)\,v_i\Bigg\|_\infty\\
\nonumber&=\sup_{|v_i|\leq 1}\,\max_{1\leq j\leq |\Theta|} \Bigg|\sum_{i=1}^{|\Theta|}\partial_{\theta_i}\partial_{\theta_j} f(\Theta,x)\,v_i\Bigg|\\
\nonumber&\leq\sup_{|v_i|\leq 1}\,\max_{1\leq j\leq |\Theta|} \sum_{i=1}^{|\Theta|}|\partial_{\theta_i}\partial_{\theta_j} f(\Theta,x)|\\
\nonumber&=\max_{1\leq j\leq |\Theta|} \sum_{i=1}^{|\Theta|}|\partial_{\theta_i}\partial_{\theta_j} f(\Theta,x)|\\
&\leq \frac{4}{N(m)}\max_{1\leq j\leq |\Theta|}\sum_{i=1}^{|\Theta|} |\mathcal{M}_i\cap \mathcal{M}_j|.
\end{align}
Using Lemma \ref{maxMM}, we have that
\begin{align}
\|dh_x(\Theta)\|_{\ell^\infty \to \ell^\infty}\leq 4\, \frac{|\mathcal{M}|^2|\mathcal{N}|}{N(m)}.
\end{align}
So, the uniform bound on the convex domain $\mathscr{P}$ of $\Theta$ ensures, by Lemma \ref{compactconvex}, that
\begin{align}
\|\nabla_\Theta f(\Theta,x)-\nabla_{\Theta} f(\Theta',x)\|_\infty\leq 4\,\frac{|\mathcal{M}|^2|\mathcal{N}|}{N(m)}\|\Theta-\Theta'\|_\infty
\end{align}
and this finally proves (\ref{lemma4}).
\end{proof}

\begin{lemma}[Lipschitzness of the model]\label{lipf}
The following inequality holds:
\[|f(\Theta,x)-f(\Theta',x)|\leq 2\frac{\Sigma_1}{N(m)}\|\Theta-\Theta'\|_\infty .\]
\end{lemma}

\begin{proof}
Using Lemma \ref{compactconvex} and Lemma \ref{lemma}:
\[|f(\Theta,x)-f(\Theta',x)|\leq \max_{\Theta\in\mathscr{P}} \|df(\Theta,x)\|_\mathcal{L} \|\Theta-\Theta'\|_\infty, \]
where
\begin{align}
\nonumber \|df(\Theta,x)\|_\mathcal{L}&=\sup_{\|v\|_\infty\leq 1} \modifica{\left|\sum_{i=1}^{Lm}\partial_{\theta_i} f(\Theta,x)\,v_i\right|}\\
&\leq \|\nabla_\Theta f(\Theta,x)\|_1\leq 2\frac{\Sigma_1}{N(m)}.
\end{align}
Therefore
\[|f(\Theta,x)-f(\Theta',x)|\leq 2\frac{\Sigma_1}{N(m)}\|\Theta-\Theta'\|_\infty. \]
\end{proof}

\begin{lemma}[Lipschitzness of the NTK]\label{kernel}
The following inequality holds:
\[ |\hat K_{\Theta_0}(x,x')-\hat K_{\Theta}(x,x')|
\leq 16\frac{\Sigma_1|\mathcal{M}|^2|\mathcal{N}|}{N_K(m)N^2(m)} \|\Theta_0-\Theta\|_\infty.\]
\end{lemma}
\begin{proof}
By Lemma \ref{lemma}, we have
\begin{align}
\nonumber |\hat K&_{\Theta_0}(x,x')-\hat K_{\Theta}(x,x')|\\
\nonumber &=\frac{1}{N_K(m)}|\nabla_\Theta f(\Theta_0,x)\cdot \nabla_\Theta f(\Theta_0,x')-\nabla_\Theta f(\modifica{\Theta},x)\cdot \nabla_\Theta f(\modifica{\Theta},x')|\\
\nonumber &=\frac{1}{N_K(m)}\Big|\big(\nabla_\Theta f(\Theta_0,x)-\nabla_\Theta f(\Theta,x)\big)\cdot \nabla_\Theta f(\Theta_0,x')\\
\nonumber &\phantom{=\frac{1}{N_K(m)}\big||} +\nabla_\Theta f(\modifica{\Theta},x)\cdot \big(\nabla_\Theta f(\Theta_0,x')-\nabla_\Theta f(\Theta,x')\big)\Big|\\
\nonumber &\leq\frac{1}{N_K(m)}\Big(\big\|\nabla_\Theta f(\Theta_0,x)-\nabla_\Theta f(\Theta,x)\big\|_\infty \big\|\nabla_\Theta f(\Theta_0,x')\big\|_1\\
\nonumber &\phantom{=\frac{1}{N_K(m)}\big(|} +\big\|\nabla_\Theta f(\modifica{\Theta},x)\big\|_1 \big\|\nabla_\Theta f(\Theta_0,x')-\nabla_\Theta f(\Theta,x')\big\|_\infty\Big)\\
\nonumber&\leq\frac{4|\mathcal{M}|^2|\mathcal{N}|}{N_K(m)N(m)}\left(\big\|\nabla_\Theta f(\Theta_0,\modifica{x'})\big\|_1+\big\|\nabla_\Theta f(\modifica{\Theta},\modifica{x})\big\|_1 \right)\|\Theta_0-\Theta\|_\infty\\
&\leq 16\frac{\Sigma_1|\mathcal{M}|^2|\mathcal{N}|}{N_K(m)N^2(m)} \|\Theta_0-\Theta\|_\infty.
\end{align}
\end{proof}

\subsubsection{Original model, linearized version and lazy training}\label{lemmalin}
The following Theorem \ref{powerful} claims that the difference between the original and the linear model at any input point is bounded by a diverging factor multiplied by the square of the infinity norm of the difference between the initial parameters and a new choice of the parameters. The square is crucial, since in the lazy regime $\|\Theta_t-\Theta_0\|_\infty$ is suppressed by a factor $1/N(m)$ (see Theorem \ref{gradfl}), which would not be sufficient to ensure the convergence to the linearized model without such square.

\begin{theorem}[Discrepancy between the original and the linearized model]
\label{powerful}
For any $\Theta_0\in\mathscr{P}$ defining the linearized model, we have the following upper bound:
\[|f(\Theta,x)-f^{\mathrm{lin}}(\Theta,x)|\leq \frac{Lm |\mathcal{M}|^2|\mathcal{N}|}{N(m)}\|\Theta-\Theta_0\|^2_\infty.\]
\end{theorem}

In order to understand how the convergence follows from the previous inequality, we can consider again the example of \cite{QLazy}, i.e.
\[ L \text{ is fixed}\qquad\text{and}\qquad N(m)=\sqrt m.\]
In this case, we can bound, by Theorem \ref{gradfl},
\[\|\Theta_t-\Theta_0\|_\infty\leq \frac{c}{\sqrt m},\]
so that
\[ |f(\Theta_t,x)-f^{\mathrm{lin}}(\Theta_t,x)|\leq c'\frac{m}{\sqrt m}\left(\frac{c}{\sqrt m}\right)^2\propto \frac{1}{\sqrt m}. \]
Therefore, in the limit $m\to\infty$ we have uniform convergence \[f(\Theta_t,\,\cdot\,)\to f^{\mathrm{lin}}(\Theta_t,\,\cdot\,).\]
Let us see in the proof how a fine estimate of the second derivative using light cones ensures the result.
\begin{proof}
We apply the multivariate version of Taylor's theorem with integral remainder to $f(\Theta,x)$ with respect to the parameters.
\[ f(\Theta,x)=f(\Theta_0,x)+\nabla f(\Theta_0,x)^T(\Theta-\Theta_0)+\sum_{i,j=1}^{Lm}R_{ij}(\Theta,x)(\theta_i-\theta_{i,0})(\theta_j-\theta_{j,0})\]
with
\[R_{ij}(\Theta,x)=\frac{1}{2}\int_0^1 (1-t) \partial_{\theta_i}\partial_{\theta_j}f(\Theta_0+t(\Theta-\Theta_0),x)\,dt.\]
By Lemma \ref{unifbounds},
\[|R_{ij}(\Theta,x)|\leq\frac{1}{2}\int_0^1 (1-t) \left(4\frac{|\mathcal{M}_i\cap\mathcal{M}_j|}{N(m)}\right)\,dt=\frac{|\mathcal{M}_i\cap\mathcal{M}_j|}{N(m)}\]
So
\begin{align}
\nonumber|f(\Theta,x)-f^{\mathrm{lin}}(\Theta,x)|&=\left|\sum_{i,j=1}^{Lm}R_{ij}(\Theta,x)(\theta_i-\theta_{i,0})(\theta_j-\theta_{j,0})\right|\\
\nonumber&\leq \left(\sum_{i,j=1}^{Lm}\frac{|\mathcal{M}_i\cap\mathcal{M}_j|}{N(m)}\right)\|\Theta-\Theta_0\|^2_\infty\\
&\leq \frac{1}{N(m)}\left(Lm \max_{1\leq j\leq Lm}\sum_{i=1}^{Lm}|\mathcal{M}_i\cap\mathcal{M}_j|\right)\|\Theta-\Theta_0\|^2_\infty.
\end{align}
Using Lemma \ref{maxMM}, we eventually estimate
\[|f(\Theta,x)-f^{\mathrm{lin}}(\Theta,x)|\leq \frac{Lm |\mathcal{M}|^2|\mathcal{N}|}{N(m)}\|\Theta-\Theta_0\|^2_\infty.\]
\end{proof}

\subsubsection{The minimum of the loss for the linear model and the convergence of the kernel}
In this subsubsection we show that the linearized model, for sufficiently large $m$ yields
\[{\min}_{\Theta\in\mathbb{R}^{Lm}}\mathcal{L}^{\mathrm{lin}}(\Theta)=0\]
with high probability, where $\mathcal{L}^{\mathrm{lin}}$ is the mean square error for the linearized model.
If the linear model is able to perfectly fit the dataset -- i.e., $\mathcal{L}^{\mathrm{lin}}(\Theta^\ast)=0$ for some $\Theta^\ast\in\mathbb{R}^{Lm}$ -- and the training converges to the global minimum of the cost, then,
in the limit of infinitely many qubits, also the original model, which is close to its linearized version in the lazy regime, will be able to perfectly fit the dataset.

\begin{remark}
In this subsubsection we will use the result of Theorem \ref{ntkconv}. Even though the proof of Theorem \ref{ntkconv} is delayed to \autoref{proofntk} for a cleaner presentation, such proof will not use any result of this subsubsection.
\end{remark}

The starting point is a convergence lemma, which will be used in different proofs.

\begin{lemma}[Convergence of the minimum eigevalue]\label{lambdamin}
Suppose that the matrix elements of a sequence of $n\times n$ symmetric matrices $M_m$ converge in probability to the matrix elements of a positive matrix $M_\infty\succ 0$:
\[\lim_{m\to\infty} \mathbb{P}[|(M_m)_{ij}-(M_\infty)_{ij}|\geq \epsilon]=0 \qquad \forall\, i,j\in\{1,\dots,n\}\quad \forall \,\delta>0.\] 
Let $\lambda_{\min}^M>0$ be the minimum eigenvalue of $M_\infty$. Then, for any $\delta>0$ and $\eta\in(0,1)$ there exist $m_0\in\mathbb{N}$ such that, for all $m\geq m_0$,
\[M_n\succ \eta\lambda_{\min}^M\id\]
with probability at least $1-\delta$.
\end{lemma}

\begin{proof} The convergence in probability implies that
\begin{align}
\nonumber \lim_{m\to\infty}\mathbb{P}[\|M_m-M_\infty\|_F\geq \epsilon]&=
\lim_{m\to\infty}\mathbb{P}\left[\sum_{i,j=1}^n|(M_m)_{ij}-(M_\infty)_{ij}|^2\geq \epsilon^2\right]\\
\nonumber &\leq \lim_{m\to\infty}\mathbb{P}\left[\max_{ij}|(M_m)_{ij}-(M_\infty)_{ij}|\geq \epsilon/n\right]\\
&\leq\lim_{m\to\infty}\sum_{i,j=1}^n\mathbb{P}\left[|(M_m)_{ij}-(M_\infty)_{ij}|\geq \epsilon/n\right]=0.
\end{align}
Therefore, for any $\delta>0$, there exist $m_1(\epsilon)\in\mathbb{N}$ such that, for any $m\geq m_0(\epsilon)$, 
\[\|M_m-M_\infty\|_F< \epsilon	\]
with probability at least $1-\delta$. Let $F_m=M_\infty-M_m$. Since \modifica{$\|F_m\|_F< \epsilon$}, the maximum \modifica{eigenvalue} of $|F_m|$ is \modifica{$\lambda_m<\epsilon$}, so
\[ F_m\preceq |F_m| \preceq \lambda_m\id\prec \epsilon \id.\]
Let $\lambda_{\min}^M$ be the minimum eigenvalue of $M_\infty$. The previous equation implies that
\[M_\infty-M_m \prec \epsilon \id \quad \to\quad M_m \succ M_\infty -\epsilon \id\succeq (\lambda_{\min}^M-\epsilon)\id.\]
Since $\lambda_{\min}^M$ is positive and $\epsilon>0$ is arbitrary, letting $\epsilon=(1-\eta)\lambda_{\min}^M>0$ -- which fixes $m_0$ -- the claim holds.
\end{proof}

Then, we need a translation of Assumption \ref{assNTK} into a property of the gradient of the function at initialization which will be used to fit the dataset in the overparameterized regime.

\begin{lemma}\label{convprob}
If (\ref{weakassumption}) holds, then for any $x,x'\in\mathcal{X}$, Assumption \ref{assNTK} implies the convergence in probability of $\hat K_{\Theta_0}(x,x')$ to $\bar K(x,x')$ over random initialization.
\begin{align}
\lim_{m\to\infty}\mathbb{P}\left[|\hat K_{\Theta_0}(x,x')-\bar K(x,x')|\leq\epsilon\right]=1\qquad \forall\, \epsilon>0.
\end{align}

\end{lemma}

\begin{proof}
Given $\epsilon>0$, by Assumption \ref{assNTK}, there exists $\bar m\in\mathbb{N}$ such that 
\begin{align}
|K(x,x')-\bar K(x,x')|\leq \frac{\epsilon}{2} \qquad \forall\, m\geq\bar m  \quad \forall\,x,x'\in\mathcal{X}.
\end{align}
Therefore, $ \forall\, m\geq\bar m  \quad \forall\,x,x'\in\mathcal{X}$,
\begin{align}
\nonumber|\hat K_{\Theta_0}(x,x')-\bar K(x,x')|&\leq|\hat K_{\Theta_0}(x,x')-K(x,x')|+|K(x,x')-\bar K(x,x')|\\
&\leq|\hat K_{\Theta_0}(x,x')-K(x,x')|+\frac{\epsilon}{2}.
\end{align}
Furthermore, by Theorem \ref{ntkconv}, if $m\geq \bar m$
\begin{align}
\nonumber\mathbb{P}\Big[|\hat K_{\Theta_0}(x,x')&-\bar K(x,x')|\leq\epsilon\Big]\geq\mathbb{P}\left[|\hat K_{\Theta_0}(x,x')-K(x,x')|+\frac{\epsilon}{2}\leq\epsilon\right]\\
\nonumber&= \mathbb{P}\left[|\hat K_{\Theta_0}(x,x')-K(x,x')|\leq\frac{\epsilon}{2}\right]\\
&\geq 1- \exp\left[-cN_K^2(m)\frac{N^4(m)}{|\mathcal{M}|^2|\mathcal{N}|^2\Sigma_2}\,\left(\frac{\epsilon}{2}\right)^2\right].
\end{align}
Hence
\begin{align}
\lim_{m\to\infty}\mathbb{P}\left[|\hat K_{\Theta_0}(x,x')-\bar K(x,x')|\leq\epsilon\right]=1.
\end{align}
\end{proof}

\begin{lemma} \label{asymptcomplete}
Let $n=|\mathcal{D}|$. If (\ref{weakassumption}) and Assumption \ref{assNTK} hold, then for any $\delta>0$, there exists $m_0\in\mathbb{N}$ such that, for any $m\geq m_0$,
\[\{\partial_{\theta_i}f(\Theta_0,X)\}_{1\leq i\leq Lm} \text{ is a complete set of vectors for } \mathbb{R}^n\]
with probability at least $1-\delta$ over random initialization.
\end{lemma}

\begin{proof}
Since, by Lemma \ref{convprob}, as $m\to\infty$
\[\hat K_{\Theta_0}(X,X^T)=\frac{1}{N_K(m)}\sum_{i=1}^{Lm}\partial_{\theta_i}f(\Theta_0,X)\partial_{\theta_i}f(\Theta_0,X^T)\]
converges in probability to $\bar K(X,X^T)$, which is strictly positive by Assumption \ref{assNTK}, there exists $m_0\in\mathbb{N}$ such that $\hat K_{\Theta_0}(X,X^T)$ is strictly positive for all $m\geq m_0$ by Lemma \ref{lambdamin} with high probability over random initialization. Calling
\[v_i=\frac{1}{\sqrt{N(m)}}\partial_{\theta_i}f(\Theta_0,X)\in\mathbb{R}^n,\]
we can rewrite
\[\hat K_{\Theta_0}\equiv \hat K_{\Theta_0}(X,X^T)=\sum_{i=1}^{Lm} v_iv_i^T.\]
Since 
\begin{enumerate}
\item $\hat K_{\Theta_0}$ is strictly positive, it is invertible, so its range is $\mathbb{R}^n$, 
\item the range of $\hat K_{\Theta_0}$ is $\text{Span}\{v_i\}_{1\leq i\leq Lm}$,
\end{enumerate}
we conclude that
\[\text{Span}\{v_i\}_{1\leq i\leq Lm}=\mathbb{R}^n.\]
\end{proof}

Now we are ready to prove the following statement.
\begin{corollary} 
For any $\delta>0$ there exists $m_0\in\mathbb{N}$ such that
\[{\min}_{\Theta\in\mathbb{R}^{Lm}}\mathcal{L}^{\mathrm{lin}}(\Theta)=0 \qquad \forall \,m\geq m_0\]
with probability at least $1-\delta$ over random initialization.
\end{corollary}

\begin{proof}
Using Lemma \ref{asymptcomplete}, 
\[\{\partial_{\theta_i}f(\Theta_0,X)\}_{1\leq i\leq Lm}\]
is definitely a complete set of vectors for $\mathbb{R}^n$ with high probability.
So, there exists $\alpha_1,\dots, \alpha_{Lm}$ such that $Y-f(\Theta_0,X)\in\mathbb{R}^n$ can be written as a linear combination:
\[\alpha_1\partial_{\theta_1}f(\Theta_0,X)+\dots+\alpha_{Lm}\partial_{\theta_{Lm}}f(\Theta_0,X)=Y-f(\Theta_0,X).\]
So, \modifica{choosing}
\[\Theta=\Theta_0+\begin{pmatrix}\alpha_1\\ \vdots\\ \alpha_{Lm}\end{pmatrix},\]
we have
\begin{align}
\nonumber\mathcal{L}^{\mathrm{lin}}(\Theta)&=\frac{1}{n}\sum_{i=1}^n\frac{1}{2}\left(f^{\mathrm{lin}}(\Theta,x^{(i)})-y^{(i)}\right)^2\\
\nonumber&=\frac{1}{2n}\|f(\Theta_0,X)+\nabla_\Theta f(\Theta_0,X)^T(\Theta-\Theta_0)-Y\|_2^2\\
&=\frac{1}{2n}\|f(\Theta_0,X)+(Y-f(\Theta_0,X))-Y\|_2^2=0.
\end{align}
\end{proof}

\subsection{Proof of Lemma \ref{boundnorm}}\label{proofboundnorm}

We need some preliminary statements in order to prove Lemma \ref{boundnorm}.

\begin{lemma}\label{checkdiagonal}
Assumption \ref{assNTK} implies that each diagonal term of $\bar K = \bar K(X,X^T)$ is strictly positive:
\[\bar K(x,x)>0 \qquad \forall\, (x,\,\cdot\,)\in\mathcal{D}.\]
\end{lemma}
\begin{proof}
By Assumption $\ref{assNTK}$, $\bar K$ is a strictly positive matrix, so its diagonal elements are strictly positive.
\end{proof}

\begin{corollary}\label{diagonalNTK}
Under Assumption \ref{assNTK}, the following property of $\bar K$ holds:
\[\exists\,x\in\mathcal{X}\quad\text{such that}\quad \bar K(x,x)> 0\quad \mathcal{K}(x,x)> 0.\]
\end{corollary}
\begin{proof}
Fix $x\in\mathcal{X}$ such that $(x,\cdot)\in\mathcal{D}$. By Lemma \ref{checkdiagonal}, we know that $\bar K(x,x)>0$. Furthermore, by Assumption \ref{zeromean}, $\mathcal{K}(x,x)>0$.
\end{proof}

\begin{lemma}\label{doppiastima}
Under the Assumption \ref{uniform}, the following inequalities hold for all $x\in \mathcal{X}$
\[4\mathbb{E}\left[f^2(\Theta,x)\right]\leq \mathbb{E}\left[\|\nabla_\Theta f(\Theta,x)\|^2\right]\leq 4\,|\mathcal{N}|\,\mathbb{E}\left[f^2(\Theta,x)\right].\]
\end{lemma}

The following statement is an improvement of \cite[Lemma 1]{napp2022quantifying} for our case: in the proof we use stronger bounds, in order to obtain a more effective constraint to our purpose.

\begin{proof}
Since $f(\Theta,x)$ is periodic with period $\pi$ in each component, we can consider the Fourier series decomposition
\[ f(\Theta,x)=\sum_{v\in\mathbb{Z}^{Lm}}\tilde f_v(x)e^{2i\Theta\cdot v} \quad\text{where}\quad \tilde f_v(x)=\int_0^{\pi}\prod_{i=1}^{Lm}\left(\frac{d\theta_i}{\pi}\right)e^{-2i\Theta\cdot v} f(\Theta,x).\]
$f(\Theta,x)$ is of the form
\begin{align}
\nonumber f(\Theta,x)&=\matrixel{0}{U^\dagger(\Theta,x)\mathcal{M}U(\Theta,x)}{0}\\
&=\smatrixel{0}{S_1^\dagger W_1^\dagger(\Theta)S_1\cdots W_L^\dagger(\Theta)S_{L+1}^\dagger \mathcal{M}S_{L+1}W_L(\Theta)\cdots S_2W_1(\Theta)S_1}{0},
\end{align}
where
\[W_i(\Theta)=\bigotimes_{i=1}^me^{-i\theta_i\mathcal{G}_i}.\]
We can consider the projectors $P_i^\pm$ on the eigenspaces of $\mathcal{G}_i$:
\[\mathcal{G}_i=P_i^+-P_i^-	\quad P_i^++P_i^-=\id\]
In this way, we can rewrite
\begin{align}
\nonumber W_i(\theta_i)&=e^{-i\theta_i\mathcal{G}_i}=e^{-i\theta_iP_i^+}e^{i\theta_iP_i^-}\\
\nonumber &=e^{-i\theta_iP_i^+}(P_i^++P_i^-)e^{i\theta_iP_i^-}\\
&=e^{-i\theta_i}P_i^++e^{i\theta_i}P_i^-.
\end{align}
This means that in 
\[ f(\Theta,x)=\smatrixel{0}{S_1^\dagger\cdots W_i^\dagger(\theta_i)\cdots S_{L+1}^\dagger \mathcal{M}S_{L+1}\cdots W_i(\theta_i)\cdots S_1}{0}
\]
the only dependence on $\theta_i$ will appear as
\[ f(\Theta,x)=e^{-2\theta_i}f^{(i,-1)}(\Theta,x)+f^{(i,0)}(\Theta,x)+e^{+2\theta_i}f^{(i,+1)}(\Theta,x), \]
where $f^{(i,\ell)}(\Theta,x)$ does not depend on $\theta_i$. Therefore, we know that
\[ \exists i\in\{1,\dots,Lm\}: v_i\notin\{0,\pm 1\}\quad\to \quad \tilde f_v(x)=0,\]
so we can restrict the sum of the Fourier expansion to
\[ f(\Theta,x)=\sum_{v\in\{0,\pm 1\}^{Lm}}\tilde f_v(x)e^{2i\Theta\cdot v}. \label{eq:dependence_theta}\]
Writing $f(\Theta,x)$ as sum of local observables depending on a limited light cone, we can further restrict the set of $v$ such that $\tilde f_v(x)\neq 0$. Let
\[ \text{supp}(v)=\{i\in\{1,\dots,Lm\}: v_i\neq 0\}.\]
Then
\begin{align}
\nonumber\tilde f_v(x)&=\int_0^{\pi}\prod_{i=1}^{Lm}\left(\frac{d\theta_i}{\pi}\right)e^{-2i\Theta\cdot v} \sum_{k=1}^mf_k(\Theta_{\mathcal{N}_k},x)\\
\nonumber&=\sum_{k=1}^m\left(\int_0^{\pi}\prod_{i\notin\mathcal{N}_k}\left(\frac{d\theta_i}{\pi}e^{-2i\theta_iv_i}\right)\right)\left(\int_0^{\pi}\prod_{i\in\mathcal{N}_k}\left(\frac{d\theta_i}{\pi}e^{-2i\theta_iv_i}\right) f_k(\Theta_{\mathcal{N}_k},x)\right)\\
&=\sum_{k=1}^m
\chi_{\mathcal{N}_k}(\text{supp}(v))
\left(\int_0^{\pi}\prod_{i\in\mathcal{N}_k}\left(\frac{d\theta_i}{\pi}e^{-2i\theta_iv_i}\right) f_k(\Theta_{\mathcal{N}_k},x)\right),
\end{align}
where
\[ \chi_{\mathcal{N}_k}(\text{supp}(v))= 
\begin{cases}
1 & \text{ if supp}(v)\subseteq \mathcal{N}_k\\
0 & \text{otherwise} 
\end{cases}.\]
Hence,
\[ \tilde f_v(x)\neq 0 \quad \to \quad \exists\, k\in\{1,\dots, m\} : \text{supp}(v)\subseteq \mathcal{N}_k. \]
This means that
\[ \tilde f_v(x)\neq 0 \quad \to \quad |\text{supp}(v)|\leq \max_{1\leq k\leq m} |\mathcal{N}_k|=|\mathcal{N}|. \]
The function $f(\Theta,x)$ is differentiable, hence
\begin{align}
\partial_{\theta_j}f(\Theta,x)&=2i\sum_{v\in\{0,\pm 1\}^{Lm}}v_j\tilde f_v(x)e^{2i\Theta\cdot v}.
\end{align}
By Parseval's identity
\begin{align}
\mathbb{E}\left[f^2(\Theta,x)\right]&=\sum_{v\in\{0,\pm 1\}^{Lm}}\left|\tilde f_v(x)\right|^2,\\
\label{due}
\mathbb{E}\left[\left(\partial_{\theta_j}f(\Theta,x)\right)^2\right]&=4\sum_{v\in\{0,\pm 1\}^{Lm}}v_j^2\left|\tilde f_v(x)\right|^2.\\
\end{align}
Using the fact that in the sum of (\ref{due}) only the terms multiplied by $v_j\in\{0,\pm 1\}$ contribute to the result,
\begin{align}
\mathbb{E}\left[\left(\partial_{\theta_j}f(\Theta,x)\right)^2\right]&=4\sum_{\substack{v\in\{0,\pm 1\}^{Lm}\\v_j\neq 0}}v_j^2\left|\tilde f_v(x)\right|^2=4\sum_{\substack{v\in\{0,\pm 1\}^{Lm}\\v_j\neq 0}}\left|\tilde f_v(x)\right|^2,\\
\nonumber\mathbb{E}\left[\|\nabla_\Theta f(\Theta,x)\|^2\right]&= 4\sum_{j=1}^{Lm}\sum_{\substack{v\in\{0,\pm 1\}^{Lm}\\v_j\neq 0}}\left|\tilde f_v(x)\right|^2\\
\nonumber&=4\sum_{v\in\{0,\pm 1\}^{Lm}}\sum_{j: v_j\neq 0}\left|\tilde f_v(x)\right|^2\\
\nonumber&=4\sum_{v\in\{0,\pm 1\}^{Lm}}\text{supp}(v)\left|\tilde f_v(x)\right|^2\\
&\leq 4\,|\mathcal{N}|\sum_{v\in\{0,\pm 1\}^{Lm}}\left|\tilde f_v(x)\right|^2=4\,|\mathcal{N}|\,
\mathbb{E}\left[f^2(\Theta,x)\right].
\end{align}
For the other inequality, we start again from (\ref{due}) and, calling $0=(0,\dots,0)\in\mathbb{Z}^{Lm}$,
\begin{align}
\nonumber\mathbb{E}\left[\|\nabla_\Theta f(\Theta,x)\|^2\right]&=\sum_{j=1}^{Lm}\mathbb{E}\left[\left(\partial_{\theta_j}f(\Theta,x)\right)^2\right]
=4\sum_{v\in\mathbb{Z}^{Lm}}\left(\sum_{j=1}^{Lm} v_j^2\right)\left|\tilde f_v(x)\right|^2\\
\nonumber&=4\sum_{v\in\mathbb{Z}^{Lm}\setminus\{0\}} v^2\left|\tilde f_v(x)\right|^2
\geq 4\sum_{v\in\mathbb{Z}^{Lm}\setminus\{0\}} \left|\tilde f_v(x)\right|^2\\
&=4\sum_{v\in\mathbb{Z}^{Lm}} \left|\tilde f_v(x)\right|^2-4\left|\tilde f_0(x)\right|^2.
\end{align}
Now we notice that
\begin{align}
\nonumber\mathbb{E}\left[f(\Theta,x)\right]&=\sum_{v\in\mathbb{Z}^{Lm}}\tilde f_v(x)\mathbb{E}\left[e^{2i\Theta\cdot v}\right]= \sum_{v\in\mathbb{Z}^{Lm}}\tilde f_v(x)\prod_{j=1}^{Lm}\mathbb{E}\left[e^{2i\theta_j v_j}\right]\\
&=\sum_{v\in\mathbb{Z}^{Lm}}\tilde f_v(x)\delta_{v,0}=\tilde f_0(x).
\end{align}
By Assumption \ref{zeromean}, 
\[ \mathbb{E}\left[f(\Theta,x)\right]=0\qquad \to \qquad \tilde f_0(x)=0.\]
Therefore
\begin{align}
\mathbb{E}\left[\|\nabla_\Theta f(\Theta,x)\|^2\right]&\geq 4\sum_{v\in\mathbb{Z}^{Lm}} \left|\tilde f_v(x)\right|^2 = 4\mathbb{E}\left[f^2(\Theta,x)\right].
\end{align}

\end{proof}

To prove Lemma \ref{boundnorm} we are going to use Lemma \ref{doppiastima} in order to bound the normalization $N_K(m)$ in terms of the normalization $N(m)$ and of the geometric quantities related to the architecture of the circuit. We will use this notation:
\[ \mathcal{K}_m(x,x')=\mathbb{E}\left[f(\Theta,x)f(\Theta,x')\right].\]
From Assumption \ref{zeromean}, there exist a limit covariance function $\mathcal{K}:\mathcal{X}\times\mathcal{X}\to \mathbb{R}$
\[ \lim_{m\to\infty}\mathcal{K}_m(x,x')=\mathcal{K}(x,x').\]

Now we are ready to prove Lemma \ref{boundnorm}.\\
As in Corollary \ref{diagonalNTK}, let $x\in\mathcal{X}$ such that both $\bar K(x,x)$ and $\mathcal{K}(x,x)$ are positive. We notice that
\[K(x,x)=\frac{1}{N_K(m)}\mathbb{E}\left[\|\nabla f(\Theta,x)\|^2\right]\qquad \mathcal{K}_m(x,x)=\mathbb{E}\left[f^2(\Theta,x)\right].\]
Therefore, the results of Lemma \ref{doppiastima} can be restated as 
\[4\mathcal{K}_m(x,x)\leq N_K(m)K(x,x)\leq 4\,|\mathcal{N}|\sqrt{\mathcal{K}_m(x,x)}\]
Dividing by $K(x,x)$,
\[4\frac{\mathcal{K}_m(x,x)}{K(x,x)}\leq N_K(m)\leq4|\mathcal{N}|\frac{\sqrt{\mathcal{K}_m(x,x)}}{K(x,x)}.\]
Since $\mathcal{K}_m\to \mathcal{K}$ and $K\to \bar K$ as $m\to \infty$, there exist $m_0$ such that, for any $m\geq m_0$,
\[ \frac{1}{2}\mathcal{K}(x,x)\leq\mathcal{K}_m(x,x)\leq 2\mathcal{K}(x,x),\qquad \frac{1}{2}\bar K(x,x)\leq K(x,x)\leq 2\bar K(x,x).\]
which implies, for all $m\geq m_0$,
\[\frac{\mathcal{K}(x,x)}{\bar K(x,x)}\leq N_K(m)\leq\left(8\sqrt 2 \frac{\sqrt{\mathcal{K}(x,x)}}{\bar K(x,x)}\right)|\mathcal{N}|.\]

\subsection{Proof of Theorem \ref{ntkconv}}\label{proofntk}

We need the following lemma.

\begin{lemma}[McDiarmid's concentration inequality \cite{mcdiarmid_1989}]
\label{mcdiarmid} 
Let $X_1,\dots,X_n$ be independent random variables, each with values in $\mathbb{X}$. Let $f:\mathbb{X}^n\to \mathbb{R}$ be a mapping such that, for every $i\in\{1,\dots, n\}$ and every $(x_1,\dots,x_n),(x'_1,\dots,x'_n)\in\mathbb{X}^n$ that differ only in the $i$-th coordinate (i.e., $\forall\, j\neq i,\, x_j=x_j'$),
\begin{align}
|f(x_1,\dots,x_n)-f(x'_1,\dots,x'_n)|\leq c_i.
\end{align}
For any $\epsilon>0$
\begin{align}
\mathbb{P}\left(f(X_1,\dots,X_n)-\mathbb{E}[f(X_1,\dots,X_n)]\geq\epsilon\right)\leq \exp\left(-\frac{2\epsilon^2}{\sum_{i=1}^n c_i^2}\right).
\end{align}
\end{lemma}

Following the strategy of \cite{QLazy}, we show that the hypotheses of McDiarmid's concentration inequality \ref{mcdiarmid} are satistied by $\hat K_\Theta$.
\begin{align}
\nonumber\hat K_\Theta(x,x')&=\frac{1}{N_K(m)}\nabla_\Theta f(\Theta,x)\cdot \nabla_\Theta f(\Theta,x')\\
\nonumber &= \frac{1}{N_K(m)}\frac{1}{N^2(m)}\sum_{k,k'=1}^m\sum_{j=1}^{|\Theta|}\partial_{\theta_j}f_k(\Theta,x)\,\partial_{\theta_j}f_{k'}(\Theta,x')\\
&= \frac{1}{N_K(m)}\frac{1}{N^2(m)}\sum_{k,k'=1}^m\sum_{j\in\mathcal{N}_k\cap\mathcal{N}_{k'}}\partial_{\theta_j}f_k(\Theta,x)\,\partial_{\theta_j}f_{k'}(\Theta,x').
\end{align}
In order to clarify when a term depends or not on a particular parameter $\theta_i$, we use the notation
\begin{align}
f_k(\Theta_{\mathcal{N}_k},x) \quad\text{instead of}\quad f_k(\Theta,x),
\end{align}
so that we emphasize that $f_k$ depends only on $\theta_i$ with $i\in\mathcal{N}_k$.\\
Let $\Gamma\subseteq \{1,\dots,m\}\times\{1,\dots,m\}\times\{1,\dots,|\Theta|\}$ be the set defined as
\begin{align}
\Gamma=\{(k,k',j):j\in\mathcal{N}_k\cap\mathcal{N}_{k'}\}.
\end{align}
Let
\begin{align}
T_{k,k',j}(\Theta_{\mathcal{N}_k\cup\mathcal{N}_{k'}})=\partial_{\theta_j}f_k(\Theta_{\mathcal{N}_k},x)\,\partial_{\theta_j}f_{k'}(\Theta_{\mathcal{N}_{k'}},x'),
\end{align}
so that we can rewrite
\begin{align}
\hat K_\Theta(x,x')=\frac{1}{N_K(m)}\frac{1}{N^2(m)}\sum_{(k,k',j)\in\Gamma}T_{k,k',j}(\Theta_{\mathcal{N}_k\cup\mathcal{N}_{k'}}).
\end{align}
We fix $i\in\{1,\dots,|\Theta|\}$ and we ask that $\theta_j=\theta'_j$ for all $j\neq i$. In order to compute $\hat K_\Theta(x,x')-\hat K_{\Theta'}(x,x')$, we notice that
\begin{align}
i\notin \mathcal{N}_k\cup\mathcal{N}_{k'} \quad \to\quad T_{k,k',j}(\Theta_{\mathcal{N}_k\cup\mathcal{N}_{k'}})-T_{k,k',j}(\Theta'_{\mathcal{N}_k\cup\mathcal{N}_{k'}})=0,
\end{align}
so we define
\begin{align}
\Gamma_i=\{(k,k',j):j\in\mathcal{N}_k\cap\mathcal{N}_{k'}, i\in\mathcal{N}_k\cup\mathcal{N}_{k'}\}.
\end{align}
Therefore, using the bound of (\ref{eq4.95}), $|T_{k,k',j}|\leq 4$. This implies
\begin{align}
\hat  K_\Theta(x,x')-\hat K_{\Theta'}(x,x')&=\frac{1}{N_K(m)}\frac{1}{N^2(m)}\sum_{(k,k',j)\in\Gamma_i}\left( T_{k,k',j}(\Theta)-T_{k,k',j}(\Theta')\right),\\
\nonumber|\hat K_\Theta(x,x')-\hat K_{\Theta'}(x,x')|&\leq\frac{1}{N_K(m)}\frac{1}{N^2(m)}\sum_{(k,k',j)\in\Gamma_i}\left(|T_{k,k',j}(\Theta)|+|T_{k,k',j}(\Theta')|\right)\\
&\leq \frac{1}{N_K(m)}\frac{|\Gamma_i|}{N^2(m)}\cdot (4+4).
\end{align}
So, we need to compute $|\Gamma_i|$. A change of perspective on the constraints in the definition of $\Gamma_i$ allows an easier computation:
\begin{align}
\Gamma_i=\{(k,k',j):(k\in\mathcal{M}_i\,\lor\,k'\in\mathcal{M}_i)\,\land\,k\in \mathcal{M}_j\, \land\, k'\in\mathcal{M}_j\}.
\end{align}
If we assume that the condition $k\in\mathcal{M}_i$ holds (first term in the following RHS), by symmetry we can estimate the cardinality of $\Gamma_i$:
\begin{align}
\nonumber\Gamma_i&=\{(k,k',j):k\in\mathcal{M}_i\,\land\,k\in \mathcal{M}_j\, \land\, k'\in\mathcal{M}_j)\}\\&\cup \{(k,k',j):k'\in\mathcal{M}_i\,\land\,k\in \mathcal{M}_j\, \land\, k'\in\mathcal{M}_j)\},\\[8pt]
\nonumber|\Gamma_i|&\leq 2|\{(k,k',j):k\in\mathcal{M}_i\,\land\,k\in \mathcal{M}_j\, \land\, k'\in\mathcal{M}_j)\}|\\
\nonumber&= 2|\{(k,k',j):k\in\mathcal{M}_i\,\land\,j\in \mathcal{N}_k\, \land\, k'\in\mathcal{M}_j)\}|\\
&\leq 2|\mathcal{M}_i|\max_k|\mathcal{N}_k|\max_j|\mathcal{M}_j|\leq 2|\mathcal{M}_i||\mathcal{M}||\mathcal{N}|.
\end{align}
Therefore, in McDiarmid's concentration inequality \ref{mcdiarmid}, we have
\begin{align}\label{c_i}
c_i&=16\,\frac{1}{N_K(m)}\frac{|\mathcal{M}||\mathcal{N}|}{N^2(m)}|\mathcal{M}_i|.
\end{align}
Summing the squares of $c_i$:
\begin{align}
\sum_{i=1}^{|\Theta|}c_i^2&=\left(16\,\frac{|\mathcal{M}||\mathcal{N}|}{N_K(m)N^2(m)}\right)^2\sum_{i=1}^{|\Theta|}|\mathcal{M}_i|^2=\left(16\,\frac{|\mathcal{M}||\mathcal{N}|}{N_K(m)N^2(m)}\right)^2\Sigma_2.
\end{align}
Calling $c=1/256$, we have
\begin{align}
\nonumber\mathbb{P}\left[|\hat K_\Theta(x,x')-K(x,x')|\geq\epsilon\right]&\leq \exp\left[-cN_K^2(m)\frac{N^4(m)}{|\mathcal{M}|^2|\mathcal{N}|^2\Sigma_2}\,\epsilon^2\right]\\
&\leq \exp\left[-cN_K^2(m)\frac{N^4(m)}{Lm|\mathcal{M}|^4|\mathcal{N}|^2}\,\epsilon^2\right].
\end{align}

\subsection{Proof of Corollaries \ref{convsol} and \ref{corgp}}\label{proofcorgp}

We need a preliminary statement.

\begin{theorem}[Slutsky's theorem \cite{Slutsky}]\label{sl}
Let $X_1,X_2,\dots$ be a sequence of random vectors or matrices converging in distribution to a random vector or \modifica{matrix} $X$
\[X_k\xrightarrow{d} X\]
and let $Y_1,Y_2,\dots$ be a sequence of random vectors or matrices converging in probability to a constant vector or \modifica{matrix} $C$
\[Y_k\xrightarrow{d} C.\]
Then
\begin{align}
X_k+Y_k&\xrightarrow{d} X+C,\\
X_kY_k&\xrightarrow{d} XC.
\end{align}
Furthermore, if $C$ is invertible,
\[X_k/Y_k\xrightarrow{d}X/C.\]
\end{theorem}

\subsubsection{Proof of Corollary \ref{convsol}}

Corollary \ref{convsol} can be proved as follows. \\
By Lemma \ref{convprob},
\[\hat K_\Theta(x,x')\xrightarrow{p}\bar K(x,x') \quad \text{as}\quad m\to\infty.\]
Furthermore, by Theorem \ref{init}
\[f(\Theta_0,\,\cdot\,)\xrightarrow{d}f^{(\infty)}(\,\cdot\,) \quad \text{as}\quad m\to\infty.\]

Let $\mathcal{F}=\{x_\alpha\}_{\alpha\in A}$ be a finite family of inputs $x_\alpha\in\mathcal{X}$ containing the inputs of the dataset
\[ \{x^{(i)}\}_{1\leq i\leq n}\subseteq \mathcal{F}\]
Since the solution (\ref{solutionevol2}) for the output corresponding to any input $x_\beta\in\mathcal{F}$ is the following linear combination of the outputs $\{f(\Theta_0,x_\alpha)\}_{\alpha\in A}$,
\[ f^{\mathrm{lin}}(\Theta^{\mathrm{lin}}_t,x_\beta)=f(\Theta_0,x_\beta)-\hat K_{\Theta_0}(x_\beta,X^T)\hat K^{-1}_{\Theta_0}\left(\id-e^{-\eta_0 \hat K_{\Theta_0}t}\right)(f(\Theta_0,X)-Y) \]
we can write
\[ f^{\mathrm{lin}}(\Theta^{\mathrm{lin}}_t,x_\beta)=\sum_{\alpha\in A}  M^{(t)}_{\beta\alpha}[\hat K_{\Theta_0}] f(\Theta_0,x_\alpha)+\left( R^{(t)}[\hat K_{\Theta_0}]\right)^TY, \label{form0}\]
where the entries $ M^{(t)}_{\beta\alpha}[\hat K_{\Theta_0}]$ and the components of $ R^{(t)}[\hat K_{\Theta_0}]$ are continous functions of the elements of matrix of the empirical NTK
\[\{\hat K_{\Theta_0}(x_\alpha,x_{\alpha'})\}_{\alpha,\alpha'\in A}.\]
By continuity, the (finite) matrix $ M^{(t)}_{\beta\alpha}[\hat K_{\Theta_0}]$ and the (finite) vector $ R^{(t)}[\hat K_{\Theta_0}]$ converge in probability to $ M^{(t)}_{\beta\alpha}[\bar K]$ and $ R^{(t)}[\bar K]$:
\[ M^{(t)}_{\beta\alpha}[\hat K_{\Theta_0}]\xrightarrow{p} M^{(t)}_{\beta\alpha}[\bar K],\qquad\qquad  R^{(t)}[\hat K_{\Theta_0}]\xrightarrow{p} R^{(t)}[\bar K].\]
By Slutsky's theorem \ref{sl}, we conclude that
\[ \{f^{\mathrm{lin}}(\Theta^{\mathrm{lin}}_t,x_\beta)\}_{x_\beta\in\mathcal{F}}\xrightarrow{d} \left\{\sum_{\alpha\in A}  M^{(t)}_{\beta\alpha}[\bar K] f^{(\infty)}(x_\alpha)+\left( R^{(t)}[\bar K]\right)^TY\right\}_{x_\beta\in\mathcal{F}} \text{as}\quad m\to\infty,\]
i.e.
\[ f^{\mathrm{lin}}(\Theta^{\mathrm{lin}}_t,\,\cdot\,)\big|_{\mathcal{F}}\xrightarrow{d}f^{(\infty)}(\,\cdot\,)\big|_{\mathcal{F}}-\bar K(\,\cdot\,,X^T)\big|_{\mathcal{F}}\bar K^{-1}\left(\id-e^{-\eta_0 \bar K t}\right)(f^{(\infty)}(X)-Y). \]
Since we assumed $\mathcal{X}$ to be finite, this is enough to prove the convergence of the distribution to the entire Gaussian process: it is sufficient to choose $\mathcal{F}=\mathcal{X}$.
\begin{remark}
    In \autoref{infinite} we will generalize the convergence of $\{f(\Theta_t,x)\}_{x\in\mathcal{X}}$ to a Gaussian process for the case of $\mathcal{X}$ being infinite. This will not require to prove that also $\{f^{\mathrm{lin}}(\Theta_t,x)\}_{x\in\mathcal{X}}$ converges to a Gaussian process: we will only need the convergence of $\{f(\Theta_t,x)\}_{x\in\mathcal{F}}$ for any $\mathcal{F}$ finite set of inputs, which, as we will see, is a corollary of the convergence of the linearized model for a finite number of inputs. Therefore, the proof given above will be enough for our purposes.
\end{remark}

\paragraph{Proof of Corollary \ref{corgp}}
In the limit $m\to \infty$, the solution (\ref{sol-limit}) is a linear combination of the Gaussian processes  $f^{\mathrm{lin}}(\Theta_0,x)=f(\Theta_0,x)$ and $F(0)$, so it is a Gaussian process as well, with
\begin{align}
\nonumber\mu_t(x)&=\mathbb{E}\left[f^{(\infty)}_t(x)\right]\\
\nonumber&=\mathbb{E}\left[f^{(\infty)}(x)\right]-\bar K(x,X^T)\bar K^{-1}\left(\id-e^{-\eta_0\bar Kt} \right)(\mathbb{E}\left[F^{(\infty)}\right]-Y)\\
&=\bar K(x,X^T)\bar K^{-1}\left(\id-e^{-\eta_0\bar Kt} \right)Y,\\
\nonumber\mathcal{K}_t(x,x')&=\mathbb{E}\left[\left(f^{(\infty)}_t(x)-\mu_t(x)\right)\left(f^{(\infty)}_t(x')-\mu_t(x')\right)\right]\\
\nonumber&=\mathbb{E}\Big[\left(f^{(\infty)}(x)-\bar K(x,X^T)\bar K^{-1}\left(\id-e^{-\eta_0\bar Kt} \right)F^{(\infty)}\right)\times\\
\nonumber&\phantom{=\mathbb{E}\Big[}\times\left(f^{(\infty)}(x')-\bar K(x',X^T)\bar K^{-1}\left(\id-e^{-\eta_0\bar Kt} \right)F^{(\infty)} \right)\Big]\\
\nonumber&=\mathcal{K}_0(x,x')\\
\nonumber&\phantom{=}- \bar K(x,X^T)\bar K^{-1}\left(\id-e^{-\eta_0\bar K t}\right) \mathcal{K}_0(X,x')\\
\nonumber&\phantom{=}-\bar K(x',X^T)\bar K^{-1}\left(\id-e^{-\eta_0\bar K t}\right) \mathcal{K}_0(X,x) \\
&\phantom{=}+\bar K(x,X^T)\bar K^{-1}\left(\id-e^{-\eta_0\bar K t}\right)\mathcal{K}_0(X,X^T)\left(\id-e^{-\eta_0\bar K t}\right)\bar K^{-1} K(X,x').
\end{align}

\subsection{Proof of Theorem \ref{gradfl}}\label{proofgradfl}

Because of Corollary \ref{corollaryR}, we can ask
\begin{align}\label{eq:unionbound}
\mathbb{P}\left(\|F(0)-Y\|_2<R\right)\geq 1-\frac{\delta}{2}\quad \forall \,m\geq m_0 \text{ for some } R=\sqrt{n\log(2n)} R_0.
\end{align}
In order to simplify the notation, we introduce
\[\rho(m)=\frac{6\sqrt n R}{N_K(m)\lambda^K_{\min}}\frac{|\mathcal{M}|}{N(m)}.\]
Let $B_r(\Theta_0)=\{\Theta: \|\Theta-\Theta_0\|_\infty< r\}$ be the ball of center $\Theta_0$ and radius $r$. By Lemma \ref{lambdamin}, there exists $\bar m\in\mathbb{N}$ such that
\begin{align}
\hat K_{\Theta_0}(X,X^T)\succ\frac{\lambda_{\min}^K}{2}\id \qquad \forall m\geq \bar m
\end{align}
with probability at least $1-\frac{\delta}{2}$. 
By uniform continuity\footnote{The uniform continuity is ensured by the fact that the derivatives of the entries of $\hat K_\Theta(X,X^T)$ with respect to the parameters are bounded and the dimension of $\hat K_\Theta(X,X^T)$ is fixed.} of $\Theta\mapsto\hat K_\Theta(X,X^T)$, if $\rho(m)$ is small enough, i.e., if $m$ is larger than $m_1\in\mathbb{N}$\footnote{By uniform continuity of $\Theta\mapsto\hat K_\Theta(X,X^T)$, $m_1$ does not depend on $\Theta_0$.} (we take $m_1\geq \bar m$),
\begin{align}
\hat K_\Theta(X,X^T)\succ\frac{\lambda_{\min}^K}{3}\id \qquad \forall \,\Theta\in B_{\rho(m)}(\Theta_0)\qquad \forall m\geq m_1.
\end{align}
Let 
\begin{align}
t_1=\inf\left\{t:\|\Theta_t-\Theta_0\|_\infty\geq \rho(m)\right\}.
\label{t1}
\end{align}
For $t\leq t_1$ we have $\hat K_{\Theta_t}\succ\frac{\lambda_{\min}^K}{3}\id$. Recalling that
\[
\frac{d}{dt} f(\Theta_t,x)=-\eta_0\hat K_{\Theta_t}(x,X^T)\cdot \left(F(t)-Y\right),
\]
we have, \modifica{with probability at least $1-\frac{\delta}{2}$},
\begin{align}
\nonumber \frac{d}{dt}\|F(t)-Y\|_2^2&=-2\eta_0(F(t)-Y)^T\hat K_{\Theta_t}(F(t)-Y)\\
&\leq -\frac{2}{3}\eta_0\lambda^K_{{\min}}\|F(t)-Y\|_2^2.
\end{align}
\modifica{Therefore,}
\begin{align}
\nonumber \|F(t)-Y\|_2^2&\leq e^{-\frac{2}{3}\eta_0\lambda^K_{\min}t}\|F(0)-Y\|_2^2\\
&\leq e^{-\frac{2}{3}\eta_0\lambda^K_{\min}t}R^2, \label{disug2}
\end{align}
\modifica{with probability at least $1-\delta$ by the union bound with \eqref{eq:unionbound}.}
Recalling also that
\[ 
\dot \Theta_t = -\frac{\eta_0}{N_K(m)} \nabla_\Theta f(\Theta_t,X)\cdot (F(t)-Y)
\]
and using Lemma \ref{lemma}
\begin{align}
\nonumber \frac{d}{dt}|\theta_i(t)-\theta_i(0)|&\leq \Big|\frac{d}{dt}\theta_i(t)\Big|=\Big|\frac{\eta_0}{N_K(m)}\partial_{\theta_i}f(\Theta,X)\cdot (F(t)-Y)\Big|\\
\nonumber &\leq \frac{\eta_0}{N_K(m)}\|\partial_{\theta_i}f(\Theta,X)\|_2\|F(t)-Y\|_2\\
 &\leq \frac{\eta_0}{N_K(m)}2\sqrt n\,\frac{|\mathcal{M}|}{N(m)}\,Re^{-\frac{1}{3}\eta_0\lambda^K_{\min}t},\\ 
 \to |\theta_i(t)-\theta_i(0)|&\leq\frac{6\sqrt n R}{N_K(m)\lambda^K_{\min}}\frac{|\mathcal{M}|}{N(m)}\left(1-e^{-\frac{1}{3}\eta_0\lambda^K_{\min}t}\right)\qquad &\forall\,t\leq t_1,\\
\to \|\Theta_t-\Theta_0\|_\infty&\leq\rho(m)\left(1-e^{-\frac{1}{3}\eta_0\lambda^K_{\min}t}\right)\qquad &\forall\,t\leq t_1,
\label{disug3}
\end{align}
\modifica{with probability at least $1-\delta$.}
If $t_1<\infty$, then  
\begin{align}
\|\Theta_{t_1}-\Theta_0\|_\infty&\leq\rho(m)\left(1-e^{-\frac{1}{3}\eta_0\lambda^K_{\min}t_1}\right)<\rho(m)\qquad \forall\,t\leq t_1,
\end{align}
but this contradicts the definition (\ref{t1}) of $t_1$, so we must have $t_1=\infty$. This implies that (\ref{disug2}) and (\ref{disug3}) holds for any $t>0$ and $m\geq \bar m =\max\{m_0,m_1\}$ \modifica{with probability at least $1-\delta$}, so, if we set $R_1=6R_0$, we have proved (\ref{grad1}) and (\ref{grad2}). 

\subsection{Proof of Corollary \ref{freezntk} and Theorem \ref{gronwall}}\label{prooffreezntk}

\subsubsection{Proof of Corollary \ref{freezntk}}
By Lemma \ref{kernel} and Theorem \ref{gradfl},
\begin{align}
\nonumber\sup_t|\hat K_{\Theta_t}(x,x')&-\hat K_{\Theta_0}(x,x')|\\
\nonumber&\leq 16\frac{\Sigma_1|\mathcal{M}|^2|\mathcal{N}|}{N_K(m)N^2(m)}\sup_t\|\Theta_0-\Theta\|_\infty\\
\nonumber&=96\frac{R_0}{\lambda_{\min}^K}n\sqrt{\log(2n)}\frac{\Sigma_1|\mathcal{M}|^3|\mathcal{N}|}{N_K^2(m)N^3(m)}
\\
&=\frac{R_2}{\lambda_{\min}^K}\,n\sqrt{\log(2n)}\,\frac{\Sigma_1|\mathcal{M}|^3|\mathcal{N}|}{N_K^2(m)N^3(m)},
\quad \text{where} \quad R_2=96R_0,
\end{align}
with probability at least $1-\delta$.

\paragraph{Proof of Theorem \ref{gronwall}}
We simply combine the discrepancy estimate of Theorem \ref{powerful} with the lazy training result of Theorem \ref{gradfl}:
\begin{align}
\nonumber\sup_t|f(\Theta_t,x)-f^{\mathrm{lin}}(\Theta_t,x)|&\leq \frac{Lm |\mathcal{M}|^2|\mathcal{N}|}{N(m)}\|\Theta_t-\Theta_0\|^2_\infty\\
\nonumber&\leq 
\frac{Lm |\mathcal{M}|^2|\mathcal{N}|}{N(m)}
\left(\frac{R_1}{\lambda^K_{\min}}n\sqrt{\log(2n)}\frac{|\mathcal{M}|}{N_K(m)N(m)}\right)^2\\
&\leq \left(\frac{R_1}{\lambda^K_{\min}}\right)^2n^2\log(2n)\frac{Lm|\mathcal{M}|^4|\mathcal{N}|}{N_K^2(m)N^3(m)}.
\end{align}

\subsection{Proof of Lemma \ref{secondorder}}\label{proofsecondorder}
In order to prove Lemma \ref{secondorder}, we need the following preliminary statement.
\begin{lemma}[Bounding the discrepancy on the examples]\label{improvedqlt} For any $\delta>0$ there exists $m_0$ such that, for any $m\geq m_0$
\[ || F(t)-F^{\mathrm{lin}}(t)||_2\leq \frac{C}{\left(\lambda_{\min}^K\right)^2}n^2\sqrt n \log(2n) \frac{Lm|\mathcal{M}|^4|\mathcal{N}|}{N^2_K(m)N^3(m)} \left(1-e^{-\frac{1}{3}\eta_0\lambda_{\min}^Kt}\right)\]
with probability at least $1-\delta$.
\end{lemma}
\begin{proof}
We follow the strategy of \cite{QLazy}, using our results to improve the final bound. Let us define
\[\Delta(t)=\|F(t)-F^{\mathrm{lin}}(t)\|_2.\]
\modifica{and, recalling that
\begin{align}
        \frac{d}{dt}f(\Theta_t,x^{(i)})=\hat K_{\Theta_t}(x^{(i)},X^T)(f(\Theta_t,X)-Y)\quad\text{and}\quad \frac{d}{dt}f^{\mathrm{lin}}(\Theta_t,x^{(i)})=\hat K_{\Theta_0}(x^{(i)},X^T)(f^{\mathrm{lin}}(\Theta_t,X)-Y), 
\end{align}
let us compute}
\begin{align}
    \nonumber
    \frac{1}{2}\frac{d}{dt}\Delta^2(t)&=\sum_{i=1}^n\frac{1}{2}\frac{d}{dt}\left(f(\Theta_t,x^{(i)})-f^{\mathrm{lin}}(\Theta^{\mathrm{lin}}_t,x^{(i)})\right)^2\\
    \nonumber
    &= \sum_{i=1}^n  \left(f(\Theta_t,x^{(i)})-f^{\mathrm{lin}}(\Theta^{\mathrm{lin}}_t,x^{(i)})\right)\left(\frac{d}{dt}f(\Theta_t,x^{(i)})-\frac{d}{dt}f^{\mathrm{lin}}(\Theta^{\mathrm{lin}}_t,x^{(i)})\right)\\
    \nonumber
    &= -\eta_0\sum_{i=1}^n \left(f(\Theta_t,x^{(i)})-f^{\mathrm{lin}}(\Theta^{\mathrm{lin}}_t,x^{(i)})\right) \times\\
    \nonumber
    & \qquad \times \left(\hat K_{\Theta_t}(x^{(i)},X^T)(f(\Theta_t,X)-Y)-\hat K_{\Theta_0}(x^{(i)},X^T)(f(\Theta^{\mathrm{lin}}_t,X)-Y)\right)\\
    \nonumber
    &= -\eta_0(F(t)-F^{\mathrm{lin}}(t))^T\hat K_{\Theta_t}(F(t)-Y)+\eta_0 (F(t)-F^{\mathrm{lin}}(t))^T\hat K_{\Theta_0}(F^{\mathrm{lin}}(t)-Y)\\
    \nonumber
    &= -\eta_0(F(t)-F^{\mathrm{lin}}(t))^T\hat K_{\Theta_t}(F(t)-Y)\\
    \nonumber
    & \quad -\eta_0 (F(t)-F^{\mathrm{lin}}(t))^T\hat K_{\Theta_0}(F(t)-F^{\mathrm{lin}}(t))\\
    &\quad +\eta_0 (F(t)-F^{\mathrm{lin}}(t))^T\hat K_{\Theta_0}(F(t)-Y).
    \label{ultimaeq}
\end{align}
Noticing that $\hat K_{\Theta_0}$ is positive semidefinite,
\begin{align}
    -\eta_0 (F(t)-F^{\mathrm{lin}}(t))^T\hat K_{\Theta_0}(F(t)-F^{\mathrm{lin}}(t))\leq 0,
\end{align}
so (\ref{ultimaeq}) reads
\begin{align}
    \nonumber
     \frac{1}{2}\frac{d}{dt}\Delta^2(t)&\leq -\eta_0(F(t)-F^{\mathrm{lin}}(t))^T\hat K_{\Theta_t}(F(t)-Y)+\eta_0 (F(t)-F^{\mathrm{lin}}(t))^T\hat K_{\Theta_0}(F(t)-Y)\\
     &=-\eta_0(F(t)-F^{\mathrm{lin}}(t))^T(\hat K_{\Theta_t}-\hat K_{\Theta_0})(F(t)-Y)
\end{align}
whence
\begin{align}
    \nonumber
     \left|\Delta(t)\frac{d}{dt}\Delta(t)\right|&\leq \eta_0\|F(t)-F^{\mathrm{lin}}(t)\|_2\|\hat K_{\Theta_t}-\hat K_{\Theta_0}\|_\mathcal{L}\|F(t)-Y\|_2\\
     &= \eta_0 \Delta(t)\|\hat K_{\Theta_t}-\hat K_{\Theta_0}\|_\mathcal{L}\|F(t)-Y\|_2.
\end{align}
Therefore,
\begin{align} \left|\frac{d}{dt}\Delta(t)\right|&\leq \eta_0\|\hat K_{\Theta_t}-\hat K_{\Theta_0}\|_{\mathcal{L}}||F(t)-Y||_2.
\end{align}
By (\ref{grad1})
\[\|F(t)-Y\|_2\leq R_0\sqrt{n\log(2n)}e^{-\frac{1}{3}\eta_0\lambda_{\min}^Kt},\]
and by Corollary \ref{freezntk}
\begin{align*}\|\hat K_{\Theta_t}-\hat K_{\Theta_0}\|_\mathcal{L}&\leq \|\hat K_{\Theta_t}-\hat K_{\Theta_0}\|_F\\
&\leq \frac{R_2}{\lambda_{\min}^K}n^2\sqrt{\log(2n)}\frac{Lm|\mathcal{M}|^4|\mathcal{N}|}{N^2_K(m)N^3(m)},
\end{align*}
so that
\[|\partial_t\Delta(t)|\leq \frac{C}{3}\frac{\eta_0}{\lambda_{\min}^K}n^2\sqrt n \log(2n) \frac{Lm|\mathcal{M}|^4|\mathcal{N}|}{N^2_K(m)N^3(m)}e^{-\frac{1}{3}\eta_0\lambda_{\min}^Kt}\qquad C=3R_0R_2\]
whence
\[\Delta(t)\leq \frac{C}{\left(\lambda_{\min}^K\right)^2}n^2\sqrt n \log(2n) \frac{Lm|\mathcal{M}|^4|\mathcal{N}|}{N^2_K(m)N^3(m)} \left(1-e^{-\frac{1}{3}\eta_0\lambda_{\min}^Kt}\right)\]
\end{proof}
Now we have all the ingredients to prove Lemma \ref{secondorder}.
We adapt to our case the strategy of \cite{CB18}. Let us compute
\begin{align}
\nonumber
\big\|\dot\Theta_t-\dot\Theta_t^{\mathrm{lin}}\big\|_\infty&=\frac{\eta}{n}\left\|\nabla_\Theta f(\Theta_t,X^T)(F(t)-Y)-\nabla_\Theta f(\Theta_0,X^T)(F^{\mathrm{lin}}(t)-Y)\right\|_\infty\\
\nonumber
&\leq \frac{\eta}{n}\left\|\left(\nabla_\Theta f(\Theta_t,X^T)-\nabla_\Theta f(\Theta_0,X^T)\right)(F(t)-Y)\right\|_\infty\\
\nonumber
&\qquad+\frac{\eta}{n}\left\|\nabla_\Theta f(\Theta_0,X^T)(F^{\mathrm{lin}}(t)-F(t))\right\|_\infty\\
\nonumber
&\leq \frac{\eta}{n}\sup_i\|\partial_{\theta_i} f(\Theta_t,X)-\partial_{\theta_i} f(\Theta_0,X)\|_2\|F(t)-Y\|_2\\
&\qquad+\frac{\eta}{n}\sup_i\|\partial_{\theta_i} f(\Theta_0,X)\|_2\|F^{\mathrm{lin}}(t)-F(t)\|_2
\end{align}
Let us bound the previous expression term by term. Combining the Lipschitzness result (\ref{eq4.88}) with the lazy training bound (\ref{grad2}), and using the convergence to the examples (\ref{grad1}), we control the first term:
\begin{align}
\nonumber
\frac{\eta}{n}\sup_i\|\partial_{\theta_i} f(\Theta_t,X)&-\partial_{\theta_i} f(\Theta_0,X)\|_2\|F(t)-Y\|_2\\
\nonumber
&\leq \frac{\eta}{n}\sqrt n\, 4\frac{|\mathcal{M}|^2|\mathcal{N}|}{N(m)}\|\Theta_t-\Theta_0\|_\infty \|F(t)-Y\|_2\\
\nonumber
&\leq \frac{\eta}{n}\sqrt n\, 4\frac{|\mathcal{M}|^2|\mathcal{N}|}{N(m)}\frac{R_1}{\lambda_{\min}^K}n\sqrt{\log(2n)}\frac{|\mathcal{M}|}{N_K(m)N(m)} \|F(t)-Y\|_2\\
\nonumber
& \leq \frac{4\eta R_1}{\lambda_{\min}^K}\sqrt{n\log(2n)}\,\frac{|\mathcal{M}|^3|\mathcal{N}|}{N_K(m)N^2(m)} \|F(t)-Y\|_2\\
\nonumber
& \leq \frac{4\eta R_1}{\lambda_{\min}^K}\sqrt{n\log(2n)}\frac{|\mathcal{M}|^3|\mathcal{N}|}{N_K(m)N^2(m)} R_0\sqrt{n\log(2n)}e^{-\frac{1}{3}\eta_0\lambda_{\min}^Kt}\\
\nonumber
&= \frac{4\eta R_0R_1}{\lambda_{\min}^K}n\log(2n)\frac{|\mathcal{M}|^3|\mathcal{N}|}{N_K(m)N^2(m)}e^{-\frac{1}{3}\eta_0\lambda_{\min}^Kt}\\
&= \frac{4\eta_0 R_0R_1}{\lambda_{\min}^K}n^2\log(2n)\frac{|\mathcal{M}|^3|\mathcal{N}|}{N_K^2(m)N^2(m)}e^{-\frac{1}{3}\eta_0\lambda_{\min}^Kt}=:A(t).
\end{align}
Regarding the second term, we need two different estimates to be used for ``small'' and ``large'' $t$, as we will show soon. The first estimate is based on the Lipschitzness of the gradient (\ref{lemma2}) and on Lemma \ref{improvedqlt}:
\begin{align*}
\frac{\eta}{n}\sup_i\|\partial_{\theta_i} &f(\Theta_0,X)\|_2\|F^{\mathrm{lin}}(t)-F(t)\|_2\\
&\leq \frac{\eta}{n}2\sqrt n \frac{|\mathcal{M}|}{N(m)}\|F^{\mathrm{lin}}(t)-F(t)\|_2\\
& \leq \frac{\eta}{n}2\sqrt n \frac{|\mathcal{M}|}{N(m)} \frac{C}{\left(\lambda_{\min}^K\right)^2}n^2\sqrt n \log(2n) \frac{Lm|\mathcal{M}|^4|\mathcal{N}|}{N^2_K(m)N^3(m)} \left(1-e^{-\frac{1}{3}\eta_0\lambda_{\min}^Kt}\right)\\
&\leq \frac{2\eta C}{\left(\lambda_{\min}^K\right)^2} n^2\log(2n)\frac{Lm|\mathcal{M}|^5|\mathcal{N}|}{N^2_K(m)N^4(m)}\\
&= \frac{2\eta_0 C}{\left(\lambda_{\min}^K\right)^2} n^3\log(2n)\frac{Lm|\mathcal{M}|^5|\mathcal{N}|}{N^3_K(m)N^4(m)}=:B(t).
\end{align*}
The second estimate exploits again (\ref{lemma2}) and the convergence to the examples for the original model (\ref{grad1}); we also need to quantify the convergence to the examples for the linearized model; this immediately follows from the analytic solution (\ref{convergenzaesempi}) and \modifica{from Theorem} \ref{gradfl}:
\begin{align}
\nonumber
\frac{\eta}{n}&\sup_i\|\partial_{\theta_i} f(\Theta_0,X)\|_2\|F^{\mathrm{lin}}(t)-F(t)\|_2\\
\nonumber
&\leq \frac{\eta}{n}2\sqrt n \frac{|\mathcal{M}|}{N(m)}\|F^{\mathrm{lin}}(t)-F(t)\|_2\\
\nonumber
&\leq \frac{2\eta}{\sqrt n} \frac{|\mathcal{M}|}{N(m)}\left(\|F^{\mathrm{lin}}(t)-Y\|_2+\|F(t)-Y\|_2\right)\\
\nonumber
&\leq \frac{2\eta}{\sqrt n} \frac{|\mathcal{M}|}{N(m)}\left(\big\|\big(e^{-\eta_0 t \hat K_{\Theta_0}}(F(0)-Y)+Y\big)-Y\big\|_2+R_0\sqrt{n\log(2n)}e^{-\frac{1}{3}\eta_0\lambda_{\min}^Kt}\right)\\
\nonumber
&\leq \frac{2\eta}{\sqrt n} \frac{|\mathcal{M}|}{N(m)}\left(e^{-\eta_0\lambda_{\min} t}\|F(0)-Y\|_2+R_0\sqrt{n\log(2n)}e^{-\frac{1}{3}\eta_0\lambda_{\min}^Kt}\right)\\
\nonumber
&\leq \frac{2\eta}{\sqrt n} \frac{|\mathcal{M}|}{N(m)}R_0\sqrt{n\log(2n)}\left(e^{-\eta_0\lambda_{\min} t}+e^{-\frac{1}{3}\eta_0\lambda_{\min}^Kt}\right)\\
\nonumber
&\leq 4R_0\eta \sqrt{\log(2n)} \frac{|\mathcal{M}|}{N(m)}e^{-\frac{1}{3}\eta_0\lambda_{\min}^Kt}\\
&\leq 4R_0\eta_0 \sqrt{\log(2n)} \frac{|\mathcal{M}|}{N(m)}e^{-\frac{1}{3}\eta_0\lambda_{\min}^Kt}=:C(t).
\end{align}
Hence
\[\big\|\dot\Theta_t-\dot\Theta_t^{\mathrm{lin}}\big\|_\infty\leq A(t)+B(t) \quad \text{\modifica{and}}\quad \big\|\dot\Theta_t-\dot\Theta_t^{\mathrm{lin}}\big\|_\infty\leq A(t)+C(t).\]
Defining
\[t^\ast = \frac{3}{\eta_0 \lambda_{\min}^K}\log N(m),\]
we integrate
\begin{align}
\nonumber
\big\|\Theta_t-\Theta_t^{\mathrm{lin}}\big\|_\infty&\leq \int_0^\infty A(t)dt + \int_0^{t^\ast} B(t)dt + \int_{t^\ast}^\infty C(t)dt\\
\nonumber
&= \frac{12R_0R_1}{(\lambda_{\min}^K)^2}n^2\log(2n)\frac{|\mathcal{M}|^3|\mathcal{N}|}{N_K^2(m)N^2(m)}\\
\nonumber
&\qquad + \frac{2\eta_0 C}{\left(\lambda_{\min}^K\right)^2} n^3\log(2n)\frac{Lm|\mathcal{M}|^5|\mathcal{N}|}{N^3_K(m)N^4(m)}\frac{3}{\eta_0 \lambda_{\min}^K}\log N(m)\\
&\qquad + \frac{12 R_0}{\lambda_{\min}^K} \sqrt{\log(2n)} \frac{|\mathcal{M}|}{N(m)}e^{-\log N(m)}
\end{align}
Using the bound of Lemma \ref{Nmax}
\[\frac{m|\mathcal{M}||\mathcal{N}|}{N^2(m)}\geq 1,\]
\modifica{
we have
\begin{align}
\big\|\Theta_t-\Theta_t^{\mathrm{lin}}\big\|_\infty\leq \left(\frac{6C}{(\lambda_{\min}^K)^3}+\frac{12R_0R_1}{(\lambda_{\min}^K)^2}+\frac{12R_0}{\lambda_{\min}^K}\right) n^3\log(2n) \frac{Lm|\mathcal{M}|^5|\mathcal{N}|^2}{N^4(m)}\log N(m)
\end{align}
and with an upper bound of the form
\begin{align}
    ax^3+bx^2+cx\leq C_1x^3+C_2
\end{align}
for suitable $C_1$ and $C_2$, where, in our case, $x=1/\lambda_{\min}^K, a=6C, b=12R_0R_1$ and  $c=12R_0$, we conclude}
\begin{align}
\big\|\Theta_t-\Theta_t^{\mathrm{lin}}\big\|_\infty&\leq \left(\frac{C_1}{(\lambda_{\min}^K)^3}+C_2\right) n^3\log(2n) \frac{Lm|\mathcal{M}|^5|\mathcal{N}|^2}{N^4(m)}\log N(m).
\end{align}

\subsection{Proof of Theorem \ref{cfevolution}}\label{proofcfevolution}
Using the result of Theorem \ref{powerful} and (\ref{lemma3})
\begin{align}
\nonumber
|f(\Theta_t,x)&-f^{\mathrm{lin}}(\Theta_t^{\mathrm{lin}},x)|\\
\nonumber
&=|f(\Theta_t,x)-f(\Theta_0,x)-\nabla_\Theta f(\Theta_0,x)^T(\Theta_t^{\mathrm{lin}}-\Theta_0)|\\
\nonumber
&= |f(\Theta_t,x)-f(\Theta_0,x)-\nabla_\Theta f(\Theta_0,x)^T(\Theta_t-\Theta_0)-\nabla_\Theta f(\Theta_0,x)^T(\Theta_t^{\mathrm{lin}}-\Theta_t)|\\
\nonumber
&\leq |f(\Theta_t,x)-f^{\mathrm{lin}}(\Theta_t,x)|+|\nabla_\Theta f(\Theta_0,x)^T(\Theta_t^{\mathrm{lin}}-\Theta_t)|\\
\nonumber
&\leq \frac{Lm|\mathcal{M}|^2|\mathcal{N}|}{N(m)}\|\Theta_t-\Theta_0\|_\infty^2+\|\nabla_\Theta f(\Theta_0,x)\|_1\|\Theta_t^{\mathrm{lin}}-\Theta_t\|_\infty\\
&\leq \frac{Lm|\mathcal{M}|^2|\mathcal{N}|}{N(m)}\|\Theta_t-\Theta_0\|_\infty^2+2\frac{Lm|\mathcal{M}|}{N(m)}\|\Theta_t^{\mathrm{lin}}-\Theta_t\|_\infty.
\end{align}
Now we use Lemma \ref{secondorder} and the lazy training inequality (\ref{grad2}):
\begin{align}
\nonumber
|f(\Theta_t,x)&-f^{\mathrm{lin}}(\Theta_t^{\mathrm{lin}},x)|\\
\nonumber
&\leq \frac{Lm|\mathcal{M}|^2|\mathcal{N}|}{N(m)}\left(\frac{R_1^2}{(\lambda_{\min}^K)^2}n^2\log(2n)\frac{|\mathcal{M}|^2}{N_K^2(m)N^2(m)}\right)\\
\nonumber
&\qquad+2\frac{Lm|\mathcal{M}|}{N(m)}\left(\frac{C_1}{(\lambda_{\min}^K)^3}+C_2\right) n^3\log(2n) \frac{Lm|\mathcal{M}|^5|\mathcal{N}|^2}{N^4(m)}\log N(m)\\
\nonumber
&\leq \frac{R_1^2}{(\lambda_{\min}^K)^2}n^2\log(2n)\frac{Lm|\mathcal{M}|^4|\mathcal{N}|}{N_K^2(m)N^3(m)}\\
&\qquad+2\left(\frac{C_1}{(\lambda_{\min}^K)^3}+C_2\right) n^3\log(2n) \frac{L^2m^2|\mathcal{M}|^6|\mathcal{N}|^2}{N^5(m)}\log N(m).
\end{align}
Therefore, is easy to show that
\begin{align}
|f(\Theta_t,x)&-f^{\mathrm{lin}}(\Theta_t^{\mathrm{lin}},x)|\leq \left(\frac{C_3}{(\lambda_{\min}^K)^3}+C_4\right) n^3\log(2n) \frac{L^2m^2|\mathcal{M}|^6|\mathcal{N}|^2}{N^5(m)}\log N(m).
\end{align}

\subsection{Proof of Theorem \ref{qnngp}}\label{proofqnngp}
\begin{lemma}\label{checkhp} The hypothesis 
\[\lim_{m\to\infty}\frac{L^2m^2|\mathcal{M}|^6|\mathcal{N}|^3}{N^5(m)}\log N(m)=0,\label{qnn1}\]
of 
Theorem \ref{qnngp} ensures that the hypotheses of Theorem \ref{init} and of Theorem \ref{ntkconv} are satisfied, i.e.
\begin{align}
\lim_{m\to\infty}\frac{m|\mathcal{M}|^2|\mathcal{N}|^2}{N^3(m)}=0\qquad \text{and}\qquad \lim_{m\to\infty}\frac{1}{N_K^2(m)}\,\frac{\Sigma_2|\mathcal{M}|^2|\mathcal{N}|^2}{N^4(m)}=0.
\end{align}
\modifica{Furthermore,}
\[\lim_{m\to\infty}\frac{Lm|\mathcal{M}|^4|\mathcal{N}|}{N_K^2(m)N^3(m)}=0,\]
i.e., the bound of Theorem \ref{powerful} is nontrivial in the limit $m\to\infty$.
\end{lemma}
\begin{proof}
We will use that $|\mathcal{M}|,|\mathcal{N}|,$ and $L$ are all bounded from below by $1$, and $N_K(m)$ is asymptotically bounded by (\ref{boundnormeq}) of Lemma \ref{boundnorm}, and we will use the inequality of Lemma \ref{Nmax} in the form
\[1\leq c\frac{m|\mathcal{M}||\mathcal{N}|}{N^2(m)}\]
so that
\[0\leq\lim_{m\to\infty}\frac{m|\mathcal{M}|^2|\mathcal{N}|^2}{N^3(m)}\leq c\lim_{m\to\infty}\frac{m^2|\mathcal{M}|^3|\mathcal{N}|^3}{N^5(m)}\leq c\lim_{m\to\infty}\frac{L^2m^2|\mathcal{M}|^6|\mathcal{N}|^3}{N^5(m)}\log N(m)=0, \]
\[ 0\leq \lim_{m\to\infty}\frac{Lm|\mathcal{M}|^4|\mathcal{N}|}{N_K^2(m)N^3(m)}\leq c\lim_{m\to\infty}\frac{Lm^2|\mathcal{M}|^5|\mathcal{N}|^2}{N_K^2(m)N^5(m)}\leq c\lim_{m\to\infty}\frac{L^2m^2|\mathcal{M}|^6|\mathcal{N}|^3}{N^5(m)}\log N(m)=0.\]
Now, noticing that (\ref{qnn1}) implies that $\lim_{m\to\infty} N(m)=\infty$,
we conclude that
\begin{align} 
\nonumber
0\leq \lim_{m\to\infty}\frac{1}{N_K^2(m)}\,\frac{\Sigma_2|\mathcal{M}|^2|\mathcal{N}|^2}{N^4(m)}&\leq 
c\lim_{m\to\infty}\frac{1}{N_K^2(m)}\,\frac{\Sigma_2m|\mathcal{M}|^3|\mathcal{N}|^3}{N^6(m)}\\
&\leq c\lim_{m\to\infty}\frac{1}{N(m)}\frac{L^2m^2|\mathcal{M}|^6|\mathcal{N}|^3}{N^5(m)}=0.
\end{align}
\end{proof}

Now we are ready to prove Theorem \ref{qnngp}.\\
Let $\bar X\in\mathbb{R}^N$ be any (finite dimensional) vector on $N$ distinct inputs $\{\bar x_1,\dots, \bar x_n\}\in \mathcal{X}$ and let $\Delta_t(\bar X)$ be the random vector
\[\Delta_t(\bar X)= f(\Theta_t,\bar X)-f^{\mathrm{lin}}(\Theta_t,\bar X).\]
We show that $\Delta_t(\bar X)\xrightarrow{p}0$ as $m\to\infty$. Let $\epsilon,\delta>0$.
By Theorem \ref{cfevolution}, there exists $m_0$ such that, for any $m\geq m_0$,
\begin{align}
\nonumber
\mathbb{P}\Bigg(&\sup_{t'}\sup_{x\in\mathcal{X}}|f(\Theta_{t'},x)-f^{\mathrm{lin}}(\Theta^{\mathrm{lin}}_{t'},x)|\\
&\qquad\qquad\leq \left(\frac{C_3}{(\lambda_{\min}^K)^3}+C_4\right) n^3\log(2n) \frac{L^2m^2|\mathcal{M}|^6|\mathcal{N}|^2}{N^5(m)}\log N(m)\Bigg)\geq 1-\delta.
\end{align}
whence
\begin{align}
\nonumber
\mathbb{P}\Bigg(&\|f(\Theta_t,\bar X)-f^{\mathrm{lin}}(\Theta^{\mathrm{lin}}_t,\bar X)\|_2\\
&\qquad\leq \sqrt N\left(\frac{C_3}{(\lambda_{\min}^K)^3}+C_4\right) n^3\log(2n) \frac{L^2m^2|\mathcal{M}|^6|\mathcal{N}|^2}{N^5(m)}\log N(m)\Bigg)\geq 1-\delta.
\end{align}
There is also $m_1\in\mathbb{N}$ such that 
\[ \sqrt{N}\left(\frac{C_3}{(\lambda_{\min}^K)^3}+C_4\right) n^3\log(2n) \frac{L^2m^2|\mathcal{M}|^6|\mathcal{N}|^2}{N^5(m)}\log N(m)<\epsilon\qquad \forall\,m\geq m_1. \label{finito}\]
So, for any $m\geq\max \{m_0,m_1\}$,
\[
\mathbb{P}\left(\|\Delta_t(\bar X)\|_2 <\epsilon \right)\geq 1-\delta,
\]
so \[\Delta_t(\bar X)\xrightarrow{p}0\quad\text{as}\quad m\to\infty. \]
By Corollary \ref{corgp}, $\{f^{\mathrm{lin}}(\Theta^{\mathrm{lin}}_t,x)\}_{x\in\{\bar x_1,\dots, \bar x_n\}}$ converges in distribution to a Gaussian process with mean $\mu_t$ and covariance $\mathcal{K}_t$ defined in the statement of the corollary:
\[ \{f^{\mathrm{lin}}(\Theta_t,x)\}_{x\in\mathcal{X}}\xrightarrow{d} \{f^{(\infty)}_t(x)\}_{x\in\mathcal{X}}.\] Now, applying Slutsky's theorem \ref{sl} we have
\[ f(\Theta_t,\bar X)=f^{\mathrm{lin}}(\Theta_t,\bar X)+\Delta_t(\bar X)\xrightarrow{d}f^{(\infty)}_t(x)\quad \text{as}\quad m\to\infty.\]
Since $\mathcal{X}$ is finite, we can choose $\bar X=\text{vec}(\mathcal{X})$ so that we have the convergence of the full multivariate Gaussian distribution, which is the Gaussian process itself.
\begin{remark}
    This argument is not valid in the case of $\mathcal{X}$ infinite because we cannot consider $\bar X=\text{vec}(\mathcal{X})$. See \autoref{infinite} for the proof in the most general case.
\end{remark}

 \section{Noisy training of quantum neural networks}\label{ch7}
\subsection{Gradient descent with a finite number of measurements}\label{7-1}
The model function $f(\Theta,x)$ is defined in terms of expectation values of quantum observables, whose result is not deterministic. Therefore, the algorithm has to be iterated a number of times which increases with the precision required. \\
Also the gradient of $f(\Theta,x)$, which is involved in the training of the neural network, has to be estimated through repeated measurements. This can be done using the parameter-shift rules. The limited precision of the model function and of its gradients has to be taken into account when we look for the convergence of the training. As depicted in \autoref{sgdfigu}, if the variance of the estimator of the gradient of $\mathcal{L}$ is not sufficiently small, the training could converge to a neighborhood of the minimum. For instance, if the variance of the estimator is fixed, when the parameters are sufficiently close to a minimum the exact gradient becomes of the same order of the statistical noise and the evolution is subject to relevant deviations. If, at a certain point, the noise becomes much larger than the exact gradient, the parameters evolve in a random direction, possibly away from the minimum. The size of the neighborhood depends on the variance of the estimator and on the shape of the loss near the minimum.\\
The procedure of estimating the model function, its gradient and of computing the evolution of the parameters is done in discrete time steps. This is why we need to move from the gradient flow to the gradient descent.

\begin{figure}[ht]
\centering
\includegraphics[width=0.3\textwidth]{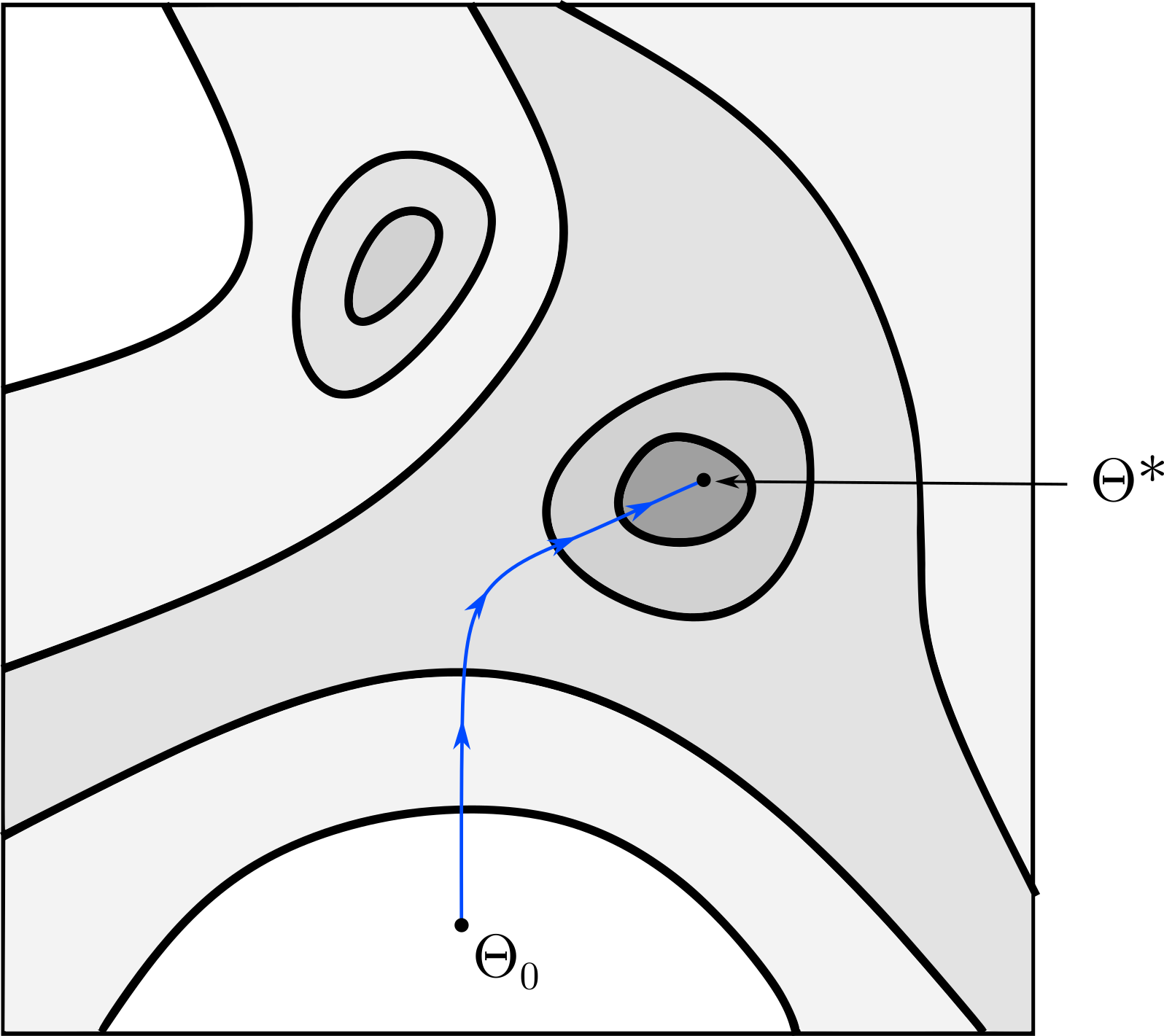}\hspace{1em}
\includegraphics[width=0.3\textwidth]{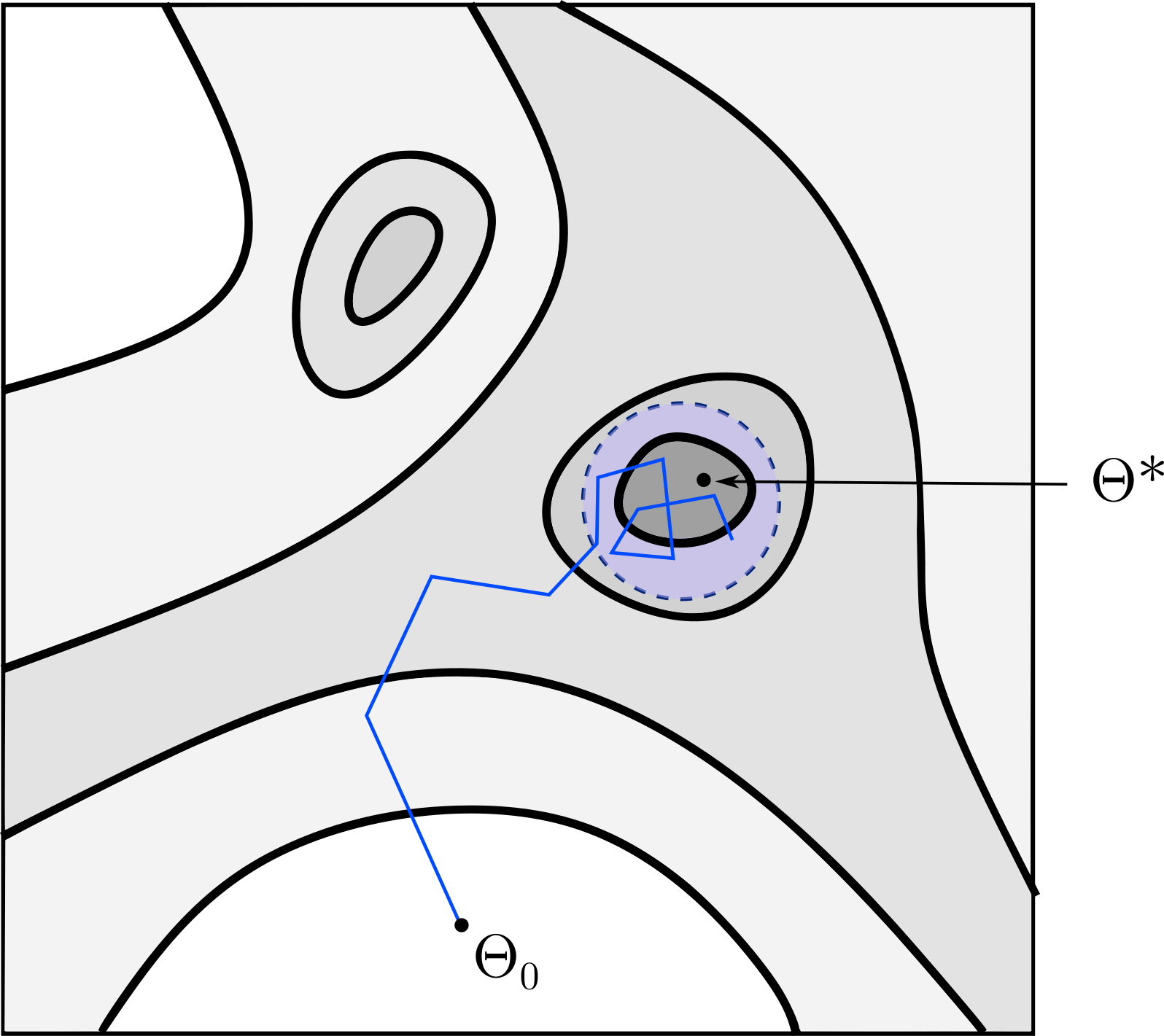}\hspace{1em}
\includegraphics[width=0.3\textwidth]{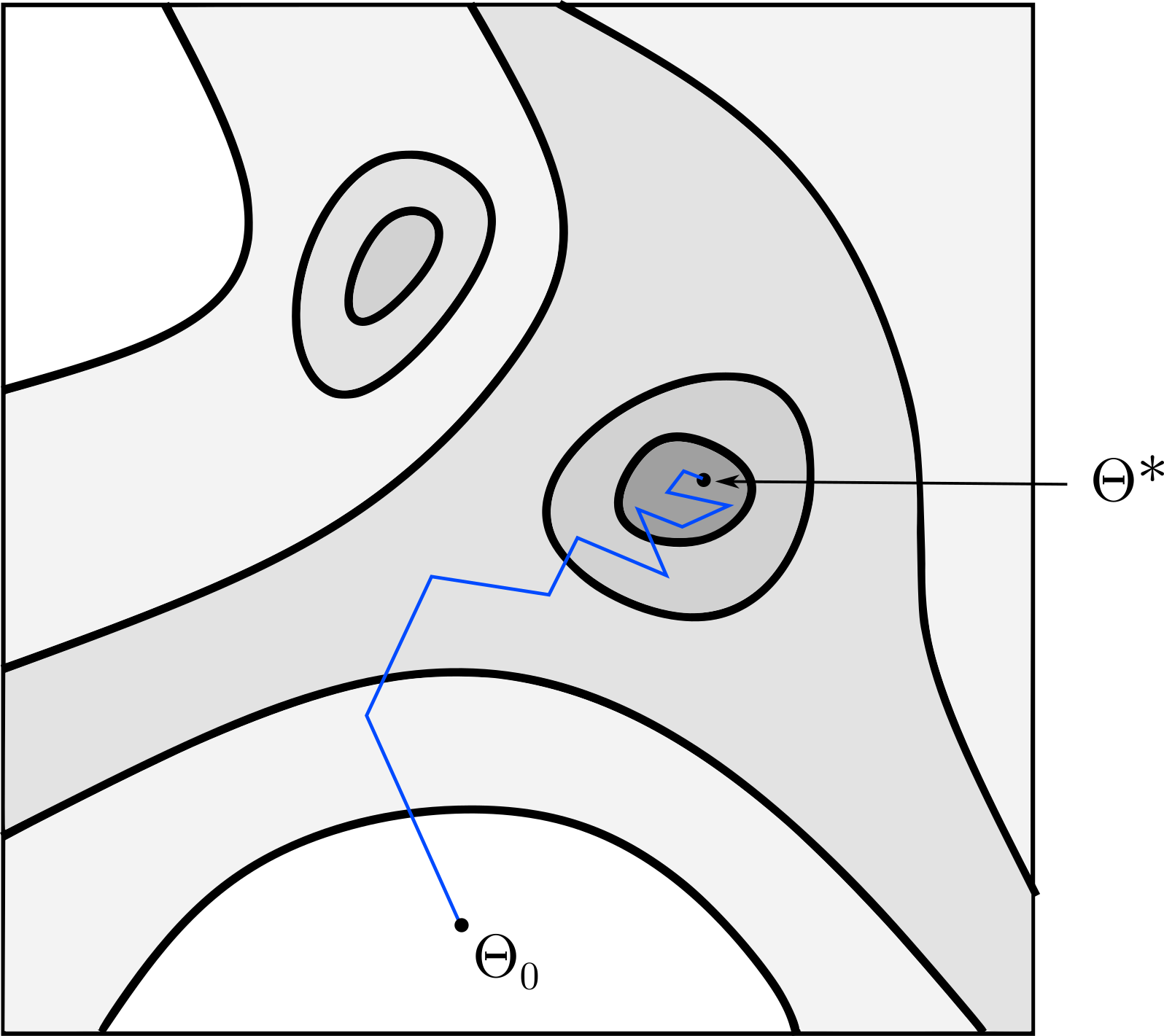}
\caption{Three possible scenarios: on the left, the gradient flow; in the middle, the stochastic gradient descent which converges to a neighborhood of the global minimum due to the high variance of the estimator of the gradient; on the right, the stochastic gradient descent with sufficiently small variance.}
\label{sgdfigu}
\end{figure}

The aim of this section is to prove that for a sufficiently high number of measurement -- which scales \textit{polynomially} with the number of qubits -- the training converges exponentially fast to the global minimum of the loss function, as we saw in the simpler case discussed in \autoref{ch6}. The strategy is based on two steps:
\begin{enumerate}
\item prove that the deterministic gradient descent converges to the global minimum;
\item show that, for a sufficiently small variance of the estimator of the gradient of the loss, the same convergence can be proved.
\end{enumerate}

Since also in this case we retrieve the lazy training regime in the limit of infinitely many qubits, the trained model function is still a Gaussian process after random initialization and stochastic training.

The outline of this section is the following.
\begin{enumerate}
\item In \autoref{7-2}, we introduce the deterministic gradient descent and we prove that it makes the empirical loss converge to zero (Theorem \ref{graddesc0}).
\item In \autoref{7-3}, we discuss the stochastic gradient descent and we prove that the convergence of the empirical loss to zero, the lazy training and the convergence to the linearized model are still valid (Theorem \ref{unbgraddesc}). As a consequence, the model is a Gaussian process in the limit of infinitely many qubits also in the case with noise (Theorem \ref{qnngpn}). These two theorems are the most important results of the section and of the article.
\item In \autoref{proofgraddesc0}, \autoref{proofunbgraddesc} and \autoref{proofqnngpn} we prove the main theorems of the previous subsections.
\end{enumerate}

\subsection{Deterministic gradient descent}\label{7-2}
The gradient descent equation is
\[ \Theta_{t+1}-\Theta_{t}=-\eta \nabla_\Theta \mathcal{L}(\Theta_t).\label{eq7.1}\]
As in the case of gradient flow, we consider the mean squared error function:
\begin{align}
\nonumber\mathcal{L}(\Theta_t)&=\frac{1}{n}\sum_{i=1}^n\frac{1}{2} \left(f(\Theta_t,x^{(i)})-y^{(i)}\right)^2\\ &=\frac{1}{2n}\|F(t)-Y\|_2^2.
\end{align}
This means that the equation (\ref{eq7.1}) can be rewritten as 
\[\Theta_{t+1}-\Theta_{t}=-\frac{\eta}{n} \nabla_\Theta F(t) \cdot (F(t)-Y). \label{eqgrd} \]
We will see, in the proof of Theorem \ref{qnngpn}, that the solution of (\ref{eqgrd}) for the linearized model gives, for $t\to\infty$, the same result obtained in the continuous time setting:
\[\label{wellfit}
\lim_{t\to\infty}f^{\mathrm{lin}}(\Theta_t,x)=f(\Theta_0,x)-\hat K_{\Theta_0}(x,X^T)\hat K_{\Theta_0}^{-1}(f(\Theta_0,X)-Y),
\]
provided that
\[\eta_0<\frac{2}{\lambda_{\min}^K+\lambda_{\max}^K}\]
and $m$ is large enough (see \autoref{proofqnngpn}).
The equation (\ref{wellfit}) shows that linearized model trained for infinite time perfectly reproduces the training labels.

\begin{theorem}[Deterministic gradient descent]
\label{graddesc0}
Let us assume that the hypotheses of Theorem \ref{init} and of Theorem \ref{ntkconv} are satisfied and that Assumption \ref{assNTK} holds. Let us furthermore assume that
\[\lim_{m\to\infty}\frac{\Sigma_1|\mathcal{M}|^3|\mathcal{N}|}{N_K^2(m)N^3(m)}=0.\]
For any $\delta>0$, there exist $\bar m\in\mathbb{N}$ and some constants $R_0,R_1,R_2$ such that, when applying gradient descent with learning rate $\eta=\frac{n}{N_K(m)}\eta_0$, where
\[\eta_0< \frac{2}{\lambda_{\min}^K+\lambda_{\max}^K},\] the following inequalities hold for all $m\geq \bar m$ with probability at least $1-\delta$ over random initialization:
\begin{align}
\label{desc1}\|F(t)-Y\|_2&\leq \sqrt{n\log(2n)}  R_0\left(1-\frac{\eta_0\lambda^K_{\min}}{3}\right)^t,\\
\label{desc2}\|\Theta_t-\Theta_{0}\|_\infty&\leq \frac{6R_0}{\lambda^K_{\min}}n\sqrt{\log(2n)}\frac{|\mathcal{M}|}{N_K(m)N(m)}.
\end{align}
\end{theorem}
\begin{proof}
See \autoref{proofgraddesc0}.
\end{proof}
\begin{remark}
The (\ref{desc1}) is equivalent to
\[\mathcal L (\Theta_t)\leq  \frac{R^2_0}{2}\log(2n)\left(1-\frac{\eta_0\lambda^K_{\min}}{3}\right)^{2t}.\]
\end{remark}
Analogously to the continous-time case we have uniform convergence to the linearized model.
\begin{theorem}
Let us assume that the hypotheses of Theorem \ref{init} and of Theorem \ref{graddesc0} are satisfied and that Assumption \ref{assNTK} holds.
For any $\delta>0$, there exist a constant $R_1>0$ and $\bar m\in\mathbb{N}$ such that the following inequality holds for all $m\geq \bar m$ with probability at least $1-\delta$ over random initialization:
\begin{align}
\sup_t\sup_{x\in\mathcal{X}}|f(\Theta_t,x)-f^{\mathrm{lin}}(\Theta_t,x)|\leq \left(\frac{6R_2}{\lambda^K_{\min}}\right)^2n^2\log(2n)\frac{Lm|\mathcal{M}|^4|\mathcal{N}|}{N_K^2(m)N^3(m)}.
\label{convunif2}
\end{align}
\end{theorem}
\begin{proof} Combining Theorem \ref{powerful} with the inequality (\ref{desc2}), the bound (\ref{convunif2}) follows immediately.
\end{proof}

Also the proof of the convergence to a Gaussian process is similar to the continuous time setting, with some differences dues to the discrete nature of time in the evolution of the model. All the details of the proof of the convergence will be given for the noisy case in the following subsection and can be immediately adapted to the case discussed in this subsection by removing the additional terms due to the statistical noise.

\subsection{Stochastic gradient descent}\label{7-3}

\modifica{
\begin{assumption}[Unbiased gradient descent]\label{sgdexact}
We consider the gradient-descent evolution
\[\Theta_{t+1}-\Theta_t=-\eta g^{(t)}(\Theta_t),\label{sgdeqn}\]
where for any training time $\modifica{t}\in\mathbb{N}$, $g^{(t)}(\Theta_{t})$ is an unbiased estimator of the gradient of the loss function in the stochastic gradient descent equation, i.e., an estimator whose average conditioned on the past history coincides with the exact value of the gradient.
More formally, each $g^{(t)}(\Theta_t)$ is a random vector such that, for any $k < t$, any $\xi_k\in\mathbb{R}$ and any $\Theta_k\in\mathbb{R}^{|\Theta|}$ we have
\[\mathbb{E}\left[\modifica{g_i^{(t)}(\Theta_t)}\,\,\,\big|\,\,\, g^{(k)}(\Theta_{k})=\xi_k\quad \forall\, k<t\right]=\partial_{\theta_i}\mathcal{L}(\Theta)\qquad \forall\,\Theta\in\mathscr{P}\label{conditioned}.\]
\end{assumption}
}

\begin{remark}
The meaning of (\ref{conditioned}) is the following: once we are at time $t$, we perform the measurements of the output of circuit and we build an the estimator \modifica{$g^{(t)}(\Theta_t)$} such that its expectation value is \modifica{$\nabla_\Theta\mathcal{L}(\Theta)$} whatever the output of \modifica{the measurements employed to determine} $\{g^{(0)},g^{(1)},\dots,g^{(t-1)}\}$ was. This is what we physically ask.
\end{remark}

\begin{figure}[ht]
\centering
\includegraphics[width=0.4\textwidth]{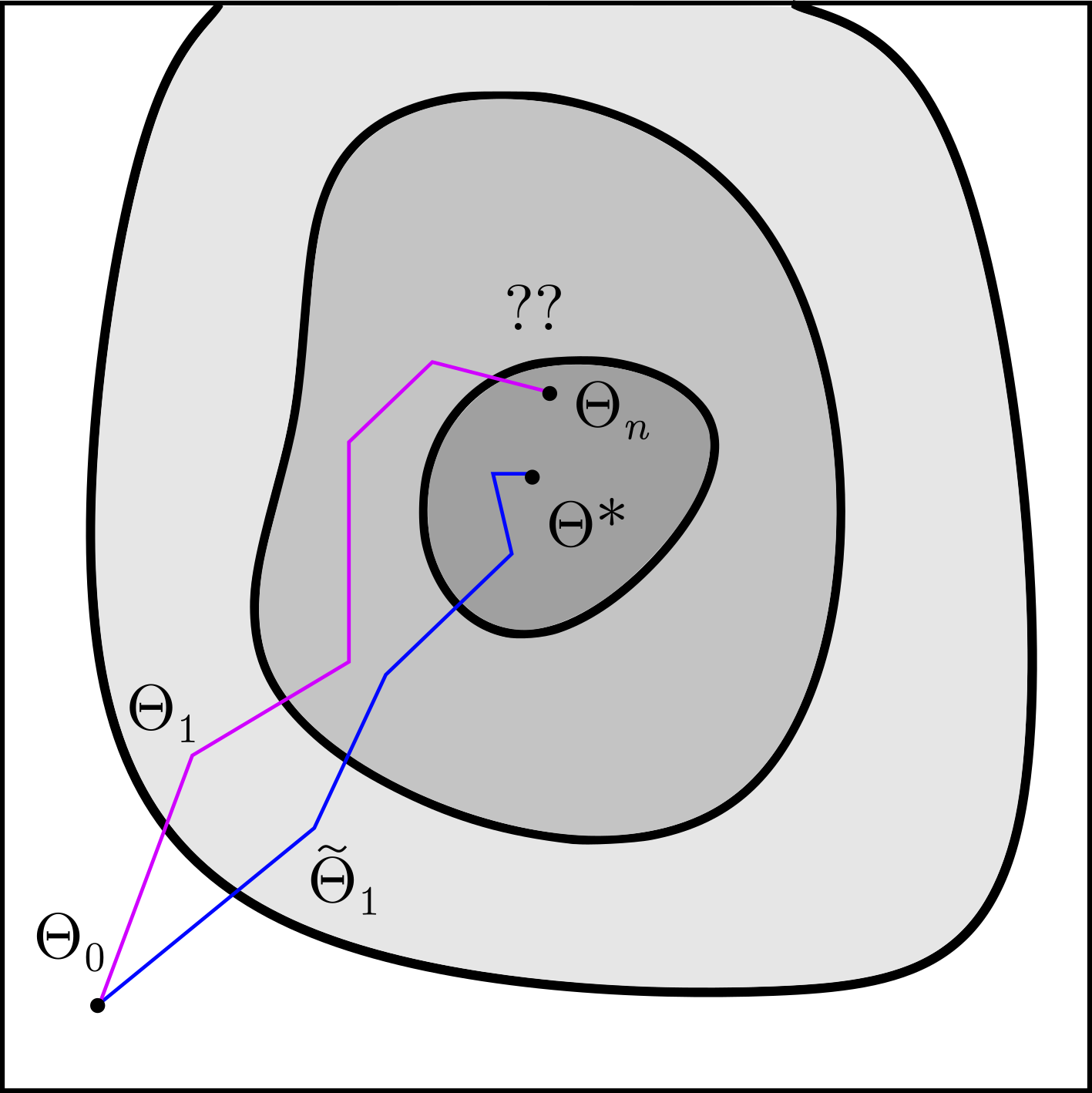}\hspace{2.5em}
\includegraphics[width=0.4\textwidth]{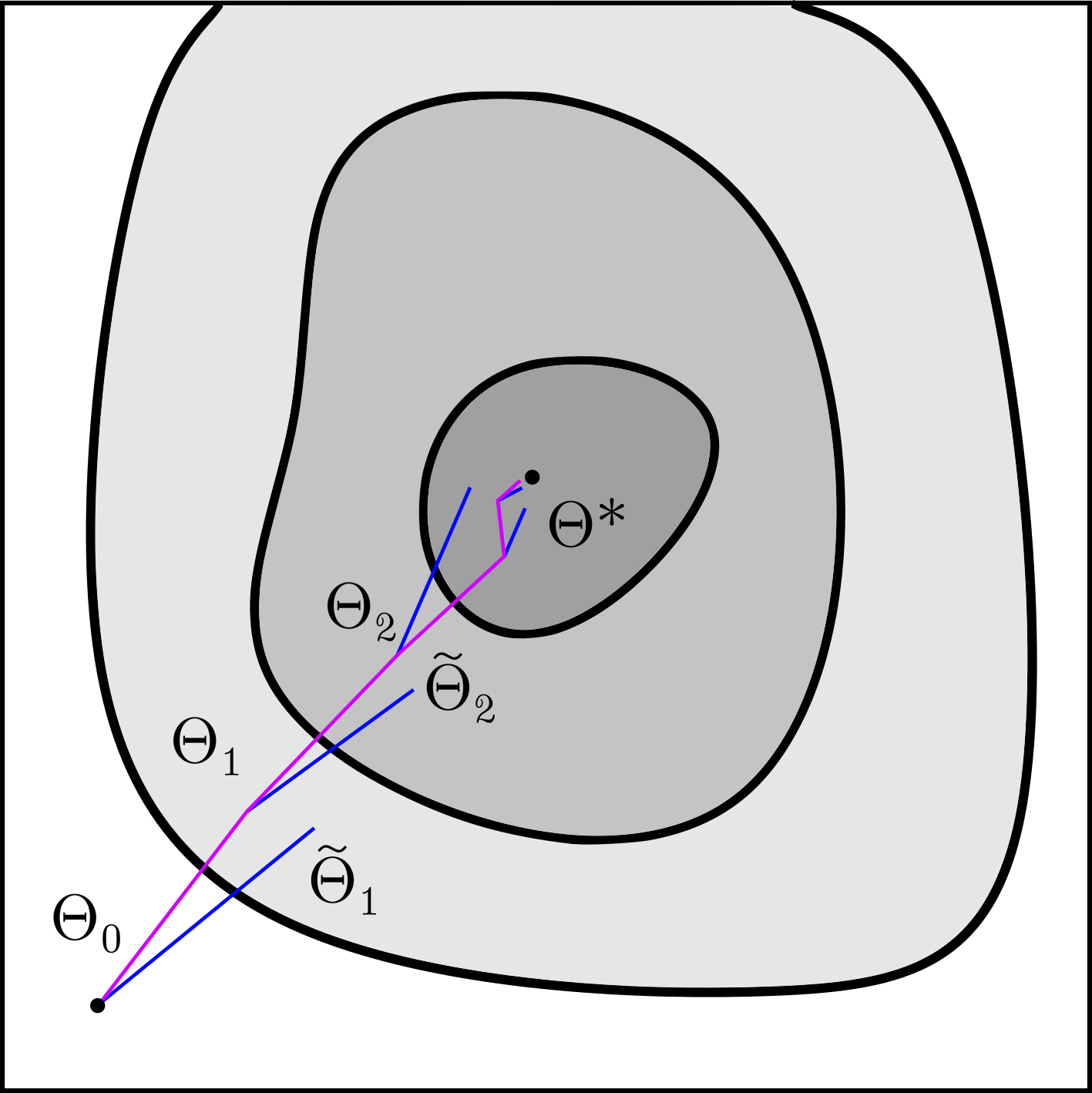}
\caption{The global and the local strategies to approach the stochastic gradient descent convergence. \modifica{With a global strategy we lose track of the path in the parameter space.}}
\label{sgdstrategy}
\end{figure}

The aim of the following theorem is to show that, with a bound on the variance which is polynomial in the number of qubits, an unbiased estimator of the gradient of the MSE loss is suitable to achieve an exponential convergence to the global minimum with a lazy evolution of the parameters.\\
The idea is to compare the stochastic evolution with the exact one. In \autoref{sgdstrategy}, a tentative strategy is depicted on the left. Given an initial vector of parameters $\Theta_0$, the stochastic evolution of the parameter is represented by the violet curve (the parameter vectors are called $\Theta_t$), while the deterministic evolution with the exact gradient is the blue curve (the parameter vectors are called $\tilde \Theta_t$), which is known to converge to the global minimum by Theorem \ref{graddesc0}. We could try to bound the distance between the two curves using some arguments based on the triangle inequality. However, some problems may arise.
\begin{enumerate}
\item At each time step $t>0$, the point $\Theta_t$ at which the estimator of the gradient is computed is different from the point $\tilde\Theta_t$ at which the exact gradient is computed; this means that, even if the exact evolution is close to a minimum ($\nabla_\Theta\mathcal{L}(\tilde\Theta_t)\leq \epsilon$), the estimated distance $\|\Theta_t-\tilde \Theta_t\|_\infty$ would affect the upper bound on $\nabla_\Theta\mathcal{L}(\Theta_t)$ with a bias due to a rough triangle inequality argument, i.e.
\[\nabla_\Theta\mathcal{L}(\Theta_t)\leq \nabla_\Theta\mathcal{L}(\tilde\Theta_t)+g(\|\Theta_t-\tilde \Theta_t\|_\infty).\label{badestimate}\]
It is easy to realize that $g$ is non decreasing. The hope is to show that somehow $\|\Theta_t-\tilde \Theta_t\|_\infty$ becomes small with $t$.
\item However, since (\ref{badestimate}) asymptotically gives few information on the convergence of $\Theta_t$ to $\Theta^\ast$ due to the bias $g$, we can only worsen the estimate:
\begin{align}
\|\Theta_{t+1}-\tilde \Theta_{t+1}\|_\infty&\leq \|\Theta_{t+1}-\Theta_{t}\|_\infty
+\|\Theta_{t}-\tilde \Theta_{t}\|_\infty
+\|\tilde \Theta_{t+1}-\tilde \Theta_t\|_\infty.
\end{align}
So, the bias $g(\|\Theta_t-\tilde \Theta_t\|_\infty)$ is non-decreasing with $t$. This hinders any conclusion on the convergence.
\end{enumerate}
The discussion of this naive strategy is the starting point to understand what goes wrong in this case and to motivate our different approach. The main problem with a \textit{global} comparison between the exact and the stochastic paths is that the iteration of triangle inequality bounding the discrepancy between $\Theta_t$ and $\tilde\Theta_t$ prodigally collects all the previous estimates, accumulating inaccuracy with time. What we need is a \textit{local} approach. This can be achieved by performing a deterministic gradient descent step starting at the last point of the stochastic evolution, as depicted in \autoref{sgdstrategy} on the right.
\[\begin{dcases}
\Theta_{t+1}=\Theta_t-\eta g^{(t)}(\Theta_t)&\quad \text{stochastic gradient descent,}\\
\tilde \Theta_{t+1}=\Theta_t-\eta \nabla_\Theta\mathcal{L}(\Theta_t)&\quad \text{locally deterministic gradient descent.}
\end{dcases}\]
Now, the variance of $g^{(t)}(\Theta_t)$ must be chosen so that $\|\Theta_{t+1}-\tilde\Theta_{t+1}\|_\infty$ is small enough to fulfill these criteria:
\begin{enumerate}
\item $\mathcal{L}(\Theta_{t+1})$ is sufficiently close to $\mathcal{L}(\tilde \Theta_{t+1})$ so that a decrease with respect to the previous cost $\mathcal{L}(\Theta_t)$ is ensured;
\item the distance between $\Theta_{t+1}$ and $\tilde\Theta_{t+1}$ is a summable function of $t$ which goes to zero as the number of qubits increases, so that the lazy training result is still valid.
\end{enumerate}
This will be the strategy to prove Theorem \ref{unbgraddesc}, which is the most important result of the article together with Theorem \ref{qnngpn}.

\begin{mdframed}
\begin{theorem}[Trainability of the model] 
\label{unbgraddesc} Let us suppose that Assumption \ref{assNTK} holds. Let us furthermore assume that
\[\lim_{m\to\infty}\frac{Lm|\mathcal{M}|^4|\mathcal{N}|^2}{N^3(m)}=0.\label{hpunbgraddesc}\]
 We suppose to train our model according to the gradient descent described in Assumption \ref{sgdexact} for a training time $T\in\mathbb{N}$ with learning rate $\eta=\frac{n}{N_K(m)}\eta_0$, where
\[\eta_0<\frac{2}{\lambda_{\min}^K+\lambda_{\max}^K}.\]
We furtherly assume the following condition on the variance of the estimator of the gradient: for any choice of $t\leq T$ and $\xi_k\in\mathbb{R}$ and $\Theta_k\in\mathbb{R}^{|\Theta|}$,
\begin{align}\label{Eg2b}
\nonumber \mathrm{Var}\Big[g_i^{(t)}(\Theta)\,\,\,&\big|\,\,\, g^{(k)}(\Theta_{k})=\xi_k\quad \forall\, k<t\Big]\\
&\leq c_0\eta_0^2\,\frac{\left(\lambda_{\min}^K\right)^4}{n^2}\frac{N_K^2(m)N^2(m)}{|\mathcal{M}|^2|\Theta|^3}\,\frac{\delta}{(t+1)^2}\,\mathcal{L}(\Theta),
\end{align}
where $c_0=\frac{1}{864\pi^2}$.\\
For any $\delta>0$, there exist $\bar m\in\mathbb{N}$ and some constants $R_0(\delta)$, $R_1(\delta)$, $R_2(\delta)$ such that for all $m\geq \bar m$ the following inequalities hold with probability at least $1-\delta$ over random initialization and stochastic training:
\begin{align}
\label{noisy21}\mathcal{L}(\Theta_t)&\leq \frac{R_0^2}{2}\log(2n)\left(1-\frac{\eta_0\lambda_{\min}^K}{3}\right)^{2t},\\
\nonumber \|\Theta_t-\Theta_0\|_\infty&\leq  R_1\eta_0 \lambda_{\min}^K\sqrt{\log(2n)}\frac{N(m)}{|\mathcal{M}||\Theta|}\\
 \label{noisy22}&\phantom{\leq}+\frac{R_2}{\lambda_{\min}^K}n\sqrt{\log(2n)} \frac{|\mathcal{M}|}{N_K(m)N(m)},\\
\nonumber \sup_{x\in\mathcal{X}}|f(\Theta_t,x)-f^{\mathrm{lin}}(\Theta_t,x)|&\leq 2R_1^2\eta_0^2\left(\lambda_{\min}^K\right)^2\log(2n)\frac{|\mathcal{N}|N(m)}{Lm}\\
\label{noisy23} &\phantom{\leq}
+\frac{2R^2_2}{\left(\lambda_{\min}^K\right)^2}n^2\log(2n)\frac{Lm|\mathcal{M}|^4|\mathcal{N}|}{N_K^2(m)N^3(m)},
\end{align}
for any $t\leq T$.
\end{theorem}
\end{mdframed}
\begin{proof}
See \autoref{proofunbgraddesc}.
\end{proof}
\begin{remark}
[Dependence of the variance on $\mathcal{L}(\Theta)$ and on $t$]
As already introduced in \autoref{7-1}, the variance of the estimator of the gradient must be sufficiently small so that the exact value of the gradient is not completely blurred by the statistical noise. In the proof of Theorem \ref{unbgraddesc} we show that a small cost function implies a small gradient. So, the cost function in the bound (\ref{Eg2b}) ensures that the variance is exponentially small (in time) when the cost and its gradient are exponentially small. Furthermore, an explicit dependence on $t$ appears in (\ref{Eg2b}). This is due to the iterative comparison of the stochastic gradient descent with the exact gradient descent of the strategy we discussed in the introduction of this subsection. Indeed, we need the probability that the estimated gradient is too far from the exact gradient to be increasingly small along the minimization procedure. Otherwise, we could not provide a union bound ensuring that the sum of the probabilities that estimated gradient is not close enough to the exact gradient is small (see (\ref{unionbounddelta})).
\end{remark}
In the following \modifica{proposition} we explain why the bound on the variance \eqref{Eg2b} implies that a polynomial number of measurements (with respect to the number of the qubits) is sufficient to 

\modifica{
\begin{proposition}[A polynomial number of measures]\label{polynmeas}
For any $\epsilon,\delta>0$, let $M(\epsilon,\delta)$ be the number of measurements required by the gradient-descent algorithm of Theorem \ref{unbgraddesc} to get a cost less than $\epsilon$ with probability at least $1-\delta$.
Then,
\[M(\epsilon,\delta)\leq C(\delta,\epsilon,\lambda_{\min}^K,n)\,\left(\frac{L\,m\left|\mathcal{M}\right|^4}{N^3(m)}\right)^\frac{4}{3}L^\frac{8}{3}m^\frac{14}{3},\]
where the term in parentheses tends to $0$ for $m\to\infty$ from the hypothesis \eqref{hpunbgraddesc} of Theorem \ref{unbgraddesc}, and the remaining terms grow at most polynomially in $m$ whenever the number of layers $L$ grows polynomially with $m$.
\end{proposition}
}
\begin{proof}
We recall that, if $X$ and $Y$ are independent random variables, then
\[\text{Var}[X+Y]=\text{Var}[X]+\text{Var}[Y],\]
therefore, if $X_1,\dots,X_M$ are i.i.d. random variables, then
\[\text{Var}\left[\frac{1}{M}\sum_{k=1}^MX_k\right]=\frac{\text{Var}[X_1]}{M}.\]
\modifica{By Assumption \ref{domain}, the generators of the unitaries encoding the trainable parameters have spectrum $\{-1,+1\}$.
Then, there is a simple strategy to compute the derivative $\partial_{\theta_i}f(\Theta,x)$. In quantum circuits like these, indeed, the derivatives of the expectation value of any observable with respect to any parameter can be written in terms of other expectation values computed for ``shifted'' parameters. These identities were introduced in \cite{Schuld_2019} (see in particular (14)) and they are known as parameter shift rules:
\[\frac{\partial}{\partial\theta_i}f(\Theta,x)=f(\Theta+\Delta^{(i)},x)-f(\Theta-\Delta^{(i)},x)\qquad \Delta^{(i)}_j\coloneqq\begin{cases}
    \frac{\pi}{4}&j=i\\ 0 & j \neq i
\end{cases}\]
so that an unbiased estimator $h^{(t,1)}_i(\Theta,x)$ of $\partial_{\theta_i}f(\Theta,x)$ can be constructed by measuring the output of the circuit at 2 different points
\[h^{(t,1)}_i(\Theta,x)=h_i^{(t,1,+)}(\Theta,x)-h_i^{(t,1,-)}(\Theta,x),\]
where $h_i^{(t,1,\pm)}(\Theta,x)$ are the outputs of the \textit{mutually independent} measurements of $\mathcal{O}$ on $\ket{\psi(\Theta\pm\Delta^{(i)},x)}=U(\Theta\pm\Delta^{(i)},x)\ket{0^m}$:
\begin{align}
\nonumber \mathbb{E}\left[h^{(t,1)}_i(\Theta,x)\right]&=\mathbb{E}\left[h_i^{(t,1,+)}(\Theta,x)-h_i^{(t,1,-)}(\Theta,x)\right]\\
&=h_i^{(t,1,+)}(\Theta,x)-h_i^{(t,1,-)}(\Theta,x)=\frac{\partial}{\partial\theta_i}f(\Theta,x).
\end{align}}
All the expectation values and all the variances are meant to be conditioned on any possible story of the previous estimators.
Iterating this procedure $M\modifica{^{(t)}}$ times, we have a family of unbiased and independent estimators 
$\{h^{(t,k)}_i(\Theta,x)\}_{k\in\{1,\dots,M\modifica{^{(t)}}\}}$
of $\frac{\partial}{\partial\theta_i}f(\Theta,x)$.
Since, by Assumption \ref{domain}, $\mathcal{O}$ is a sum of $m$ observables with spectrum in $[-1,1]$, a very rough uniform bound on the variance of $h_i^{(t,1,j)}(\Theta+\Delta^{(j)},x)$ can be easily provided:
\[\text{Var}\left[h_i^{(t,k,j)}(\Theta+\Delta^{(j)},x)\right]\leq\frac{m^2}{N^2(m)},\]
whence
\[\text{Var}\left[h_i^{(t,k)}(\Theta,x)\right]\leq\frac{\modifica{2}m^2}{N^2(m)}.\]
Therefore, 
\[g^{(t)}_i(\Theta,x)=\frac{1}{M\modifica{^{(t)}}}\sum_{k=1}^{\modifica{M^{(t)}}} h_i^{(t,k)}(\Theta,x)\]
is an unbiased estimator of $\frac{\partial}{\partial\theta_i}f(\Theta,x)$ with variance
\[\text{Var}\left[g_i^{(t)}(\Theta,x)\right]\leq \frac{1}{M\modifica{^{(t)}}}\frac{\modifica{2}m^2}{N^2(m)}.\]
By (\ref{Eg2b}) we notice that, in the physical implementation of the quantum neural network, the variance required at each time step is of the form
\[\mathrm{Var}\Big[g_i^{(t)}(\Theta_t)\,\,\,\big|\,\,\, g^{(k)}(\Theta_{k})=\xi_k\quad \forall\, k<t\Big]\leq c(\delta, \lambda_{\min}^K,n)\frac{N^2_K(m)N^2(m)}{|\mathcal{M}|^2L^3m^3}\frac{\mathcal{L}(\Theta_t)}{(t+1)^2}.\]
This is ensured if the number of measurements $M\modifica{^{(t)}}$ is sufficiently large
\[\frac{1}{M\modifica{^{(t)}}}\frac{\modifica{2}m^2}{N^2(m)}\leq c(\delta, \lambda_{\min}^K,n)\frac{N^2_K(m)N^2(m)}{|\mathcal{M}|^2L^3m^3}\frac{\mathcal{L}(\Theta_t)}{(t+1)^2},\]
i.e.
\[
M\modifica{^{(t)}}=\frac{\modifica{2}}{c(\delta, \lambda_{\min}^K,n)}\frac{|\mathcal{M}|^2L^3m^5}{N^2_K(m)N^4(m)}\frac{(t+1)^2}{\mathcal{L}(\Theta_t)}
\]
is sufficient.
Since this procedure must be repeated for each parameter, the total number of measurements required at time $t$ is
\[M_{tot}^{(t)}=|\Theta|M\modifica{^{(t)}}=LmM\modifica{^{(t)}}.\]
To get $\epsilon$-close to the minimum of the cost, by (\ref{noisy21}) the number of time steps required is, at most,
\[ T_{\max}^{(\epsilon)}=\frac{\log\left(\frac{1}{\epsilon}\right)+\log\big(R_0^2\log(2n)\big)}{-2\log\left(1-\frac{\eta_0\lambda_{\min}^K}{3}\right)}.\]
So the total number of measurements required is bounded by
\[\modifica{M(\epsilon,\delta)}\leq \sum_{t=1}^{T_{\max}^{(\epsilon)}}M_{tot}^{(t)}\leq \frac{\modifica{2}T_{\max}^{(\epsilon)}(T_{\max}^{(\epsilon)}+1)^2}{c(\delta, \lambda_{\min}^K,n)\epsilon}\frac{|\mathcal{M}|^2L^4m^6}{N^2_K(m)N^4(m)}.\]
\modifica{Setting
\[C(\delta,\epsilon,\lambda_{\min}^K,n)\coloneqq\frac{\modifica{2}T_{\max}^{(\epsilon)}(T_{\max}^{(\epsilon)}+1)^2}{c(\delta, \lambda_{\min}^K,n)\epsilon}\]
we get
\begin{align}
M(\epsilon,\delta) &\leq C(\delta,\epsilon,\lambda_{\min}^K,n)\,\frac{|\mathcal{M}|^2L^4m^6}{N^2_K(m)N^4(m)} \overset{(\mathrm{a})}{\le}C(\delta,\epsilon,\lambda_{\min}^K,n)\,\left(\frac{L\,m\left|\mathcal{M}\right|^4}{N^3(m)}\right)^\frac{4}{3}\frac{L^\frac{8}{3}m^\frac{14}{3}}{|\mathcal{M}|^\frac{10}{3}}\nonumber\\
&\le C(\delta,\epsilon,\lambda_{\min}^K,n)\,\left(\frac{L\,m\left|\mathcal{M}\right|^4}{N^3(m)}\right)^\frac{4}{3}L^\frac{8}{3}m^\frac{14}{3}\,,
\end{align}
where (a) follows from Lemma \ref{boundnorm}.
The claim follows.}
\end{proof}

Theorem \ref{unbgraddesc} is the fundamental result to prove that also in the noisy case the model is a Gaussian process during the training in the limit of many qubits. However, we still need to verify that the noisy trajectory in the parameter space is close to the exact trajectory with the linearized evolution equation. This result is ensured by the following statements, which require a further suppression $\xi(m)\to 0$ in the bound for variance of the estimator for the gradient of the cost function; we are going to choose $\xi(m)$ in order to ensure the rate of convergence needed to obtain the convergence to a Gaussian process. We will also replace $\delta$ with $\delta/4$: this is just a technical detail which will be used in the proof for a union bound.

\begin{lemma}[The parameters differ at the second order]\label{secondorder2} Let us assume a condition on the variance of the estimator of the gradient stronger than (\ref{Eg2b}) by a factor $\xi(m)/4$: for any choice of $t\leq T$ and $\xi_k\in\mathbb{R}$ and $\Theta_k\in\mathbb{R}^{|\Theta|}$,
\begin{align}\label{VarXi}
\nonumber \mathrm{Var}\Big[g_i^{(t)}(\Theta)\,\,\,&\big|\,\,\, g^{(k)}(\Theta_{k})=\xi_k\quad \forall\, k<t\Big]\\
&\leq c_0\eta_0^2\,\frac{\left(\lambda_{\min}^K\right)^4}{n^2}\frac{N_K^2(m)N^2(m)}{|\mathcal{M}|^2|\Theta|^3}\xi(m)\,\frac{\delta/4}{(t+1)^2}\,\mathcal{L}(\Theta),
\end{align}
where $c_0=\frac{1}{864\pi^2}$.
Then, for any fixed $\delta>0$, there exist some constant $C_1, C_2$ and an integer $\bar m\in\mathbb{N}$ such that, for any $m\geq \bar m$,
\[\big\|\Theta_t-\Theta_t^{\mathrm{lin}}\big\|_\infty\leq \left(\frac{C_1}{(\lambda_{\min}^K)^3}+C_2\right)n^3\log(2n)\frac{Lm|\mathcal{M}|^5|\mathcal{N}|^3}{N^4(m)}\log N(m) \frac{1}{2}\left(1+N(m)\sqrt{\xi(m)}\right)\]
with probability at least $1-\delta$.
\end{lemma}
\begin{proof}
    See \autoref{proofsecondorder2}.
\end{proof}
Now it is clear why we need the further suppression $\xi(m)$ in the bound for the variance of the estimator of the gradient of the cost function, as we will discuss more in detail in Corollary (\ref{nuovavarianza}) and in Theorem (\ref{confronto_finale}): the divergent term $N(m)$ must be cancelled by $\xi(m)$ in order to ensure the following statements.

\begin{corollary}\label{nuovavarianza}
    If we choose, in particular \[\xi(m)=\frac{1}{N^2(m)},\] which means requiring that
    \begin{align}\label{VarNew}
\nonumber \mathrm{Var}\Big[g_i^{(t)}(\Theta)\,\,\,&\big|\,\,\, g^{(k)}(\Theta_{k})=\xi_k\quad \forall\, k<t\Big]\\
&\leq c_0\eta_0^2\,\frac{\left(\lambda_{\min}^K\right)^4}{n^2}\frac{N_K^2(m)}{|\mathcal{M}|^2|\Theta|^3}\,\frac{\delta/4}{(t+1)^2}\,\mathcal{L}(\Theta),
\end{align}
where $c_0=\frac{1}{864\pi^2}$, then, for any $\delta>0$, there exist $C_1, C_2$ and $m_0$ such that, for any $m\geq m_0$,
\begin{align}
\big\|\Theta_t-\Theta_t^{\mathrm{lin}}\big\|_\infty&\leq \left(\frac{C_1}{(\lambda_{\min}^K)^3}+C_2\right)n^3\log(2n)\frac{Lm|\mathcal{M}|^5|\mathcal{N}|^3}{N^4(m)}\log N(m)
\end{align}
with probability at least $1-\delta$.
\end{corollary}
Now we have all the ingredients to state the generalization of Theorem \ref{qnngp} to the discrete time setting with statistical noise.
\begin{theorem}[The original evolution is close to the linear evolution]\label{confronto_finale}
Let us suppose that the hypotheses of the previous lemmas are satisfied; in particular, we ask (\ref{VarNew}) for the variance of the estimator of the gradient of the cost function. Then, in the presence or in the absence of statistical noise in the gradient descent, we have the following bound: for any $\delta>0$ there exist some constants $C_3$ and $C_4$ and a number of qubits $m_0\in\mathbb \mathcal{N}$ such that, for any $m\geq m_0$
\begin{align}
|f(\Theta_t,x)-f^{\mathrm{lin}}(\Theta_t^{\mathrm{lin}},x)|\leq \left(\frac{C_3}{(\lambda_{\min}^K)^3}+C_4\right)n^3\log(2n)\frac{L^2m^2|\mathcal{M}|^6|\mathcal{N}|^4}{N^5(m)}\log N(m)
\end{align}
with probability at least $1-\delta$.
\end{theorem}
\begin{proof}
    See \autoref{proofconfronto}.
\end{proof}

\begin{mdframed}
\begin{theorem}[Quantum neural networks as Gaussian processes -- noisy setting]\label{qnngpn}
Let us assume that a circuit satisfying Assumption \ref{assNTK} and such that
\[\lim_{m\to\infty}\frac{L^2m^2|\mathcal{M}|^6|\mathcal{N}|^4}{N^5(m)}\log N(m)=0\label{hpq}\]
is randomly initialized according to Assumption \ref{zeromean} and is trained using the noisy gradient descent with the following bound on the variance of the estimator of the gradient of the cost function: for any choice of $t\leq T$ and $\xi_k\in\mathbb{R}$ and $\Theta_k\in\mathbb{R}^{|\Theta|}$,
\begin{align}\label{Eg2c}
\nonumber \mathrm{Var}\Big[g_i^{(t)}(\Theta)\,\,\,&\big|\,\,\, g^{(k)}(\Theta_{k})=\xi_k\quad \forall\, k<t\Big]\\
&\leq c_0\eta_0^2\,\frac{\left(\lambda_{\min}^K\right)^4}{n^2}\frac{N_K^2(m)}{|\mathcal{M}|^2|\Theta|^3}\,\frac{\delta/4}{(t+1)^2}\,\mathcal{L}(\Theta).
\end{align} Then, in the limit of infinitely many qubits $m\to\infty$, $\{f(\Theta_t,x)\}_{x\in\mathcal{X}}$ converges in distribution to a Gaussian process $\{f^{(\infty)}_t(x)\}_{x\in\mathcal{X}}$ with mean and covariance
\begin{align}
\mu_t(x)&=\bar K(x,X^T)\bar K^{-1}\left(\id-\left(\id-\eta_0\bar K\right)^t\right) Y\\
\nonumber\mathcal{K}_t(x,x')&=\mathcal{K}_0(x,x'),\\
\nonumber&\phantom{=}- \bar K(x,X^T)\bar K^{-1}\left(\id-\left(\id-\eta_0\bar K\right)^t\right) \mathcal{K}_0(X,x')\\
\nonumber&\phantom{=}-\bar K(x',X^T)\bar K^{-1}\left(\id-\left(\id-\eta_0\bar K\right)^t\right) \mathcal{K}_0(X,x) \\
\nonumber&\phantom{=}+\bar K(x,X^T)\bar K^{-1}\left(\id-\left(\id-\eta_0\bar K\right)^t\right)\times\\
&\phantom{=========}\times\mathcal{K}_0(X,X^T)\left(\id-\left(\id-\eta_0\bar K\right)^t\right)\bar K^{-1} \bar K(X,x').
\end{align}
\end{theorem}
\end{mdframed}
\begin{proof}
See \autoref{proofqnngpn}.
\end{proof}

\subsection{Proof of Theorem \ref{graddesc0}}\label{proofgraddesc0}
\begin{lemma}[Lagrange theorem for vector valued functions in several variables]
\label{lagrange}
Let $\mathcal{U}\subseteq \mathbb{R}^m$ be a compact convex set and let $f:\mathcal{U}\to \mathbb{R}^n$ be continuously differentiable. Then
\begin{align}
f(y)-f(x)&=\left[\int_0^1df(x+t(y-x))dt\right] (y-x).
\end{align}
\end{lemma}
\begin{proof}
See e.g. \cite{giaquinta2010mathematical}.
\end{proof}

The following lemma is crucial both for the deterministic and the stochastic case. We show that, if we are asymptotically close enough to the parameters at initialization, then a gradient descent step contracts the distance between the output of the model and the labels provided by the dataset.

Let $B_r(\Theta_0)=\{\Theta: \|\Theta-\Theta_0\|_\infty< r\}$ be the ball of center $\Theta_0$ and radius $r$.

\begin{lemma}[Single step of gradient descent]\label{step}
Let us assume that the hypothesis of Theorem \ref{ntkconv} is satisfied and that Assumption \ref{assNTK} holds. Let $\delta>0$, $\zeta\in (0,1)$ and suppose
\[\eta=\frac{n}{N_K(m)}\eta_0,\qquad \eta_0<\frac{2}{\lambda_{\min}^K+\lambda_{\max}^K}.\]
Furthermore, let $r(m):\mathbb{N}\to \mathbb{R}^+$
such that
\[\lim_{m\to\infty} \frac{\Sigma_1|\mathcal{M}|^2|\mathcal{N}|}{N_K(m)N^2(m)}r(m)=0 \label{hplim}.\]
 There exists $\bar m\in\mathbb{N}$ such that, with probability at least $1-\delta$ over random initialization, 
given any $\Theta$ and $\Theta'$ both in $B_{r(m)}(\Theta_0)$ and satisfying
\[\Theta'-\Theta=-\eta\nabla_\Theta\mathcal{L}(\Theta), \label{sopraqui}\]
where $\mathcal{L}$ is the MSE loss, the following inequality holds for any $m\geq \bar m$:
\begin{align} 
\|f(\Theta',X)-Y\|_2\leq \left(1-\zeta\eta_0\lambda_{\min}^K\right)\|f(\Theta,X)-Y\|_2.
\end{align}
\end{lemma}

\begin{proof}

Let us start rewriting
\begin{align}
\|f(\Theta',x)-Y\|_2&=\|(f(\Theta',x)-f(\Theta,x))+(f(\Theta,x)-Y)\|_2.
\end{align}
Calling\footnote{We use the convention of considering $\nabla f^T$ as a row vector.}
\[DF(\tau)\equiv\begin{pmatrix} \nabla_\Theta f(\Theta+\tau(\Theta'-\Theta),x^{(1)})^T\\\nabla_\Theta f(\Theta+\tau(\Theta'-\Theta),x^{(2)})^T\\\vdots\\\nabla_\Theta f(\Theta+\tau(\Theta'-\Theta),x^{(n)})^T \end{pmatrix}\, \]
we can apply Lagrange theorem \ref{lagrange} to obtain
\begin{align}
\nonumber\|f(\Theta',X)-Y\|_2&=\Bigg\|\left[\int_0^1 DF(\tau)d\tau\right](\Theta'-\Theta)+(f(\Theta,X))-Y)\Bigg\|_2\\
\nonumber&= \Bigg\|\int_0^1 \big(DF(\tau)(\Theta'-\Theta)+(f(\Theta,X)-Y)\big)d\tau\Bigg\|_2\\
\nonumber&\leq \int_0^1d\tau\big\| DF(\tau)(\Theta'-\Theta)+f(\Theta,X)-Y\big\|_2\\
&= \big\| DF(\xi)(\Theta'-\Theta)+f(\Theta,X)-Y\big\|_2,
\end{align}
where $\xi\in[0,1]$.
By the definition of gradient descent (\ref{sopraqui}), we also have
\begin{align}
\nonumber \|f(\Theta',X)-Y\|_2&\leq \big\| -\frac{\eta_0}{N_K(m)}DF(\xi)\nabla_\Theta\mathcal{L}(\Theta) +f(\Theta,X)-Y\big\|_2 \\
\nonumber &\leq \big\| -\frac{\eta_0}{N_K(m)}DF(\xi)\nabla_\Theta f(\Theta,X^T) (f(\Theta,X)-Y)+f(\Theta,X)-Y\big\|_2\\
\nonumber &=\Bigg\| \left(\id-\frac{\eta_0}{N_K(m)}DF(\xi)\nabla_\Theta f(\Theta,X^T) \right) (f(\Theta,X)-Y)\Bigg\|_2 \\
&\leq \Bigg\| \left(\id-\frac{\eta_0}{N_K(m)}DF(\xi)\nabla_\Theta f(\Theta,X^T)\right)\Bigg\|_\mathcal{L} \|f(\Theta,X)-Y\|_2.
\end{align}
Calling $\tilde \Theta = \Theta+\xi(\Theta'-\Theta)$, by convexity of the ball,
\[ \tilde \Theta\in B_{r(m)}(\Theta_0)\quad  \text{because}\quad 
\begin{cases} \|\Theta-\Theta_0\|_\infty\leq r(m)\\
\|\Theta'-\Theta_0\|_\infty\leq r(m)\end{cases}.\]
Now we have to estimate
\begin{align}
\nonumber \Bigg\| \Bigg(\id&-\frac{\eta_0}{N_K(m)}DF(\xi)(\nabla_\Theta f(\Theta,X^T)\Bigg)\Bigg\|_\mathcal{L}\\
\nonumber&=\Bigg\| \left(\id-\frac{\eta_0}{N_K(m)}\sum_{i=1}^{Lm}\partial_{\theta_i}f(\tilde \Theta, X)\partial_{\theta_i}f(\Theta,X^T)\right)\Bigg\|_\mathcal{L}\\
\nonumber&=\Bigg\| \Bigg(\id-\eta_0 \Big(K(X,X^T) - K(X,X^T)+ \hat K_{\Theta_0}(X,X^T)\\
\nonumber&\phantom{=\Bigg\| \Bigg(\id-\eta_0 \Big(}
-\hat K_{\Theta_0}(X,X^T)+\frac{1}{N_K(m)}\sum_{i=1}^{Lm}\partial_{\theta_i}f(\tilde \Theta, X)\partial_{\theta_i}f(\Theta,X^T)\Big)\Bigg)\Bigg\|_\mathcal{L}\\
\nonumber&\leq \|\id-\eta_0 K(X,X^T)\|_\mathcal{L}+\eta_0\|K(X,X^T)-\hat K_{\Theta_0}(X,X^T)\|_\mathcal{L}\\
\nonumber&\phantom{\leq} +\eta_0\Bigg\|\sum_{i=1}^{Lm}\partial_{\theta_i}f(\Theta_0, X)\partial_{\theta_i}f(\Theta_0,X^T)- \sum_{i=1}^{Lm}\partial_{\theta_i}f(\tilde \Theta, X)\partial_{\theta_i}f(\Theta,X^T)\Bigg\|_\mathcal{L}\\
\nonumber&\leq \|\id-\eta_0 K(X,X^T)\|_\mathcal{L}+\eta_0\|K(X,X^T)-\hat K_{\Theta_0}(X,X^T)\|_\mathcal{L}\\
\nonumber&\phantom{\leq} +\frac{\eta_0}{N_K(m)}\Bigg\|\sum_{i=1}^{Lm}\big(\partial_{\theta_i}f(\Theta_0, X)-\partial_{\theta_i}f(\tilde \Theta, X\big)\partial_{\theta_i}f(\Theta_0,X^T)\Bigg\|_\mathcal{L}\\
&\phantom{\leq} +\frac{\eta_0}{N_K(m)}\Bigg\|\sum_{i=1}^{Lm}\partial_{\theta_i}f(\tilde \Theta, X)\big(\partial_{\theta_i}f(\Theta_0,X^T)-\partial_{\theta_i}f(\Theta,X^T)\big)\Bigg\|_\mathcal{L}.
\end{align}
Each term can be bounded as follows. Let us start with
\begin{align}
\|\id-\eta_0 K(X,X^T)\|_\mathcal{L}\leq \max_{\lambda\in \text{Spec} K}|1-\eta_0\lambda|.
\end{align}
Since $f(\lambda)=|1-\eta_0\lambda|$ is convex,
\begin{align}
\|\id-\eta_0 K(X,X^T)\|_\mathcal{L}\leq \max\{|1-\eta_0\lambda_{\min}^K|,|1-\eta_0\lambda_{\max}^K|\}
\end{align}
and, since $\eta_0<\frac{2}{\lambda_{\min}^K+\lambda_{\max}^K}$,
\[1-\eta_0\lambda^\ast= 0 \quad \iff \quad \lambda^\ast=\frac{1}{\eta_0}>\frac{\lambda_{\min}^K+\lambda_{\max}^K}{2}.\]
Therefore, we can conclude that
\[\|\id-\eta_0 K(X,X^T)\|_\mathcal{L}\leq 1-\eta_0\lambda_{\min}^K\]
as illustrated in \autoref{flambda}.
\begin{figure}[ht]
\centering
\includegraphics[width=0.70\textwidth]{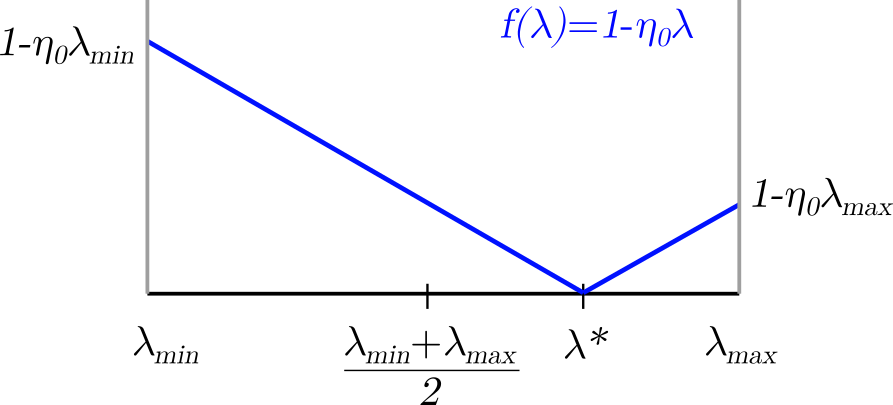}
\caption{Behaviour of $f(\lambda)$}
\label{flambda}
\end{figure}
Now we consider
\begin{align}
\eta_0\|K(X,X^T)-\hat K_{\Theta_0}(X,X^T)\|_\mathcal{L}\leq\|K(X,X^T)-\hat K_{\Theta_0}(X,X^T)\|_F.
\end{align}
We can claim, by Theorem \ref{ntkconv} and Lemma \ref{lambdamin}, the existence of $m_1\in\mathbb{N}$ such that
\[ \|K(X,X^T)-\hat K_{\Theta_0}(X,X^T)\|_F\leq \left(\frac{1-\zeta}{2}\right)\lambda_{\min}^K \qquad \forall \,m\geq m_1\]
with probability at least $1-\delta$. \\
Finally, using Lemma \ref{lemma},
\begin{align}
\nonumber\Bigg\|\sum_{i=1}^{Lm}\big(&\partial_{\theta_i}f(\Theta_0, X)-\partial_{\theta_i}f(\tilde \Theta, X\big)\partial_{\theta_i}f(\Theta_0,X^T)\Bigg\|_\mathcal{L}\\
\nonumber&\leq \Bigg\|\sum_{i=1}^{Lm}\big(\partial_{\theta_i}f(\Theta_0, X)-\partial_{\theta_i}f(\tilde \Theta, X\big)\partial_{\theta_i}f(\Theta_0,X^T)\Bigg\|_F\\
\nonumber&\leq \sqrt{\sum_{j,j'=1}^n\left(\sum_{i=1}^{Lm}\big(\partial_{\theta_i}f(\Theta_0, x^{(j)})-\partial_{\theta_i}f(\tilde \Theta, x^{(j)}\big)\partial_{\theta_i}f(\Theta_0,x^{(j')})\right)^2}\\
\nonumber&\leq \sqrt{\sum_{j,j'=1}^n\left(\sup_{i'}\big|\partial_{\theta_{i'}}f(\Theta_0, x^{(j)})-\partial_{\theta_{i'}}f(\tilde \Theta, x^{(j)})\big|\sum_{i=1}^{Lm}\left|\partial_{\theta_i}f(\Theta_0,x^{(j')})\right|\right)^2}\\
\nonumber&\leq \sup_x \sup_{i'}\big|\partial_{\theta_{i'}}f(\Theta_0, x)-\partial_{\theta_{i'}}f(\tilde \Theta, x)\big|\sqrt{\sum_{j,j'=1}^n\left(\sum_{i=1}^{Lm}\left|\partial_{\theta_i}f(\Theta_0,x^{(j')})\right|\right)^2}\\
\nonumber&\leq \sup_x \sup_{i'}\big|\partial_{\theta_{i'}}f(\Theta_0, x)-\partial_{\theta_{i'}}f(\tilde \Theta, x)\big|\sup_{x'}\sum_{i=1}^{Lm}\left|\partial_{\theta_i}f(\Theta_0,x'))\right|\sqrt{\sum_{j,j'=1}^n 1}\\
\nonumber& \leq n \sup_{x} \big\|\nabla_\Theta f(\Theta_0, x)-\nabla_\Theta f(\tilde \Theta, x\big)\big\|_\infty \sup_{x'}\big\|\nabla_\Theta f(\Theta_0,x')\big\|_1\\
&\leq n\cdot 4 \,\frac{|\mathcal{M}|^2|\mathcal{N}|}{N(m)}\|\Theta_0-\tilde\Theta\|_\infty\cdot \frac{2\Sigma_1}{N(m)}\leq 8n \,\frac{\Sigma_1|\mathcal{M}|^2|\mathcal{N}|}{N^2(m)}r(m).
\end{align}
Hence
\begin{align}
\nonumber\Bigg\| \Bigg(\id&-\frac{\eta_0}{N_K(m)}DF(\xi)(\nabla_\Theta f(\Theta,X^T)\Bigg)\Bigg\|_\mathcal{L}\\
&\leq 1-\eta_0\lambda_{\min}^K+\left(\frac{1-\zeta}{2}\right)\eta_0\lambda_{\min}^K+8n\eta_0 \,\frac{\Sigma_1|\mathcal{M}|^2|\mathcal{N}|}{N_K(m)N^2(m)}r(m).
\end{align}
Because of (\ref{hplim}), there exist $m_2\in\mathbb{N}$ such that
\[8n\eta_0 \,\frac{\Sigma_1|\mathcal{M}|^2|\mathcal{N}|}{N_K(m)N^2(m)}r(m)\leq \left(\frac{1-\zeta}{2}\right)\eta_0\lambda_{\min}^K\qquad \forall m\geq m_2.\]
This yields
\begin{align}
\nonumber\Bigg\| \Bigg(\id&-\frac{\eta_0}{N_K(m)}DF(\xi)(\nabla_\Theta f(\Theta,X^T)\Bigg)\Bigg\|_\mathcal{L}\\
&\leq 1-\eta_0\lambda_{\min}^K+(1-\zeta)\eta_0\lambda_{\min}^K = 1-\zeta\eta_0\lambda_{\min}^K.
\end{align}
Therefore
\[ \|f(\Theta',X)-Y\|_2\leq (1-\zeta\eta_0\lambda_{\min}^K)\|f(\Theta,X)-Y\|_2\]
for all $m\geq \bar m = \max \{m_1,m_2\}$.

\end{proof}

Now we have all the ingredients to prove Theorem \ref{graddesc0}.

Because of Corollary \ref{corollaryR}, we can ask
\begin{align}
\mathbb{P}\left(\|F(0)-Y\|_2<R\right)\geq 1-\frac{\delta}{2}\quad \forall m\geq m_0 \text{ for some } R=\sqrt{n\log(2n)} R_0
\end{align}
with probability at least $1-\frac{\delta}{2}$.

As in the proof of the Theorem \ref{gradfl}, we introduce
\[\rho(m)=\frac{6R_0}{\lambda^K_{\min}}n \sqrt{\log(2n)}\frac{|\mathcal{M}|}{N_K(m)N(m)}.\]

We prove (\ref{desc1}) and (\ref{desc2}) by induction. If $t=0$, then
\[\|F(0)-Y\|_2<R\]
holds because of the previous assumption; (\ref{desc2}) does not have to be proved for $t=0$.
Let us suppose that (\ref{desc1}) and (\ref{desc2}) hold for any $t\leq t^\ast$.  Recalling that
\[ 
\Theta_{t+1}-\Theta_{t} = -\frac{\eta_0}{N_K(m)} \nabla_\Theta f(\Theta_t,X)\cdot (F(t)-Y)
\]
and using Lemma \ref{lemma}, we have
\begin{align}
\nonumber|\theta_i(t+1)-\theta_i(t)|&\leq\Big|\frac{\eta_0}{N_K(m)}\partial_{\theta_i}f(\Theta_t,X)\cdot (F(t)-Y)\Big|\\
&\leq \frac{\eta_0}{N_K(m)}\modifica{2\sqrt n \frac{|\mathcal{M}|}{N(m)}}\|F(t)-Y\|_2
\end{align}
Hence, because (\ref{desc1}) holds for any $t\leq t^\ast$,
\begin{align}
\nonumber |\theta_i(t+1)-\theta_i(t)|&\leq \frac{\eta_0}{N_K(m)}2\sqrt n\,\frac{|\mathcal{M}|}{N(m)}\sqrt {n\log(2n)} R_0\left(1-\frac{1}{3}\eta_0\lambda_{\min}^K\right)^{t}\\ 
&= 2R_0\eta_0n\sqrt{\log(2n)}\,\frac{|\mathcal{M}|}{N_K(m)N(m)}\left(1-\frac{1}{3}\eta_0\lambda_{\min}^K\right)^{t}\qquad \forall \, t\leq t^\ast.
\end{align}
Therefore, using the triangle inequality and considering the supremum over $i$ of the previous inequalities,
\begin{align}
\nonumber \|\Theta_{t^\ast+1}-\Theta_0\|_\infty&\leq \sum_{k=0}^{t^\ast}\|\Theta_{k+1}-\Theta_k\|_\infty\\
\nonumber &\leq 2R_0n\sqrt{\log(2n)}\eta_0\,\frac{|\mathcal{M}|}{N_K(m)N(m)}\sum_{k=0}^{t^\ast}\left(1-\frac{1}{3}\eta_0\lambda_{\min}^K\right)^k\\
\nonumber &\leq 2R_0n\sqrt{\log(2n)}\eta_0\,\frac{|\mathcal{M}|}{N_K(m)N(m)}\sum_{k=0}^\infty\left(1-\frac{1}{3}\eta_0\lambda_{\min}^K\right)^k\\
&=\frac{6R_0}{\lambda^K_{\min}}n\sqrt{\log(2n)}\frac{|\mathcal{M}|}{N_K(m)N(m)}=\rho(m).
\end{align}
This means that (\ref{desc2}) holds also for $t=t^\ast+1$.
Now, let us notice that
\[\lim_{m\to\infty} \frac{\Sigma_1|\mathcal{M}|^2|\mathcal{N}|}{N_K(m)N^2(m)}\rho(m)=\lim_{m\to\infty}\frac{6 R_0}{\lambda^K_{\min}}n\sqrt{\log(2n)}\frac{\Sigma_1|\mathcal{M}|^3|\mathcal{N}|}{N_K^2(m)N^3(m)}=0\]
by hypothesis. We invoke\footnote{Since $r(m)$ and $\zeta$ are always the same in the proof by induction, the lemma has to be claimed once.} Lemma \ref{step} with probablity to fail $\frac{\delta}{2}$ instead of $\delta$ and choosing the following parameters:
\[r(m)=\rho(m),\qquad \Theta'=\Theta_{t^\ast+1},\qquad \Theta=\Theta_{t^\ast},\qquad \zeta=\frac{1}{3}.\]
Therefore
\[ \|F(t^\ast+1)-Y\|_2\leq R_0\sqrt{n\log(2n)}\left(1-\frac{\eta_0\lambda_{\min}^K}{3}\right)^{t^\ast+1}\]
which is (\ref{desc1}) for $t=t^\ast+1$.

\subsection{Proof of Theorem \ref{unbgraddesc}}\label{proofunbgraddesc}

\begin{lemma}\label{llhp} The hypothesis 
\[\lim_{m\to\infty}\frac{Lm|\mathcal{M}|^4|\mathcal{N}|^2}{N^3(m)}=0\]
of 
Theorem \ref{unbgraddesc} ensures that the hypotheses of Theorem \ref{init}, of Theorem \ref{ntkconv} and of Theorem \ref{graddesc0} are satisfied, i.e.
\[\lim_{m\to\infty}\frac{m|\mathcal{M}|^2|\mathcal{N}|^2}{N^3(m)}=0,\label{checkhp1}\] \[\lim_{m\to\infty}\frac{1}{N_K^2(m)}\,\frac{\Sigma_2|\mathcal{M}|^2|\mathcal{N}|^2}{N^4(m)}=0,\label{checkhp2}\]
\[\lim_{m\to\infty}\frac{\Sigma_1|\mathcal{M}|^2|\mathcal{N}|}{N_K^2(m)N^3(m)}=0.\label{checkhp3}\]
\end{lemma}
\begin{proof}
It is easy to see that
\[0\leq\lim_{m\to\infty}\frac{m|\mathcal{M}|^2|\mathcal{N}|^2}{N^3(m)}\leq\lim_{m\to\infty}\frac{Lm|\mathcal{M}|^4|\mathcal{N}|^2}{N^3(m)}=0,\]
\[0\leq\frac{1}{N_K^2(m)}\,\frac{\Sigma_2|\mathcal{M}|^2|\mathcal{N}|^2}{N^4(m)}\leq\lim_{m\to\infty}\frac{Lm|\mathcal{M}|^4|\mathcal{N}|^2}{N^3(m)}=0,\]
\[0\leq\lim_{m\to\infty}\frac{\Sigma_1|\mathcal{M}|^2|\mathcal{N}|}{N_K^2(m)N^3(m)}\leq\lim_{m\to\infty}\frac{Lm|\mathcal{M}|^4|\mathcal{N}|^2}{N^3(m)}=0.\]
\end{proof}

Two simple lemmas are proved now, so that the computations will be simplified later.
\begin{lemma}[An elementary inequality]\label{elemineq}
Let $x\in\mathbb{R}$ such that
\[x^2\leq A+A'+Bx,\]
where $A,A',B>0$. Then
\[x\leq \sqrt{A}+\sqrt{A'}+B.\]
\end{lemma}
\begin{proof}
The equality $x^2= A+A'+Bx$ holds if and only if
\[x=x_1=\frac{B+\sqrt{B^2+4(A+A')}}{2} \quad \text{or} \quad x=x_2=\frac{B-\sqrt{B^2+4(A+A')}}{2},\]
so the inequality \modifica{$x^2\leq A+A'+Bx$} implies
\[x_2 \leq x\leq x_1.\]
Since
\begin{align}
\nonumber 2 x_1 &= B+\sqrt{B^2+4(A+A')}=B+\sqrt{B^2+(2\sqrt A)^2+(2\sqrt{A'})^2} \\
&\leq B+B+2\sqrt A+2\sqrt{A'}= 2(B+\sqrt A+\sqrt{A'}),
\end{align}
we have
\[x\leq \sqrt A+\sqrt{A'}+B.\]
\end{proof}

\begin{lemma}[Another elementary inequality]\label{elemineq2}
Let $x\in[0,1]$. Then
\[\frac{x}{4}\leq \sqrt{1-\frac{x}{2}}-\sqrt{1-x}.\]
\end{lemma}
\begin{proof}
First we notice that
\[\sqrt{1-\frac{x}{2}}-\sqrt{1-x}\geq 0.\]
Then we use that, for $x\in [0,1]$,
\[\sqrt{1-\frac{x}{2}}+\sqrt{1-x}\leq 2,\]
whence
\begin{align}
\nonumber \left(\sqrt{1-\frac{x}{2}}+\sqrt{1-x}\right)\left(\sqrt{1-\frac{x}{2}}-\sqrt{1-x}\right)&\leq 2\left(\sqrt{1-\frac{x}{2}}-\sqrt{1-x}\right)\\
\nonumber \frac{1}{2}\left(1-\frac{x}{2}-(1-x)\right)&\leq \sqrt{1-\frac{x}{2}}-\sqrt{1-x}\\
\frac{x}{4}&\leq \sqrt{1-\frac{x}{2}}-\sqrt{1-x}.
\end{align}
\end{proof}
Now we are ready to prove Theorem \ref{unbgraddesc}.
As usual, because of Corollary \ref{corollaryR} we can ask
\begin{align}
\mathbb{P}\left(\|F(0)-Y\|_2<R\right)\geq 1-\frac{\delta}{2}\quad \forall m\geq m_0
\end{align}
for some $R=\sqrt {n\log(2n)} R_0$ with probability at least $1-\frac{\delta}{2}$.

At each temporal step, let us introduce two different evolutions of the same parameter vector $\Theta_t$, one stochastic and one deterministic. More precisely, let us consider two stochastic processes $\Theta_t$ and $\tilde \Theta_t$ defined for $t\geq 0$ by
\[
\begin{dcases}
\Theta_{t+1}&=\Theta_t-\eta g^{(t)}(\Theta_t)\\
\tilde \Theta_{t+1}&=\Theta_t-\eta \nabla_\Theta\mathcal{L}(\Theta_t)
\end{dcases}.
\]
It is crucial to notice that the parameter vector appering in the RHS of the second equation is $\Theta_t$, not $\tilde \Theta_t$.\\
When we discussed the strategy to prove the theorem, we emphasized the necessity to bound the discrepancy between the loss function evalued at the new points $\Theta_{t+1}$ and $\tilde\Theta_{t+1}$. An upper bound can be computed starting from the Lipschitzness of the model. By Lemma \ref{lipf},
\[|f(\Theta_1,x)-f(\Theta_2,x)|\leq 2|\Theta|\frac{|\mathcal{M}|}{N(m)}\|\Theta_1-\Theta_2\|_\infty.\]
This implies the following bound on the growth $\mathcal{L}(\Theta)$:
\begin{align}
\nonumber |\mathcal{L}(\Theta_1)-\mathcal{L}(\Theta_2)|&=\frac{1}{2n}\sum_{i=1}^n\Big((f(\Theta_1,x^{(i)})-y^{(i)})^2-(f(\Theta_2,x^{(i)})-y^{(i)})^2\Big)\\
\nonumber &=\frac{1}{2n}\sum_{i=1}^n\big(f(\Theta_1,x^{(i)})-f(\Theta_2,x^{(i)})\big)\times\\
\nonumber &\phantom{=\frac{1}{2n}sum\big(f}\times\big(f(\Theta_1,x^{(i)})-y^{(i)}+f(\Theta_2,x^{(i)})-y^{(i)})\big)\\
\nonumber &\leq \frac{1}{\sqrt {2n}}\|f(\Theta_1,X)-f(\Theta_2,X)\|_2\times\\
\nonumber &\phantom{\leq \frac{1}{\sqrt 2n}||f(\Theta_1}\times\frac{1}{\sqrt {2n}}\Big(\|f(\Theta_1,X)-Y\|_2+\|f(\Theta_2,X)-Y\|_2\Big)\\
&\leq \sqrt 2|\Theta|\frac{|\mathcal{M}|}{N(m)}\big( \sqrt{\mathcal{L}(\Theta_1)}+\sqrt{\mathcal{L}(\Theta_2)}\big)\|\Theta_1-\Theta_2\|_\infty.
\end{align}
In particular,
\begin{align}
\nonumber \mathcal{L}(\Theta_1)&\leq \mathcal{L}(\Theta_2)+\sqrt 2|\Theta|\frac{|\mathcal{M}|}{N(m)}\sqrt{\mathcal{L}(\Theta_2)}\|\Theta_1-\Theta_2\|_\infty \\
& \phantom{\leq \mathcal{L}(\Theta_2)}+  \sqrt 2|\Theta|\frac{|\mathcal{M}|}{N(m)}\|\Theta_1-\Theta_2\|_\infty \sqrt{\mathcal{L}(\Theta_1)}.
\end{align}
So we can apply Lemma \ref{elemineq} with $x=\sqrt{\mathcal{L}(\Theta_1)}$ and
\[A=\mathcal{L}(\Theta_2),\qquad A'= \sqrt 2|\Theta|\frac{|\mathcal{M}|}{N(m)}\sqrt{\mathcal{L}(\Theta_2)}\|\Theta_1-\Theta_2\|_\infty,\]
\[B= \sqrt 2|\Theta|\frac{|\mathcal{M}|}{N(m)}\|\Theta_1-\Theta_2\|_\infty,\]
obtaining
\begin{align}
\nonumber\sqrt{\mathcal{L}(\Theta_1)}&\leq \sqrt{\mathcal{L}(\Theta_2)} + \sqrt[4]{\mathcal{L}(\Theta_2)}\sqrt{\modifica{\sqrt{2}}|\Theta|\frac{|\mathcal{M}|}{N(m)}}\|\Theta_1-\Theta_2\|_\infty^{1/2}\\
\nonumber &\phantom{\leq}+\sqrt 2|\Theta|\frac{|\mathcal{M}|}{N(m)}\|\Theta_1-\Theta_2\|_\infty\\
&\leq \left(\sqrt[4]{\mathcal{L}(\Theta_2)}+\sqrt{\sqrt 2|\Theta|\frac{|\mathcal{M}|}{N(m)}}\|\Theta_1-\Theta_2\|_\infty^{1/2}\right)^2.
\label{later}
\end{align}
The discrepancy between the loss function after a deterministic evolution of the parameters and the loss function after a stochastic evolution of the parameters is small when the outcome of the stochastic evolution is close to the result of the deterministic evolution. This can be done controlled by a concentration inequality. For any choice of $t\leq T$ and $\xi_k\in\mathbb{R}$ and $\Theta_k\in\mathbb{R}^{|\Theta|}$,

\[\Theta_{t+1}-\tilde \Theta_{t+1}=\eta (\nabla_\Theta\mathcal{L}(\Theta_t)-g^{(t)}(\Theta_t)),\]
\[\mathbb{E}\left[\theta_i(t+1)-\tilde \theta_i(t+1)\,\,\,\big|\,\,\, g^{(k)}(\Theta_{k})=\xi_k\quad \forall\, k<t\right]=0,\]
\begin{align}
\nonumber \text{Var}\Big[\theta_i(t+1)-\tilde \theta_i(t+1)\,\,\,&\big|\,\,\, g^{(k)}(\Theta_{k})=\xi_k\quad \forall\, k<t\Big]\\
&=\eta^2\text{Var}\left[g^{(t)}_i(\Theta_t)\,\,\,\big|\,\,\, g^{(k)}(\Theta_{k})=\xi_k\quad \forall\, k<t\right].
\end{align}
In order to have a cleaner notation, from now on we will not explitictly write that the variance and the probability are contitioned on any possible outcome of the previous estimators, even though we will always assume this.
By Chebyshev's inequality,
\[\mathbb{P}\left(|\theta_i(t+1)-\tilde \theta_i(t+1)|>k\eta\sqrt{\text{Var}\left[g^{(t)}_i(\Theta_t)\right]}\right)\leq\frac{1}{k^2}.\]
If we set 
\[k= \pi\sqrt{\frac{|\Theta|}{3\delta}}(t+1),\]
then
\begin{align}
\nonumber \mathbb{P}\Bigg(|\theta_i(t+1)-&\tilde \theta_i(t+1)|<\eta\pi\sqrt{\frac{|\Theta|}{3\delta}}(t+1)\sqrt{\text{Var}\left[g^{(t)}_i(\Theta_t)\right]}\quad \forall\, t\leq T, \forall i \Bigg)\\
\nonumber &\geq 1-\frac{\delta}{2}\, \frac{1}{|\Theta|}\sum_{i=1}^{Lm}\frac{6}{\pi^2}\sum_{t=0}^T\frac{1}{(t+1)^2}\geq  1-\frac{\delta}{2}\, \frac{6}{\pi^2}\sum_{t=0}^\infty\frac{1}{(t+1)^2}\\
&= 1-\frac{\delta}{2}.\label{unionbounddelta}
\end{align}
So, if
\begin{align}\label{stimag}
    \|\Theta_{t+1}-\tilde\Theta_{t+1}\|_\infty\leq \eta\pi\sqrt{\frac{|\Theta|}{3\delta}}(t+1)\sup_i\sqrt{\text{Var}\left[g^{(t)}_i(\Theta_t)\right]},
\end{align}
we can bound the loss function evalued at $\Theta_{t+1}$ by using (\ref{later})
\begin{align}
\nonumber \sqrt[4]{\mathcal{L}(\Theta_{t+1})}&\leq \sqrt[4]{\mathcal{L}(\tilde\Theta_{t+1})}+\sqrt{\sqrt 2|\Theta|\frac{|\mathcal{M}|}{N(m)}}\|\Theta_{t+1}-\tilde\Theta_{t+1}\|_\infty^{1/2}\\
&\leq \sqrt[4]{\mathcal{L}(\tilde\Theta_{t+1})}+\sqrt[4]{\frac{2\eta^2\pi^2}{3\delta}|\Theta|^3\frac{|\mathcal{M}|^2}{N^2(m)}(t+1)^2\sup_i\text{Var}\left[g^{(t)}_i(\Theta_t)\right]}.
\end{align}

Now we prove the following bounds by induction on $t$:
\[
\begin{dcases}
\|\Theta_t-\Theta_0\|_\infty\leq 2R_0{\lambda^K_{\min}}n\sqrt{\log(2n)}\frac{|\mathcal{M}|}{N_K(m)N(m)}\sum_{k=0}^{t-1}\left(1-\frac{\eta_0\lambda_{\min}^K}{3}\right)^k\\
\mathcal{L}(\Theta_t)\leq \frac{1}{2}R_0^2\log(2n) \left(1-\frac{\eta_0\lambda_{\min}^K}{3}\right)^{2t}
\end{dcases}.
\label{boundst}\]
If $t=0$, it is clear that they hold. Now, let us assume that they hold for $t=t^\ast$.

Using the deterministic evolution
\[ 
\tilde \Theta_{t^\ast+1}-\Theta_{t^\ast} = -\frac{\eta_0}{N_K(m)} \nabla_\Theta f(\Theta_{t^\ast},X)\cdot (F(t)-Y)
\]
and using Lemma \ref{lemma}
\begin{align}
\nonumber|\tilde\theta_i(t^\ast+1)-\theta_i(t^\ast)|&\leq\Big|\frac{\eta_0}{N_K(m)}\partial_{\theta_i}f(\Theta_{t^\ast},X)\cdot (F(t^\ast)-Y)\Big|\\
\nonumber&\leq \frac{\eta_0}{N_K(m)}\|\partial_{\theta_i}f(\Theta_{t^\ast},X)\|_2\|F(t^\ast)-Y\|_2\\
\nonumber&\leq \frac{\eta_0}{N_K(m)}2\sqrt n\,\frac{|\mathcal{M}|}{N(m)}\sqrt {n\log(2n)} R_0\left(1-\frac{1}{3}\eta_0\lambda_{\min}^K\right)^{t^\ast}\\ 
&= 2R_0\eta_0n\modifica{\sqrt{\log(2n)}}\,\frac{|\mathcal{M}|}{N_K(m)N(m)}\left(1-\frac{1}{3}\eta_0\lambda_{\min}^K\right)^{t^\ast}.
\end{align}
By the triangle inequality, using the inductive hypothesis,
\begin{align}
\nonumber\|\tilde\Theta_{t^\ast+1}-\Theta_0\|_\infty&\leq \|\tilde\Theta_{t^\ast+1}-\Theta_{t^\ast}\|_\infty+\|\Theta_{t^\ast}-\Theta_0\|_\infty\\
&\leq 2R_0\eta_0n\sqrt{\log(2n)}\,\frac{|\mathcal{M}|}{N_K(m)N(m)}\sum_{k=0}^{t^\ast}\left(1-\frac{1}{3}\eta_0\lambda_{\min}^K\right)^k.
\end{align}
In particular, both $\|\Theta_{t^\ast}-\Theta_0\|_\infty$ and $\|\tilde \Theta_{t^\ast+1}-\Theta_0\|_\infty$ are bounded by
\begin{align}
\nonumber 2R_0\eta_0n\sqrt{\log(2n)}\,&\frac{|\mathcal{M}|}{N_K(m)N(m)}\sum_{k=0}^\infty\left(1-\frac{1}{3}\eta_0\lambda_{\min}^K\right)^k\\
&\qquad=\frac{6 R_0}{N_K(m)\lambda^K_{\min}}n\sqrt{\log(2n)}\frac{|\mathcal{M}|}{N(m)}=\rho(m).
\end{align}

Now, let us notice that
\[\lim_{m\to\infty} \frac{\Sigma_1|\mathcal{M}|^2|\mathcal{N}|}{N_K(m)N^2(m)}\rho(m)=\lim_{m\to\infty}\frac{6 R_0}{\lambda^K_{\min}}n\sqrt{\log(2n)}\frac{\Sigma_1|\mathcal{M}|^3|\mathcal{N}|}{N_K^2(m)N^3(m)}=0\]
by hypothesis. We invoke\footnote{Since $r(m)$ and $\zeta$ are always the same in the proof by induction, the lemma has to be claimed once.} Lemma \ref{step} with probablity to fail $\frac{\delta}{4}$ instead of $\delta$ and choosing the following parameters:
\[r(m)=\rho(m),\qquad \Theta'=\tilde\Theta_{t^\ast+1},\qquad \Theta=\Theta_{t^\ast},\qquad \zeta=\frac{2}{3},\]
Therefore
\[ \|f(\tilde \Theta_{t^\ast+1},X)-Y\|_2\leq \left(1-\frac{2}{3}\eta_0\lambda_{\min}^K\right)\|f(\Theta_{t^\ast},X)-Y\|_2,\]
which yields
\begin{align}
\nonumber \sqrt[4]{\mathcal{L}(\Theta_{t^\ast+1})}
&\leq \sqrt{1-\frac{2}{3}\eta_0\lambda_{\min}^K}\sqrt[4]{\mathcal{L}(\Theta_{t^\ast})}\\
&\phantom{\leq}+\sqrt[4]{\frac{2\eta^2\pi^2}{3\delta}|\Theta|^3\frac{|\mathcal{M}|^2}{N^2(m)}(t+1)^2\text{Var}\left[g^{(t)}_i(\Theta_{t^\ast})\right]}.
\end{align}
We recall that 
\[\eta=\frac{n}{N_K(m)}\eta_0.\]
By hypothesis (\ref{Eg2b}), the variance is sufficiently small:
\begin{align}
\nonumber\sqrt[4]{\mathcal{L}(\Theta_{t^\ast+1})}
&\leq \sqrt{1-\frac{2}{3}\eta_0\lambda_{\min}^K}\sqrt[4]{\mathcal{L}(\Theta_{t^\ast})}+\sqrt[4]{\frac{2\eta_0^2\pi^2}{3}
c_0\eta_0^2\left(\lambda_{\min}^K\right)^4\mathcal{L}(\Theta_{t^\ast})}\\
\nonumber&\leq \left(\sqrt{1-\frac{2}{3}\eta_0\lambda_{\min}^K}+\eta_0\lambda_{\min}^K\sqrt[4]{\frac{2\pi^2}{3}
\frac{1}{864\pi^2}}\right)\sqrt[4]{\mathcal{L}(\Theta_{t^\ast})}\\
&\leq \left(\sqrt{1-\frac{2}{3}\eta_0\lambda_{\min}^K}+\frac{\eta_0\lambda_{\min}^K}{6}\right)\sqrt[4]{\mathcal{L}(\Theta_{t^\ast})}.
\end{align}
Using Lemma \ref{elemineq2} with $x/4=\eta_0\lambda_{\min}^K/6$,
\begin{align}
\nonumber \sqrt[4]{\mathcal{L}(\Theta_{t^\ast+1})}
&\leq \left(\sqrt{1-\frac{2}{3}\eta_0\lambda_{\min}^K}+\sqrt{1-\frac{\eta_0\lambda_{\min}^K}{3}}-\sqrt{1-\frac{2}{3}\eta_0\lambda_{\min}^K}\right)\sqrt[4]{\mathcal{L}(\Theta_{t^\ast})}\\
&=\sqrt{1-\frac{\eta_0\lambda_{\min}^K}{3}}\sqrt[4]{\mathcal{L}(\Theta_{t^\ast})},
\end{align}
whence
\begin{align}
\mathcal{L}(\Theta_{t^\ast+1})
&\leq \left(1-\frac{\eta_0\lambda_{\min}^K}{3}\right)^2 \mathcal{L}(\Theta_{t^\ast}).
\end{align}
Then, using the inductive hypothesis,
\[\mathcal{L}(\Theta_{t^\ast+1})\leq \frac{R_0^2}{2}\log(2n) \left(1-\frac{\eta_0\lambda_{\min}^K}{3}\right)^{2(t^\ast+1)}.\]
So, both the bounds (\ref{boundst}) hold for $t=t^\ast+1$. By induction, they are valid for any $t\in\mathbb{N}$.\\
Using Lemma \ref{lemma},
\begin{align}
\nonumber|\tilde\theta_i(t+1)-\theta_i(t)|&=\eta|\partial_{\theta_i}\mathcal{L}(\Theta_t)|=\frac{\eta}{n}\left|\sum_{j=1}^n\partial_{\theta_i}f(\Theta_t,x^{(j)})\big(f(\Theta_t,x^{(j)})-y^{(j)}\big)\right|\\
\nonumber &\leq \frac{\sqrt 2\eta}{\sqrt n}\|\partial_{\theta_i}f(\Theta,X)\|_2\frac{1}{\sqrt {2n}}\|F(t)-Y\|_2\\
&=2\sqrt 2 \eta \frac{|\mathcal{M}|}{N(m)}\sqrt{\mathcal{L}(\Theta_t)}.
\end{align}
By the triangle inequality,
\begin{align}
\nonumber\|\Theta_{t+1}&-\Theta_{t}\|_\infty\leq \|\Theta_{t+1}-\tilde \Theta_{t+1}\|_\infty+\|\tilde\Theta_{t+1}-\Theta_{t}\|_\infty\\
\nonumber&\leq \eta\pi\sqrt{\frac{|\Theta|}{3\delta}}(t+1)\sqrt{\text{Var}\left[g^{(t)}_i(\Theta_t)\right]}+2\sqrt 2 \eta \frac{|\mathcal{M}|}{N(m)}\sqrt{\mathcal{L}(\Theta_t)}\\
\nonumber&\leq 
\eta\pi\sqrt{\frac{|\Theta|}{3\delta}}(t+1)\sqrt{c_0}\eta_0\frac{\left(\lambda_{\min}^K\right)^2}{n}\frac{N_K(m)N(m)}{|\mathcal{M}||\Theta|^{3/2}}\,\frac{\sqrt \delta}{(t+1)}\,\sqrt{\mathcal{L}(\Theta_t)}\\
\nonumber&\phantom{\leq}+2\sqrt 2 \eta \frac{|\mathcal{M}|}{N(m)}\sqrt{\mathcal{L}(\Theta_t)}\\
\nonumber &\leq \eta_0 \frac{R_0}{\sqrt 2}\sqrt{\log(2n)}\,\times\\
&\phantom{\leq}\times\left(\eta_0\frac{\pi\sqrt{c_0}}{\sqrt 3}\left(\lambda_{\min}^K\right)^2\frac{N(m)}{|\mathcal{M}||\Theta|}+2\sqrt 2\,n \frac{|\mathcal{M}|}{N_K(m)N(m)}\right)\left(1-\frac{\eta_0\lambda_{\min}^K}{3}\right)^t.
\label{triangle}
\end{align}
Hence, using again the triangle inequality,
\begin{align}
\nonumber\|\Theta_T-\Theta_0\|_\infty&\leq\sum_{t=0}^{T-1}\|\Theta_{t+1}-\Theta_{t}\|_\infty\\
\nonumber&\leq \eta_0 R_0\sqrt{\log(2n)}\left(\eta_0\frac{\pi\sqrt{c_0}}{\sqrt 6}\left(\lambda_{\min}^K\right)^2\frac{N(m)}{|\mathcal{M}||\Theta|}+2n \frac{|\mathcal{M}|}{N_K(m)N(m)}\right)\times \\
\nonumber& \qquad \times\sum_{t=0}^\infty\left(1-\frac{\eta_0\lambda_{\min}^K}{3}\right)^t\\
\nonumber&=\frac{3R_0\sqrt{\log(2n)}}{\modifica{\lambda^K_{\min}}}\left(\eta_0\frac{\pi\sqrt{c_0}}{\sqrt 6}\left(\lambda_{\min}^K\right)^2\frac{N(m)}{|\mathcal{M}||\Theta|}+2n \frac{|\mathcal{M}|}{N_K(m)N(m)}\right)\\
&= \sqrt{\log(2n)}\left(R_1(\delta)\eta_0 \lambda_{\min}^K\frac{N(m)}{|\mathcal{M}||\Theta|}+R_2(\delta)\frac{n}{\lambda_{\min}^K} \frac{|\mathcal{M}|}{N_K(m)N(m)}\right),
\end{align}
where we recalled the dependence $R_0=R_0(\delta)$. This proves (\ref{noisy22}).
Finally, we use Theorem \ref{powerful} to give the upper bound
\begin{align}
\nonumber |f(\Theta,x)&-f^{\mathrm{lin}}(\Theta,x)|\\
\nonumber &\leq \frac{Lm|\mathcal{M}|^2|\mathcal{N}|}{N(m)}\log(2n)\left(R_1(\delta)\eta_0 \lambda_{\min}^K\frac{N(m)}{|\mathcal{M}||\Theta|}+R_2(\delta)\frac{n}{\lambda_{\min}^K} \frac{|\mathcal{M}|}{N_K(m)N(m)}\right)^2\\
&\leq 2R_1^2(\delta)\log(2n)\eta_0^2\left(\lambda_{\min}^K\right)^2\frac{|\mathcal{N}|N(m)}{Lm}+2\frac{R^2_2(\delta)}{\left(\lambda_{\min}^K\right)^2}n^2\log(2n)\frac{Lm|\mathcal{M}|^4|\mathcal{N}|}{N_K^2(m)N^3(m)},
\end{align}
\modifica{where in the last line we have used that $(x+y)^2\leq 2x^2+2y^2$}. And this concludes the proof.

\subsection{Proof of Lemma \ref{secondorder2}}\label{proofsecondorder2}
We need to introduce some definitions in order to deal with the discretization of time and with the noise.
Let us consider the case in which the gradient of the cost is exactly known and let us define $r(t,x)$ as
\[f(\Theta_{t+1},x)=:f(\Theta_t,x)+\nabla_\Theta f(\Theta_t,x)^T(\Theta_{t+1}-\Theta_t)+r(t,x),\]
so that the discrete evolution reads
\begin{align}
\nonumber
\Delta f(\Theta_t,x) := f(\Theta_{t+1},x)-f(\Theta_t,x)&=\nabla_\Theta f(\Theta_t,x)^T(\Theta_{t+1}-\Theta_t)+r(t,x)\\
\nonumber
&= -\eta\nabla_\Theta f(\Theta_t,x)^T \nabla_\Theta\mathcal{L}(\Theta_t)+r(t,x)\\
& = -\eta_0 \hat K_{\Theta_t}(x,X^T)(F(t)-Y)+r(t,x).
\end{align}
Let us now generalize the previous equation in the presence of the noise due the finite number of measurements; as before, we introduce $r(t,x)$ as
\[f(\Theta_{t+1},x)=:f(\Theta_t,x)+\nabla_\Theta f(\Theta_t,x)^T(\Theta_{t+1}-\Theta_t)+r(t,x),\]
but we also define
\[g^{(t)}(\Theta_t)=:\nabla \mathcal{L}(\Theta_t)+\epsilon(t) \qquad \text{and}\qquad \delta(t,x)=:-\eta \nabla_\Theta f(\Theta_t,x)^T \epsilon(t)\]
so that
\[\Theta_{t+1}-\Theta_t =-\eta g^{(t)}(\Theta_t) = -\eta \big(\nabla_\Theta\mathcal{L}(\Theta_t)+\epsilon(t)\big) \label{thetaepsilon} \]
whence
\begin{align}
\nonumber
\Delta f(\Theta_t,x) := f(\Theta_{t+1},x)-f(\Theta_t,x)&=\nabla_\Theta f(\Theta_t,x)^T(\Theta_{t+1}-\Theta_t)+r(t,x)\\
\nonumber
&= -\eta\nabla_\Theta f(\Theta_t,x)^T \big(\nabla_\Theta\mathcal{L}(\Theta_t)+\epsilon(t)\big)+r(t,x)\\
\label{connoise}
& = -\eta_0 \hat K_{\Theta_t}(x,X^T)(F(t)-Y)+r(t,x)+\delta(t,x),
\end{align}
therefore the noiseless case can be recovered from (\ref{connoise}) by setting $\delta(t)=0$. We extend the time parameter $t$ to semi-integers as follows
\begin{align}
\nonumber f(\Theta_{t+1/2},x) :&= \frac{f(\Theta_{t+1},x)+f(\Theta_{t},x)}{2}\\
\nonumber &=\frac{1}{2}\big(2f(\Theta_t,x)+\Delta f(\Theta_t,x)\big)\\
& =f(\Theta_t,x) -\frac{\eta_0}{2} \hat K_{\Theta_t}(x,X^T)(F(t)-Y)+\frac{1}{2}(r(t,x)+\delta(t,x)).\label{ff1}
\end{align}
This compact notation will be useful in the upcoming computations. It is simple to verify that
\begin{align}
    \nonumber 
    f^{\mathrm{lin}}(\Theta_{t+1/2},x) :&= \frac{f^{\mathrm{lin}}(\Theta_{t+1},x)+f^{\mathrm{lin}}(\Theta_{t},x)}{2}\\
    \label{ff2}
    &= f^{\mathrm{lin}}(\Theta_t,x) -\frac{\eta_0}{2} \hat K_{\Theta_0}(x,X^T)(F^{\mathrm{lin}}(t)-Y).
\end{align}
Later we will also use that, for any function $g(t)$, the following identity holds
\begin{align}
\nonumber \frac{1}{2}\Delta g^2(t) &= \frac{1}{2}\big(g^2(t+1)-g^2(t)\big)\\
&=\frac{1}{2}\big(g(t+1)+g(t)\big)\big(g(t+1)-g(t)\big)=g(t+\um)\Delta g(t),
\end{align}
where we used the previous convention on semi-integers time values.

The following lemma is a crucial ingredient to prove Lemma \ref{secondorder2}.
\begin{lemma}[Bounding the discrepancy on the examples, discrete time case]\label{improveqltn}
Let us assume a condition on the variance of the estimator of the gradient stronger than (\ref{Eg2b}) by a factor $\xi(m)$: for any choice of $t\leq T$ and $\xi_k\in\mathbb{R}$ and $\Theta_k\in\mathbb{R}^{|\Theta|}$,
\begin{align}
\nonumber \mathrm{Var}\Big[g_i^{(t)}(\Theta)\,\,\,&\big|\,\,\, g^{(k)}(\Theta_{k})=\xi_k\quad \forall\, k<t\Big]\\
&\leq c_0\eta_0^2\,\frac{\left(\lambda_{\min}^K\right)^4}{n^2}\frac{N_K^2(m)N^2(m)}{|\mathcal{M}|^2|\Theta|^3}\xi(m)\,\frac{\delta/4}{(t+1)^2}\,\mathcal{L}(\Theta),
\end{align}
where $c_0=\frac{1}{864\pi^2}$.
Then, for any $\delta>0$, there exists $C,C'$ and $m_0\in\mathbb{N}$ such that, for each $m\geq m_0$
\begin{align}
    \|F(t)-&F^{\mathrm{lin}}(t)\|_2
    \leq \frac{C}{(\lambda_{\min}^K)^2}n^2\sqrt n \log(2n)\frac{Lm|\mathcal{M}|^4|\mathcal{N}|^3}{N^3(m)}+ \Delta(m)
\end{align}
with probability at least $1-\delta$, where 
\[ \Delta(m)=  C'\eta_0\lambda_{\min}^K\sqrt{n\log(2n)}\sqrt{\xi(m)} \label{defdelta}\]
\end{lemma}
\begin{proof}
We will use the compact notation mentioned above:
\[F(t)=f(\Theta_t,X),\quad F^{\mathrm{lin}}(t)=f^{\mathrm{lin}}(\Theta_t^{\mathrm{lin}},X)\]
\[F(t+\um)=\frac{f(\Theta_{t+1},X)+f(\Theta_{t},X)}{2},\quad F^{\mathrm{lin}}(t+\um)=\frac{f^{\mathrm{lin}}(\Theta^{\mathrm{lin}}_{t+1},X)+f(\Theta^{\mathrm{lin}}_{t},X)}{2}\]
We recall that $\hat K_{\Theta_t}:=\hat K_{\Theta_t}(X,X^T)$.
\begin{align}
\nonumber
\frac{1}{2}\Delta ||F(t)&-F^{\mathrm{lin}}(t)||_2^2\\
\nonumber
&=\frac{1}{2}\left(||F(t+1)-F^{\mathrm{lin}}(t+1)||_2^2-||F(t)-F^{\mathrm{lin}}(t)||_2^2\right)\\
\nonumber
&= \left(F(t+1)-F(t)-F^{\mathrm{lin}}(t+1)+F^{\mathrm{lin}}(t)\right)^T\left(F(t+\um)-F^{\mathrm{lin}}(t+\um)\right)\\
\nonumber
&= \left(\Delta F(t)-\Delta F^{\mathrm{lin}}(t)\right)^T\left(F(t+\um)-F^{\mathrm{lin}}(t+\um)\right)\\
\nonumber
&=-\eta_0\left(\hat K_{\Theta_t}(F(t)-Y)-\hat K_{\Theta_0}(F^{\mathrm{lin}}(t)-Y)\right)^T\left(F(t+\um)-F^{\mathrm{lin}}(t+\um)\right)\\
\nonumber
&\qquad + \big(r(t,X)+\delta(t,X)\big)^T\left(F(t+\um)-F^{\mathrm{lin}}(t+\um)\right)\\
\nonumber
&=-\eta_0\left((\hat K_{\Theta_t}-\hat K_{\Theta_0})(F(t)-Y)-\hat K_{\Theta_0}(F^{\mathrm{lin}}(t)-F(t))\right)^T\times\\
\nonumber& \qquad \times \left(F(t+\um)-F^{\mathrm{lin}}(t+\um)\right)\\
\nonumber
&\qquad + \big(r(t,X)+\delta(t,X)\big)^T\left(F(t+\um)-F^{\mathrm{lin}}(t+\um)\right)\\
\nonumber
&=-\eta_0(F(t)-Y)^T(\hat K_{\Theta_t}-\hat K_{\Theta_0})\left(F(t+\um)-F^{\mathrm{lin}}(t+\um)\right)\\
\nonumber
&\qquad + \eta_0(F^{\mathrm{lin}}(t)-F(t))^T\hat K_{\Theta_0}\left(F(t+\um)-F^{\mathrm{lin}}(t+\um)\right)\\
\nonumber
&\qquad + \big(r(t,X)+\delta(t,X)\big)^T\left(F(t+\um)-F^{\mathrm{lin}}(t+\um)\right)\\
&= A(t)+B(t)+C(t)
\end{align}
where
\begin{align}
    A(t)&:= -\eta_0(F(t)-Y)^T(\hat K_{\Theta_t}-\hat K_{\Theta_0})\left(F(t+\um)-F^{\mathrm{lin}}(t+\um)\right)\\
    B(t)&:= \eta_0(F^{\mathrm{lin}}(t)-F(t))^T\hat K_{\Theta_0}\left(F(t+\um)-F^{\mathrm{lin}}(t+\um)\right)\\
    C(t)&:=\big(r(t,X)+\delta(t,X)\big)^T\left(F(t+\um)-F^{\mathrm{lin}}(t+\um)\right)
\end{align}
Now, in the absence of statistical noise thanks to Theorem \ref{graddesc0} -- which we invoke with probability at least $1-\frac{\delta}{4}$ -- we have (\ref{desc1}). In the presence of statistical noise, we invoke Theorem \ref{unbgraddesc} with probability $1-\frac{\delta}{4}$, so that (\ref{VarXi}) ensures that (\ref{Eg2b}) is verified. Using the analogue of Corollary \ref{freezntk} in the discrete time setting (the proof would be identical) with probability at least $1-\frac{\delta}{4}$, we have the following bound, which holds with probability at least $1-\frac{\delta}{2}$:
\begin{align}
\nonumber
|A(t)|&\leq \eta_0\|F(t)-Y\|_2|\|\hat K_{\Theta_t}-\hat K_{\Theta_0}\|_\mathcal{L}\left\|F(t+\um)-F^{\mathrm{lin}}(t+\um)\right\|_2\\
\nonumber
&\leq \eta_0R_0\sqrt{n\log(2n)}\left(1-\frac{\eta_0\lambda_{\min}^K}{3}\right)^t\|\hat K_{\Theta_t}-\hat K_{\Theta_0}\|_F\times\\
\nonumber
&\qquad \times\frac{1}{2}\left\|F(t+1)-F^{\mathrm{lin}}(t+1)+F(t)-F^{\mathrm{lin}}(t)\right\|_2\\
\nonumber
&\leq \eta_0R_0\sqrt{n\log(2n)}\left(1-\frac{\eta_0\lambda_{\min}^K}{3}\right)^t\frac{R_2}{\lambda_{\min}^K}n^2\sqrt{\log(2n)}\frac{Lm|\mathcal{M}|^4|\mathcal{N}|}{N^2_K(m)N^3(m)}\times\\
\nonumber
&\qquad \times \frac{1}{2}\bigg(\big\|F(t+1)-F^{\mathrm{lin}}(t+1)\big\|+\big\|F(t)-F^{\mathrm{lin}}(t)\big\|_2\bigg)\\
\nonumber
&= \frac{R_0R_2\eta_0}{\lambda_{\min}^K}n^2\sqrt n \log(2n) \frac{Lm|\mathcal{M}|^4|\mathcal{N}|}{N^2_K(m)N^3(m)}\left(1-\frac{\eta_0\lambda_{\min}^K}{3}\right)^t\times\\
&\qquad \times \frac{1}{2}\bigg(\big\|F(t+1)-F^{\mathrm{lin}}(t+1)\big\|+\big\|F(t)-F^{\mathrm{lin}}(t)\big\|_2\bigg).
\end{align}
Using (\ref{ff1}) and (\ref{ff2}),
\begin{align}
\nonumber
B(t)&=\eta_0(F^{\mathrm{lin}}(t)-F(t))^T\hat K_{\Theta_0}\left(F(t+\um)-F^{\mathrm{lin}}(t+\um)\right)\\
\nonumber
&=\eta_0(F^{\mathrm{lin}}(t)-F(t))^T\hat K_{\Theta_0}\Big( F(t) -\frac{\eta_0}{2} \hat K_{\Theta_t}(F(t)-Y)\\
\nonumber
&\phantom{\eta_0(F^{\mathrm{lin}}(t)-F(t))^T\hat K_{\Theta_0}\Big(\qquad}-F^{\mathrm{lin}}(t) +\frac{\eta_0}{2} \hat K_{\Theta_0}(F^{\mathrm{lin}}(t)-Y)\Big)\\
\nonumber
&\qquad + \frac{\eta_0}{2}(F^{\mathrm{lin}}(t)-F(t))^T\hat K_{\Theta_0}(r(t,\modifica{X})+\delta(t,\modifica{X}))\\
\nonumber
&=-\eta_0(F^{\mathrm{lin}}(t)-F(t))^T\hat K_{\Theta_0}( F^{\mathrm{lin}}(t)-F(t)) \\
\nonumber
&\phantom{=} -\frac{\eta_0^2}{2}(F^{\mathrm{lin}}(t)-F(t))^T\hat K_{\Theta_0}\Big(\hat K_{\Theta_t}(F(t)-Y)- \hat K_{\Theta_0}(F^{\mathrm{lin}}(t)-Y)\Big)\\
\nonumber
&\phantom{=} + \frac{\eta_0}{2}(F^{\mathrm{lin}}(t)-F(t))^T\hat K_{\Theta_0}(r(t,X)+\delta(t,X))\\
\nonumber
&=-\eta_0(F^{\mathrm{lin}}(t)-F(t))^T\hat K_{\Theta_0}( F^{\mathrm{lin}}(t)-F(t)) \\
\nonumber
&\phantom{=} -\frac{\eta_0^2}{2}(F^{\mathrm{lin}}(t)-F(t))^T\hat K_{\Theta_0}\Big(\hat K_{\Theta_t}(F(t)-Y)-\hat K_{\Theta_0}(F^{\mathrm{lin}}(t)-Y)\Big)\\
\nonumber
&\phantom{=} + \frac{\eta_0}{2}(F^{\mathrm{lin}}(t)-F(t))^T\hat K_{\Theta_0}(r(t,X)+\delta(t,X))\\
\nonumber
&=-\eta_0(F^{\mathrm{lin}}(t)-F(t))^T\hat K_{\Theta_0}( F^{\mathrm{lin}}(t)-F(t)) \\
\nonumber
&\phantom{=} -\frac{\eta_0^2}{2}(F^{\mathrm{lin}}(t)-F(t))^T\hat K_{\Theta_0}\Big((\hat K_{\Theta_t}-\hat K_{\Theta_0})(F(t)-Y)-\hat K_{\Theta_0}(F^{\mathrm{lin}}(t)-F(t))\Big)\\
\nonumber
&\phantom{=} + \frac{\eta_0}{2}(F^{\mathrm{lin}}(t)-F(t))^T\hat K_{\Theta_0}(r(t,X)+\delta(t,X))\\
\nonumber
&=(F^{\mathrm{lin}}(t)-F(t))^T\left(-\eta_0\hat K_{\Theta_0}+\frac{\eta_0^2}{2}\hat K^2_{\Theta_0}\right) (F^{\mathrm{lin}}(t)-F(t))\\
\nonumber
&\phantom{=} -\frac{\eta_0^2}{2}(F^{\mathrm{lin}}(t)-F(t))^T\hat K_{\Theta_0}(\hat K_{\Theta_t}-\hat K_{\Theta_0})(F(t)-Y)\\
&\phantom{=} + \frac{\eta_0}{2}(F^{\mathrm{lin}}(t)-F(t))^T\hat K_{\Theta_0}(r(t,X)+\delta(t,X)).
\end{align}
We recall that, in the discrete time setting, we always require that
\[\eta_0<\frac{2}{\lambda_{\min}^K+\lambda_{\max}^K},\]
in particular
\[\eta_0<\frac{2}{\lambda_{\max}^K},\]
\modifica{therefore} it exists $\rho>1$ such that
\[\eta_0=\frac{1}{\rho}\frac{2}{\lambda_{\max}^K}.\]
The spectrum of $\hat K_{\Theta_0}$ converges to the spectrum of $\bar K$, so, if $m$ is large enough,
\[\text{Spec}(\hat K_{\Theta_0})\subseteq [0,\rho\lambda_{\max}]=[0,2/\eta_0]. \]
Now, we notice that
\[ \phi(\xi)=\frac{1}{2}\xi^2-\xi\leq 0 \quad \iff \quad \xi\in[0,2]\]
hence, for any $\xi$ of the form
\[\xi = \eta_0\lambda\qquad \lambda\in\text{Spec}(\hat K_{\Theta_0}),\]
we have $\phi(\xi)\leq 0$, so the matrix
\[ M=-\eta_0\hat K_{\Theta_0}+\frac{\eta_0^2}{2}\hat K^2_{\Theta_0}\]
is negative semidefinite. Whence,
\begin{align}
\nonumber
B(t)&\leq -\frac{\eta_0^2}{2}(F^{\mathrm{lin}}(t)-F(t))^T\hat K_{\Theta_0}(\hat K_{\Theta_t}-\hat K_{\Theta_0})(F(t)-Y)\\
&\qquad + \frac{\eta_0}{2}(F^{\mathrm{lin}}(t)-F(t))^T\hat K_{\Theta_0}(r(t,X)+\delta(t,X)).
\end{align}
We recall that we just required the maximal eigenvalue of $\hat K_{\Theta_0}$ to be smaller than $2/\eta_0$, so
\[\|\hat K_{\Theta_0}\|_\mathcal{L}\leq \frac{2}{\eta_0},\]
which will be useful soon. Indeed,
\begin{align}
\nonumber
|B(t)|&\leq \frac{\eta_0^2}{2}\|F^{\mathrm{lin}}(t)-F(t)\|_2\|\hat K_{\Theta_0}\|_\mathcal{L}\|\hat K_{\Theta_t}-\hat K_{\Theta_0}\|_\mathcal{L}\|F(t)-Y\|_2\\
\nonumber
&\qquad+ \frac{\eta_0}{2}\|F^{\mathrm{lin}}(t)-F(t)\|_2\|\hat K_{\Theta_0}\|_\mathcal{L}\|r(t,X)+\delta(t,X)\|_2\\
\nonumber
&\leq \|F^{\mathrm{lin}}(t)-F(t)\|_2\Bigg(\frac{\eta_0^2}{2} \frac{2}{\eta_0} \|\hat K_{\Theta_t}-\hat K_{\Theta_0}\|_FR_0\sqrt{n\log(2n)}\left(1-\frac{\eta_0\lambda_{\min}^K}{3}\right)^t\\
\nonumber
&\phantom{\leq \|F^{\mathrm{lin}}(t)-F(t)\|_2\Bigg(}\quad + \frac{\eta_0}{2}\frac{2}{\eta_0}\|r(t,X)+\delta(t,X)\|_2\Bigg)\\
\nonumber
&\leq \|F^{\mathrm{lin}}(t)-F(t)\|_2\Bigg(\frac{\eta_0^2}{2} \frac{2}{\eta_0} \frac{R_2}{\lambda_{\min}^K}n^2\sqrt{\log(2n)}\frac{Lm|\mathcal{M}|^4|\mathcal{N}|}{N_K^2(m)N^3(m)}R_0\sqrt{n\log(2n)}\times\\
\nonumber
&\phantom{\leq \|F^{\mathrm{lin}}(t)-F(t)\|_2\Bigg(}\quad \times \left(1-\frac{\eta_0\lambda_{\min}^K}{3}\right)^t + \frac{\eta_0}{2}\frac{2}{\eta_0}\|r(t,X)+\delta(t,X)\|_2\Bigg)\\
\nonumber
&\leq \|F^{\mathrm{lin}}(t)-F(t)\|_2\Bigg(\frac{R_0R_2\eta_0}{\lambda_{\min}^K}n^2\sqrt{n}\log(2n)\frac{Lm|\mathcal{M}|^4|\mathcal{N}|}{N_K^2(m)N^3(m)}\left(1-\frac{\eta_0\lambda_{\min}^K}{3}\right)^t\\
&\phantom{\leq \|F^{\mathrm{lin}}(t)-F(t)\|_2\Bigg(}\quad +\|r(t,X)\|_2+\|\delta(t,X)\|_2\Bigg)
\end{align}

Finally,
\begin{align}
\nonumber
|C(t)|\leq &\|r(t,X)+\delta(t,X)\|_2\left\|F(t+\um)-F^{\mathrm{lin}}(t+\um)\right\|_2\\
&\leq \left(\|r(t,X)\|_2+\|\delta(t,X)\|_2\right)\times \frac{1}{2}\bigg(\big\|F(t+1)-F^{\mathrm{lin}}(t+1)\big\|+\big\|F(t)-F^{\mathrm{lin}}(t)\big\|_2\bigg)
\end{align}

Now, as we discussed above, we use that
\[\frac{1}{2}\Delta g^2(t)=g(t+\um)\Delta g(t)\]
with
\[g(t)=\|F(t)-F^{\mathrm{lin}}(t)\|_2\]
so that
\begin{align}
\nonumber
\frac{1}{2}\big(\|F(t+1)-F^{\mathrm{lin}}(t+1)\|_2&+\|F(t)-F^{\mathrm{lin}}(t)\|_2\big)\Delta \|F(t)-F^{\mathrm{lin}}(t)\|_2\\
&=\frac{1}{2}\Delta \|F(t)-F^{\mathrm{lin}}(t)\|_2^2=A(t)+B(t)+C(t).
\end{align}
Hence,
\begin{align}
\big|\Delta \|F(t)-F^{\mathrm{lin}}(t)\|_2\big|\leq \frac{|A(t)|+|B(t)|+|C(t)|}{\frac{1}{2}\big(\|F(t+1)-F^{\mathrm{lin}}(t+1)\|_2+\|F(t)-F^{\mathrm{lin}}(t)\|_2\big)}.
\end{align}
Using our bounds and recalling that, for any $a,b>0$,
\[\frac{a}{a+b}\leq 1, \]
we find that
\begin{align}
\nonumber
\big|\Delta \|F(t)-F^{\mathrm{lin}}(t)\|_2\big|&\leq \Bigg(\frac{R_0R_2\eta_0}{\lambda_{\min}^K}n^2\sqrt n \log(2n) \frac{Lm|\mathcal{M}|^4|\mathcal{N}|}{N^2_K(m)N^3(m)}\left(1-\frac{\eta_0\lambda_{\min}^K}{3}\right)^t+\\
\nonumber
&\qquad\qquad+ \|r(t,X)\|_2+\|\delta(t,X)\|_2 \Bigg)\times\\
\nonumber
&\qquad \times \Bigg(1+\frac{\|F(t)-F^{\mathrm{lin}}(t)\|_2}{\frac{1}{2}\big(\|F(t+1)-F^{\mathrm{lin}}(t+1)\|_2+\|F(t)-F^{\mathrm{lin}}(t)\|_2\big)}\Bigg)\\
\nonumber
&\leq 3\Bigg(\frac{R_0R_2\eta_0}{\lambda_{\min}^K}n^2\sqrt n \log(2n) \frac{Lm|\mathcal{M}|^4|\mathcal{N}|}{N^2_K(m)N^3(m)}\left(1-\frac{\eta_0\lambda_{\min}^K}{3}\right)^t+\\
&\qquad\qquad+ \|r(t,X)\|_2+\|\delta(t,X)\|_2 \Bigg)\label{tosum}
\end{align}
Due to Theorem \ref{powerful},
\[|r(t,x)|\leq \frac{Lm|\mathcal{M}|^2|\mathcal{N}|}{N(m)}\|\Theta_{t+1}-\Theta_t\|_\infty^2\]
so, recalling (\ref{triangle})
\begin{align}\nonumber
\|\Theta_{t+1}-\Theta_{t}\|_\infty &\leq \eta_0 \frac{R_0}{\sqrt 2}\sqrt{\log(2n)}\,\times\\
&\phantom{\leq}\quad\times\left(\eta_0\frac{\pi\sqrt{c_0}}{\sqrt 3}\left(\lambda_{\min}^K\right)^2\frac{N(m)}{|\mathcal{M}||\Theta|}+2\sqrt 2\,n \frac{|\mathcal{M}|}{N_K(m)N(m)}\right)\left(1-\frac{\eta_0\lambda_{\min}^K}{3}\right)^t,
\end{align}
we can bound
\begin{align}\nonumber
\sum_{t=0}^\infty|r(t,x)| &\leq 2\eta_0^2 \frac{R_0^2}{ 2}\log(2n)\,\times\\
\nonumber
&\phantom{\leq}\quad\times\left(\eta_0^2\frac{\pi^2c_0}{3}\left(\lambda_{\min}^K\right)^2\frac{N(m)|\mathcal{N}|}{Lm}+8\,n^2 \frac{Lm|\mathcal{M}|^4|\mathcal{N}|}{N_K^2(m)N^3(m)}\right)\frac{3}{2\eta_0\lambda_{\min}^K},\\
&\leq \frac{3 R_0^2}{\lambda_{\min}^K}\eta_0\log(2n)\left(\eta_0^2\frac{\pi^2c_0}{3}\left(2\lambda_{\min}^K\right)^2\frac{N(m)|\mathcal{N}|}{Lm}+8\,n^2 \frac{Lm|\mathcal{M}|^4|\mathcal{N}|}{N_K^2(m)N^3(m)}\right),
\end{align}
where we used that
\begin{align}
    \sum_{t=0}^\infty\left(1-\frac{\eta_0\lambda_{\min}^K}{3}\right)^{2t}=\frac{1}{\frac{2\eta_0\lambda_{\min}^K}{3}-\left(\frac{\eta_0\lambda_{\min}^K}{3}\right)^2}\leq\frac{3}{2\eta_0\lambda_{\min}^K}.
\end{align}
Let us claim that the variance of the statistical noise is small enough
\[\sum_{t=0}^\infty 3\|\delta(t,X)\|_2\leq \Delta(m)\]
with high probability, where $\Delta(m)\to 0$ fast enough when $m\to\infty$. More precisely, we have previously invoked Theorem \ref{unbgraddesc}, so we can assume that (\ref{stimag}) holds, i.e.
\begin{align}
    \nonumber
    \|\eta(\nabla_\Theta\mathcal{L}(\Theta_t)-g^{(t)}(\Theta_t))\|_\infty&=:\|\Theta_{t+1}-\tilde\Theta_{t+1}\|_\infty\\
    &\leq \eta\pi\sqrt{\frac{|\Theta|}{3\delta}}(t+1)\sup_i\sqrt{\text{Var}\left[g^{(t)}_i(\Theta_t)\right]}.
\end{align}
Recalling that
\[g^{(t)}(\Theta_t)=\nabla \mathcal{L}(\Theta_t)+\epsilon(t) \qquad \text{and}\qquad \delta(t,x)=-\eta \nabla_\Theta f(\Theta_t,x)^T \epsilon(t)\]
and using (\ref{noisy21}) in (\ref{VarXi}) to write explicitly
\begin{align}
    \nonumber
    \|\epsilon(t)\|_\infty&=\|\nabla_\Theta\mathcal{L}(\Theta_t)-g^{(t)}(\Theta_t)\|_\infty\\
    \nonumber
    &\leq \pi\sqrt{\frac{|\Theta|}{3\delta}}(t+1)\sup_i\sqrt{\text{Var}\left[g^{(t)}_i(\Theta_t)\right]}\\
    \nonumber
    & \leq \pi\sqrt{\frac{|\Theta|}{3\delta}}(t+1)
    \Bigg(
    c_0\eta_0^2\,\frac{\left(\lambda_{\min}^K\right)^4}{n^2}\frac{N_K^2(m)N^2(m)}{|\mathcal{M}|^2|\Theta|^3}\times\\
    & \qquad\qquad\qquad\qquad\quad\times
    \frac{\delta/4}{(t+1)^2}\frac{R_0^2}{2}\log(2n)\left(1-\frac{1}{3}\eta_0\lambda_{\min}^K\right)^{\modifica{2}t}
    \Bigg)^{1/2}\sqrt{\xi(m)},
    \label{boundepsilon}
\end{align}
we finally know that, with the estimate (\ref{lemma3}),
\begin{align}
    \nonumber
    \|\delta(t,X)\|_2&\leq \eta\sqrt{n}\sup_x\|\nabla_\Theta f(\Theta_t,x)\|_1\|\epsilon(t)\|_\infty\\
    \nonumber
    & \leq\eta\sqrt{n}  \left(2L\frac{m}{N(m)}|\mathcal{M}|\right)\pi\sqrt{\frac{|\Theta|}{3\delta}}(t+1) \sqrt{\xi(m)}\times\\
    \nonumber
    & \quad
    \times\Bigg(
    c_0\eta_0^2\,\frac{\left(\lambda_{\min}^K\right)^4}{n^2}\frac{N_K^2(m)N^2(m)}{|\mathcal{M}|^2|\Theta|^3}
    \frac{\delta/4}{(t+1)^2}\frac{R_0^2}{2}\log(2n)\left(1-\frac{1}{3}\lambda_{\min}^K\right)^{\modifica{2}t}
    \Bigg)^{1/2}\\
    \nonumber
    & \leq  R_0\sqrt{c_0}\sqrt{\log(2n)}\frac{\left(\lambda_{\min}^K\right)^2}{\sqrt{n}}\left(N_K(m)\eta\eta_0\right)\pi\sqrt{\frac{1}{6}}\left(1-\frac{1}{3}\eta_0\lambda_{\min}^K\right)^{\modifica{t}}\sqrt{\xi(m)}\\
    & \leq   R_0\pi\sqrt{\frac{c_0}{6}}\sqrt{n\log(2n)}\left(\eta_0\lambda_{\min}^K\right)^2\left(1-\frac{1}{3}\eta_0\lambda_{\min}^K\right)^{\modifica{t}}\sqrt{\xi(m)}
\end{align}
where we recalled that $N_K(m)\eta=n\eta_0$. Then
\begin{align}
    \nonumber
    \sum_{t=0}^\infty 3\|\delta(t,X)\|_2
    &\leq \frac{3\pi R_0}{\modifica{\frac{1}{3}\eta_0\lambda_{\min}^K}}\sqrt{\frac{c_0}{6}}\sqrt{n\log(2n)}\left(\eta_0\lambda_{\min}^K\right)^2\sqrt{\xi(m)}\\
    &\leq \modifica{9}\pi R_0\sqrt{\frac{c_0}{6}}\eta_0\lambda_{\min}^K\sqrt{n\log(2n)}\sqrt{\xi(m)}=:\Delta(m)
    \label{simile}
\end{align}
Summing the bound (\ref{tosum}), we conclude that
\begin{align}
    \nonumber
    \|F(t)-&F^{\mathrm{lin}}(t)\|_2\\
    \nonumber
    &\leq \frac{9 R_0^2}{\lambda_{\min}^K}\eta_0\sqrt n \log(2n)\left(\eta_0^2\frac{\pi^2c_0}{3}\left(2\lambda_{\min}^K\right)^2\frac{N(m)|\mathcal{N}|}{Lm}+8\,n^2 \frac{Lm|\mathcal{M}|^4|\mathcal{N}|}{N_K^2(m)N^3(m)}\right)\\
    \nonumber
    &+ \frac{3R_0R_2}{(\lambda_{\min}^K)^2}n^2\sqrt n \log(2n) \frac{Lm|\mathcal{M}|^4|\mathcal{N}|}{N^2_K(m)N^3(m)} +\Delta(m)\\
    \nonumber
    &\leq C_0\sqrt n \log(2n)\left(\eta_0^3\lambda_{\min}^K\frac{N(m)|\mathcal{N}|}{Lm}+n^2\left(\frac{\eta_0}{\lambda_{\min}^K}+\frac{1}{(\lambda_{\min}^K)^2}\right) \frac{Lm|\mathcal{M}|^4|\mathcal{N}|}{N_K^2(m)N^3(m)}\right)\\
    &+ \Delta(m)
\end{align}
for some $C_0$. This can be simplified further by recalling Lemma \ref{Nmax}, which ensures that for some $c>0$ 
\begin{align}\label{magicN}
    m^2\geq \frac{N^4(m)}{c|\mathcal{M}|^2|\mathcal{N}|^2} \qquad \to \qquad \frac{N(m)|\mathcal{N}|}{Lm}\leq \frac{cm|\modifica{\mathcal{M}}|^2|\mathcal{N}|^3}{LN^{\modifica{3}}(m)},
\end{align}
and by noticing that, by assumption,
\begin{align}\label{magicEta}
    \eta_0<\frac{2}{\lambda_{\min}^K+\lambda_{\max}^K}\leq\frac{1}{\lambda_{\min}^K},
\end{align}
 whence
\begin{align}
    \nonumber
    \|F(t)-&F^{\mathrm{lin}}(t)\|_2\\
    \nonumber
    &\leq C_1\sqrt n \log(2n)\left(\eta_0^3\lambda_{\min}^K+n^2\left(\frac{\eta_0}{\lambda_{\min}^K}+\frac{1}{(\lambda_{\min}^K)^2}\right) \right)\frac{Lm|\mathcal{M}|^4|\mathcal{N}|^3}{N^3(m)}+ \Delta(m)\\
    \nonumber
    &\leq C_1\sqrt n \log(2n)\frac{1+2n^2}{(\lambda_{\min}^K)^2} \frac{Lm|\mathcal{M}|^4|\mathcal{N}|^3}{N^3(m)}+ \Delta(m)\\
    &\leq \frac{C}{(\lambda_{\min}^K)^2}n^2\sqrt n \log(2n)\frac{Lm|\mathcal{M}|^4|\mathcal{N}|^3}{N^3(m)}+ \Delta(m)
\end{align}
for some $C_1$ and $C$.
\end{proof}

Now we are ready to prove Lemma \ref{secondorder2}.
Many estimates are similar to the continuous time setting. Let us start from
\begin{align}
\nonumber
\|\Delta(\Theta_t&-\Theta_t^{\mathrm{lin}})\|_\infty\\
\nonumber
&\leq \frac{\eta}{n}\left\|\nabla_\Theta f(\Theta_t,X^T)(F(t)-Y) +n\epsilon(t)-\nabla_\Theta f(\Theta_0,X^T)(F^{\mathrm{lin}}(t)-Y)\right\|_\infty\\
\nonumber
&\leq \frac{\eta}{n}\sup_i\|\partial_{\theta_i} f(\Theta_t,X)-\partial_{\theta_i} f(\Theta_0,X)\|_2\|F(t)-Y\|_2\\
&\qquad+\frac{\eta}{n}\sup_i\|\partial_{\theta_i} f(\Theta_0,X)\|_2\|F^{\mathrm{lin}}(t)-F(t)\|_2+\eta\|\epsilon(t)\|_\infty.
\label{sommabound}
\end{align}
Similarly to the continuous time setting,
\begin{align}\nonumber
\frac{\eta}{n}\sup_i\|\partial_{\theta_i} f(\Theta_t,X)&-\partial_{\theta_i} f(\Theta_0,X)\|_2\|F(t)-Y\|_2\\
&\leq \frac{4\eta_0 R_0R_1}{\lambda_{\min}^K}n^2\log(2n)\frac{|\mathcal{M}|^3|\mathcal{N}|}{N_K^2(m)N^2(m)}\left(1-\frac{\eta_0\lambda_{\min}^K}{3}\right)^t=:A(t),
\end{align}
\begin{align}
\nonumber
\frac{\eta}{n}\sup_i\|\partial_{\theta_i} &f(\Theta_0,X)\|_2\|F^{\mathrm{lin}}(t)-F(t)\|_2 \\
&\leq 4R_0\eta_0 \sqrt{\log(2n)} \frac{|\mathcal{M}|}{N(m)}\left(1-\frac{\eta_0\lambda_{\min}^K}{3}\right)^t=:C(t),
\end{align}
while, using Lemma \ref{improveqltn},
\begin{align}
\nonumber
\frac{\eta}{n}\sup_i\|\partial_{\theta_i} &f(\Theta_0,X)\|_2\|F^{\mathrm{lin}}(t)-F(t)\|_2\\
\nonumber
&\leq \frac{\eta}{n}2\sqrt n \frac{|\mathcal{M}|}{N(m)}\|F^{\mathrm{lin}}(t)-F(t)\|_2\\
\nonumber
& \leq \frac{\eta}{n}2\sqrt n \frac{|\mathcal{M}|}{N(m)}
    \Bigg(\frac{C}{(\lambda_{\min}^K)^2}n^2\sqrt n \log(2n)\frac{Lm|\mathcal{M}|^4|\mathcal{N}|^3}{N^3(m)}+ \Delta(m)\Bigg)\\
&= \frac{2\eta_0 C}{(\lambda_{\min}^K)^2}n^3 \log(2n)\frac{Lm|\mathcal{M}|^5|\mathcal{N}|^3}{N_K(m)N^4(m)}+2\eta_0\sqrt n \frac{|\mathcal{M}|}{N_K(m)N(m)} \Delta(m)
:=B(t)
\end{align}
Therefore, (\ref{sommabound}) implies the following two bounds:
\[\|\Delta(\Theta_t-\Theta_t^{\mathrm{lin}})\|_\infty\leq A(t)+B(t)+\eta\|\epsilon(t)\|_\infty \]
and
\[\|\Delta(\Theta_t-\Theta_t^{\mathrm{lin}})\|_\infty\leq A(t)+C(t)+\eta\|\epsilon(t)\|_\infty.\]
Due to the hypothesis (\ref{VarXi}) the variance of $g^{(t)}(\Theta_t)$ is small enough for the bounds we are interested in: indeed, with probability at least $1-\frac{\delta}{4}$, by (\ref{boundepsilon}) we have
\begin{align}
    \nonumber
    \|\epsilon(t)\|_\infty
    & \leq \pi\sqrt{\frac{|\Theta|}{3\delta}}(t+1)
    \Bigg(
    c_0\eta_0^2\,\frac{\left(\lambda_{\min}^K\right)^4}{n^2}\frac{N_K^2(m)N^2(m)}{|\mathcal{M}|^2|\Theta|^3}\times\\
    \nonumber
    & \qquad\qquad\qquad\qquad\quad\times
    \frac{\delta/4}{(t+1)^2}\frac{R_0^2}{2}\log(2n)\left(1-\frac{1}{3}\eta_0\lambda_{\min}^K\right)^{\modifica{2}t}
    \Bigg)^{1/2}\sqrt{\xi(m)},\\
    &=\frac{R_0\pi}{2}\sqrt{\log(2n)}\frac{\eta_0}{n}(\lambda_{\min}^K)^2\sqrt{\frac{c_0}{6}}\frac{N_K(m)N(m)}{|\mathcal{M}|Lm}\left(1-\frac{1}{3}\eta_0\lambda_{\min}^K\right)^{\modifica{t}}\sqrt{\xi(m)}
\end{align}
whence, similarly to (\ref{simile})
\begin{align}
    \nonumber
    \sum_{t=0}^\infty \eta\|\epsilon(t)\|_\infty
    &\leq \frac{\pi R_0\modifica{/2}}{\modifica{\frac{1}{3}\eta_0\lambda_{\min}^K}}\sqrt{\frac{c_0}{6}}\sqrt{\log(2n)}\left(\eta_0\lambda_{\min}^K\right)^2\frac{N(m)}{|\mathcal{M}|Lm}\sqrt{\xi(m)}\\
    \nonumber
    &= \frac{1}{\modifica{6}\sqrt n} \left(\modifica{9}\pi R_0\sqrt{\frac{c_0}{6}}\sqrt{n\log(2n)}\eta_0\lambda_{\min}^K\sqrt{\xi(m)}\right)\frac{N(m)}{|\mathcal{M}|Lm}\\
    &=\frac{1}{\modifica{6}\sqrt n}\frac{N(m)}{|\mathcal{M}|Lm}\Delta(m).
\end{align}
As in the continuous time setting, we define
\[t^\ast = \left\lfloor\frac{3}{\eta_0 \lambda_{\min}^K}\log N(m)\right\rfloor.\]
It is useful to notice that
\[1+x\leq e^x \quad \forall\,x\in\mathbb{R}\qquad \to\qquad (1-x)^t\leq e^{-xt}\quad \forall\,x\leq 1\quad \forall\,t\geq 0.\]
Indeed,
\begin{align}
\nonumber
\sum_{t=t^\ast}^\infty \left(1-\frac{\eta_0\lambda_{\min}^K}{3}\right)^t&=\frac{3}{\eta_0\lambda_{\min}^K}-\sum_{t=0}^{t^\ast-1} \left(1-\frac{\eta_0\lambda_{\min}^K}{3}\right)^t\\
\nonumber
&=\frac{3}{\eta_0\lambda_{\min}^K}-\frac{1-(1-\eta_0\lambda_{\min}^K/3)^{t^\ast}}{\eta_0\lambda_{\min}^K/3}\\
\nonumber
&=\frac{3}{\eta_0\lambda_{\min}^K}\left(1-\frac{\eta_0\lambda_{\min}^K}{3}\right)^{t^\ast}\\
\label{potenzaexp}
&\leq \frac{3}{\eta_0\lambda_{\min}^K}e^{-\frac{1}{3}\eta_0\lambda_{\min}^Kt^\ast}\leq \frac{3}{\eta_0\lambda_{\min}^K}e^{-\log N(m)}=\frac{3}{\eta_0\lambda_{\min}^K}\frac{1}{N(m)}
\end{align}
So we can sum the bounds for (\ref{sommabound}), using (\ref{potenzaexp}) for the one concerning $C(t)$:
\begin{align}
\nonumber
\big\|\Theta_t-\Theta_t^{\mathrm{lin}}\big\|_\infty&\leq \sum_{t=0}^\infty A(t) + \sum_{t=0}^{t^\ast-1} B(t) + \sum_{t=t^\ast}^\infty C(t)+\sum_{t=0}^\infty \eta\|\epsilon(t)\|_\infty\\
\nonumber
&\leq \frac{12R_0R_1}{(\lambda_{\min}^K)^2}n^2\log(2n)\frac{|\mathcal{M}|^3|\mathcal{N}|}{N_K^2(m)N^2(m)}\\
\nonumber
&\qquad + \frac{6 C}{(\lambda_{\min}^K)^3}n^3 \log(2n)\frac{Lm|\mathcal{M}|^5|\mathcal{N}|^3}{N_K(m)N^4(m)}\log N(m)\\
&\qquad+\frac{6}{\lambda_{\min}^K}\sqrt n \frac{|\mathcal{M}|}{N_K(m)N(m)}\log N(m) \Delta(m)\\
\nonumber
&\qquad + \frac{12 R_0}{\lambda_{\min}^K} \sqrt{\log(2n)} \frac{|\mathcal{M}|}{N^2(m)}\\
&\qquad + \frac{3}{2\sqrt n}\frac{N(m)}{|\mathcal{M}|Lm}\Delta(m).
\end{align}
Now, using the definition (\ref{defdelta}) of $\Delta(m)$ and recalling that $\eta_0\lambda_{\min}^K< 1$
\[ \Delta(m)=  C'\eta_0\lambda_{\min}^K\sqrt{n\log(2n)}\sqrt{\xi(m)}<C'\sqrt{n\log(2n)}\sqrt{\xi(m)}\]
we conclude that, thanks to the estimate (\ref{magicN}),
\begin{align}
\nonumber
\big\|\Theta_t-\Theta_t^{\mathrm{lin}}\big\|_\infty&\leq \frac{12R_0R_1\sqrt c}{(\lambda_{\min}^K)^2}n^2\log(2n)\frac{|\mathcal{M}|^4|\mathcal{N}|^2}{N_K^2(m)N^4(m)}\\
\nonumber
&\qquad + \frac{6 C}{(\lambda_{\min}^K)^3}n^3 \log(2n)\frac{Lm|\mathcal{M}|^5|\mathcal{N}|^3}{N_K(m)N^4(m)}\log N(m)\\
\nonumber
&\qquad +\frac{6C'\sqrt c}{\lambda_{\min}^K}n \sqrt{\log(2n)}\frac{Lm|\mathcal{M}|^2|\mathcal{N}|}{N_K(m)N^3(m)} \log N(m)\sqrt{\xi(m)}\\
\nonumber
&\qquad + \frac{12 R_0\sqrt c}{\lambda_{\min}^K} \sqrt{\log(2n)} \frac{Lm|\mathcal{M}|^2|\mathcal{N}|}{N^4(m)}\\
&\qquad + \frac{3C'c}{2}\sqrt{\log(2n)}\frac{Lm|\mathcal{M}||\mathcal{N}|^2}{N^3(m)}\sqrt{\xi(m)}.
\end{align}
Since $\lim_{m\to\infty}N(m) =\infty$ under our assumptions, we can suppose to take $m$ large enough so that $\log N(m)\geq 1$. Therefore, we can finally simplify the above expression as follows:
\begin{align}
\nonumber
\big\|\Theta_t&-\Theta_t^{\mathrm{lin}}\big\|_\infty\\
&\leq \left(\frac{C_1}{(\lambda_{\min}^K)^3}+C_2\right)n^3\log(2n)\frac{Lm|\mathcal{M}|^5|\mathcal{N}|^3}{N^4(m)}\log N(m) \frac{1}{2}\left(1+N(m)\sqrt{\xi(m)}\right).
\end{align}

\subsection{Proof of Theorem \ref{confronto_finale}}\label{proofconfronto}
Let $r\in\{0,1\}$ be the binary variable indicating the presence ($r=1$) or the absence ($r=0$) of noise. We will use it to unify (\ref{convunif2}) with (\ref{noisy23}):
\begin{align}
\nonumber
|f(\Theta_t,x)-f^{\mathrm{lin}}(\Theta_t,x)|&\leq \frac{M_1}{(\lambda_{\min}^K)^2}n^2\log(2n)\frac{Lm|\mathcal{M}|^4|\mathcal{N}|}{N_K^2(m)N^3(m)}\\
&\qquad+ rM_2(\lambda_{\min}^K)^2\eta_0^2\log(2n)\frac{|\mathcal{N}|N(m)}{Lm}.
\end{align}
Therefore, considering also (\ref{lemma3}) and (\ref{magicN})
\begin{align}
\nonumber
|f(\Theta_t,x)&-f^{\mathrm{lin}}(\Theta_t^{\mathrm{lin}},x)|\\
\nonumber
&=|f(\Theta_t,x)-f(\Theta_0,x)-\nabla_\Theta f(\Theta_0,x)^T(\Theta_t^{\mathrm{lin}}-\Theta_0)|\\
\nonumber
&= |f(\Theta_t,x)-f(\Theta_0,x)-\nabla_\Theta f(\Theta_0,x)^T(\Theta_t-\Theta_0)-\nabla_\Theta f(\Theta_0,x)^T(\Theta_t^{\mathrm{lin}}-\Theta_t)|\\
\nonumber
&\leq |f(\Theta_t,x)-f^{\mathrm{lin}}(\Theta_t,x)|+|\nabla_\Theta f(\Theta_0,x)^T(\Theta_t^{\mathrm{lin}}-\Theta_t)|\\
\nonumber
&\leq |f(\Theta_t,x)-f^{\mathrm{lin}}(\Theta_t,x)|+\|\nabla_\Theta f(\Theta_0,x)\|_1\|\Theta_t^{\mathrm{lin}}-\Theta_t\|_\infty\\
\nonumber
&\leq \frac{M_1}{(\lambda_{\min}^K)^2}n^2\log(2n)\frac{Lm|\mathcal{M}|^4|\mathcal{N}|}{N_K^2(m)N^3(m)}+ rM_2c(\lambda_{\min}^K)^2\eta_0^2\log(2n)\frac{Lm|\mathcal{M}|^2|\mathcal{N}|^3}{N^3(m)}\\
&\qquad+2\frac{Lm|\mathcal{M}|}{N(m)}\|\Theta_t^{\mathrm{lin}}-\Theta_t\|_\infty.
\end{align}
Combining this result with Corollary (\ref{nuovavarianza}), we have that
\begin{align}
\nonumber
|f(\Theta_t,x)&-f^{\mathrm{lin}}(\Theta_t^{\mathrm{lin}},x)|\\
\nonumber
&\leq \frac{M_1}{(\lambda_{\min}^K)^2}n^2\log(2n)\frac{Lm|\mathcal{M}|^4|\mathcal{N}|}{N_K^2(m)N^3(m)}+ rM_2c(\lambda_{\min}^K)^2\eta_0^2\log(2n)\frac{Lm|\mathcal{M}|^2|\mathcal{N}|^3}{N^3(m)}\\
\nonumber
&\qquad+2\left(\frac{C_1}{(\lambda_{\min}^K)^3}+C_2\right)n^3\log(2n)\frac{L^2m^2|\mathcal{M}|^6|\mathcal{N}|^3}{N^5(m)}\log N(m)\\
\nonumber
&\leq \frac{M_1}{(\lambda_{\min}^K)^2}n^2\log(2n)\frac{L^2m^2|\mathcal{M}|^5|\mathcal{N}|^2}{N^5(m)}+ rM_2c\log(2n)\frac{L^2m^2|\mathcal{M}|^3|\mathcal{N}|^4}{N^3(m)}\\
\nonumber
&\qquad+2\left(\frac{C_1}{(\lambda_{\min}^K)^3}+C_2\right)n^3\log(2n)\frac{L^2m^2|\mathcal{M}|^6|\mathcal{N}|^3}{N^5(m)}\log N(m)\\
&\leq \left(\frac{C_3}{(\lambda_{\min}^K)^3}+C_4\right)n^3\log(2n)\frac{L^2m^2|\mathcal{M}|^6|\mathcal{N}|^4}{N^5(m)}\log N(m),
\end{align}
where we supposed, as above, that $m$ is large enough so that $\log N(m)\geq 1$ \modifica{and, in order to get the above bound simplified, we also used  (\ref{magicEta}) and Lemma \ref{Nmax} in order to get an inequality as in (\ref{magicN}).}

\subsection{Proof of Theorem \ref{qnngpn}}\label{proofqnngpn}
\begin{lemma}\label{llhp2} The hypothesis 
\[\lim_{m\to\infty}\frac{L^2m^2|\mathcal{M}|^6|\mathcal{N}|^4}{N^5(m)}\log N(m)=0,\]
of 
Theorem \ref{qnngpn} ensures that the hypotheses of Theorem \ref{init}, of Theorem \ref{ntkconv} and of Theorem \ref{unbgraddesc} are satisfied, i.e.
\[\lim_{m\to\infty}\frac{m|\mathcal{M}|^2|\mathcal{N}|^2}{N^3(m)}=0,\label{xcheckhp1}\] \[\lim_{m\to\infty}\frac{1}{N_K^2(m)}\,\frac{\Sigma_2|\mathcal{M}|^2|\mathcal{N}|^2}{N^4(m)}=0,\label{xcheckhp2}\]
\[\lim_{m\to\infty}\frac{Lm|\mathcal{M}|^4|\mathcal{N}|^2}{N^3(m)}=0.\label{xcheckhp3}\]
Furthermore,
the bounds (\ref{noisy22}) and (\ref{noisy23}) of Theorem \ref{unbgraddesc} are nontrivial as $m\to\infty$, i.e.
\[\lim_{m\to\infty}\frac{N(m)}{|\mathcal{M}||\Theta|}=0, \qquad\qquad \lim_{m\to\infty}\frac{|\mathcal{M}|}{N_K(m)N(m)}=0,\label{extra1}\]
\[\lim_{m\to\infty}\frac{|\mathcal{N}|N(m)}{Lm}=0, \qquad\qquad \lim_{m\to\infty}\frac{Lm|\mathcal{M}|^4|\mathcal{N}|}{N_K^2(m)N^3(m)}=0.\label{extra2}\]
\end{lemma}
\begin{proof}
(\ref{xcheckhp1}) and (\ref{xcheckhp2}) are ensured by (\ref{xcheckhp3}) because of Lemma \ref{llhp}. So we just need to prove (\ref{xcheckhp3}). Using the inequality of Lemma \ref{Nmax} in the form
\[1\leq c\frac{m|\mathcal{M}||\mathcal{N}|}{N^2(m)}\label{eqNmax}\]
we have
\begin{align}
    \nonumber
    0\leq \lim_{m\to\infty}\frac{Lm|\mathcal{M}|^4|\mathcal{N}|^2}{N^3(m)}&\leq \lim_{m\to\infty}\frac{Lm^2|\mathcal{M}|^5|\mathcal{N}|^3}{N^5(m)}\\
    &\leq \lim_{m\to\infty}\frac{L^2m^2|\mathcal{M}|^6|\mathcal{N}|^4}{N^5(m)}\log N(m)=0
\end{align}
Now we notice that
\begin{align}
\nonumber 0&\leq \lim_{m\to\infty}\frac{N(m)}{|\mathcal{M}||\Theta|}\leq \lim_{m\to\infty}\frac{|\mathcal{N}|N(m)}{Lm}
=\lim_{m\to\infty}\frac{|\mathcal{N}|N^6(m)}{LmN^5(m)}\\
&\leq c^3\lim_{m\to\infty}\frac{|\mathcal{N}|\big(m^3|\mathcal{M}|^3|\mathcal{N}|^3\big)}{LmN^5(m)}
\leq c^3\lim_{m\to\infty}\frac{L^2m^2|\mathcal{M}|^6|\mathcal{N}|^4}{N^5(m)}\log N(m)=0,
\end{align}
where we used again \ref{eqNmax}. This proves the first requirements of (\ref{extra1}) and (\ref{extra2}). Finally, with a similar strategy
\begin{align}
\nonumber 0&\leq\lim_{m\to\infty}\frac{|\mathcal{M}|}{N_K(m)N(m)}\leq c^2 \lim_{m\to\infty}\frac{m^2|\mathcal{M}|^3|\mathcal{N}|^2}{N^5(m)}\\
& \leq c^2 \lim_{m\to\infty}\frac{L^2m^2|\mathcal{M}|^6|\mathcal{N}|^4}{N^5(m)}\log N(m)=0,
\end{align}
 and
\begin{align}
\nonumber
0\leq\lim_{m\to\infty}\frac{Lm|\mathcal{M}|^4|\mathcal{N}|}{N_K^2(m)N^3(m)}&\leq \lim_{m\to\infty}c\frac{Lm^2|\mathcal{M}|^5|\mathcal{N}|^4}{N^5(m)}\\
&\leq \lim_{m\to\infty}c\frac{L^2m^2|\mathcal{M}|^6|\mathcal{N}|^4}{N^5(m)}\log N(m)=0 .
\end{align}
This concludes the proof of Lemma \ref{llhp2}.
\end{proof}

As in \autoref{ch6}, we need to show that the linearized model converges to a Gaussian process in the limit of many qubits. It will be useful to recall that
\[\sum_{k=0}^{t-1}\left(\id-\eta_0\hat K_{\Theta_0}\right)^k=\frac{\id-\left(\id-\eta_0\hat K_{\Theta_0}\right)^t}{\id-\left(\id-\eta_0\hat K_{\Theta_0}\right)}=\frac{1}{\eta_0}\left(\id-\left(\id-\eta_0\hat K_{\Theta_0}\right)^t\right)\hat K_{\Theta_0}^{-1},\]
provided that $\hat K_{\Theta_0}$ is invertible.
\begin{lemma}\label{solutiondiscrete}
The solution of the gradient descent equation (\ref{eqgrd}) for the linearized model is
\[f^{\mathrm{lin}}(\Theta_t,x)=f(\Theta_0,x)-\hat K_{\Theta_0}(x,X^T)\left(\id-\left(\id-\eta_0\hat K_{\Theta_0}\right)^t\right)\hat K_{\Theta_0}^{-1}(f(\Theta_0,X)-Y),\]
provided that $\hat K_{\Theta_0}$ is invertible.
\end{lemma}
\begin{proof}
We start from (\ref{eqgrd}) for the linearized model:
\[\Theta_{t+1}-\Theta_{t}=-\frac{\eta}{n} \nabla_\Theta f^{\mathrm{lin}}(\Theta_t,X^T)(f^{\mathrm{lin}}(\Theta_t,X)-Y), \]
which can be rewritten as
\begin{align}
\nonumber\bar\Theta_{t+1}-\bar\Theta_{t}&=-\frac{\eta}{n} \nabla_\Theta f(\Theta_0,X^T)(f(\Theta_0,X)+\nabla^T_\Theta f(\Theta_0,X)(\bar\Theta_t-Y)\\
&=-\frac{\eta}{n} \nabla_\Theta f(\Theta_0,X^T)\nabla^T_\Theta f(\Theta_0,X)\,\bar\Theta_t+b,
\end{align}
where $\bar \Theta_{t}\equiv \Theta_{t}-\Theta_{0}$ and
\[b=-\frac{\eta}{n} \nabla_\Theta f(\Theta_0,X^T)(f(\Theta_0,X)-Y).\]
So, we have to solve the recurrence relation
\[
\begin{dcases}
\bar\Theta_{t+1}=\left(\id-\frac{\eta}{n} \nabla_\Theta f(\Theta_0,X^T)\nabla^T_\Theta f(\Theta_0,X)\right)\bar\Theta_t+b\\
\bar\Theta_0=0
\end{dcases}.
\]
Calling
\[M_{\Theta_0}=\frac{1}{N_K(m)}\nabla_\Theta f(\Theta_0,X^T)\nabla^T_\Theta f(\Theta_0,X)\]
and redefining $\eta=\frac{n}{N_K(m)}\eta_0$, we have
\[
\begin{dcases}
\bar\Theta_{t+1}=\left(\id-\eta_0M_{\Theta_0}\right)\bar\Theta_t+b\\
\bar\Theta_0=0
\end{dcases}.
\]
It is easy to see that
\begin{align}
\nonumber\bar\Theta_1&=b\\
\nonumber\bar\Theta_2&=b+\left(\id-\eta_0M_{\Theta_0}\right)b\\
\nonumber\bar\Theta_3&=b+\left(\id-\eta_0M_{\Theta_0}\right)b+\left(\id-\eta_0M_{\Theta_0}\right)^2b\\
\nonumber&\,\,\;\vdots\\
\bar\Theta_t&=\sum_{k=0}^{t-1}\left(\id-\eta_0M_{\Theta_0}\right)^kb,
\end{align}
whence
\begin{align}
\nonumber f^{\mathrm{lin}}(\Theta_t,x)&=f(\Theta_0,x)+\nabla_\Theta^Tf(\Theta_0,x)\bar\Theta_t\\
\nonumber&=f(\Theta_0,x)+\nabla_\Theta^Tf(\Theta_0,x)\sum_{k=0}^{t-1}\left(\id-\eta_0M_{\Theta_0}\right)^kb\\
&=f(\Theta_0,x)-\frac{\eta}{n}\nabla_\Theta^Tf(\Theta_0,x)\sum_{k=0}^{t-1}\left(\id-\eta_0M_{\Theta_0}\right)^k\nabla_\Theta f(\Theta_0,X^T)(f(\Theta_0,X)-Y).
\end{align}
By the definition of $M_{\Theta_0}$.
\begin{align}
\nonumber&M_{\Theta_0}\nabla_\Theta f(\Theta_0,X^T)\\
\nonumber&\qquad\qquad=\frac{1}{N_K(m)}\Big(\nabla_\Theta f(\Theta_0,X^T)\nabla^T_\Theta f(\Theta_0,X\Big))\nabla_\Theta f(\Theta_0,X^T)\\
\nonumber&\qquad\qquad=\nabla_\Theta f(\Theta_0,X^T)\frac{1}{N_K(m)}\Big(\nabla^T_\Theta f(\Theta_0,X)\nabla_\Theta f(\Theta_0,X^T)\Big)\\
&\phantom{\qquad\qquad=\frac{1}{N_K(m)}\nabla_\Theta f(\Theta_0,X^T)\nabla^T_\Theta f(\Theta_0,X)\nabla_\Theta f(\Theta}
=\nabla_\Theta f(\Theta_0,X^T)\hat K_{\Theta_0},
\end{align}
which yields
\begin{align}
\nonumber f^{\mathrm{lin}}(\Theta_t,x)&=f(\Theta_0,x)-\frac{\eta}{n}\nabla_\Theta^Tf(\Theta_0,x)\nabla_\Theta f(\Theta_0,X^T)\sum_{k=0}^{t-1}\left(\id-\eta_0\hat K_{\Theta_0}\right)^k(f(\Theta_0,X)-Y)\\
&=f(\Theta_0,x)-\eta_0\hat K_{\Theta_0}(x,X^T)\sum_{k=0}^{t-1}\left(\id-\eta_0\hat K_{\Theta_0}\right)^k(f(\Theta_0,X)-Y)\nonumber\\
&=f(\Theta_0,x)-\hat K_{\Theta_0}(x,X^T)\left(\id-\left(\id-\eta_0\hat K_{\Theta_0}\right)^t\right)\hat K_{\Theta_0}^{-1}(f(\Theta_0,X)-Y).
\label{questasoluzione}
\end{align}
\end{proof}
\begin{remark}\label{remeta}
Let 
\[\eta_0<\frac{2}{\lambda_{\min}^K+\lambda_{\max}^K},\]
i.e.
\[\eta_0=\frac{2}{\lambda_{\min}^K+\lambda_{\max}^K+2\epsilon}\qquad\text{for some}\qquad \epsilon>0.\]
Because of Assumption \ref{assNTK} and Lemma \ref{convprob}, $\hat K_{\Theta_0}(X,X^T)$ converges to $\bar K$, which is invertible, as $m\to\infty$. So, if $m$ is large enough, then
\[\text{Spec}(\hat K_{\Theta_0})\subseteq \left[\frac{\lambda_{\min}^K}{2},\lambda_{\max}^K+\epsilon\right],\]
whence
\[1-\eta_0\frac{\lambda_{\min}^K}{2}<1,\]
\[ 1-\eta_0(\lambda_{\max}^K+\epsilon)>1-\frac{2\lambda_{\max}+2\epsilon}{\lambda_{\max}^K+2\epsilon}=-1+\frac{2\epsilon}{\lambda_{\max}^K+2\epsilon}>-1. \]
Therefore
\[-1< 1-\eta_0\lambda< 1\qquad\forall\,\lambda\in\text{Spec}(\hat K_{\Theta_0})\qquad \to\qquad -\id \prec \id-\eta_0\hat K_{\Theta_0} \prec \id \label{questaqui}\]
so, in the limit $t\to\infty$, 
\[\left(\id-\eta_0\hat K_{\Theta_0}\right)^t\to 0\]
and (\ref{questasoluzione}) converges to
\[
\lim_{t\to\infty}f^{\mathrm{lin}}(\Theta_t,x)=f(\Theta_0,x)- \hat K_{\Theta_0}(x,X^T)\hat K_{\Theta_0}^{-1}(f(\Theta_0,X)-Y).
\]
\end{remark}

\begin{lemma}\label{convloc}
In the limit $m\to \infty$, $\{f^{\mathrm{lin}}(\Theta^{\mathrm{lin}}_t,x)\}_{x\in\mathcal{X}}$ converges in distribution to a Gaussian process $\{f_t^{(\infty)}(x)\}_{x\in\mathcal{X}}$ with mean and variance
\begin{align}
\mu_t(x)&=\bar K(x,X^T)K^{-1}\left(\id-\left(\id-\eta_0\bar K\right)^t\right)\bar Y,\\
\nonumber\mathcal{K}_t(x,x')&=\mathcal{K}_0(x,x'),\\
\nonumber&\phantom{=}- \bar K(x,X^T)\bar K^{-1}\left(\id-\left(\id-\eta_0\bar K\right)^t\right) \mathcal{K}_0(X,x')\\
\nonumber&\phantom{=}-\bar K(x',X^T)\bar K^{-1}\left(\id-\left(\id-\eta_0\bar K\right)^t\right) \mathcal{K}_0(X,x) \\
\nonumber&\phantom{=}+\bar K(x,X^T)\bar K^{-1}\left(\id-\left(\id-\eta_0\bar K\right)^t\right)\times\\
&\phantom{=========}\times\mathcal{K}_0(X,X^T)\left(\id-\left(\id-\eta_0\bar K\right)^t\right)\bar K^{-1} \bar K(X,x').
\end{align}
\end{lemma}

\begin{proof}
By Lemma \ref{convprob},
\[\hat K_\Theta(x,x')\xrightarrow{p}\bar K(x,x') \quad \text{as}\quad m\to\infty.\]
Furthermore, by Theorem \ref{init}
\[f(\Theta_0,\,\cdot\,)\xrightarrow{d}f^{(\infty)}(\,\cdot\,) \quad \text{as}\quad m\to\infty.\]

Let $\mathcal{F}=\{x_\alpha\}_{\alpha\in A}$ be a finite family of inputs $x_\alpha\in\mathcal{X}$ containing the inputs of the dataset
\[ \{x^{(i)}\}_{1\leq i\leq n}\subseteq \mathcal{F}\]
Since the solution (\ref{questasoluzione}) for the output corresponding to any input $x_\beta\in\mathcal{F}_A$ is the following linear combination of the outputs $\{f(\Theta_0,x_\alpha)\}_{\alpha\in A}$ and of $Y$,
\[ f^{\mathrm{lin}}(\Theta^{\mathrm{lin}}_t,x_\beta)=f(\Theta_0,x_\beta)-\eta_0\hat K_{\Theta_0}(x,X^T)\left(\id-\left(\id-\eta_0\hat K_{\Theta_0}\right)^t\right)\hat K_{\Theta_0}^{-1}(f(\Theta_0,X)-Y), \]
we can write
\[ f^{\mathrm{lin}}(\Theta_t^{\mathrm{lin}},x_\beta)=\sum_{\alpha\in A} \tilde M^{(t)}_{\beta\alpha}[\hat K_{\Theta_0}] f(\Theta_0,x_\alpha)+\left(\tilde R^{(t)}[\hat K_{\Theta_0}]\right)^TY, \label{form}\]
where the entries $\tilde M^{(t)}_{\beta\alpha}[\hat K_{\Theta_0}]$ and the components of $\tilde R^{(t)}[\hat K_{\Theta_0}]$ are continous functions of the elements of matrix of the empirical NTK
\[\{\hat K_{\Theta_0}(x_\alpha,x_{\alpha'})\}_{\alpha,\alpha'\in A}.\]
By continuity, the (finite) matrix $\tilde M^{(t)}_{\beta\alpha}[\hat K_{\Theta_0}]$ and the (finite) vector $\tilde R^{(t)}[\hat K_{\Theta_0}]$ converge in probability to $\tilde M^{(t)}_{\beta\alpha}[\bar K]$ and $\tilde R^{(t)}[\bar K]$:
\[\tilde M^{(t)}_{\beta\alpha}[\hat K_{\Theta_0}]\xrightarrow{p}\tilde M^{(t)}_{\beta\alpha}[\bar K],\qquad\qquad \tilde R^{(t)}[\hat K_{\Theta_0}]\xrightarrow{p}\tilde R^{(t)}[\bar K].\]
By Slutsky's theorem \ref{sl}, we conclude that
\[ \{f^{\mathrm{lin}}(\Theta^{\mathrm{lin}}_t,x_\beta)\}_{x_\beta\in\mathcal{F}}\xrightarrow{d} \left\{\sum_{\alpha\in A}  \tilde M^{(t)}_{\beta\alpha}[\bar K] f^{(\infty)}(x_\alpha)+\left( R^{(t)}[\bar K]\right)^TY\right\}_{x_\beta\in\mathcal{F}} \quad \text{as}\quad m\to\infty,\]
i.e.
\[ f^{\mathrm{lin}}(\Theta^{\mathrm{lin}}_t,\,\cdot\,)\big|_{\mathcal{F}}\xrightarrow{d}f^{(\infty)}(\,\cdot\,)\big|_{\mathcal{F}}-\bar K(\,\cdot\,,X^T)\big|_{\mathcal{F}}\left(\id-\left(\id-\eta_0\bar K\right)^t\right)\bar K^{-1}(f^{(\infty)}(X)-Y).\label{quisopra} \]
In the limit $m\to \infty$, the solution (\ref{quisopra}) is a linear combination of the Gaussian processes  $\{f^{\mathrm{lin}}(\Theta_0,x)\}_{x\in\mathcal{F}}=\{f(\Theta_0,x)\}_{x\in\mathcal{F}}$ and $F(0)$, so it is a Gaussian process $\{f_t^{(\infty)}(x)\}_{x\in\mathcal{F}}$ as well, with
\begin{align}
\nonumber\mu_t(x)&=\mathbb{E}\left[f^{(\infty)}_t(x)\right]\\
\nonumber&=\mathbb{E}\left[f^{(\infty)}(x)\right]-\bar K(x,X^T)\left(\id-\left(\id-\eta_0\bar K\right)^t\right)\bar K^{-1}(\mathbb{E}\left[F^{(\infty)}\right]-Y)\\
&=\bar K(x,X^T)\left(\id-\left(\id-\eta_0\bar K\right)^t\right)\bar K^{-1}Y,\\
\nonumber\mathcal{K}_t(x,x')&=\mathbb{E}\left[\left(f^{(\infty)}_t(x)-\mu_t(x)\right)\left(f^{(\infty)}_t(x')-\mu_t(x')\right)\right]\\
\nonumber&=\mathbb{E}\Bigg[\left(f^{(\infty)}(x)-\bar K(x,X^T)\left(\id-\left(\id-\eta_0\bar K\right)^t\right)\bar K^{-1}F^{(\infty)}\right)\times\\
\nonumber&\phantom{=\mathbb{E}\Big[}\times\left(f^{(\infty)}(x')-\bar K(x',X^T)\left(\id-\left(\id-\eta_0\bar K\right)^t\right)\bar K^{-1}F^{(\infty)} \right)\Bigg]\\
\nonumber&=\mathcal{K}_0(x,x')\\
\nonumber&\phantom{=}- \bar K(x,X^T)\bar K^{-1}\left(\id-\left(\id-\eta_0\bar K\right)^t\right) \mathcal{K}_0(X,x')\\
\nonumber&\phantom{=}-\bar K(x',X^T)\bar K^{-1}\left(\id-\left(\id-\eta_0\bar K\right)^t\right) \mathcal{K}_0(X,x) \\
\nonumber&\phantom{=}+\bar K(x,X^T)\bar K^{-1}\left(\id-\left(\id-\eta_0\bar K\right)^t\right)\bar \times\\
&\phantom{=========}\times\mathcal{K}_0(X,X^T)\left(\id-\left(\id-\eta_0\bar K\right)^t\right)\bar K^{-1} \bar K(X,x').
\end{align}

Since we assumed $\mathcal{X}$ to be finite, this is enough to prove the convergence of the distribution to the entire Gaussian process: it is sufficient to choose $\mathcal{F}=\mathcal{X}$.
\end{proof}

\begin{remark}
    In \autoref{infinite} we will generalize the convergence of $\{f(\Theta_t,x)\}_{x\in\mathcal{X}}$ to a Gaussian process for the case of $\mathcal{X}$ being infinite. This will not require to prove that also $\{f^{\mathrm{lin}}(\Theta_t,x)\}_{x\in\mathcal{X}}$ converges to a Gaussian process: we will only need the convergence of $\{f(\Theta_t,x)\}_{x\in\mathcal{F}}$ for any $\mathcal{F}$ finite set of inputs, which, as we will see, is a corollary of the convergence of the linearized model for a finite number of inputs. Therefore, the proof given above will be enough for our purposes.
\end{remark}

Now we are ready to prove Theorem \ref{qnngpn}.\\
Let $\bar X\in\mathbb{R}^N$ be any (finite dimensional) vector on $N$ distinct inputs $\{\bar x_1,\dots, \bar x_n\}\in \mathcal{X}$ and let $\Delta_t(\bar X)$ be the random vector
\[\Delta_t(\bar X)= f(\Theta_t,\bar X)-f^{\mathrm{lin}}(\Theta_t,\bar X).\]
We show that $\Delta_t(\bar X)\xrightarrow{p}0$ as $m\to\infty$. Let $\epsilon,\delta>0$.
By Theorem \ref{confronto_finale}, there exists $m_0$ such that, for any $m\geq m_0$,
\begin{align}
\nonumber
\mathbb{P}\Bigg(&\sup_{t'}\sup_{x\in\mathcal{X}}|f(\Theta_{t'},x)-f^{\mathrm{lin}}(\Theta^{\mathrm{lin}}_{t'},x)|\\
&\qquad\qquad\leq \left(\frac{C_3}{(\lambda_{\min}^k)^3}+C_4\right) n^3\log(2n) \frac{L^2m^2|\mathcal{M}|^6|\mathcal{N}|^4}{N^5(m)}\log N(m)\Bigg)\geq 1-\delta.
\end{align}
whence
\begin{align}
\nonumber
\mathbb{P}\Bigg(&\|f(\Theta_t,\bar X)-f^{\mathrm{lin}}(\Theta^{\mathrm{lin}}_t,\bar X)\|_2\\
&\qquad\leq \sqrt N\left(\frac{C_3}{(\lambda_{\min}^k)^3}+C_4\right) n^3\log(2n) \frac{L^2m^2|\mathcal{M}|^6|\mathcal{N}|^4}{N^5(m)}\log N(m)\Bigg)\geq 1-\delta.
\end{align}
There is also $m_1\in\mathbb{N}$ such that 
\[ \sqrt{N}\left(\frac{C_3}{(\lambda_{\min}^k)^3}+C_4\right) n^3\log(2n) \frac{L^2m^2|\mathcal{M}|^6|\mathcal{N}|^4}{N^5(m)}\log N(m)<\epsilon\qquad \forall\,m\geq m_1. \label{finito2}\]
So, for any $m\geq\max \{m_0,m_1\}$,
\[
\mathbb{P}\left(\|\Delta_t(\bar X)\|_2 <\epsilon \right)\geq 1-\delta,
\]
so \[\Delta_t(\bar X)\xrightarrow{p}0\quad\text{as}\quad m\to\infty. \]
By Corollary \ref{corgp}, $\{f^{\mathrm{lin}}(\Theta^{\mathrm{lin}}_t,x)\}_{x\in\mathcal{X}}$ converges in distribution to a Gaussian process with mean $\mu_t$ and covariance $\mathcal{K}_t$ defined in the statement of the corollary:
\[ \{f^{\mathrm{lin}}(\Theta_t,x)\}_{x\in\mathcal{X}}\xrightarrow{d} \{f^{(\infty)}_t(x)\}_{x\in\mathcal{X}}.\] Now, applying Slutsky's theorem \ref{sl} we have
\[ f(\Theta_t,\bar X)=f^{\mathrm{lin}}(\Theta_t,\bar X)+\Delta_t(\bar X)\xrightarrow{d}f^{(\infty)}_t(x)\quad \text{as}\quad m\to\infty.\]
Since $\mathcal{X}$ is finite, we can choose $\bar X=\text{vec}(\mathcal{X})$ so that we have the convergence of the full multivariate Gaussian distribution, which is the Gaussian process itself. 
And this concludes the proof of Theorem \ref{qnngpn}.
\begin{remark}
    This argument is not valid in the case of $\mathcal{X}$ infinite because we cannot consider $\bar X=\text{vec}(\mathcal{X})$. See \autoref{infinite} for the proof in the most general case.
\end{remark}

\section{The case of an infinite input space}\label{infinite}
So far we considered the case of a finite input space $\mathcal{X}$ in order to provide simpler proofs for the reader interested in the physical case of having a finite number of digits to codify an input belonging to a compact space. However, for a rigorous and complete mathematical treatment of the topic, we will extend all our results to the more general case of an infinite input space. In particular, in order to prove the convergence when $\mathcal{X}$ is infinite, we need the following steps:
\begin{enumerate}
    \item the functions generated by quantum neural networks are equicontinuous and equibounded with high probability (Lemma \ref{equicont} and Lemma \ref{equibound});
    \item therefore, with high probability, they belong to a compact set by Ascoli-Arzelà's theorem;
    \item as a consequence, Prokhorov's theorem can be invoked: the family of the distributions of the functions generated by quantum neural networks is compact;
    \item we will use this fact to prove that the limit of that stochastic process is the Gaussian process defined by the linearized model (Theorem \ref{final}).
\end{enumerate}
Let us see all the steps in detail. All the statements will be proved for the case of discrete time with noise for a more concise presentation. Most of the proofs in the continuous time setting would be identical or, when slightly different, simpler; the reader can easily reconstruct the complete proofs for that case.

In order to prove the equicontinuity of the model function (Lemma \ref{equicont}) we need the following concentration inequality:

\begin{lemma}[Corollary of the Efron–Stein inequality {\cite{Boucheron2013}}]\label{efron}
    Let $X_1, \dots, X_n$ be independent random variables taking values in $\mathbb{X}$ and let $Z = f(X)$ be a square-integrable function of $X = (X_1, \dots , X_n)$ with the bounded differences property: for some constants $c_1,\dots, c_n$ it holds that
    \begin{align}\label{boundeddiff}
        \sup_{x_1,\dots,x_n,x_i'\in\mathbb{X}}\left|f(x_1,\dots,x_n)-f(x_1,\dots,x_{i-1},x_i',x_{i+1},\dots,x_n)\right|\leq c_i \quad 1\leq i\leq n.
    \end{align}
    Then,
    \[\mathrm{Var}(Z)\leq \frac{1}{4}\sum_{i=1}^nc_i^2\,.\]
\end{lemma}

Now we can prove the following lemma:

\begin{lemma}[Equicontinuity with high probability]\label{equicont} Under the same hypotheses of Theorem \ref{init}, let us suppose that
\[\lim_{m\to\infty}\frac{Lm|\mathcal{M}|^4|\mathcal{N}|^2}{N_K^2(m)N^4(m)}=0.\]
Then, for any $\delta>0$ there exists $M>0$ such that for any $m\geq 1$, with probability at least $1-\delta$ we have
\[ \max_{i=1,\,\ldots,\,\dim\mathcal{X}} \sup_{\substack{x\in\mathcal{X}\\t\geq 0}}\left|\partial_{x_i}f(\Theta_t,x)\right| \le M\,.\]
\end{lemma}

\begin{proof}
As introduced in \autoref{encinput}, $f(\Theta,x)$ is periodic with period $\pi$ in each input $x_i$. So, we can consider the Fourier series decomposition
\[ f(\Theta,x)=\sum_{q\in\mathbb{Z}^{\dim \mathcal{X}}} f^{(q)}(\Theta)e^{2i q\cdot x} \quad\text{where}\quad f^{(q)}(\Theta)=\int_\mathcal{X}\frac{dx}{(\pi)^{\dim\mathcal{X}}}e^{-2iq\cdot x} f(\Theta,x).\]
Recalling Definition \ref{defpruning} and Lemma \ref{reduction}, we can consider the pruning of the model (\ref{modelloconx}):
\begin{align}
\nonumber f_k(\Theta,x)&=\frac{1}{N(m)}\matrixel{0}{U^\dagger(\Theta,x)O_k U(\Theta,x)}{0}\\
\nonumber
&=\frac{1}{N(m)}\matrixel{0}{\left[U^\dagger(\Theta,x)\right]_kO_k \left[U(\Theta,x)\right]_k}{0}\\
\nonumber
&=\frac{1}{N(m)}\smatrixel{0}{\left[W_1^\dagger(\Theta) V_1^\dagger(x) \cdots W_L^\dagger(\Theta)V_L^\dagger(x)\right]_k
\\ &\qquad\qquad\qquad\qquad
O_k\left[V_L(x)W_L(\Theta)\cdots V_1(x)W_1(\Theta)\right]_k}{0},
\end{align}
where the only encoding unitaries in (\ref{Vx})
\[V_\ell(x)=\prod_{E\in S_\ell}\prod_{j=1}^{\dim\mathcal{X}} U_{E,j}^{(\ell)} \prod_{i\in E} e^{-i x_j\mathcal{K}_j^{\ell,i}}\]
which may contribute to $f_k(\Theta,x)$ are the ones such that $i\in \mathcal{N}^\ell_k$, as a consequence of the pruning property (see Lemma \ref{reduction}), and such that $\mathcal{K}_j^{\ell,i}\neq 0$. Therefore, let $N_{j,k}$ be the number of times that $x_j$ appears nontrivially (i.e., $\mathcal{K}_j^{\ell,i}\neq 0$) in the light cone of $f_k(\Theta,x)$.
We can consider the projectors $P^\pm_{j,\ell,i}$ on the eigenspaces of $\mathcal{K}_j^{\ell,i}$:
\[\mathcal{K}_j^{\ell,i}=P^+_{j,\ell,i}-P^-_{j,\ell,i}	\quad P^+_{j,\ell,i}+P^-_{j,\ell,i}=\id\]
In this way, we can rewrite
\begin{align}
\nonumber
e^{-ix_j\mathcal{K}_j^{\ell,i}}&=e^{-ix_jP_{j,\ell,i}^+}e^{ix_jP^-_{j,\ell,i}}\\
\nonumber &=e^{-ix_jP^+_{j,\ell,i}}(P^+_{j,\ell,i}+P^-_{j,\ell,i})e^{ix_jP_{j,\ell,i}^-}\\
&=e^{-ix_j}P_{j,\ell,i}^++e^{ix_j}P_{j,\ell,i}^-.
\end{align}
This means that in $f_k(\Theta,x)$ the dependence on $x$ will appear as
\[ f_k(\Theta,x)=\sum_{q \in Q_k}e^{2iq\cdot x}f^{(q,k)}(\Theta),\]
where $f^{(q,k)}(\Theta)$ does not depend on $x$ and where
\[Q_k:=\bigtimes_{j=1}^{\dim\mathcal{X}}\{-N_{j,k}, -N_{j,k}+1,\dots,N_{j,k}-1, N_{j,k}\}.\] 
We define $N_{\max}= \max_{j,k} N_{j,k}$ so that we can extend 
\[Q_k\subseteq Q := \{-N_{\max}, -N_{\max}+1,\dots,N_{\max}-1, N_{\max}\}^{\dim \mathcal{X}} \]
and
\[f^{(q,k)}(\Theta):=0\quad  \text{when}\quad q\in  Q\setminus Q_k,\] 
so that we can rewrite
\[ f_k(\Theta,x)=\sum_{q \in Q}e^{2iq\cdot x}f^{(q,k)}(\Theta).\]
Then
\[ f(\Theta,x)=\sum_{k=1}^mf_k(\Theta,x)=\sum_{q\in Q}e^{2iq\cdot x_i}\left(\sum_{k=1}^mf^{(q,k)}(\Theta)\right).\]
Calling
\[f^{(q)}(\Theta):=\sum_{k=1}^mf^{(q,k)}(\Theta),\]
we have proved the representation
\[ f(\Theta,x)=\sum_{q\in Q} e^{2iq\cdot x}f^{(q)}(\Theta),\]
where $f^{(q)}(\Theta)$ does not depend on $x$, whence
\[\partial_{x_i}f(\Theta,x)=\sum_{q \in Q} 2iq_ie^{2iq\cdot x}f^{(q)}(\Theta),\label{eqcont2}\]
so that, for any $x\in\mathcal{X},$
\begin{align}
    \nonumber
    \left|\partial_{x_i}f(\Theta,x)\right|&\leq 2\sum_{q\in Q} |q_i|\left|f^{(q)}(\Theta)\right|\\
    \nonumber
    &\leq 2\left(\sum_{q\in Q} |q_i|^2\right)^{1/2}\left(\sum_{q\in Q}\left|f^{(q)}(\Theta)\right|^2\right)^{1/2}\\
    \label{stimaindip}
    &\leq 2(2 N_{\max}+1)^{\dim \mathcal{X}/2}N_{\max}\left(\int_\mathcal{X}f(\Theta,x')^2\frac{dx'}{(\pi)^{\dim\mathcal{X}}}\right)^{1/2}.
\end{align}
We notice that the final bound (\ref{stimaindip}) does not depend on $x$. Calling $C=C(N_{\max}, \dim\mathcal{X}):=2(2 N_{\max}+1)^{\dim \mathcal{X}}N_{\max}^2$ and $d\mu(x):=dx/(\pi)^{\dim\mathcal{X}}$, which is a probability measure on $\mathcal{X}$, we therefore have that
\begin{align}
    \mathbb{E}\left[\sup_{x\in\mathcal{X}}\left|\partial_{x_i}f(\Theta,x)\right|^2\right]&\leq  C\int_\mathcal{X}\mathbb{E}\left[\left(f(\Theta,x)\right)^2\right]d\mu(x).
    \label{eqcont6}
\end{align}

Let us consider $\Theta=\Theta_0$. For any $m\geq 1$ the functions $\mathbb{E}\left[\left(f(\Theta_0,x)\right)^2\right]$ are continuous for $x\in\mathcal{X}$ and they converge uniformly to $\mathcal{K}(x,x)$ by (\ref{convergence}). Then, on the compact domain $\mathcal{X}$, we are allowed to exchange the limit and the integral as follows:
\begin{align}
    \nonumber
    \lim_{m\to\infty}\mathbb{E}\left[\sup_{x\in\mathcal{X}}\left|\partial_{x_i}f(\Theta_0,x)\right|^2\right]&\leq  C\int_\mathcal{X}\lim_{m\to\infty}\mathbb{E}\left[\left(f(\Theta_0,x)\right)^2\right]d\mu(x)=  C\int_\mathcal{X}\mathcal{K}(x,x)d\mu(x)
\end{align}
Being $\mathcal{K}(x,x)$ the uniform limit of continuous functions, it is continuous too, so it has a maximum over its compact domain. This means that the above integral is finite. As a consequence, there exists a uniform bound $M_0$ on $\mathbb{E}\left[\sup_{x\in\mathcal{X}}\left|\partial_{x_i}f(\Theta_0,x)\right|^2\right]$:
\[\mathbb{E}\left[\sup_{x\in\mathcal{X}}\left|\partial_{x_i}f(\Theta_0,x)\right|^2\right]\leq M_0\qquad \forall m\geq 1.\label{eqcont1}\]
Let us now generalize the previous result to $\Theta=\Theta_t$ for any $t\geq 0$. By Theorem \ref{confronto_finale}, there exists $m_0\in\mathbb{N}$ such that, with high probability,
\[\sup_{x\in\mathcal{X}}\left|f(\Theta_t,x)-f^{\mathrm{lin}}(\Theta^{\mathrm{lin}}_t,x)\right|\leq B(m)\qquad \text{with}\qquad \lim_{m\to\infty}B(m)=0,\]
for all $m\geq m_0$ and for any $t\geq 0$, whence $B=\max_{m\geq m_0} B(m)<\infty$ and
\[|f(\Theta_t,x)|\leq |f^{\mathrm{lin}}(\Theta^{\mathrm{lin}}_t,x)|+B \qquad \forall\, x\in\mathcal{X},\quad\forall\,m\geq m_0.\label{eqcont7}\]
Furthermore, we know that for $m$ large enough, $\hat K_{\Theta_0}(X,X^T)$ is invertible with high probability and therefore, by Lemma \ref{solutiondiscrete} and Remark \ref{remeta} (see, in particular, (\ref{questaqui})),
\begin{align}
    \nonumber
    |f^{\mathrm{lin}}(\Theta^{\mathrm{lin}}_t,x)-f(\Theta_0,x)|&\leq \left|\hat K_{\Theta_0}(x,X^T)\hat K^{-1}_{\Theta_0}\left(\id-\left(\id-\eta_0\hat K_{\Theta_0}t\right)^t\right)(F(0)-Y)\right|\\
    &\leq \|\hat K_{\Theta_0}(x,X^T)\|_2\|\hat K^{-1}_{\Theta_0}\|_F\|F(0)-Y\|_2.
    \label{eqcont3}
\end{align}
Both $\|\hat K^{-1}_{\Theta_0}\|_F$ and $\|F(0)-Y\|_2$ are bounded with high probability by Corollary \ref{corollaryR} and Lemma \ref{convprob} for $m$ large enough; let their product be bounded by $\sqrt{C'}$ with high probability. We will denote with $\mathcal{E}$ the event in which all the previous bounds are satisfied. We can suppose that $\mathbb{P}(\mathcal{E})\geq 1-\delta/2$. \\
We will use a couple of times the following inequality. Let $X\geq 0$ be a non-negative random variable. Then
\begin{align}
    \nonumber
    \mathbb{E}\left[X\right]&=\mathbb{P}(\mathcal{E})\mathbb{E}\left[\;X\; \middle|\;\mathcal{E}\;\right]+\mathbb{P}(\mathcal{E}^c)\mathbb{E}\left[\;X\; \middle|\;\mathcal{E}^c\;\right]\\
    &\geq \mathbb{P}(\mathcal{E})\mathbb{E}\left[\;X\; \middle|\;\mathcal{E}\;\right]
    \label{eqcont4}
\end{align}
Let us use (\ref{eqcont4}) for this bound:
\[\mathbb{E}\left[\left(f(\Theta_0,x)\right)^2\; \middle|\;\mathcal{E}\;\right]\leq \frac{2}{2-\delta} \mathbb{E}\left[\left(f(\Theta_0,x)\right)^2\right]\leq 2\mathcal{K}(x,x)\label{eqcont5}\]
Let us now combine (\ref{eqcont6}), (\ref{eqcont7}) and (\ref{eqcont3}):
\begin{align}
    \nonumber
        \limsup_{m\to\infty}&\mathbb{E}\left[\sup_{\substack{x\in\mathcal{X}\\t\geq 0}}\left|\partial_{x_i}f(\Theta_t,x)\right|^2\; \middle|\;\mathcal{E}\;\right]\\
    \nonumber
    &\leq \limsup_{m\to\infty} 2C\int_\mathcal{X}\left(\mathbb{E}\left[\sup_{t\geq 0}\left(f^{\mathrm{lin}}(\Theta_t,x)\right)^2\; \middle|\;\mathcal{E}\;\right]+B^2\right)d\mu(x)\\
    \nonumber
    & \leq \limsup_{m\to\infty} 2C\int_\mathcal{X}\Big(2\mathbb{E}\left[\left(f(\Theta_0,x)\right)^2\; \middle|\;\mathcal{E}\;\right]+B^2\\
    \nonumber
    &\phantom{\leq \limsup_{m\to\infty} 2C\int_\mathcal{X}}\qquad +2\mathbb{E}\left[\|\hat K_{\Theta_0}(x,X^T)\|_2^2\|\hat K^{-1}_{\Theta_0}\|_F^2\|F(0)-Y\|_2^2\; \middle|\;\mathcal{E}\;\right]\Big)d\mu(x)\\
    \nonumber
    & \leq 2C\Big( 4\mathcal{K}(x,x)+B^2\\
    & \phantom{\leq 2C}\qquad +2C'\limsup_{m\to\infty} \int_\mathcal{X} \,\mathbb{E}\left[\|\hat K_{\Theta_0}(x,X^T)\|_2^2\; \middle|\;\mathcal{E}\;\right]d\mu(x)\Big)
\end{align}
Regarding the last term, we need a bound of the form
\[\limsup_{m\to\infty} \int_\mathcal{X}\,\mathbb{E}\left[\|\hat K_{\Theta_0}(x,X^T)\|_2^2\; \middle|\;\mathcal{E}\;\right]d\mu(x)\leq A.\]
Let us recall that
\begin{align}
    \nonumber
    \left(\hat K_{\Theta_0}(x,x')\right)^2&=\left(\nabla_\Theta f(\Theta_0,x)\cdot \nabla_\Theta f(\Theta_0,x')\right)^2\\
    \nonumber
    &\leq \|\nabla_\Theta f(\Theta_0,x)\|_2^2\|\nabla_\Theta f(\Theta_0,x')\|_2^2\\
    &\leq \frac{1}{2}\|\nabla_\Theta f(\Theta_0,x)\|_2^2+\frac{1}{2}\|\nabla_\Theta f(\Theta_0,x')\|_2^2=\frac{1}{2}\hat K_{\Theta_0}^2(x,x)+\frac{1}{2}\hat K_{\Theta_0}^2(x',x')
\end{align}
and write more explicitly
\begin{align}
    \nonumber
    \mathbb{E}\left[\|\hat K_{\Theta_0}(x,X^T)\|_2^2\; \middle|\;\mathcal{E}\;\right]&=\sum_{j=1}^n \mathbb{E}\left[\hat K^2_{\Theta_0}(x,x^{(j)})\; \middle|\;\mathcal{E}\;\right]\\
    \nonumber
    &=\frac{1}{2}\sum_{j=1}^n \left(\mathbb{E}\left[\hat K^2_{\Theta_0}(x,x)\; \middle|\;\mathcal{E}\;\right]+\mathbb{E}\left[\hat K^2_{\Theta_0}(x^{(j)},x^{(j)})\; \middle|\;\mathcal{E}\;\right]\right)\\
    \nonumber
    &= \frac{n}{2} \mathbb{E}\left[\hat K^2_{\Theta_0}(x,x)\; \middle|\;\mathcal{E}\;\right]+\frac{1}{2}\sum_{j=1}^n\mathbb{E}\left[\hat K^2_{\Theta_0}(x^{(j)},x^{(j)})\; \middle|\;\mathcal{E}\;\right]\\
    \nonumber
    &\leq \frac{n}{2}\left(\text{Var}[\;Y_0\; \middle|\;\mathcal{E}\;]+\left(\mathbb{E}[\;Y_0\; \middle|\;\mathcal{E}\;]\right)^2\right)\\
    &\qquad +\frac{1}{2}\sum_{j=1}^n \left(\text{Var}[\;Y_j\; \middle|\;\mathcal{E}\;]+\left(\mathbb{E}[\;Y_j\; \middle|\;\mathcal{E}\;]\right)^2\right).
\end{align}
where $Y_j$ is the random variable
\[Y_j=\hat K_{\Theta_0}(x^{(j)},x^{(j)}), \qquad \text{with} \qquad x^{(0)}:=x.\]
By Assumption \ref{assNTK},
\begin{align}
\lim_{m\to\infty} \sup_{x,x'\in\mathcal{X}}\left|\mathbb{E}[\hat K_{\Theta_0}(x,x')]-\bar K(x,x')\right|=0 \qquad\forall x,x'\in\mathcal{X},
\end{align}
i.e., $\mathbb{E}[Y_j]$ converges uniformly to $\bar K(x^{(j)},x^{(j)})$. 
Noticing that, by definition, $Y_j\geq 0$ for any $m>0$, we can use (\ref{eqcont4})
\begin{align}
    \nonumber
    \mathbb{E}\left[Y_j\right]&\geq \mathbb{P}(\mathcal{E})\mathbb{E}\left[\;Y_j\; \middle|\;\mathcal{E}\;\right]\geq \left(1-\frac{\delta}{2}\right)\mathbb{E}\left[\;Y_j\; \middle|\;\mathcal{E}\;\right],
\end{align}
whence
\[\mathbb{E}\left[\;Y_j\; \middle|\;\mathcal{E}\;\right]\leq \mathbb{E}\left[Y_j\right]+\frac{\delta}{2-\delta} \mathbb{E}\left[Y_j\right]\]
\[\limsup_{m\to\infty} \sup_{\substack{0\leq j\leq n\\x^{(0)}=x\in\mathcal{X}}}\left|\mathbb{E}\left[\;Y_j\; \middle|\;\mathcal{E}\;\right]-\bar K(x^{(j)},x^{(j)})\right|\leq \frac{\delta}{2-\delta} \sup_{x\in\mathcal{X}}\bar K(x,x)\leq \sup_{x\in\mathcal{X}}\bar K(x,x)\]
Furthermore, being $\bar K(x,x')$ the uniform limit of continuous function, it is continuous, too; as a consequence, on its compact domain $\bar K(x,x')$ is bounded.
Then, there is a constant $G$ (which is uniform in $j=0,\dots, n$ and in $x^{(0)}=x\in\mathcal{X}$) such that
\[\mathbb{E}\left[\;Y_j\; \middle|\;\mathcal{E}\;\right]\leq G\]
for $m$ large enough.
\[\limsup_{m\to\infty}\int_\mathcal{X}\left(\mathbb{E}[\;Y_j\; \middle|\;\mathcal{E}\;]\right)^2d\mu(x)\leq G^2\label{eqcont10}\]
We recall that in (\ref{c_i}) we proved that (\ref{boundeddiff}) holds for the model $f(\Theta,x)$ as a function of the parameters, i.e., for any $x,x'\in\mathcal{X}$
\[\sup_{\substack{\Theta\in\mathscr{P}\\ \theta_k'\in[0,2\pi)}}\left|\hat K_\Theta(x,x')-\hat K_{\Theta'}(x,x')\right|\leq c_k \quad 1\leq k\leq Lm,\]
where $\Theta':=(\theta_1,\dots,\theta_{k-1},\theta_k',\theta_{k+1},\dots,\theta_{Lm})$,
with
\[c_k=16\,\frac{1}{N_K(m)}\frac{|\mathcal{M}||\mathcal{N}|}{N^2(m)}|\mathcal{M}_k|.\]
Furthermore, for any $m>0$ and $x,x'\in\mathcal{X}$, the function $\Theta\mapsto \hat K_\Theta (x,x')$ is continuous on the compact domain $\mathscr{P}\ni \Theta$: it is therefore bounded and, as a consequence, square-integrable.
Whence, by Lemma \ref{efron},
\[\text{Var}[Y_j]\leq 64 \frac{Lm|\mathcal{M}^4|\mathcal{N}|^2}{N_K^2(m)N^4(m)}=:\Lambda(m).\]
Using again that $Y_j\geq 0$, we can bound
\begin{align}
    \nonumber
    \text{Var}\left[\;Y_j\right]&=\mathbb{P}(\mathcal{E})\mathbb{E}\left[\;Y_j^2\; \middle|\;\mathcal{E}\;\right]+\mathbb{P}(\mathcal{E}^c)\mathbb{E}\left[\;Y_j^2\; \middle|\;\mathcal{E}^c\;\right]\\
    \nonumber
    &\qquad +\big(\mathbb{P}(\mathcal{E})\mathbb{E}\left[\;Y_j\; \middle|\;\mathcal{E}\;\right]+\mathbb{P}(\mathcal{E}^c)\mathbb{E}\left[\;Y_j\; \middle|\;\mathcal{E}^c\;\right]\big)^2\\
    \nonumber
    &\geq \big(\mathbb{P}(\mathcal{E})\big)^2 (\mathbb{E}\left[\;Y_j^2\; \middle|\;\mathcal{E}\;\right]+(\mathbb{E}\left[\;Y_j\; \middle|\;\mathcal{E}\;\right])^2\big)\\
    &\geq \left(1-\frac{\delta}{2}\right)^2\text{Var}\left[\,Y_j\; \middle|\;\mathcal{E}\;\right]\geq \frac{1}{4}\text{Var}\left[\;Y_j\; \middle|\;\mathcal{E}\;\right]
\end{align}
Since $\lim_{m\to\infty}\Lambda(m)=0$, we can compute
\[\limsup_{m\to\infty}\int_\mathcal{X}\text{Var}\left[\;Y_j\; \middle|\;\mathcal{E}\;\right]d\mu(x)\leq 4\limsup_{m\to\infty}\int_\mathcal{X}\text{Var}\left[Y_j\right]d\mu(x) =\limsup_{m\to\infty}\Lambda(m)=0.\]
So, for any $j=0,\dots,n$,
\[\limsup_{m\to\infty} \int_\mathcal{X}\,\mathbb{E}\left[|\hat K_{\Theta_0}(x^{(j)},x^{(j)})|^2\right]d\mu(x)\leq G^2, \label{eqcont9}\]
whence
\begin{align}
    \nonumber
    \limsup_{m\to\infty} \int_\mathcal{X}\,&\mathbb{E}\left[\|\hat K_{\Theta_0}(x,X^T)\|_2^2\right]d\mu(x)\leq nG^2 =: A,
\end{align}
therefore there exists a constant $\tilde M$ which does not depend on $i$ and $t$ such that
\[\mathbb{E}\left[\sup_{\substack{x\in\mathcal{X}\\t\geq 0}}\left|\partial_{x_i}f(\Theta_t,x)\right|^2\; \middle|\;\mathcal{E}\;\right]\leq \tilde M \qquad \forall\, t\geq 0\]
for any $m$ large enough.
\[\limsup_{m\to\infty} \int_\mathcal{X}\,\mathbb{E}\left[|\hat K_{\Theta_0}(x,x)|^2\right]d\mu(x)\leq G^2\label{eqcont11}\]
Finally, by Markov's inequality
\[\mathbb{P}\left(|X|\geq a\right)\leq \frac{\mathbb{E}\left[|X|^2\right]}{a^2}\]
with
\[X=\sup_{\substack{x\in\mathcal{X}\\t\geq 0}}\left|\partial_{x_i}f(\Theta_t,x)\right| \qquad \text{and}\qquad a= \sqrt{\frac{\tilde M}{2\delta}},\]
we have
\[\mathbb{P}\left(\sup_{\substack{x\in\mathcal{X}\\t\geq 0}}\left|\partial_{x_i}f(\Theta_t,x)\right|\geq \sqrt{\frac{\tilde M}{2\delta}}\; \middle|\;\mathcal{E}\;\right)\leq \frac{\delta}{2}\qquad \forall x\in\mathcal{X}, \quad 1\leq i\leq \dim\mathcal{X},\]
therefore the claim holds with probability at least $1-\delta$.
So far we proved that there exists $\bar m\in\mathbb{N}$ such that, with probability at least $1-\delta$,
\[ \max_{i=1,\,\ldots,\,\dim\mathcal{X}} \sup_{x\in\mathcal{X}}\left|\partial_{x_i}f(\Theta_t,x)\right| \le \bar M\,\]
for any $m\geq \bar m$. However, we know that
\[\sup_{\substack{\Theta\in\mathscr{P}, x\in\mathcal{X}\\1\leq i\leq \dim\mathcal{X}}}|\partial_{x_i}f(\Theta,x)|=M(m),\]
where, for any $m\in\mathbb{N}^\ast$, $M(m)<\infty$. Whence, if we set
\[M=\max\left\{M(1),\dots, M(\bar m-1),\sqrt{\frac{\tilde M}{2\delta}}\right\},\]
then
\[\sup_{\substack{x\in\mathcal{X}\\t\geq 0}}\left|\partial_{x_i}f(\Theta_t,x)\right|\leq M\]
for any $i$, with probability at least $1-\delta$, for all $m\geq 1$.
\end{proof}

\begin{lemma}[Equiboundedness with high probability]\label{equibound} Let us suppose that the hypotheses of Lemma \ref{equicont} are satisfied. Then, for any $\delta>0$ and there exists $C>0$ such that for any $m\geq 1$, with probability at least $1-\delta$ we have
\[ \sup_{\substack{x\in\mathcal{X}\\t\geq 0}}\left|f(\Theta_t,x)\right| \le C.\]
\end{lemma}

\begin{proof}
Let us fix $\bar x=x^{(1)}\in\mathcal{X}$. Let
\begin{align}
\sigma^2_t=\mathcal{K}_t(\bar x,\bar x), \quad \sigma^2=\sup_{t\geq 0}\sigma_t^2, \quad \text{and}\quad \bar y=y^{(1)}.
\end{align}
By Lemma \ref{convloc}, it is easy to see that $\sigma^2<\infty$. Furthermore, by Lemma \ref{convloc}, the random variable $f(\Theta_t,\bar x)$ converges in distribution to $f^{(\infty)}_t(\bar x)$, which are Gaussian random variables with mean $\mu_t(\bar x)$ and variance $\sigma^2_t$.
Given a constant $R$ to be fixed later, we can compute
\begin{align}\nonumber
\mathbb{P}\bigg[|f^{(\infty)_t}&(\bar x))-\bar y|\geq R \bigg]\\
&=\mathbb{P}\left[f^{(\infty)}(\bar x)\geq R+\bar y \right]+
\mathbb{P}\left[-f^{(\infty)}(\bar x))\geq R-\bar y \right], 
\end{align}
which can be estimated using Chernoff bound \cite{chernoff}:
\begin{align}\nonumber
\mathbb{P}\bigg[|f^{(\infty)}_t(\bar x)-\bar y|\geq R\bigg]
&\leq \exp\left[-\frac{1}{2\sigma_t^2}\left(R+\bar y\right)^2 \right]+
\exp\left[-\frac{1}{2\sigma_t^2}\left(R-\bar y\right)^2 \right]\\
&\leq 2\exp\left[-\frac{1}{2\sigma_t^2}\left(R-|\bar y|\right)^2 \right].\label{cher2}
\end{align}
Let
\[R(\delta,\bar y, \sigma^2) =\sqrt 2\sigma\sqrt{\log\left(\frac{8}{\delta}\right)}+|\bar y|.\]
Then we can bound (\ref{cher2}) as follows
\begin{align}
\nonumber
\mathbb{P}\bigg[|f^{(\infty)}_t(\bar x)-\bar y|\geq R \bigg] &\leq 2\exp\left[-\frac{1}{2\sigma_i^2}\left(2\sigma^2\log\left(\frac{8}{\delta}\right)\right) \right]\\
\nonumber &\leq 2\exp\left[-\log \frac{8}{\delta} \right]=\frac{\delta}{4},
\end{align}
where we used that $\sigma^2/\sigma^2_t\geq 1$.
Since the random variable $f(\Theta_t,\bar x)$ converges in distribution to $f^{(\infty)}_t(\bar x)$, there exists $m_0\in\mathbb{N}$ such that
\[ \mathbb{P}\bigg[|f(\Theta_t,\bar x)-\bar y|\geq R \bigg]\leq 2 \mathbb{P}\bigg[|f^{(\infty)}_t(\bar x)-\bar y|\geq R \bigg]\qquad \forall\, m\geq m_0\]
because of Lemma \ref{lemmaconvdistr}.
Therefore,
\[  \mathbb{P}\bigg[|f(\Theta_t,\bar x)-\bar y|\geq R \bigg]\leq\frac{\delta}{2}.\]
Now, we have that, for any $x\in\mathcal{X}$
\[f(\Theta_t,x)=f(\Theta_t,\bar x)+\int_0^1 \left(\frac{d}{d\xi} f(\Theta_t,\xi x+(1-\xi)\bar x)\right) d\xi \]
whence, by Lemma \ref{equicont}, there is a constant $M(\delta/2)$
\begin{align}
    \nonumber
    |f(\Theta_t,x)|&\leq |f(\Theta_t,\bar x)|+\left|\int_0^1 (x-\bar x)\cdot \nabla_x f(\Theta_t,\xi x+(1-\xi)\bar x) d\xi\right|\\
    \nonumber
    &\leq |f(\Theta_t,\bar x)|+\|x-\bar x\|_\infty \left(\dim\mathcal{X} \sup_{1\leq i\leq \dim\mathcal{X}}\sup_{x_0\in\mathcal{X}}\left|\partial_{x_i}f(\Theta_t,x_0) \right|\right)\\
    &\leq |f(\Theta_t,\bar x)|+ 2\pi\dim\mathcal{X} \,M\left(\frac{\delta}{2}\right)
\end{align}
with probability at least $1-\delta/2$. Therefore, for any $m\geq m_0$
\begin{align}
    |f(\Theta_t,x)|\leq R+|\bar y|+ 2\pi\dim\mathcal{X} M=\tilde C
\end{align}
with probability at least $1-\delta$, where $\tilde C$ does not depend on $x$ or $m$. However, as discussed at the end of the proof of Lemma \ref{equicont}, this result can be generalized as follows: for any $\delta$ there is a constant $C$ such that
\begin{align}
    |f(\Theta_t,x)|\leq C \qquad \forall x\in\mathcal{X}, \forall t\geq 0
\end{align}
with probability at least $1-\delta$ for any $m\geq 1$, where $\tilde C$ does not depend on $x$ or $m$.
\end{proof}

\begin{mdframed}
\begin{theorem}[Convergence to a Gaussian process]\label{final}
Let us consider any feature space $\mathcal{X}\subseteq [0,\pi]^{\dim\mathcal{X}}$, which may also be an infinite set. Let us suppose that
\[\lim_{m\to\infty}\frac{L^2m^2|\mathcal{M}|^6|\mathcal{N}|^4}{N^5(m)}\log N(m)=0. \label{hpinfinite}\]
Then, for any fixed $t\geq 0$, as $m\to\infty$,
\[\{f(\Theta_t,\,\cdot\,)\}_{x\in\mathcal{X}}\xrightarrow{d}\{f_t^{(\infty)}(\,\cdot\,)\}_{x\in\mathcal{X}},\]
where $\{f_t^{(\infty)}(\,\cdot\,)\}_{x\in\mathcal{X}}$ is the Gaussian process characterized by mean and covariance
\begin{align}
\mu_t(x)&=\bar K(x,X^T)K^{-1}\left(\id-\left(\id-\eta_0\bar K\right)^t\right) Y\\
\nonumber\mathcal{K}_t(x,x')&=\mathcal{K}_0(x,x'),\\
\nonumber&\phantom{=}- \bar K(x,X^T)\bar K^{-1}\left(\id-\left(\id-\eta_0\bar K\right)^t\right) \mathcal{K}_0(X,x')\\
\nonumber&\phantom{=}-\bar K(x',X^T)\bar K^{-1}\left(\id-\left(\id-\eta_0\bar K\right)^t\right) \mathcal{K}_0(X,x) \\
\nonumber&\phantom{=}+\bar K(x,X^T)\bar K^{-1}\left(\id-\left(\id-\eta_0\bar K\right)^t\right)\times\\
&\phantom{=========}\times\mathcal{K}_0(X,X^T)\left(\id-\left(\id-\eta_0\bar K\right)^t\right)\bar K^{-1} \bar K(X,x').
\end{align}
\end{theorem}
\end{mdframed}

\begin{proof}
The hypothesis (\ref{hpinfinite}) is the same of Theorem \ref{qnngpn}, so by Lemma \ref{llhp2} we immediately see that the hypotheses of Lemma \ref{equicont} and Lemma \ref{equibound} are satisfied.
We recall that both the parameter space and the model function depends on the number of qubits:
\[\Theta^{(m)}\in\mathscr{P}=\mathscr{P}^{(m)},\qquad f(\Theta,x)=f^{(m)}(\Theta^{(m)},x).\]
Let $\mu_{m,t}$ the probability distribution of the function generated by the quantum neural network with $m$ qubits at time $t$; it is the probability distribution on $C^0(\mathcal{X},\mathbb{R})$, the set continuous function from $\mathcal{X}$ to $\mathbb{R}$ (endowed with the sup norm), induced by the random initialization of the parameters of the circuit and, in case, by the statistical noise during the training. Let us fix $\delta>0$. Then, by Lemma \ref{equicont} and Lemma \ref{equibound}, there exist $M_\delta>0$ and $C_\delta>0$ such that, for any $m\geq 1$ and $t\geq 0$
\[\mathbb{P}\left(\sup_{\substack{x\in\mathcal{X}\\1\leq i\leq \dim\mathcal{X}}}|\partial_{x_i}f^{(m)}(\Theta_t,x)|\leq M_\delta, \quad \sup_{x\in\mathcal{X}} |f^{(m)}(\Theta_t,x)|\leq C_\delta \right)>1-\delta.\]
 The family
\[K_\delta :=\{g\in C^0(\mathcal{X},\mathbb{R}): |g(x)-g(y)|\leq M_\delta\dim\mathcal{X}\|x-y\|_\infty,\, |g(x)|\leq C_\delta\,\forall\, x,y\in\mathcal{X}\}\]
is a set of equicontinuous and equibounded continuous functions, so by Ascoli-Arzelà's theorem it is compact \cite{Rudin1976}. In particular, for any $m$ and $t$,
\[\mu_{m,t}\left(K_\delta\right)>1-\delta\]
This means that, fixing $t\geq 0$, the family of measures $\{\mu_{m,t}\}_{m\geq 1}$ is tight. Then, by Prokhorov's theorem \cite{Billingsley1999}, each subsequence $\mu_{m(k),t}$ of $\{\mu_{m,t}\}_{m\geq 1}$ has a convergence subsubsequence $\mu_{m(k(n)),t}$ to a measure $\mu_{\infty,t}$ on $C^0(\mathcal{X},\mathbb{R})$ according to the weak topology. This implies the convergence in distribution of the entire sequence $\{f^{(m)}(\Theta_t,\,\cdot\,)\}_{m\geq 1}$ to a unique stochastic process $\{g_t^{(\infty)}(x)\}_{x\in \mathcal{X}}$:
\[f^{(m)}(\Theta_t,\,\cdot\,)\xrightarrow{d}g_t^{(\infty)}(\,\cdot\,).\]
Let us characterize this stochastic process. Fixed any finite family $\mathcal{F}$ of inputs, by Theorem \ref{qnngpn} we \modifica{know}

that the limit distribution of $\{f(\Theta_t,x)\}_{x\in \mathcal{F}}$ is a multivariate Gaussian with mean function $\mu_t(x)$ and covariance $\mathcal{K}_t(x,x')$. Therefore, arbitrariness of $F$, by the limit distribution is
\[f^{(m)}(\Theta_t,\,\cdot\,)\xrightarrow{d}f_t^{(\infty)}(\,\cdot\,)\]
i.e., the unique limit is $\{g_t^{(\infty)}(x)\}_{x\in \mathcal{X}}\equiv\{f_t^{(\infty)}(x)\}_{x\in \mathcal{X}}$; this proves the claim.
\end{proof}

 \section{Conclusions}\label{concl}

We have proved that in the limit of infinite width, quantum neural networks can always be trained in polynomial time with respect to the number of qubits as long as they do not suffer from barren plateaus, and that the probability distribution of the function generated by the trained network converges in distribution to a Gaussian process, whose mean and covariance can be computed analytically.

More precisely, we have proved the following results.
First, the probability distribution of the function generated by a quantum neural network with randomly initialized parameters converges in distribution to a Gaussian process when the cardinalities of the light cones are sufficiently small (Theorem \ref{init}).
Second, the trained model is able to perfectly fit the training set exponentially fast in time, and the probability distribution of the function generated by the trained network converges in distribution to a Gaussian process whose mean and covariance can be computed analytically (Theorem \ref{qnngp}). We have also provided a quantitative bound (\ref{grad2}) that shows that the training happens in the lazy regime (Theorem \ref{gradfl} and Theorem \ref{gronwall}).
Furthermore, we have taken into account the statistical noise due to the finite number of measurements to estimate the gradients of the cost function. In particular, we have studied the evolution of the model under gradient descent with an unbiased estimator of the gradient of the cost function (\ref{sgdeqn}). We have proved that a sufficiently large number of measurements ensures all the convergence results of the noiseless setting with high probability (Theorem \ref{unbgraddesc} and Theorem \ref{qnngpn}). Such number grows polynomially with the number of qubits (\modifica{Proposition} \ref{polynmeas}), so that any quantum advantage is not hindered by the training procedure.
From the mathematical point of view we also have given a novel contribution in this field: going beyond the standard assumption that the feature space is a finite set, we have rigorously proved that the convergence to a Gaussian process, both at initialization and during the training, is valid also when we consider an infinite feature space (Theorem \ref{final}).

For a real quantum circuit, the number of qubits will unlikely be large enough to completely suppress the non-Gaussianities in the probability distribution of the function generated by the quantum neural network. As in the case of classical neural networks, this should not be seen as a problem: on the contrary, the deviations from the Gaussian limit make the neural tangent kernel trainable and can help to improve the generalization power of the neural network \cite{misiakiewicz2023lectures}. However, a rigorous understanding of such deviations is still an open problem in both the classical and the quantum case.

We stress that, despite Ref. \cite{cerezo2023does} conjectured that all the architectures for quantum neural networks that do not suffer from barren plateaus can be simulated efficiently on a classical computer, our results are valid also in regimes that have hope of quantum advantage.
Indeed, from the discussion in \autoref{combinazioni}, a quantum neural network with the qubits arranged on a $d$-dimensional lattice and the two-qubit gates acting between neighboring qubits does not allow for naive efficient classical simulations whenever $d\ge2$ and the depth grows at least logarithmically with the number of qubits.
On the other hand, assuming that the variance of the function generated by the quantum neural networks considered in our paper behaves as in the architecture of Ref. \cite{napp2022quantifying}\footnote{Formally, Ref. \cite{napp2022quantifying} considers random nonparametric two-qubit gates sampled from a 2-design.
However, there are no reasons to believe that the architectures that we consider, where the two-qubit gates are fixed, will have a different scaling of the variance of the generated function.}, for any $d$ the hypotheses of our theorems are always satisfied for a logarithmic scaling of the depth $L \simeq \epsilon\log m$ with $\epsilon$ small enough but independent on $m$.

Our results open several questions:
\begin{itemize}
\modifica{
\item Is it possible to determine quantitative bounds for the distance between the probability distribution of the function generated by a given trained quantum neural network with a finite number of qubits and the Gaussian process associated to the covariance at initialization and the analytic neural tangent kernel of the network?
Answering this question would avoid the need to consider a sequence of quantum neural networks with increasing number of qubits and would determine the width at which the results of this paper become valid.
}
\item How does the convergence of the probability distribution of the generated function to a Gaussian process depends on the size of the training set, \emph{i.e.}, how should the number of qubits scale with the number of examples for our results to be valid?
For classical neural networks, a width polynomial in the number of examples is enough \cite{arora2019on}, and we expect the same to hold in the quantum setting.
\item Is it possible to prove nontrivial bounds to the normalization constant $N(m)$ of a generic quantum circuit and therefore determine whether the circuit suffers from barren plateaus?
The study of the detailed architecture of the class of circuits satisfying the hypotheses of our theorems would be essential to determine which architectures can have hopes of quantum advantage.

\item Can the class of circuits satisfying of our theorems be extended? More precisely, we wonder whether there are circuits whose light cones grow faster than the requirements of the form \eqref{condizione} that appears in the hypotheses of our theorems, but whose output is trainable and converges to a Gaussian process. We do expect that the hypotheses of our theorems can be relaxed, but, in order to envisage a larger class of variational circuits for which our results hold, it may be necessary to formulate new proofs using different strategies.
\end{itemize}

All these questions will be considered in our future works, with the hope to answer to some of them and to give new insights on the working principles of quantum neural networks.
Finally, we hope that our results can open the way to answering the main open problem in quantum machine learning: are there problems of practical relevance that variationally trained quantum neural networks can solve better than any classical algorithm?

\appendix

\section{A counterexample for large light cones}\label{ch5-5}
In this section, we will show that, if we relax the hypothesis of Theorem \ref{init} concerning the growth of the light cones, we can exhibit a circuit which generates a function whose distribution at initialization is not a Gaussian process.

\begin{figure}[ht]
\centering
\includegraphics[width=0.82\textwidth]{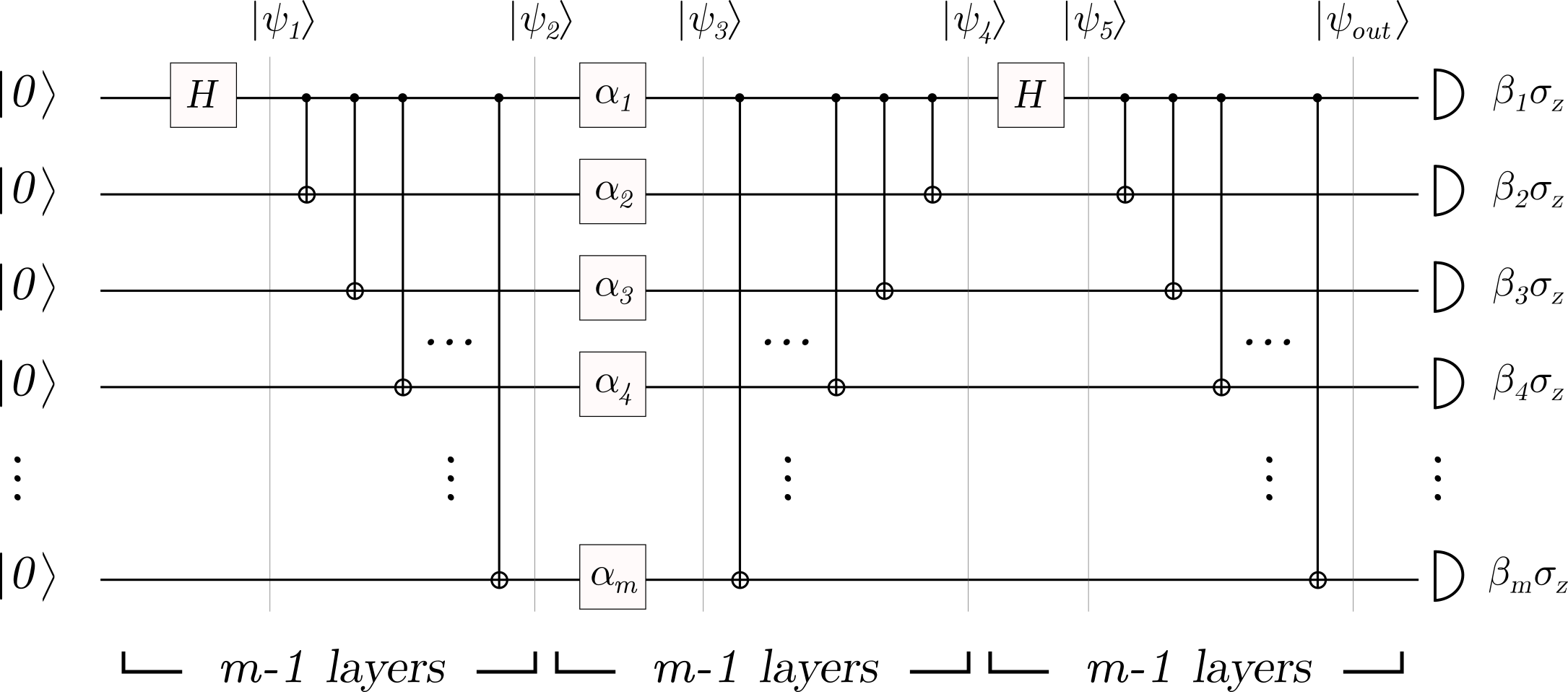}
\caption{Pathological circuit}
\end{figure}
Let us consider the circuit in the figure. It is a variational circuit in which the feature encoding gates and most parameter gates are trivial, since they act as the identity on the qubits. For this reason they are not represented in the figure. The only exception is the layer $\ell = m$, in which we identify $\alpha_k = \theta_{[mk]}$.  Each layer has one C-NOT as an entangling gate and, in the layers $\ell=1$ and $\ell=2m-1$, a Hadamard gate appears. The total number of layers is $L=3m-3=O(m)$ and the light cones have cardinality $O(m)$. The local observables are Pauli $Z$ rescaled by a bounded factor $0\leq \beta_k\leq 1$.

The initial state, after the Hadamard gate, becomes
\begin{align}
|00\dots 0\rangle \xmapsto{H_1} \ket{\psi_1}\frac{1}{\sqrt 2}\left(|00\dots 0\rangle+|10\dots 0\rangle\right).
\end{align}
After the first $m-1$ layers the state becomes
\begin{align}
\ket{\psi_2}=\frac{1}{\sqrt 2}\left(|00\dots 0\rangle+|11\dots 1\rangle\right)
\end{align}
The single qubit gates labelled by $\alpha_k$ are phase gates
\begin{align}
P(\alpha_i)=\begin{pmatrix} 1 & \\ & e^{i\alpha_i} \end{pmatrix} \qquad \text{in the computational basis.}
\end{align}
The state after the action of these gates is
\begin{align}
\ket{\psi_3}=\frac{1}{\sqrt 2}\left(|00\dots 0\rangle+e^{i\sum_{k=1}^m\alpha_k}|11\dots 1\rangle\right).
\end{align}
If we assume that $\alpha_k$ are i.i.d. random variables distributed as
\begin{align}
\mathbb{P}(\alpha_k=0)=\frac{1}{2},\qquad \mathbb{P}(\alpha_k=\pi)=\frac{1}{2},
\end{align}
we have, by symmetry,
\begin{align}
\mathbb{P}\left(\sum_{k=1}^m\alpha_k\equiv 0\text{ mod } 2\pi \right)=\mathbb{P}\left(\sum_{k=1}^m\alpha_k\equiv \pi\text{ mod } 2\pi\right)=\frac{1}{2}.
\end{align}
So the following states are equally likely
\begin{align}
\frac{1}{\sqrt 2}\left(|00\dots 0\rangle\pm|11\dots 1\rangle\right).
\end{align}
The following C-NOTs produce the state
\begin{align}
\ket{\psi_4}=\frac{1}{\sqrt 2}\left(|00\dots 0\rangle\pm|10\dots 0\rangle\right)=\frac{1}{\sqrt 2}(|0\rangle\pm|1\rangle)\otimes|0\dots 0\rangle.
\end{align}
Then, by the Hadamard gate,
\begin{align}
\ket{\psi_4}=\frac{1}{\sqrt 2}(|0\rangle\pm|1\rangle)\otimes|0\dots 0\rangle \xmapsto{H_1} 
\ket{\psi_5}=\begin{cases}
|0\rangle\otimes |0\dots 0\rangle\\
|1\rangle\otimes |0\dots 0\rangle
\end{cases}.
\end{align}
The final C-NOTs yield
\begin{align}
\ket{\psi_5}=\begin{cases}
|0\rangle\otimes |0\dots 0\rangle\\
|1\rangle\otimes |0\dots 0\rangle
\end{cases}
\mapsto\quad|\psi_{out}\rangle=
\begin{cases}
|00\dots 0\rangle\\
|11\dots 1\rangle
\end{cases}.
\end{align}
The local observables are given by
\begin{align}
O_k=\beta_k\sigma_z^{(k)}\qquad\text{with}\qquad \beta_k=\sqrt k-\sqrt{k-1}.
\end{align}
Then, the output function is
\begin{align}
\nonumber f(\Theta)&=\frac{1}{N(m)}\smatrixel{\psi_{out}}{\mathcal{O}}{\psi_{out}}\qquad\text{with}\qquad N(m)=\sqrt m\\
\nonumber &=\frac{\pm 1}{\sqrt m}\big(\sqrt 1 + (\sqrt 2 - \sqrt 1) + (\sqrt 3-\sqrt 2)+\cdots +(\sqrt m-\sqrt{m-1})\big)\\&=\pm 1,
\end{align}
which means that
\begin{align}
\mathbb{P}(f(\Theta)=1)=\mathbb{P}(f(\Theta)=-1)=\frac{1}{2}.
\end{align}
This distribution does not converge to a Gaussian process as $m\to \infty$.

This circuit has a remarkable property: if we perturb the first qubit by a local operation which cancels the first Hadamard gate, the result is completely changed and the result does not depend on the parameters $\alpha_i$.
\begin{figure}[H]
\centering
\includegraphics[width=0.82\textwidth]{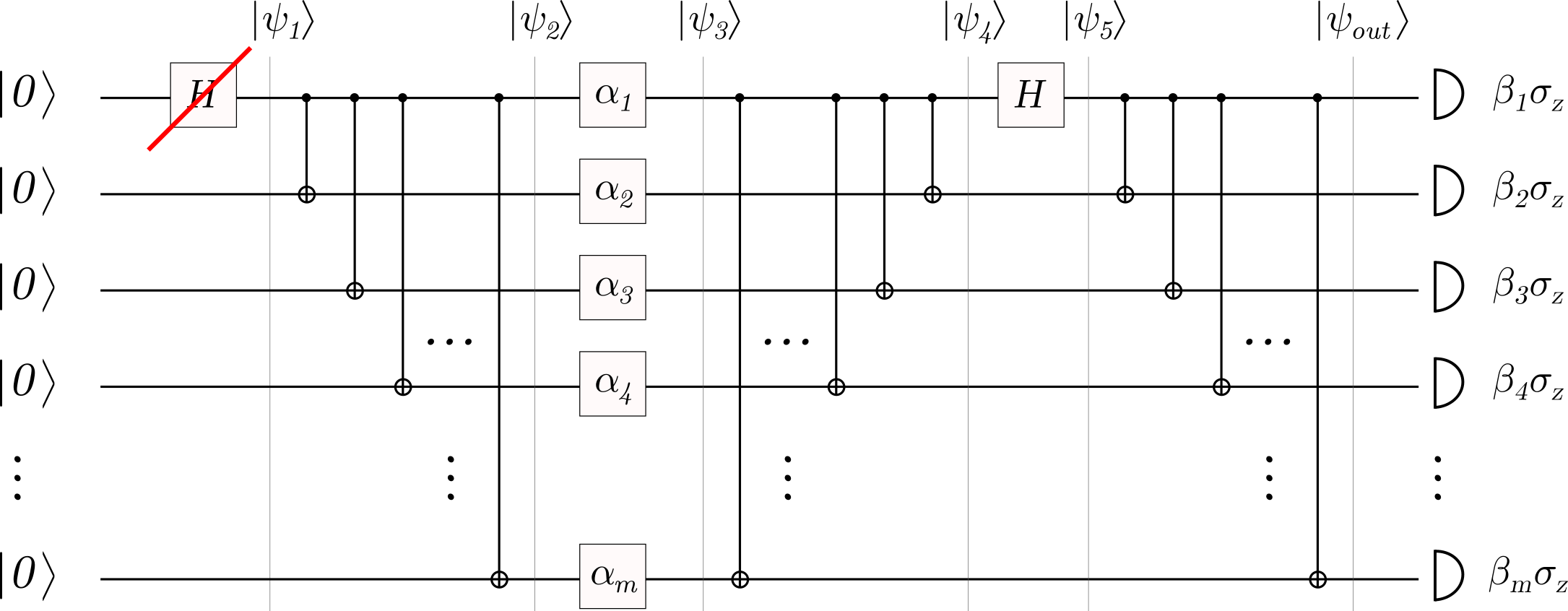}
\end{figure}
\[\ket{\psi_1}=\ket{\psi_2}=\ket{\psi_3}=\ket{\psi_4}=\ket{00\cdots 0},\qquad \ket{\psi_5}=\ket{+}\otimes\ket{0\cdots 0},\]
\[\ket{\psi_{out}}=\frac{1}{\sqrt 2}(\ket{00\cdots 0}+\ket{11\cdots 1}).\]

If we also cancel the second Hadamard gate, the output changes again: it is equal to the input.
\begin{figure}[H]
\centering
\includegraphics[width=0.82\textwidth]{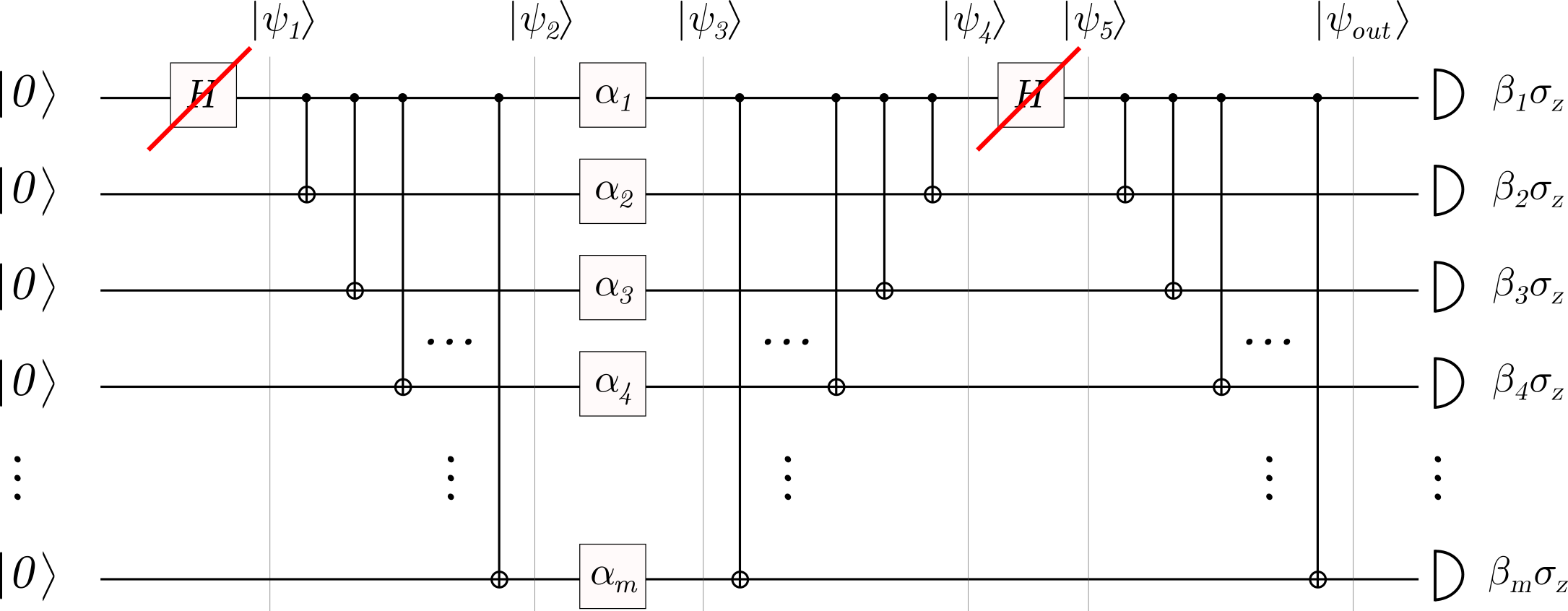}
\end{figure}
\[\ket{\psi_1}=\ket{\psi_2}=\ket{\psi_3}=\ket{\psi_4}=\ket{\psi_5}=\ket{\psi_{out}}=\ket{00\cdots 0}.\]
In both cases, a ``small'' perturbation of the circuit (i.e., a unitary transformation acting on a single qubit) has important consequences in the final state of all the qubits. 
A more robust behaviour of the output state with respect to such perturbations would be a very desirable property for a variational circuit.

\section{Architecture-independent bounds for light cones} \label{architecture}

\begin{lemma}[Architecture-independent bounds are exponential in $L$] \label{cardinalities}
Without any knowledge on the family of interactions $\mathcal{I}_U$ of a circuit $U$, we can estimate the cardinalities of the light cones as follows:
\begin{align}
|\mathcal{N}_k| &\leq 2^{L+1},\label{stimaN}\\
|\mathcal{M}_{[\ell k]}|&\leq 2^{L+1-\ell}.\label{stimaM}
\end{align}
\end{lemma}
\begin{proof}
Since we use only two-qubit gates, $|\mathcal{I}_{\ell,k}|\leq 2$.\\
By (\ref{J}) and (\ref{Nl}), we have
\[|\mathcal{N}_k^\ell|=|\mathcal{J}_k^\ell|\leq
\begin{dcases}
|\mathcal{I}_{L,k}|\leq 2 &\quad \ell=L,\\
|\mathcal{J}^{\ell+1}_k|\max_{k'}|\mathcal{I}_{\ell,k'}|\leq 2|\mathcal{J}^{\ell+1}_k| &\quad \ell<L.
\end{dcases}\]
So, inductively,
\[\begin{dcases}
|\mathcal{N}_k^\ell|\leq 2\big|\mathcal{N}_k^{(\ell+1)}\big|\\
|\mathcal{N}_k^L|\leq 2
\end{dcases}
\qquad\to\qquad \big|\mathcal{N}_k^{(L-\ell)}\big|\leq 2^{\ell+1}.\]
Whence, by Definition \ref{lightcones},
\[|\mathcal{N}_k|\leq \sum_{\ell=0}^{L-1}\big|\mathcal{N}_k^{(L-\ell)}\big|\leq \sum_{\ell=1}^{L}2^\ell\leq 2^{L+1},\]
which is (\ref{stimaN}). Now, by Lemma \ref{constructive}, if $\ell< L$,
\[|\mathcal{M}_{[\ell k]}|\leq |\mathcal{I}_{\ell,k}|\max_{k'}|\mathcal{M}_{[(\ell+1)\,k']}|\leq 2\max_{k'}|\mathcal{M}_{[(\ell+1)\,k']}|.\]
and, since
\[\mathcal{M}_{[L k]}=\mathcal{I}_{L,k} \quad \to \quad |\mathcal{M}_{[L k]}|\leq 2,\]
we have the recurrence relation
\[
\begin{dcases}
\max_k|\mathcal{M}_{[\ell k]}|\leq 2\max_{k}|\mathcal{M}_{[(\ell+1)\,k]}|& \quad\ell<L\\
\max_k|\mathcal{M}_{[L k]}|\leq 2
\end{dcases}.
\]
Therefore
\[|\mathcal{M}_{[\ell k]}|\leq \max_{\bar k}|\mathcal{M}_{[\ell \bar k]}|\leq 2^{L-\ell+1}.\]
\end{proof}

\begin{corollary} The bounds of Lemma \ref{cardinalities} imply the following bounds on the maximal cardinalities and the sums defined above
\[ |\mathcal{M}|\leq 2^L\qquad |\mathcal{N}|\leq 2^{L+1}\qquad \Sigma_n\leq 2m2^{nL}\]
\end{corollary}
\begin{proof}
The first two bounds are immediate, while the third can be proved as follows:
\begin{align}
\nonumber\Sigma_n&=\sum_{i=1}^{Lm}|\mathcal{M}_i|^n= \sum_{k=1}^{m}\sum_{\ell=1}^{L}|\mathcal{M}_{[\ell k]}|^n\leq \sum_{k=1}^{m}\sum_{\ell=1}^{L}\left(2^{L+1-\ell}\right)^n\\
&=m\sum_{\ell=1}^{L}\left(2^{\ell}\right)^n
\nonumber \leq m\sum_{\ell=0}^{L}\left(2^n\right)^\ell = m\frac{\left(2^n\right)^{L+1}-1}{2^n-1}\\
&\leq\left(\frac{2^n}{2^n-1}\right) m2^{nL}\leq 2m2^{nL}.
\end{align}
\end{proof}

\modifica{
\section{Computation of the analytic NTK}\label{app:NTK}
In this appendix we explain how $K(x,x')\coloneqq \mathbb{E}\left[\hat K_\Theta(x,x')\right]$ can be analytically computed \rimodifica{for any given circuit with finite width}. Using the fundamental rules of quantum mechanics, a laborious but elementary linear algebra computation allows to write $f(\Theta,x)$, and therefore $\hat K_\Theta(x,x')$, as an analytic function of $\Theta$ for any fixed $x$ and $x'$. Clearly, this naive method does not provide a classically efficient algorithm, since it is based on computations in the local Hilbert spaces of each observable; such spaces typically have superpolynomial dimension with respect to the number of qubits. In the spirit of the proof of Lemma \ref{doppiastima}, we already know which kind of dependence on $\Theta$ will appear: since $\hat K_\Theta(x,x')$ is defined in terms of the derivatives of $f(\Theta,x)$ and $f(\Theta,x')$, by virtue of \eqref{eq:dependence_theta} we will obtain an expression of the form
\begin{align}
    \hat K_\Theta(x,x')&=\frac{1}{N_K(m)}\left(\sum_{v\in\{0,\pm 1\}^{Lm}}\tilde f_v(x)e^{2i\Theta\cdot v}2iv\right)\cdot \left(\sum_{w\in\{0,\pm 1\}^{Lm}}\tilde f_w(x')e^{2i\Theta\cdot w}2iw\right)\\
    &=\frac{1}{N_K(m)}\sum_{v,w\in\{0,\pm 1\}^{Lm}}\tilde f_v(x)\tilde f_w(x')e^{2i\Theta\cdot (v+w)}(-4v\cdot w),\\
\end{align} 
where the coefficients $f_v(\,\cdot\,)$ are the results of the linear algebra computations in the local Hilbert spaces of each observable.
Therefore, we see that the computation of the expectation value at initialization is elementary as well, since the parameters are uniform random variables (see Assumption \ref{uniform}):
\begin{align}
    K(x,x')&=\mathbb{E}\left[\hat K_\Theta(x,x')\right]\\
    &=\frac{-4}{N_K(m)}\sum_{v,w\in\{0,\pm 1\}^{Lm}}\tilde f_v(x)\tilde f_w(x')(v\cdot w)\prod_{j=1}^{Lm}\mathbb{E}_{\theta_j}\left[e^{2i(v_j+w_j)\theta_j}\right]\\
    &=\frac{-4}{N_K(m)}\sum_{v,w\in\{0,\pm 1\}^{Lm}}\tilde f_v(x)\tilde f_w(x')(v\cdot w)\prod_{j=1}^{Lm}\id_{v_j=-w_j}\\
    &=\frac{4}{N_K(m)}\sum_{v\in\{0,\pm 1\}^{Lm}}\tilde f_v(x)\tilde f_{-v}(x')(v\cdot v),
\end{align} 
where $\id_{v_j=-w_j}$ is the characteristic function
\begin{align}
    \id_{v_j=-w_j}=\begin{cases}
        1 & v_j=-w_j \\
        0 & v_j\neq-w_j
    \end{cases}.
\end{align}
This shows that an analytical form for $K(x,x')$ can be computed.

We stress that the algorithm outlined above will in general have a runtime that grows polynomially in the dimension of the local Hilbert spaces of Definition \ref{locH} and will not be efficient.
However, our results are not meant as an alternative algorithm to replace the sequence of quantum neural networks, but just as a means to mathematically characterize the probability distribution of the functions generated by the trained quantum neural networks.
Indeed, if the covariance at initialization and the analytic neural tangent kernel could be computed efficiently on a classical computer, the associated kernel regression algorithm would also be efficient to run on a classical computer and therefore quantum neural networks would have no advantage in the regime where our results apply.
}

\section*{Declarations}
\subsection*{Funding and/or Conflicts of interests/Competing interests}
GDP has been supported by the HPC Italian National Centre for HPC, Big Data and Quantum Computing - Proposal code CN00000013 - CUP J33C22001170001 and by the Italian Extended Partnership PE01 - FAIR Future Artificial Intelligence Research - Proposal code PE00000013 - CUP J33C22002830006 under the MUR National Recovery and Resilience Plan funded by the European Union - NextGenerationEU.
Funded by the European Union - NextGenerationEU under the National Recovery and Resilience Plan (PNRR) - Mission 4 Education and research - Component 2 From research to business - Investment 1.1 Notice Prin 2022 - DD N. 104 del 2/2/2022, from title ``understanding the LEarning process of QUantum Neural networks (LeQun)'', proposal code 2022WHZ5XH – CUP J53D23003890006.
GDP is a member of the ``Gruppo Nazionale per la Fisica Matematica (GNFM)'' of the ``Istituto Nazionale di Alta Matematica ``Francesco Severi'' (INdAM)''.\\
No funding was received by FG for conducting this study.\\

The authors have no competing interests to declare that are relevant to the content of this article.

\subsection*{Data availability}
We do not analyse or generate any datasets, because our work proceeds within a theoretical and mathematical approach.

\subsection*{Acknowledgements}
We thank Dario Trevisan for useful discussions.

\bibliography{biblio}

@article{QLazy,
   title={Quantum Lazy Training},
   volume={7},
   ISSN={2521-327X},
   url={http://dx.doi.org/10.22331/q-2023-04-27-989},
   DOI={10.22331/q-2023-04-27-989},
   journal={Quantum},
   publisher={Verein zur Forderung des Open Access Publizierens in den Quantenwissenschaften},
   author={Abedi, Erfan and Beigi, Salman and Taghavi, Leila},
   year={2023},
   month=apr, pages={989} }

@Inbook{ModPHI,
author="F{\'e}ray, Valentin
and M{\'e}liot, Pierre-Lo{\"i}c
and Nikeghbali, Ashkan",
title="Dependency graphs and mod-Gaussian convergence",
bookTitle="Mod-$\phi$ Convergence: Normality Zones and Precise Deviations",
year="2016",
publisher="Springer International Publishing",
address="Cham",
pages="95--110",
isbn="978-3-319-46822-8",
doi="10.1007/978-3-319-46822-8_9",
url="https://doi.org/10.1007/978-3-319-46822-8_9"
}

@misc{mcdiarmid_1989, place={Cambridge}, series={London Mathematical Society Lecture Note Series}, title={On the method of bounded differences}, DOI={10.1017/CBO9781107359949.008}, booktitle={Surveys in Combinatorics, 1989: Invited Papers at the Twelfth British Combinatorial Conference}, publisher={Cambridge University Press}, author={McDiarmid, Colin}, editor={Siemons, J.Editor}, year={1989}, pages={148–188}, collection={London Mathematical Society Lecture Note Series}}

@article{Lee2020,
doi = {10.1088/1742-5468/abc62b},
url = {https://dx.doi.org/10.1088/1742-5468/abc62b},
year = {2020},
month = {12},
publisher = {IOP Publishing and SISSA},
volume = {2020},
number = {12},
pages = {124002},
author = {Jaehoon Lee and Lechao Xiao and Samuel S. Schoenholz and Yasaman Bahri and Roman Novak and Jascha Sohl-Dickstein and Jeffrey Pennington},
title = {Wide neural networks of any depth evolve as linear models under gradient descent*},
journal = {Journal of Statistical Mechanics: Theory and Experiment}
}

@InProceedings{PLBook,
author="Karimi, Hamed
and Nutini, Julie
and Schmidt, Mark",
editor="Frasconi, Paolo
and Landwehr, Niels
and Manco, Giuseppe
and Vreeken, Jilles",
title="Linear Convergence of Gradient and Proximal-Gradient Methods Under the Polyak-{\L}ojasiewicz Condition",
booktitle="Machine Learning and Knowledge Discovery in Databases",
year="2016",
publisher="Springer International Publishing",
address="Cham",
pages="795--811",
isbn="978-3-319-46128-1"
}

@article{McClean_2018,
   title={Barren plateaus in quantum neural network training landscapes},
   volume={9},
   ISSN={2041-1723},
   url={http://dx.doi.org/10.1038/s41467-018-07090-4},
   DOI={10.1038/s41467-018-07090-4},
   number={1},
   pages={4812},
   journal={Nature Communications},
   publisher={Springer Science and Business Media LLC},
   author={McClean, Jarrod R. and Boixo, Sergio and Smelyanskiy, Vadim N. and Babbush, Ryan and Neven, Hartmut},
   year={2018},
   month=nov }

@article{Anschuetz2022,
title = {Quantum variational algorithms are swamped with traps},
  volume = {13},
  ISSN = {2041-1723},
  url = {http://dx.doi.org/10.1038/s41467-022-35364-5},
  DOI = {10.1038/s41467-022-35364-5},
  number = {1},
  journal = {Nature Communications},
  publisher = {Springer Science and Business Media LLC},
  author = {Anschuetz,  Eric R. and Kiani,  Bobak T.},
  year = {2022},
  month = dec 
}

@misc{napp2022quantifying,
      title={Quantifying the barren plateau phenomenon for a model of unstructured variational ans\"{a}tze}, 
      author={John Napp},
      year={2022},
      eprint={2203.06174},
      archivePrefix={arXiv},
      primaryClass={quant-ph}
}

@article{Schuld_2021,
	doi = {10.1103/physreva.103.032430},
	url = {https://doi.org/10.1103%2Fphysreva.103.032430},
	year = 2021,
	month = {03},
	publisher = {American Physical Society ({APS})},
	volume = {103},
	number = {3},
	author = {Maria Schuld and Ryan Sweke and Johannes Jakob Meyer},
	title = {Effect of data encoding on the expressive power of variational quantum-machine-learning models},
	journal = {Physical Review A}
}

@article{Cerezo_2021,
  title = {Cost function dependent barren plateaus in shallow parametrized quantum circuits},
  volume = {12},
  ISSN = {2041-1723},
  url = {http://dx.doi.org/10.1038/s41467-021-21728-w},
  DOI = {10.1038/s41467-021-21728-w},
  number = {1},
  journal = {Nature Communications},
  publisher = {Springer Science and Business Media LLC},
  author = {Cerezo,  M. and Sone,  Akira and Volkoff,  Tyler and Cincio,  Lukasz and Coles,  Patrick J.},
  year = {2021},
  month = mar 
}

@misc{marrero2021entanglement,
      title={Entanglement Induced Barren Plateaus}, 
      author={Carlos Ortiz Marrero and Mária Kieferová and Nathan Wiebe},
      year={2021},
      eprint={2010.15968},
      archivePrefix={arXiv},
      primaryClass={quant-ph}
}

@misc{you2022convergence,
      title={A Convergence Theory for Over-parameterized Variational Quantum Eigensolvers}, 
      author={Xuchen You and Shouvanik Chakrabarti and Xiaodi Wu},
      year={2022},
      eprint={2205.12481},
      archivePrefix={arXiv},
      primaryClass={quant-ph}
}

@misc{lee2018deep,
      title={Deep Neural Networks as Gaussian Processes}, 
      author={Jaehoon Lee and Yasaman Bahri and Roman Novak and Samuel S. Schoenholz and Jeffrey Pennington and Jascha Sohl-Dickstein},
      year={2018},
      eprint={1711.00165},
      archivePrefix={arXiv},
      primaryClass={stat.ML}
}

@misc{doring2021method,
      title={The Method of Cumulants for the Normal Approximation}, 
      author={Hanna D\"oring and Sabine Jansen and Kristina Schubert},
      year={2021},
      eprint={2102.01459},
      archivePrefix={arXiv},
      primaryClass={math.PR}
}

@misc{Goodfellow-et-al-2016,
    title={Deep Learning},
    author={Ian Goodfellow and Yoshua Bengio and Aaron Courville},
    publisher={MIT Press},
    note={\url{http://www.deeplearningbook.org}},
    year={2016}
}

@article{Cerezo_review,
  title = {Variational quantum algorithms},
  volume = {3},
  ISSN = {2522-5820},
  url = {http://dx.doi.org/10.1038/s42254-021-00348-9},
  DOI = {10.1038/s42254-021-00348-9},
  number = {9},
  journal = {Nature Reviews Physics},
  publisher = {Springer Science and Business Media LLC},
  author = {Cerezo,  M. and Arrasmith,  Andrew and Babbush,  Ryan and Benjamin,  Simon C. and Endo,  Suguru and Fujii,  Keisuke and McClean,  Jarrod R. and Mitarai,  Kosuke and Yuan,  Xiao and Cincio,  Lukasz and Coles,  Patrick J.},
  year = {2021},
  month = aug,
  pages = {625–644}
}

@Inbook{Kallenberg2021,
author="Kallenberg, Olav",
title="Gaussian and Poisson Convergence",
bookTitle="Foundations of Modern Probability ",
year="2021",
publisher="Springer International Publishing",
address="Cham",
pages="125--146",
isbn="978-3-030-61871-1",
doi="10.1007/978-3-030-61871-1_7",
url="https://doi.org/10.1007/978-3-030-61871-1_7"
}

@BOOK{Rasmussen2005,
  title     = "Gaussian processes for machine learning",
  author    = "Rasmussen, Carl Edward and Williams, Christopher K I",
  publisher = "MIT Press",
  series    = "Adaptive Computation and Machine Learning series",
  month     =  nov,
  year      =  2005,
  address   = "London, England",
  language  = "en"
}

@misc{Romano2017,
  title     = "Counterexamples in probability and statistics",
  author    = "Romano, Joseph P and Siegel, Andrew F",
  publisher = "Routledge",
  month     =  nov,
  year      =  2017
}

@BOOK{Hogg2004,
  title     = "Introduction to mathematical statistics",
  author    = "Hogg, Robert and Craig, Allen T",
  publisher = "Pearson",
  edition   =  6,
  month     =  6,
  year      =  2004,
  address   = "Upper Saddle River, NJ",
  language  = "en"
}

@article{10.1214/aop/1176991903,
author = {Svante Janson},
title = {{Normal Convergence by Higher Semiinvariants with Applications to Sums of Dependent Random Variables and Random Graphs}},
volume = {16},
journal = {The Annals of Probability},
number = {1},
publisher = {Institute of Mathematical Statistics},
pages = {305 -- 312},
year = {1988},
doi = {10.1214/aop/1176991903},
URL = {https://doi.org/10.1214/aop/1176991903}
}

@article{D_ring_2012,
	doi = {10.1007/s10959-012-0437-0},
	url = {https://doi.org/10.1007%2Fs10959-012-0437-0},
	year = 2012,
	month = {06},
	publisher = {Springer Science and Business Media {LLC}
},
	volume = {26},
	number = {2},
	pages = {360--385},
	author = {Hanna D\"oring and Peter Eichelsbacher},
	title = {Moderate Deviations via Cumulants},
	journal = {Journal of Theoretical Probability}
}

@article{De_Palma_2021,
	doi = {10.1109/tit.2021.3076442},
	url = {https://doi.org/10.1109%2Ftit.2021.3076442},
	year = 2021,
	month = {10},
	publisher = {Institute of Electrical and Electronics Engineers ({IEEE})},
	volume = {67},
	number = {10},
	pages = {6627--6643},
	author = {Giacomo De Palma and Milad Marvian and Dario Trevisan and Seth Lloyd},
	title = {The Quantum Wasserstein Distance of Order 1},
	journal = {{IEEE} Transactions on Information Theory}
}

@misc{williams_1991, place={Cambridge}, title={Probability with Martingales}, DOI={10.1017/CBO9780511813658}, publisher={Cambridge University Press}, author={Williams, David}, year={1991}}

@misc{vaart_1998, place={Cambridge}, series={Cambridge Series in Statistical and Probabilistic Mathematics}, title={Asymptotic Statistics}, DOI={10.1017/CBO9780511802256}, publisher={Cambridge University Press}, author={Vaart, A. W. van der}, year={1998}, collection={Cambridge Series in Statistical and Probabilistic Mathematics}}

@article{stirling,
 ISSN = {00029890, 19300972},
 URL = {http://www.jstor.org/stable/2308012},
 author = {Herbert Robbins},
 journal = {The American Mathematical Monthly},
 number = {1},
 pages = {26--29},
 publisher = {Mathematical Association of America},
 title = {A Remark on Stirling's Formula},
 urldate = {2023-07-30},
 volume = {62},
 year = {1955}
}

@article{Havl_ek_2019,
	doi = {10.1038/s41586-019-0980-2},
	url = {https://doi.org/10.1038%2Fs41586-019-0980-2},
        ISSN = {1476-4687},
	year = 2019,
	month = {3},
	publisher = {Springer Science and Business Media {LLC}
},
	volume = {567},
	number = {7747},
	pages = {209--212},
	author = {Vojt{\v{e}}ch Havl{\'{i}}{\v{c}}ek and Antonio D. C{\'{o}}rcoles and Kristan Temme and Aram W. Harrow and Abhinav Kandala and Jerry M. Chow and Jay M. Gambetta},
	title = {Supervised learning with quantum-enhanced feature spaces},
	journal = {Nature}
}

@article{Liu_2021,
  title = {A rigorous and robust quantum speed-up in supervised machine learning},
  volume = {17},
  ISSN = {1745-2481},
  url = {http://dx.doi.org/10.1038/s41567-021-01287-z},
  DOI = {10.1038/s41567-021-01287-z},
  number = {9},
  journal = {Nature Physics},
  publisher = {Springer Science and Business Media LLC},
  author = {Liu,  Yunchao and Arunachalam,  Srinivasan and Temme,  Kristan},
  year = {2021},
  month = jul,
  pages = {1013–1017}
}

@article{chernoff,
author = {Herman Chernoff},
title = {{A Measure of Asymptotic Efficiency for Tests of a Hypothesis Based on the sum of Observations}},
volume = {23},
journal = {The Annals of Mathematical Statistics},
number = {4},
publisher = {Institute of Mathematical Statistics},
pages = {493 -- 507},
year = {1952},
doi = {10.1214/aoms/1177729330},
URL = {https://doi.org/10.1214/aoms/1177729330}
}

@Misc{Slutsky,
 Author = {Slutsky, E.},
 Title = {{\"U}ber stochastische {Asymptoten} und {Grenzwerte}.},
 Year = {1925},
 Language = {German},
 HowPublished = {Metron 5, {Nr}. 3, 3-89 (1925).},
 zbMATH = {2592216},
 JFM = {51.0380.03}
}

@misc{giaquinta2010mathematical,
  title={Mathematical Analysis: An Introduction to Functions of Several Variables},
  author={Giaquinta, M. and Modica, G.},
  isbn={9780817646127},
  lccn={2009922164},
  url={https://books.google.it/books?id=0YE\_AAAAQBAJ},
  year={2010},
  publisher={Birkh{\"a}user Boston}

}

@BOOK{Fristedt1996-ps,
  title     = "A modern approach to probability theory",
  author    = "Fristedt, Bert E and Gray, Lawrence F",
  publisher = "Birkhauser Boston",
  series    = "Probability and Its Applications",
  edition   =  1997,
  month     =  12,
  year      =  1996,
  address   = "Secaucus, NJ",
  language  = "en"
}

@misc{Hanin,
  author        = {B. Hanin},
  title         = {Neural Networks and Large Width and Depth (lecture notes of the mini-course given by the author in the Rome Center on Mathematics for Modeling and Data Science at Tor Vergata in January 2023)},
  month         = {1},
  year          = {2023},
}

@misc{JGH18,
      title={Neural Tangent Kernel: Convergence and Generalization in Neural Networks}, 
      author={Arthur Jacot and Franck Gabriel and Clément Hongler},
      year={2020},
      eprint={1806.07572},
      archivePrefix={arXiv},
      primaryClass={cs.LG}
}

@misc{CB18,
      title={On Lazy Training in Differentiable Programming}, 
      author={Lenaic Chizat and Edouard Oyallon and Francis Bach},
      year={2020},
      eprint={1812.07956},
      archivePrefix={arXiv},
      primaryClass={math.OC}
}

@misc{Han18,
      title={Which Neural Net Architectures Give Rise To Exploding and Vanishing Gradients?}, 
      author={Boris Hanin},
      year={2018},
      eprint={1801.03744},
      archivePrefix={arXiv},
      primaryClass={stat.ML}
}

@misc{Han21a,
author = {Hanin, Boris},
year = {2022},
month = {4},
pages = {},
title = {Correlation Functions in Random Fully Connected Neural Networks at Finite Width}
}

@misc{Han21b,
      title={Random Neural Networks in the Infinite Width Limit as Gaussian Processes}, 
      author={Boris Hanin},
      year={2021},
      eprint={2107.01562},
      archivePrefix={arXiv},
      primaryClass={math.PR}
}

@Inbook{Neal1996,
author="Neal, Radford M.",
title="Priors for Infinite Networks",
bookTitle="Bayesian Learning for Neural Networks",
year="1996",
publisher="Springer New York",
address="New York, NY",
pages="29--53",
isbn="978-1-4612-0745-0",
doi="10.1007/978-1-4612-0745-0_2",
url="https://doi.org/10.1007/978-1-4612-0745-0_2"
}

@article{Schuld_2019,
	doi = {10.1103/physreva.99.032331},
	url = {https://doi.org/10.1103%2Fphysreva.99.032331},
	year = 2019,
	month = {3},
	publisher = {American Physical Society ({APS})},
	volume = {99},
	number = {3},
	author = {Maria Schuld and Ville Bergholm and Christian Gogolin and Josh Izaac and Nathan Killoran},
	title = {Evaluating analytic gradients on quantum hardware},
	journal = {Physical Review A}
}

@article{Schuld_2019a,
	doi = {10.1103/physrevlett.122.040504},
	url = {https://doi.org/10.1103%2Fphysrevlett.122.040504},
	year = 2019,
	month = {2},
	publisher = {American Physical Society ({APS})},
	volume = {122},
	number = {4},
	author = {Maria Schuld and Nathan Killoran},
	title = {Quantum Machine Learning in Feature Hilbert Spaces},
	journal = {Physical Review Letters}
}

@misc{yang2020scaling,
      title={Scaling Limits of Wide Neural Networks with Weight Sharing: Gaussian Process Behavior, Gradient Independence, and Neural Tangent Kernel Derivation}, 
      author={Greg Yang},
      year={2020},
      eprint={1902.04760},
      archivePrefix={arXiv},
      primaryClass={cs.NE}
}

@BOOK{Boucheron2013,
  title     = "Concentration inequalities",
  author    = "Boucheron, Stephane and Lugosi, Gabor and Massart, Pascal",
  publisher = "Oxford University Press",
  month     =  03,
  year      =  2013,
  address   = "London, England",
  language  = "en"
}

@BOOK{Rudin1976,
  title     = "Principles of mathematical analysis",
  author    = "Rudin, Walter",
  publisher = "McGraw-Hill Professional",
  series    = "International series in pure and applied mathematics",
  edition   =  3,
  month     =  2,
  year      =  1976,
  address   = "New York, NY",
  language  = "en"
}

@BOOK{Billingsley1999,
  title     = "Convergence of probability measures",
  author    = "Billingsley, Patrick",
  publisher = "John Wiley \& Sons",
  series    = "Wiley Series in Probability and Statistics",
  edition   =  2,
  month     =  7,
  year      =  1999,
  address   = "Nashville, TN",
  language  = "en"
}

@article{mnih2015human,
  title = {Human-level control through deep reinforcement learning},
  volume = {518},
  ISSN = {1476-4687},
  url = {http://dx.doi.org/10.1038/nature14236},
  DOI = {10.1038/nature14236},
  number = {7540},
  journal = {Nature},
  publisher = {Springer Science and Business Media LLC},
  author={Mnih, Volodymyr and Kavukcuoglu, Koray and Silver, David and Rusu, Andrei A and Veness, Joel and Bellemare, Marc G and Graves, Alex and Riedmiller, Martin and Fidjeland, Andreas K and Ostrovski, Georg and others},
  year = {2015},
  month = feb,
  pages = {529–533}
}

@article{lecun2015deep,
  title = {Deep learning},
  volume = {521},
  ISSN = {1476-4687},
  url = {http://dx.doi.org/10.1038/nature14539},
  DOI = {10.1038/nature14539},
  number = {7553},
  journal = {Nature},
  publisher = {Springer Science and Business Media LLC},
  author = {LeCun,  Yann and Bengio,  Yoshua and Hinton,  Geoffrey},
  year = {2015},
  month = may,
  pages = {436–444}
}

@misc{radford2015unsupervised,
      title={Unsupervised Representation Learning with Deep Convolutional Generative Adversarial Networks}, 
      author={Alec Radford and Luke Metz and Soumith Chintala},
      year={2016},
      eprint={1511.06434},
      archivePrefix={arXiv},
      primaryClass={cs.LG}
}

@article{schmidhuber2015deep,
  title={Deep learning in neural networks: An overview},
  author={Schmidhuber, J{\"u}rgen},
  journal={Neural networks},
  volume={61},
  pages={85--117},
  year={2015},
  publisher={Elsevier}
}

@misc{vaswani2017attention,
 author = {Vaswani, Ashish and Shazeer, Noam and Parmar, Niki and Uszkoreit, Jakob and Jones, Llion and Gomez, Aidan N and Kaiser, \L ukasz and Polosukhin, Illia},
 booktitle = {Advances in Neural Information Processing Systems},
 editor = {I. Guyon and U. Von Luxburg and S. Bengio and H. Wallach and R. Fergus and S. Vishwanathan and R. Garnett},
 pages = {},
 publisher = {Curran Associates, Inc.},
 title = {Attention is All you Need},
 url = {https://proceedings.neurips.cc/paper_files/paper/2017/file/3f5ee243547dee91fbd053c1c4a845aa-Paper.pdf},
 volume = {30},
 year = {2017}
}

@misc{you2021exponentially,
  title = 	 {Exponentially Many Local Minima in Quantum Neural Networks},
  author =       {You, Xuchen and Wu, Xiaodi},
  booktitle = 	 {Proceedings of the 38th International Conference on Machine Learning},
  pages = 	 {12144--12155},
  year = 	 {2021},
  editor = 	 {Meila, Marina and Zhang, Tong},
  volume = 	 {139},
  series = 	 {Proceedings of Machine Learning Research},
  month = 	 {7},
  publisher =    {PMLR},
  url = 	 {https://proceedings.mlr.press/v139/you21c.html}
}

@article{Kiani_2022,
doi = {10.1088/2058-9565/ac79c9},
url = {https://dx.doi.org/10.1088/2058-9565/ac79c9},
year = {2022},
month = {7},
publisher = {IOP Publishing},
volume = {7},
number = {4},
pages = {045002},
author = {Bobak Toussi Kiani and Giacomo De Palma and Milad Marvian and Zi-Wen Liu and Seth Lloyd},
title = {Learning quantum data with the quantum earth mover’s distance},
journal = {Quantum Science and Technology}
}

@misc{yang2021tensorI,
      title={Tensor Programs I: Wide Feedforward or Recurrent Neural Networks of Any Architecture are Gaussian Processes}, 
      author={Greg Yang},
      year={2021},
      eprint={1910.12478},
      archivePrefix={arXiv},
      primaryClass={cs.NE}
}

@misc{yang2020tensorII,
      title={Tensor Programs II: Neural Tangent Kernel for Any Architecture}, 
      author={Greg Yang},
      year={2020},
      eprint={2006.14548},
      archivePrefix={arXiv},
      primaryClass={stat.ML}
}

@misc{yang2021tensorIIb,
      title={Tensor Programs IIb: Architectural Universality of Neural Tangent Kernel Training Dynamics}, 
      author={Greg Yang and Etai Littwin},
      year={2021},
      eprint={2105.03703},
      archivePrefix={arXiv},
      primaryClass={cs.LG}
}

@misc{yang2021tensorIII,
      title={Tensor Programs III: Neural Matrix Laws}, 
      author={Greg Yang},
      year={2021},
      eprint={2009.10685},
      archivePrefix={arXiv},
      primaryClass={cs.NE}
}

@misc{yang2023tensorIVb,
      title={Tensor Programs IVb: Adaptive Optimization in the Infinite-Width Limit}, 
      author={Greg Yang and Etai Littwin},
      year={2023},
      eprint={2308.01814},
      archivePrefix={arXiv},
      primaryClass={cs.LG}
}

@misc{yang2022tensorV,
      title={Tensor Programs V: Tuning Large Neural Networks via Zero-Shot Hyperparameter Transfer}, 
      author={Greg Yang and Edward J. Hu and Igor Babuschkin and Szymon Sidor and Xiaodong Liu and David Farhi and Nick Ryder and Jakub Pachocki and Weizhu Chen and Jianfeng Gao},
      year={2022},
      eprint={2203.03466},
      archivePrefix={arXiv},
      primaryClass={cs.LG}
}

@misc{yang2023tensorVI,
      title={Tensor Programs VI: Feature Learning in Infinite-Depth Neural Networks}, 
      author={Greg Yang and Dingli Yu and Chen Zhu and Soufiane Hayou},
      year={2023},
      eprint={2310.02244},
      archivePrefix={arXiv},
      primaryClass={cs.NE}
}

@misc{yang2022featureIV,
      title={Feature Learning in Infinite-Width Neural Networks}, 
      author={Greg Yang and Edward J. Hu},
      year={2022},
      eprint={2011.14522},
      archivePrefix={arXiv},
      primaryClass={cs.LG}
}

@article{Larocca_2023,
   title={Theory of overparametrization in quantum neural networks},
   volume={3},
   ISSN={2662-8457},
   url={http://dx.doi.org/10.1038/s43588-023-00467-6},
   DOI={10.1038/s43588-023-00467-6},
   number={6},
   journal={Nature Computational Science},
   publisher={Springer Science and Business Media LLC},
   author={Larocca, Martín and Ju, Nathan and García-Martín, Diego and Coles, Patrick J. and Cerezo, Marco},
   year={2023},
   month={6},
pages={542–551} }

@article{liu2022representation,
  title = {Representation Learning via Quantum Neural Tangent Kernels},
  author = {Liu, Junyu and Tacchino, Francesco and Glick, Jennifer R. and Jiang, Liang and Mezzacapo, Antonio},
  journal = {PRX Quantum},
  volume = {3},
  issue = {3},
  pages = {030323},
  numpages = {12},
  year = {2022},
  month = {8},
  publisher = {American Physical Society},
  doi = {10.1103/PRXQuantum.3.030323},
  url = {https://link.aps.org/doi/10.1103/PRXQuantum.3.030323}
}

@article{liu2023analytic,
  title = {Analytic Theory for the Dynamics of Wide Quantum Neural Networks},
  author = {Liu, Junyu and Najafi, Khadijeh and Sharma, Kunal and Tacchino, Francesco and Jiang, Liang and Mezzacapo, Antonio},
  journal = {Phys. Rev. Lett.},
  volume = {130},
  issue = {15},
  pages = {150601},
  numpages = {7},
  year = {2023},
  month = {4},
  publisher = {American Physical Society},
  doi = {10.1103/PhysRevLett.130.150601},
  url = {https://link.aps.org/doi/10.1103/PhysRevLett.130.150601}
}

@misc{cerezo2023does,
      title={Does provable absence of barren plateaus imply classical simulability? Or, why we need to rethink variational quantum computing}, 
      author={M. Cerezo and Martin Larocca and Diego García-Martín and N. L. Diaz and Paolo Braccia and Enrico Fontana and Manuel S. Rudolph and Pablo Bermejo and Aroosa Ijaz and Supanut Thanasilp and Eric R. Anschuetz and Zoë Holmes},
      year={2023},
      eprint={2312.09121},
      archivePrefix={arXiv},
      primaryClass={quant-ph}
}

@misc{misiakiewicz2023lectures,
      title={Six Lectures on Linearized Neural Networks}, 
      author={Theodor Misiakiewicz and Andrea Montanari},
      year={2023},
      eprint={2308.13431},
      archivePrefix={arXiv},
      primaryClass={stat.ML}
}

@misc{arora2019on,
 author = {Arora, Sanjeev and Du, Simon S and Hu, Wei and Li, Zhiyuan and Salakhutdinov, Russ R and Wang, Ruosong},
 booktitle = {Advances in Neural Information Processing Systems},
 editor = {H. Wallach and H. Larochelle and A. Beygelzimer and F. d\textquotesingle Alch\'{e}-Buc and E. Fox and R. Garnett},
 pages = {},
 publisher = {Curran Associates, Inc.},
 title = {On Exact Computation with an Infinitely Wide Neural Net},
 url = {https://proceedings.neurips.cc/paper_files/paper/2019/file/dbc4d84bfcfe2284ba11beffb853a8c4-Paper.pdf},
 volume = {32},
 year = {2019}
}

@misc{you2023analyzing,
      title={Analyzing Convergence in Quantum Neural Networks: Deviations from Neural Tangent Kernels}, 
      author={Xuchen You and Shouvanik Chakrabarti and Boyang Chen and Xiaodi Wu},
      year={2023},
      eprint={2303.14844},
      archivePrefix={arXiv},
      primaryClass={quant-ph}
}

@misc{garciamartin2023deep,
      title={Deep quantum neural networks form Gaussian processes}, 
      author={Diego García-Martín and Martin Larocca and M. Cerezo},
      year={2023},
      eprint={2305.09957},
      archivePrefix={arXiv},
      primaryClass={quant-ph}
}

@misc{rad2023deep,
      title={Deep Quantum Neural Networks are Gaussian Process}, 
      author={Ali Rad},
      year={2023},
      eprint={2305.12664},
      archivePrefix={arXiv},
      primaryClass={quant-ph}
}

@article{brandao2016product,
  title={Product-state approximations to quantum states},
  author={Brandao, Fernando GSL and Harrow, Aram W},
  journal={Communications in Mathematical Physics},
  volume={342},
  pages={47--80},
  year={2016},
  publisher={Springer}
}

@article{seeger2002pac,
  title={PAC-Bayesian generalisation error bounds for Gaussian process classification},
  author={Seeger, Matthias},
  journal={Journal of machine learning research},
  volume={3},
  number={Oct},
  pages={233--269},
  year={2002}
}

\end{document}